\documentclass[11pt]{report}

\usepackage{fullpage}

\usepackage{phonetic, amsmath,dsfont, upgreek, tipa, tipx}	
\usepackage{amsfonts} 
\usepackage{graphicx }	
\usepackage{booktabs, array, mathrsfs, verbatim,bbding, fancyhdr}

\allowdisplaybreaks[2] 
\newtheorem{thm}{Theorem}[section]

\newtheorem{cor}{Corollary}[section]
\newtheorem{appxlem}{Lemma}[chapter]
\newtheorem{example}{Example}   
\newbox\qedbox
\setbox\qedbox=\hbox{$\diamond$}
\newenvironment{proof}{\smallskip\noindent{\bf Proof.}\hskip \labelsep}
                        {\hfill\penalty10000\copy\qedbox\par\medskip}

\begin{document}
\allowdisplaybreaks


   \thispagestyle{empty}%
        \null\vskip2in%
        \begin{center}
                {\LARGE{Generalized EMP and nonlinear Schr\"odinger-type reformulations\\ of some scalar field cosmological models}}
        \end{center}
        \vfill        \begin{center}
                 A Dissertation Presented\\
                        \hfill \\
                            by\\
                        \hfill \\
                    {Jennie D'Ambroise}  
        \end{center}
        \vfill
        \begin{center}
          Submitted to the Graduate School of the \\
          University of Massachusetts Amherst in partial fulfillment \\
          of the requirements for the degree of\\
                \hfill \\
                Doctor of Philosophy \\
                \hfill \\
                May 2010\\
                \hfill \\
                Department of Mathematics and Statistics \\
        \end{center}\vskip.5in\newpage

        \thispagestyle{empty}%
        \null\vfill
        \begin{center}
                \copyright\ Copyright by Jennie D'Ambroise \@2010\\
                All Rights Reserved
        \end{center}
        \vfill\newpage

\newpage
   \thispagestyle{empty}%
   \renewcommand{\abstractname}{\large Dedication}
   \begin{abstract}
\begin{center}
This paper is dedicated to my family,\\ for supporting me even when I am incomprehensible.
\end{center}
\end{abstract}
\pagebreak





\newpage
   \thispagestyle{empty}%
   \renewcommand{\abstractname}{\large Acknowledgements}
\begin{abstract}
\begin{quote}
One would like to believe that life is simple and that mathematics is easy to understand.  I would like to thank Floyd Williams for showing me that sometimes it is possible for both to be true.
 \end{quote}
 \vspace{-.09in}
 \begin{quote}
To my thesis committee members, Professors Panayotis Kevrekidis, Robert Kusner and Jennie Traschen, thank you for many useful suggestions on my research.  Also, thanks to Michael Satz for always entertaining mathematical conversation.
\end{quote}
 \vspace{-.09in}
\begin{quote}
I would like to thank the countless individuals who have made possible a rich and diverse experience over the course of more than a decade of education at the University of Massachusetts at Amherst.  The programs in place have allowed me to have numerous abroad and professional experiences, and the friendly and inquisitive nature of the individuals in the Department of Mathematics and Statistics has made for an enjoyable and rigorous graduate experience. 
\end{quote}
 \vspace{-.09in}
\begin{quote}
To my students, thank you for surprising me with your willingness to listen to me talk about mathematics.
\end{quote}
 \vspace{-.09in}
\begin{quote}
To my fellow martial artists, thank you for always welcoming me without question.
\end{quote}
 \vspace{-.09in}
\begin{quote}
To my dear friend John Gusler, thank you for always helping me to get where I need to go.
\end{quote}
\end{abstract}

   \renewcommand{\abstractname}{\large Abstract}
\begin{abstract}
We show that Einstein's gravitational field equations for the Friedmann-Robertson-Lema\^itre-Walker (FRLW) and for two conformal versions of the Bianchi I and Bianchi V  perfect fluid scalar field cosmological models, can be equivalently reformulated in terms of a single equation of either generalized Ermakov-Milne-Pinney (EMP) or (non)linear Schr\"odinger (NLS) type.  This work generalizes or presents an alternative to similar reformulations published by the authors who inspired this thesis:  R. Hawkins, J. Lidsey, T. Christodoulakis, T. Grammenos, C. Helias, P. Kevrekidis, G. Papadopoulos and F. Williams.  In particular we cast much of these authors' works into a single framework via straightforward derivations of the EMP and NLS equations from a simple linear combination of the relevant Einstein equations.  By rewriting the resulting expression in terms of the volume expansion factor and performing a change of variables, we obtain an uncoupled EMP or NLS equation that is independent of the imposition of additional conservation equations.  Since the correspondences shown here present an alternative route for obtaining exact solutions to Einstein's equations, we reconstruct many known exact solutions via their EMP or NLS  counterparts and show by numerical analysis the stability properties of many solutions.\end{abstract}

\pagenumbering{roman}
\tableofcontents
\listoftables
\listoffigures


\newpage

\pagenumbering{arabic}
\chapter{Introduction}                     

\section{Einstein's field equations (EFE)}
Einstein's gravitational field equations (EFE)
\begin{equation}G_{ij}=-\kappa T_{ij}+\Lambda g_{ij}\label{eq: theEFE}\end{equation}
 for $i,j\in\{0, \dots, d\}$, are the essential equations of general relativity in $d+1$ spacetime dimensions.  The Einstein tensor is $G_{ij}\stackrel{def.}{=}R_{ij}-\frac{1}{2}Rg_{ij}$ where $R_{ij}$ is the Ricci tensor, $R$ is the scalar curvature and $g_{ij}$ is the metric.  Also $\Lambda$ is the cosmological constant and $\kappa$ is a generalization of $8\pi G$, for $G$ Newton's constant, to $d+1$ spacetime dimensions.  In terms of Christoffel symbols of the second kind
\begin{equation}\Gamma_{ij}^k\stackrel{def.}{=}\frac{1}{2}\displaystyle\sum_{s=0}g^{sk}\left(g_{si,j}-g_{ij,s}+g_{js,i}\right)\label{eq: christoffel}\end{equation}
for $i,j,k\in\{0, 1, \dots, d\}$, we define the Ricci tensor to be
\begin{equation}R_{ij}\stackrel{def.}{=}\displaystyle\sum_{k=0}^d\left(\Gamma_{kj,i}^k-\Gamma_{ij,k}^k\right)+\displaystyle\sum_{m=0}^d\displaystyle\sum_{n=0}^d\Gamma_{im}^n\Gamma_{nj}^m-\displaystyle\sum_{m=0}^d\displaystyle\sum_{n=0}^d\Gamma_{ij}^m\Gamma_{nm}^n\label{eq: Rij-Gammaij}\end{equation}
for $i,j\in\{0,1,\dots, d\}$ (note that others may define the Ricci tensor to be the negative of (\ref{eq: Rij-Gammaij}), in which case the Einstein equations would be $G_{ij}=\kappa T_{ij}-\Lambda g_{ij}$).   In (\ref{eq: christoffel}), the subscript $_{,n}$ denotes differentiation with respect to $x_n$.

We will consider an energy-momentum tensor
\begin{equation}T_{ij}=T_{ij}^{(1)}+T_{ij}^{(2)}\label{eq: Tij=sumTijs}\end{equation}
that is the sum of two terms.  The first term is the energy-momentum tensor for a minimally coupled scalar field $\phi$ with potential $V$ so that 
\begin{equation}T_{ij}^{(1)}=
\partial_i\phi\partial_j\phi-g_{ij}\left[\frac{1}{2}\displaystyle\sum_{k=0}^d\partial^k\phi\partial_k\phi + V\circ\phi\right],\label{eq: Tijminimallycoupledphi}\end{equation}
where  $\circ$ denotes composition and $\partial^k\phi \stackrel{def.}{=} \sum_{l=0}^d g^{kl}\partial_l\phi$.    We take $\phi(t)$ to depend only on time $x_0=t$ so that (\ref{eq: Tijminimallycoupledphi}) reduces to
\begin{equation}T_{ij}^{(1)}=\delta_{0i}\delta_{0j}\dot\phi^2-g_{ij}\left[\frac{1}{2}g^{00}\dot\phi^2+V\circ\phi\right]\label{eq: Tijminimallycoupledphi(tonly)}\end{equation}
where $\delta  _{0i}\stackrel{def.}{=}\left\{\begin{array}{ll}1&\mbox{ if }i=0\\0&\mbox{ if }i\neq 0\end{array}\right.$.  
The second term in (\ref{eq: Tij=sumTijs}) is defined as
\begin{equation}T^{(2)}=\left( \begin{array}{ccccc} -\rho(t)g_{00}&&&&\\ &p(t)g_{11}&&&\\ &&p(t)g_{22}&&\\ &&&\ddots &\\ &&&&p(t)g_{dd} \end{array}\right)\label{eq: T2matrix}\end{equation}
for density and pressure functions $\rho(t), p(t)$, and where the off-diagonal entries are zero.

In the special case when the metric is diagonal and $g_{00}=-1$, (\ref{eq: Tijminimallycoupledphi(tonly)}) shows that 
\begin{equation}T_{00}^{(1)}=\frac{1}{2}\dot\phi^2+V\circ\phi\label{eq: T00scalarfldperffld}\end{equation}
and
\begin{equation}T_{ii}^{(1)}= g_{ii}\left(\frac{1}{2}\dot\phi^2-V\circ\phi\right)\label{eq: Tiiscalarfldperffld}\end{equation}
for $1\leq i\leq d$.
That is, $T_{ij}^{(1)}$ reduces to the energy-momentum tensor for a perfect fluid with density and pressure
\begin{equation}\rho_\phi(t)=\frac{1}{2}\dot\phi^2+V\circ\phi\quad\mbox{and} \quad p_\phi(t)=\frac{1}{2}\dot\phi^2-V\circ\phi,\end{equation}
respectively, in terms of the scalar field and potential.

In Chapter 2 we show the equivalence of solving three types of ordinary differential equations:  a generalized Ermakov-Milne-Pinney type equation, a non-linear Schr\"odinger type equation and a third equation that we will see arises in each of the cosmological models considered in this thesis.  Since the third type of equation is derived from Einstein's field equations in terms of scale factors, we will refer to it as a {\it scale factor equation}.    Each subsection of Chapters 3-7 shows the application of a correspondence established in Chapter 2 to a specific cosmological model.   Following each theorem in Chapters 3-7, we consider special cases where exact solutions of Einstein's equations are derived from exact solutions of an NLS or EMP.  For many examples we show numerical (in-)stability graphs of the exact solutions.  The appendices show some extra calculations including some computations with exact solutions to EMP and NLS equations \cite{Bazeia, Dabrowska, Levai}.  

This thesis is a partial response to the proposal by R. Hawkins and J. Lidsey that EMP equations may appear in certain pure scalar field or other classes of cosmological models \cite{HL}.  The results in Chapters 3-7 either generalize or present an alternative to the reformulations of Einstein's equations seen in \cite{1, 1.5, JDconfBI, JD1, JDFW, HL, 4, L, LP, 8, FWEMPBI}.   In \cite{1} and \cite{FWEMPBI}, the presence of an exponential term in the EMP formulation of a Bianchi I model couples the system to a second equation; in contrast, we find a simpler term so that the Bianchi I EMP here is not coupled to a second equation.  The methods here have been noticed by other researchers \cite{Gumjud, Gumjudexact} and have been used in some cosmological applications.  A brief overview of the methods in this thesis can be found in a shorter paper by the author \cite{JDVarna}.   For future work one may consider whether the methods presented here can be extended to non-Bianchi universes such as a Kantowski-Sachs.

\section{Conservation equations}

The divergence of Einstein's tensor $G_{ij}=R_{ij}-\frac{1}{2}Rg_{ij}$ is zero.  Therefore by Einstein's equations (\ref{eq: theEFE}) the divergence of the energy-momentum tensor is zero.  That is, the conservation equation
\begin{equation}div(T)_l=\displaystyle\sum_{i,j}g^{ij}\left[\partial_i T_{lj}-\displaystyle\sum_k\Gamma^k_{ij}T_{lk}-\displaystyle\sum_k\Gamma^k_{il}T_{kj}\right]=0\label{eq: divTij=0}\end{equation}
for $i, j, k, l\in\{0, 1, \dots d\}$ is automatically satisfied by any solution $g_{ij}$ of Einstein's equations.  When $T_{ij}$ is taken to be a sum as in (\ref{eq: Tij=sumTijs}), one can further impose the condition 
\begin{equation}div(T^{(2)})_l=0\label{eq: T2divzero}\end{equation}
for $0\leq l\leq d$.  The results of this thesis do not rely on $T_{ij}^{(2)}$ satisfying (\ref{eq: T2divzero}), which is not imposed in the theorems in Chapters 3-7.  For example, if the scalar field $\phi(t)$ and the potential $V$ are non-constant, the conservation equation (\ref{eq: T2divzero}) may not hold for arbitrary density $\rho(t)$ and pressure $p(t)$ satisfying a corresponding EMP or NLS.  To obtain $T_{ij}^{(2)}$ satisfying (\ref{eq: T2divzero}) one can impose an analogue conservation equation on the EMP or NLS side of the correspondence.  We record in Appendix F equivalent versions of the conservation equation (\ref{eq: T2divzero}) for $l=0$, translated into EMP and NLS variables for each of the theorems in Chapters 3-7.

\section{Guide to numerical and exact solutions }
\label{sec: super}

The mapping of closed form solutions of EMP or NLS equations to solutions of Einstein's equations cannot always be executed analytically, but we will consider some solutions for which an exact solution to Einstein's equations can indeed be  (re-)derived via the correspondences in this thesis.  Some exact solutions of Einstein's equations in the literature \cite{BJ, BBL, CardenasCampo, EM, ndimFRLW, GarcCatCamp, Gumjud, Gumjudexact, Joseph, 4, Lima, LP, OT, Padman, RRS} will be seen to arise from an exact solution of an EMP or NLS equation.  

For some solutions we show a stability graph, where the exact solution is shown in bold font for comparison.  The numerical solutions were generated by running the Livermore solver (LSODE) \cite{lsode} on the second order EMP or NLS equation, coupled to the differential equation for the reparameterization function.  In all cases we graph the volume expansion factor and show that most exact solutions seen here are unstable.  Solutions are said to be unstable here if the difference between the exact and numerical solutions grow by at least two orders of magnitude over the graphed time interval.

\chapter{The General Correspondences}
\section{Scale factor and generalized EMP equations}

In this thesis, the Einstein equations (\ref{eq: theEFE}) for a number of cosmological models are reduced to a scale factor equation of the form 
\begin{equation}
\dot{H}(t)+\delta    H(t)^2+\varepsilon\dot{\phi}(t)^2=\displaystyle\sum_{i=0}^N \frac{G_i(t)}{a(t)^{A_i}}\label{eq: AEFE}
\end{equation}
for $H(t)\stackrel{def.}{=}\dot{a}(t)/a(t)$ and $\delta  , \varepsilon, A_i\in\mathds{R}$.  In this section we show a mapping between solution sets $\left(a(t), \phi(t), G_0(t), \dots, G_N(t)\right)$ of (\ref{eq: AEFE}) and solution sets $\left(Y(\tau), \right.$ $\left. Q(\tau),\lambda_0(\tau), \dots, \lambda_N(\tau)\right)$ of the generalized EMP equation
\begin{equation}
Y''(\tau)+Q(\tau)Y(\tau)=\displaystyle\sum_{i=0}^N \frac{\lambda_i(\tau)}{Y(\tau)^{B_i}}\label{eq: AEMP}
\end{equation}
for $B_i\in\mathds{R}$.  The dictionary between solutions of (\ref{eq: AEFE}) and (\ref{eq: AEMP}) is as follows:
\begin{eqnarray}a(t)&=&Y(\tau (t))^{1/q}\label{eq: Aa-Y}\\
q\varepsilon \varphi '(\tau)^2&=&Q(\tau)\label{eq: Avarphi-Q}\\
G_i(t)&=&\frac{\theta^2}{q}\lambda_i(\tau(t)) \label{eq: AGi-lambdai}\end{eqnarray}
for $0\leq i\leq N\in\mathds{N}$,
\begin{equation}\phi(t)=\varphi(\tau(t))\label{eq: Aphi-varphi}\end{equation}
and where $\tau(t)$ is a solution to the differential equation
\begin{equation}\dot\tau(t)=\theta a(t)^{q-\delta  }\label{eq: Atau-a}\end{equation}
for some constants $\theta>0$ and $q \in\mathds{R}\backslash\{0\}$.  Here dot denotes differentiation with respect to $t$ and prime denotes differentiation with respect to $\tau$.  Also the powers $A_i$ and $B_i$ in (\ref{eq: AEFE}) and (\ref{eq: AEMP}) respectively, are related by the equation
\begin{equation}B_i=\frac{A_i+q-2\delta    }{q}.\label{eq: ABi-Ai}\end{equation}

\begin{thm}\label{thm: EFE-EMP}
Suppose you are given a twice differentiable function $a(t)>0$, a once differentiable function $\phi(t)$, and also functions $G_0(t),\dots, G_N(t)$ which satisfy the scale factor equation (\ref{eq: AEFE}) for some $\delta  ,\varepsilon, A_0,\dots, A_N\in\mathds{R}$ and $N\in\mathds{N}$.  If $f(\tau)$ is the inverse of a function $\tau(t)$ which satisfies (\ref{eq: Atau-a}) for some $\theta>0$ and $q\in\mathds{R}\backslash\{0\}$, then by (\ref{eq: Aa-Y})-(\ref{eq: AGi-lambdai}) the functions 
\begin{eqnarray}Y(\tau)&=&a(f(\tau))^q\label{eq: AY-a}\\
Q(\tau)&=&q\varepsilon\varphi '(\tau)^2\label{eq: AQ-phi}\\
\lambda_i(\tau)&=&\frac{q}{\theta^2}G_i(f(\tau))\label{eq: Alambdai-Gi}\end{eqnarray}
solve the generalized EMP equation (\ref{eq: AEMP}) for $B_i$ as in (\ref{eq: ABi-Ai}) and for 
\begin{equation}\varphi(\tau)=\phi(f(\tau))\label{eq: Avarphi-phi}\end{equation}
(by (\ref{eq: Aphi-varphi})).  Note that since the function $a(t)$ and the constant $\theta$ are both positive, $\dot\tau(t)>0$ so that $\tau(t)$ is an increasing function mapping $t\in\mathds{R}$ to $\tau\in\mathds{R}$ and therefore the inverse function $f(\tau)$ exists. 

Conversely, suppose you are given a twice differentiable function $Y(\tau)>0$, a continuous function $Q(\tau)$, and also functions $\lambda_0(\tau), \dots, \lambda_N(\tau)$ which satisfy the generalized EMP equation (\ref{eq: AEMP}) for some $B_i\in\mathds{R}$ and $N\in\mathds{N}$.  In order to construct functions which solve (\ref{eq: AEFE}), first find $\tau(t)$ and $\varphi(\tau)$ which solve 
\begin{equation}\dot\tau(t)=\theta Y(\tau(t))^{(q-\delta  )/q}\label{eq: Atau-Y}\end{equation}
and (\ref{eq: Avarphi-Q}) respectively, for some $\theta>0, q\in\mathds{R}\backslash\{0\}$ and $\delta  \in\mathds{R}$ (note that (\ref{eq: Atau-Y}) was obtained by combining (\ref{eq: Aa-Y}) and (\ref{eq: Atau-a})).  Then the set of functions $\left(a(t),\phi(t),\right.$ $\left. G_0, \dots,\right.$ $\left. G_N\right)$ given by (\ref{eq: Aa-Y}), (\ref{eq: Aphi-varphi}) and (\ref{eq: AGi-lambdai}) solves the scale factor equation (\ref{eq: AEFE}) for $A_i$ as in (\ref{eq: ABi-Ai}).  That is, the powers $A_i$ are given in terms of $B_i$ by the equation
\begin{equation}A_i=q(B_i-1)+2\delta   \label{eq: AAi-Bi}.\end{equation}
\end{thm}

\begin{proof}
To prove the forward implication, we begin by computing $f'(\tau)$, a quantity that will be required to simplify the derivatives of $Y(\tau)$.  Since $f  (\tau (t))=t$ we differentiate this relation with respect to $t$ to obtain $f  '(\tau (t))\dot\tau  (t)=1$ so that $ f  '(\tau)=1/\dot\tau  (f  (\tau))$, and by (\ref{eq: Atau-a}) we have 
\begin{equation}f'(\tau)=\frac{1}{\theta  }a(f  (\tau))^{\delta    -q}.\label{eq: Af-a}\end{equation}
Differentiating the definition (\ref{eq: AY-a}) of $Y(\tau)$ and using (\ref{eq: Af-a}) we obtain
\begin{eqnarray}Y'(\tau)&=&qa(f  (\tau))^{q-1}\dot{a}(f  (\tau))f  '(\tau)\notag\\
&=&\frac{q}{\theta  }H(f  (\tau))a(f  (\tau))^{\delta    }.\label{eq: yprimegeneral-Ha}\end{eqnarray}
Differentiating again and using (\ref{eq: Af-a}) we obtain
\begin{eqnarray}Y''(\tau)&=&\frac{q}{\theta  }f  '(\tau)\left[\dot{H}(f  (\tau))a(f  (\tau))^{\delta    }+\delta    H(f  (\tau))^2a(f  (\tau))^{\delta    }   \right]\notag\\
&=&\frac{q}{\theta  ^2}a(f  (\tau))^{2\delta    -q}\left[\dot{H}(f  (\tau))+\delta    H(f  (\tau))^2\right].\end{eqnarray}
Since $a(t)$ is assumed to satisfy the scale factor equation (\ref{eq: AEFE}),  (\ref{eq: yprimegeneral-Ha}) can be written 
\begin{eqnarray}
Y''(\tau)&=&\frac{q}{\theta  ^2}a(f  (\tau))^{2\delta    -q}\left[-\varepsilon\dot{\phi}(f  (\tau))^2+\displaystyle\sum_{i=0}^N \frac{G_i(f  (\tau))}{a(f  (\tau)))^{A_i}}\right]\notag\\
&=&-\frac{q\varepsilon}{\theta  ^2}\dot{\phi}(f  (\tau))^2a(f  (\tau))^{2\delta    -2q}a(f  (\tau))^{q}+\displaystyle\sum_{i=0}^N\frac{\frac{q}{\theta  ^2}G_i(f  (\tau))}{a(f  (\tau))^{A_i+q-2\delta    }}.\label{eq: AY''-a}
\end{eqnarray}
Differentiating the definition (\ref{eq: Avarphi-phi}) of $\varphi(\tau)$ and again using (\ref{eq: Af-a}) we have 
\begin{equation}\varphi '(\tau)=\dot\phi(f(\tau))f'(\tau)=\frac{1}{\theta}\dot\phi(f(\tau))a(f(\tau))^{\delta  -q},\end{equation}
so that the definition (\ref{eq: AQ-phi}) of $Q(\tau)$ can be written as
\begin{equation}Q(\tau)=\frac{q\varepsilon}{\theta  ^2}\dot{\phi}(f  (\tau))^2a(f  (\tau))^{2\delta    -2q}.\label{eq: AQ-a}\end{equation}
By (\ref{eq: AQ-a}) and the definitions (\ref{eq: AY-a}), (\ref{eq: Alambdai-Gi}) and (\ref{eq: ABi-Ai}) of  $Y(\tau), \lambda_i(\tau)$ and $B_i\in\mathds{R}$ respectively, (\ref{eq: AY''-a})   becomes
\begin{eqnarray}Y''(\tau)+Q(\tau)Y(\tau)=\displaystyle\sum_{i=0}^N\frac{\lambda_i(\tau)}{Y(\tau)^{B_i}}.
\end{eqnarray}

To prove the converse statement, differentiate the definition (\ref{eq: Aa-Y}) of $a(t)$ and use the definition (\ref{eq: Atau-Y}) of $\tau(t)$ to obtain
\begin{eqnarray}\dot{a}(t)&=&\frac{1}{q}Y(\tau (t))^{(1-q)/q}Y'(\tau (t))\dot\tau  (t)\notag\\
&=&\frac{\theta  }{q}Y(\tau (t))^{(1-\delta    )/q}Y'(\tau (t)).\label{eq: Adota-Y}\end{eqnarray}
Dividing by $a(t)$, we have that 
\begin{equation}H(t)\stackrel{def.}{=}\frac{\dot{a}(t)}{a(t)}=\frac{\theta  }{q}Y'(\tau (t))Y(\tau (t))^{-\delta    /q}.\label{eq: AH-Y}
\end{equation}
Differentiating (\ref{eq: AH-Y}) and again using the definition (\ref{eq: Atau-Y}) of $\tau(t)$, we obtain
\begin{eqnarray}\dot{H}(t)
&=&\frac{\theta  }{q}\dot\tau  (t)\left[Y''(\tau (t))Y(\tau (t))^{-\delta    /q}-\frac{\delta    }{q}Y'(\tau (t))^2Y(\tau (t))^{-(\delta    +q)/q}\right]\notag\\
&=&\frac{\theta  ^2}{q}Y(\tau (t))^{(q-\delta  )/q}\left[Y''(\tau (t))Y(\tau (t))^{-\delta    /q}-\frac{\delta    }{q}Y'(\tau (t))^2Y(\tau (t))^{-(\delta    +q)/q}\right]\notag\\
&=&\frac{\theta^2}{q}Y(\tau(t))^{(q-2\delta  )/q}Y''(\tau)-\frac{\delta   \theta  ^2 }{q^2}Y'(\tau (t))^2Y(\tau (t))^{-2\delta    /q}\label{eq: generaldotH-generalY}
.\end{eqnarray}
Since $Y(\tau)$ is assumed to satisfy the generalized EMP equation (\ref{eq: AEMP}), equation (\ref{eq: generaldotH-generalY}) can be written as 
\begin{eqnarray}
\dot{H}(t)&=&\frac{\theta  ^2}{q}Y(\tau (t))^{(q-2\delta  )/q}
\left[
-Q(\tau (t))Y(\tau (t))+\displaystyle\sum_{i=0}^N\frac{\lambda_i(\tau (t))}{Y(\tau (t))^{B_i}}
\right]\notag\\
&&\qquad\qquad -\frac{\delta   \theta  ^2 }{q^2}Y'(\tau (t))^2Y(\tau (t))^{-2\delta    /q}\notag\\
&=&-\frac{\theta  ^2}{q}Q(\tau (t))Y(\tau (t))^{2(q-\delta  )/q}+\displaystyle\sum_{i=0}^N\frac{\frac{\theta  ^2}{q}\lambda_i(\tau (t))}{Y(\tau (t))^{(q(B_i-1)+2\delta  )/q}}\notag\\
&& \qquad\qquad -\frac{\delta   \theta  ^2 }{q^2}Y'(\tau (t))^2Y(\tau (t))^{-2\delta    /q}.\label{eq: AdotH-Y}
\end{eqnarray}
By definitions (\ref{eq: Aphi-varphi}) and (\ref{eq: Atau-Y}) of $\phi(t)$ and $\tau(t)$ respectively,  we have that 
\begin{equation}\dot\phi(t)=\varphi '( \tau(t))\dot\tau(t)=\theta \varphi '(\tau(t))  Y(\tau (t))^{(q-\delta  )/q}.\label{eq: Adotphi-Y}\end{equation}
Using definition (\ref{eq: Avarphi-Q}) of $\varphi(\tau)$ in terms of $Q(\tau)$ and squaring (\ref{eq: Adotphi-Y}), we have
\begin{equation}\dot\phi(t)^2= \frac{\theta^2 }{q\varepsilon}Q(\tau(t))   Y(\tau (t))^{2(q-\delta  )/q}.\end{equation}
This shows that the first term in (\ref{eq: AdotH-Y}) is equal to $-\varepsilon\dot\phi(t)^2$.    Noting that by (\ref{eq: AH-Y}) the last term of (\ref{eq: AdotH-Y}) is equal to $-\delta   H(t)^2$, and using definitions (\ref{eq: Aa-Y}), (\ref{eq: AGi-lambdai}) and (\ref{eq: AAi-Bi}) of $a(t), G_i(t)$ and $A_i$ respectively, (\ref{eq: AdotH-Y}) becomes
\begin{equation}\dot{H}(t)=-\varepsilon\dot{\phi}(t)^2+\displaystyle\sum_{i=0}^N\frac{G_i(t)}{a(t)^{A_i}}-\delta    H(t)^2.\end{equation}
This proves the theorem.

\end{proof}


\section{Generalized EMP and Schr\"odinger-type equations}
      
We now record a mapping between any solution set $\left(Y(\tau),Q(\tau), \lambda(\tau), \lambda_1(\tau), \right.$ $\left. \dots,\lambda_N(\tau)\right)$ of the generalized EMP equation
\begin{equation}Y''(\tau)+Q(\tau)Y(\tau)=\frac{\lambda(\tau)}{Y(\tau)^B}+\displaystyle\sum_{i=1}^N\frac{\lambda_i(\tau)}{Y(\tau)^{B_i}}\label{eq: BEMP}\end{equation}
for $B\in\mathds{R}\backslash\{-1,1\}, B_i\in\mathds{R}\backslash\{-1\}, N\in\mathds{N}$, and  a corresponding solution set $\left(u(\sigma), E(\sigma), P(\sigma),\right.$ $\left. F_1(\sigma), \dots, F_N(\sigma)\right)$ of what we will call a non-linear Schr\"odinger-type equation
 \begin{equation}u''(\sigma)+[E(\sigma)-P(\sigma)]u(\sigma)=\displaystyle\sum_{i=1}^N\frac{F_i(\sigma)}{u(\sigma)^{C_i}}\label{eq: BNLS}\end{equation}
 for $C_i\in\mathds{R}$.   
 
 In general, the Schr\"odinger-type equation (\ref{eq: BNLS}) contains one less non-linear term than the generalized EMP equation (\ref{eq: BEMP}).  Therefore although the correspondence does hold when $\lambda=\lambda_i=E=0$, the point is that a nonzero nonlinear term $\lambda(\tau)/Y(\tau)^B$ in (\ref{eq: BEMP}) transforms to the linear term $E(\sigma)u(\sigma)$ in (\ref{eq: BNLS}).  Therefore the ``Schr\"odinger" nature of this latter equation is most apparent when $\lambda_i(\tau)=0$ for $1\leq i\leq N$ and when the function $\lambda(\tau)=\lambda$ is constant  in (\ref{eq: BEMP}).  In this case (as we will see in this section), solutions to the generalized EMP 
 \begin{equation}Y''(\tau)+Q(\tau)Y(\tau)=\frac{\lambda}{Y(\tau)^B}\label{eq: B1termEMP}\end{equation} correspond to solutions of a one-dimensional \emph{linear} Schr\"odinger equation 
 \begin{equation}u''(\sigma)+[E-P(\sigma)]u(\sigma)=0\label{eq: BLS}\end{equation}
 for $E$ constant.  This slightly generalizes a result of F. Williams in which solutions of a classical EMP (that is, for $B=3$ in (\ref{eq: B1termEMP})) are shown to be in correspondence with solutions of a linear Schr\"odinger equation (\ref{eq: BLS}).   One can also refer to the paper of W. Milne \cite{Milne}.

The dictionary between solutions to (\ref{eq: BEMP}) and (\ref{eq: BNLS}) is as follows:
\begin{eqnarray}Y(\tau(t))^{B-1}&=&u(\sigma(t))^{-2}\label{eq: BYtau-usigma}\\
\vartheta^2(B-1)Q(\tau(t))^{(B-1)/(B+1)}&=&2P(\sigma(t))u(\sigma(t))^{2}\label{eq: BQtau-usigma}\\
\vartheta^2(B-1)\lambda(\tau(t))&=&2E(\sigma(t))\label{eq: Blambdatau-Esigma}\\
\vartheta^2(B-1)\lambda_i(\tau(t))&=&2F_i(\sigma(t))\label{eq: Blambdaitau-Fisigma}\end{eqnarray}
where $\tau_\sigma(t)$ and $\sigma(t)$ are solutions to the differential equations
\begin{equation}\dot\tau_\sigma(t)=\sqrt{\vartheta}Y(\tau_\sigma(t))^{\frac{1}{4}(B+1)}\label{eq: Bdottau-Y}\end{equation}
and
\begin{equation}\dot\sigma(t)=\frac{1}{\sqrt{\vartheta}}u(\sigma(t))^{\frac{1}{2}(B+1)/(B-1)}\label{eq: Bdotsigma-u}\end{equation}
respectively, for some $\vartheta>0$.  
Also the powers $B_i$ and $C_i$ in (\ref{eq: BEMP}) and (\ref{eq: BNLS}) respectively, are related by the equation
\begin{equation}C_i=1-2\frac{(B_i-1)}{(B-1)}.\label{eq: BCi-Bi}\end{equation}
Note that by (\ref{eq: BYtau-usigma}), (\ref{eq: Bdottau-Y}) and (\ref{eq: Bdotsigma-u}), we have
\begin{equation}\dot\tau_\sigma(t)=\frac{1}{\dot\sigma(t)}\label{eq: Bdottau-dotsigma}.\end{equation}   
We notate the function $\tau_\sigma(t)$ with the subscript $\sigma$ in order to distinguish it from the separate quantity $\tau(t)$ which appears in section 2.1.  

\begin{thm}\label{thm: EMP-NLS}
Suppose you are given a twice differentiable function $Y(\tau)>0$ and also functions $Q(\tau), \lambda(\tau), \lambda_1(\tau), \dots, \lambda_N(\tau)$ which satisfy the generalized EMP equation (\ref{eq: BEMP}) for some $B\in\mathds{R}\backslash\{-1,1\}, B_i\in\mathds{R}\backslash\{-1\}$ and $N\in\mathds{N}$.  In order to construct a set of functions which solve the Schr\"odinger-type equation (\ref{eq: BNLS}), begin by solving  for the function $\tau_\sigma(t)$ in (\ref{eq: Bdottau-Y}) (for any $\vartheta  >0$) and then solve (\ref{eq: Bdottau-dotsigma}) for $\sigma(t)$ .  Let $g(\sigma)$ denote the inverse of $\sigma(t)$ (which exists since $\dot\sigma(t)>0$ for all $t$).  Then by (\ref{eq: BYtau-usigma})-(\ref{eq: Blambdaitau-Fisigma}) the following functions solve the Schr\"odinger-type equation (\ref{eq: BNLS}):
\begin{eqnarray}u(\sigma)&=&Y(\tau_\sigma (g (\sigma)))^{\frac{1}{2}(1-B)}\label{eq: Bu-Ytausigma}\\
P(\sigma)&=&\frac{\vartheta  ^2}{2}(B-1)Q(\tau_\sigma (g (\sigma)))Y(\tau_\sigma (g (\sigma)))^{B+1}\label{eq: BP-Qtausigma}\\
E(\sigma)&=&\frac{\vartheta  ^2}{2}(B-1)\lambda(\tau_\sigma (g (\sigma)))\label{eq: BE-lambdatausigma}\\
F_i(\sigma)&=&\frac{\vartheta  ^2}{2}(1-B)\lambda_i(\tau_\sigma (g (\sigma)))\label{eq: BFi-lambdaitausigma}
\end{eqnarray}
for $C_i$ as in (\ref{eq: BCi-Bi}).




Conversely, suppose you are given a twice differentiable function $u(\sigma)>0$ and also functions $E(\sigma), P(\sigma), F_i(\sigma)$ which satisfy the Schr\"odinger-type equation (\ref{eq: BNLS}) for some $C_i\in\mathds{R}$ and $N\in\mathds{N}$.  In order to construct functions which solve (\ref{eq: BEMP}), first solve (\ref{eq: Bdotsigma-u}) for $\sigma(t)$ and for any constants $\vartheta>0$ and $B\in\mathds{R}\backslash\{-1,1\}$.  Then solve  for $\tau_\sigma (t)$ in  (\ref{eq: Bdottau-dotsigma}) and let $f_\sigma(\tau)$ denote its inverse (which exists since $\dot\tau_\sigma(t)>0$ for all $t$).
By (\ref{eq: BYtau-usigma})-(\ref{eq: Blambdaitau-Fisigma}),  the functions
\begin{eqnarray}Y(\tau)&=&u(\sigma (f_\sigma (\tau)))^{2/(1-B)}\label{eq: BY-usigmatau}\\
Q(\tau)&=&\frac{2}{\vartheta  ^2 (B-1)}P(\sigma (f_\sigma (\tau)))u(\sigma (f_\sigma (\tau)))^{2(B+1)/(B-1)}\label{eq: BQ-Psigmatau}\\
\lambda(\tau)&=&\frac{2}{\vartheta  ^2(B-1)}E(\sigma (f_\sigma (\tau)))\label{eq: Blambda-Esigmatau}\\
\lambda_i(\tau)&=&\frac{2}{\vartheta  ^2(1-B)}F_i(\sigma (f_\sigma (\tau)))\label{eq: Blambdai-Fisigmatau}
\end{eqnarray}
satisfy the generalized EMP (\ref{eq: BEMP}) for 
\begin{equation}B_i=\frac{1}{2}\left((1-B)C_i+(B+1)\right)\label{eq: BBi-Ci})\end{equation}
(by (\ref{eq: BCi-Bi})).
\end{thm}

\begin{proof}
To prove the forward statement, we begin by computing $g '(\sigma)$, a quantity that will be required to simplify the derivatives of $u(\sigma)$.  Since $g (\sigma (t))=t$, we differentiate this relation with respect to $t$ and obtain $g '(\sigma (t))\dot\sigma (t)=1$ so that $g '(\sigma)=1/\dot\sigma (g (\sigma))=\dot{\tau}_\sigma (g (\sigma))$.  Therefore by (\ref{eq: Bdottau-Y}) we have
\begin{equation}\dot\tau_\sigma(g(\sigma))g'(\sigma)=\vartheta \ Y(\tau_\sigma (g (\sigma)))^{(B+1)/2}.\label{eq: Bgprime-Y}\end{equation}
Differentiating the definition (\ref{eq: Bu-Ytausigma}) of $u(\sigma)$ and using (\ref{eq: Bgprime-Y}), we obtain \begin{eqnarray}u'(\sigma)
&=&\frac{1}{2}(1-B)Y(\tau_\sigma (g (\sigma)))^{-(B+1)/2}Y'(\tau_\sigma (g (\sigma)))\dot{\tau}_\sigma (g (\sigma))g '(\sigma)\notag\\
&=&\frac{\vartheta  }{2}(1-B)Y(\tau_\sigma (g (\sigma)))^{-(B+1)/2+(B+1)/2}Y'(\tau_\sigma (g (\sigma)))\notag\\
&=&\frac{\vartheta  }{2}(1-B)Y'(\tau_\sigma (g (\sigma))).
\end{eqnarray}
Differentiating again and using (\ref{eq: Bgprime-Y}), we see that 
\begin{eqnarray}u''(\sigma)&=&\frac{\vartheta  }{2}(1-B)Y''(\tau_\sigma (g (\sigma)))\dot{\tau}_\sigma (g (\sigma))g '(\sigma) \notag\\
&=&  \frac{\vartheta^2  }{2}(1-B)Y''(\tau_\sigma (g (\sigma))) Y(\tau_\sigma (g (\sigma)))^{(B+1)/2}.\label{eq: generalu''-generaly''}\end{eqnarray}
Since $Y(\tau)$ is assumed to satisfy the generalized EMP equation (\ref{eq: BEMP}), (\ref{eq: generalu''-generaly''}) becomes
\begin{eqnarray}u''(\sigma)&=&\frac{\vartheta  ^2}{2}(1-B)Y(\tau_\sigma (g (\sigma)))^{(B+1)/2}
\left[
-Q(\tau_\sigma (g (\sigma)))Y(\tau_\sigma (g (\sigma)))\right.\notag\\
&&\qquad\left.+\frac{\lambda(\tau_\sigma (g (\sigma)))}{Y(\tau_\sigma (g (\sigma)))^B}+\displaystyle\sum_{i=1}^N \frac{\lambda_i(\tau_\sigma (g (\sigma)))}{Y(\tau_\sigma (g (\sigma)))^{B_i}}
\right]\notag\\
&=&\frac{\vartheta  ^2}{2}(B-1)Q(\tau_\sigma (g (\sigma)))Y(\tau_\sigma (g (\sigma)))^{B+1}Y(\tau_\sigma (g (\sigma)))^{(1-B)/2}\notag\\
&& \qquad -\frac{\vartheta  ^2 (B-1)\lambda(\tau_\sigma (g (\sigma)))}{2Y(\tau_\sigma (g (\sigma)))^{(B-1)/2}} +\displaystyle\sum_{i=1}^N\frac{\vartheta  ^2 (1-B)\lambda_i(\tau_\sigma (g (\sigma)))}{2Y(\tau_\sigma (g (\sigma)))^{B_i-\frac{1}{2}(B+1)}}.\notag\\
\end{eqnarray}
By the definitions (\ref{eq: Bu-Ytausigma}),(\ref{eq: BP-Qtausigma}), (\ref{eq: BE-lambdatausigma}) and (\ref{eq: BFi-lambdaitausigma}) of $u(\sigma), P(\sigma), E(\sigma)$ and $F_i(\sigma)$ respectively, we have that 
\begin{eqnarray}
u''(\sigma)&=&P(\sigma)u(\sigma)-\frac{E(\sigma)}{u(\sigma)^{-1}}+\displaystyle\sum_{i=1}^N\frac{F_i(\sigma)}{u(\sigma)^{\frac{2B_i-(B+1)}{1-B}}}.\notag\\
&=&P(\sigma)u(\sigma)-E(\sigma)u(\sigma)+\displaystyle\sum_{i=1}^N\frac{F_i(\sigma)}{u(\sigma)^{C_i}}\end{eqnarray}
for $C_i$ as in (\ref{eq: BCi-Bi}).  This proves the forward implication.

To prove the converse statement, we will need $f_\sigma '(\tau)$ in order to simplify the derivatives of $Y(\tau)$.  Differentiating the relation $f _\sigma(\tau_\sigma (t))=t$ with respect to $t$ implies $f_\sigma '(\tau_\sigma (t))\dot{\tau}_\sigma (t)=1$ therefore $f_\sigma '(\tau)=1/\dot{\tau}_\sigma (f _\sigma(\tau))=\dot\sigma (f_\sigma (\tau))$.  By (\ref{eq: Bdotsigma-u}) we can form the useful quantity 
\begin{equation}\dot\sigma(f_\sigma(\tau))f_\sigma '(\tau)=\frac{1}{\vartheta}u(\sigma (f_\sigma (\tau)))^{\frac{(B+1)}{(B-1)}}.\label{eq: Bfprime-u}\end{equation}
  Now differentiating the definition of $Y(\tau)$ and using (\ref{eq: Bfprime-u}), we see that 
\begin{eqnarray}Y'(\tau)
&=&\frac{2}{(1-B)} u(\sigma (f_\sigma (\tau)))^{\frac{(B+1)}{(1-B)}}u'(\sigma (f _\sigma(\tau)))\dot\sigma (f _\sigma(\tau))f_\sigma '(\tau)\notag\\
&=&\frac{2}{\vartheta   (1-B)}u(\sigma (f_\sigma (\tau)))^{\frac{(B+1)}{(1-B)}+\frac{(B+1)}{(B-1)}}u'(\sigma (f_\sigma (\tau)))\notag\\
&=&\frac{2}{\vartheta   (1-B)}u'(\sigma (f_\sigma (\tau))).\label{eq: CYprime-u}
\end{eqnarray}
Differentiating $Y'(\tau)$ and again using (\ref{eq: Bfprime-u}), we obtain
\begin{eqnarray}Y''(\tau)
&=&\frac{2}{\vartheta (1-B)}u''(\sigma (f_\sigma (\tau)))\dot\sigma (f_\sigma (\tau))f_\sigma '(\tau)\notag\\
&=&\frac{2}{\vartheta^2 (1-B)}u''(\sigma (f_\sigma (\tau)))u(\sigma (f_\sigma (\tau)))^{\frac{(B+1)}{(B-1)}}.\label{eq: CYprimeprime-u}\end{eqnarray}
Since $u(\sigma)$ is assumed to satisfy the Schr\"odinger-type equation (\ref{eq: BNLS}), the equation (\ref{eq: CYprimeprime-u}) can be written as
\begin{eqnarray}
Y''(\tau)&=&\frac{2}{\vartheta  ^2 (1-B)}u(\sigma (f_\sigma (\tau)))^{(B+1)/(B-1)}\left[
-E(\sigma (f _\sigma(\tau)))u(\sigma (f_\sigma (\tau)))\right.\notag\\
&&\qquad\qquad\left.+P(\sigma (f_\sigma (\tau)))u(\sigma (f_\sigma (\tau)))+\displaystyle\sum_{i=1}^N\frac{F_i(\sigma (f_\sigma (\tau)))}{u(\sigma (f_\sigma (\tau)))^{C_i}}
\right]\notag\\
&=&\frac{2E(\sigma (f_\sigma (\tau)))}{\vartheta  ^2 (B-1)u(\sigma (f_\sigma (\tau)))^{-2B/(B-1)}}\notag\\
&&-\frac{2}{\vartheta  ^2 (B-1)}P(\sigma (f_\sigma (\tau)))u(\sigma (f_\sigma (\tau)))^{2(B+1)/(B-1)}u(\sigma (f_\sigma (\tau)))^{2/(1-B)}\notag\\
 &&+\displaystyle\sum_{i=1}^N\frac{2F_i(\sigma (f_\sigma (\tau)))/(1-B)}{\vartheta  ^2 u(\sigma (f_\sigma (\tau)))^{((B-1)C_i-(B+1))/(B-1)}}.
\end{eqnarray}
By the definitions (\ref{eq: BY-usigmatau}), (\ref{eq: BQ-Psigmatau}), (\ref{eq: Blambda-Esigmatau}) and (\ref{eq: Blambdai-Fisigmatau}) of $Y(\tau), Q(\tau), \lambda(\tau)$ and $\lambda_i(\tau)$ respectively, we obtain
\begin{equation}Y''(\tau)=\frac{\lambda(\tau)}{Y(\tau)^B}-Q(\tau)Y(\tau)+\displaystyle\sum_{i=1}^N\frac{\lambda_i(\tau)}{Y(\tau)^{B_i}}
\end{equation}
for $C_i$ as in (\ref{eq: BCi-Bi}).  This proves the theorem.

\end{proof}


\section{Scale factor and Schr\"odinger-type equations}


By composing the maps in Theorems \ref{thm: EFE-EMP} and \ref{thm: EMP-NLS}, one can see that a direct reformulation of the scale factor equation
\begin{equation}
\dot{H}(t)+\delta   H(t)^2+\varepsilon \dot{\phi}(t)^2=\frac{G(t)}{a(t)^A}+\displaystyle\sum_{i=1}^N \frac{G_i(t)}{a(t)^{A_i}}\label{eq: CEFE}
\end{equation}
(again $H(t)=\dot{a}(t)/a(t)$) in terms of the nonlinear Schrodinger-type equation
\begin{equation}
u''(\sigma)+[E(\sigma)-P(\sigma)]u(\sigma)=\displaystyle\sum_{i=1}^N\frac{F_i(\sigma)}{u(\sigma)^{C_i }}\label{eq: CNLS}\end{equation}
is possible by identifying $B_0$ in Theorem \ref{thm: EFE-EMP} with $B$ in Theorem \ref{thm: EMP-NLS}.   The theorem below is exactly the resulting statement.  As we noted in Section 2.2, this reformulation of the scale factor equation (\ref{eq: CEFE}) in terms of the non-linear Schr\"odinger-type equation (\ref{eq: CNLS})  will have one less non-linear term than the alternate generalized EMP reformulation seen in Section 2.1.

For the forward implication of the theorem in this section, composing the respective notations of  Theorems \ref{thm: EFE-EMP} and \ref{thm: EMP-NLS} is convenient.  However this is not true for the converse implication, in which we now use new notation which is equivalent to  identifying $(A_0-2\delta   )/q$ and $A_0/(q-\delta   )$ of Theorem \ref{thm: EFE-EMP} with $n/3$ and $m$, respectively, used here.  Here we also rename $A_0$ of Theorem \ref{thm: EFE-EMP} to just $A$.

In summary, the dictionary between functions which solve (\ref{eq: CEFE}) and functions which solve (\ref{eq: CNLS}) is as follows:
\begin{eqnarray}  a(f(\tau_\sigma(t)))^{(A-2\delta  )} & =&   u(\sigma(t))^{-2}\label{eq: Cataus-usigma}\quad\quad\\
\varepsilon (A-2\delta  )\psi'(\sigma)^2&=& 2P(\sigma)\label{eq: Cpsi-P}\\
(A-2\delta  )\theta^{A/(q-\delta  )}G(f(\tau_\sigma(t)))&=&2E(\sigma(t))\label{eq: CGtaus-Esigma}\\
(2\delta  -A)\theta^{A/(q-\delta  )}G_i(f(\tau_\sigma(t)))&=&2F_i(\sigma(t))\label{eq: CGitaus-Fisigma}\end{eqnarray}
where $f(\tau)$ is the inverse function of $\tau(t)$ and each of $\tau(t)$ and $\sigma(t)$ are solutions to the differential equations
\begin{eqnarray}\dot\tau(t)&=&\theta a(t)^{q-\delta  }\label{eq: Cdottau-a},\\
\dot\sigma(t)&=&\frac{1}{\sqrt{\vartheta}}u(\sigma(t))^{(n+6)/2n}\label{eq: Cdotsigma-u}\end{eqnarray}
respectively for some constants $\theta, \vartheta>0, q\in\mathds{R}\backslash\{\delta  \}$ and $n\in\mathds{R}\backslash\{0\}$.  Also $\phi$ and $\psi$ are related by  the composition
\begin{equation}\phi(f(\tau_\sigma(t)))=\psi(\sigma(t)),\label{eq: Cphitaus-psisigma}\end{equation}
and the powers $A_i$ and $C_i$ in (\ref{eq: CEFE}) and (\ref{eq: CNLS}) respectively, are related by the equation
\begin{equation}C_i=1-2\frac{(A_i-2\delta  )}{(A-2\delta  )}\label{eq: CCi-Ai}.\end{equation}
In addition, depending on which direction one is mapping solutions, it may be convenient to define $\tau_\sigma(t)$ as the solution to either 
\begin{equation}\dot\tau_\sigma(t)=\frac{1}{\dot\sigma(t)}\label{eq: Cdottaus-dotsigma}\end{equation}
or equivalently by (\ref{eq: Cdotsigma-u}), (\ref{eq: Cataus-usigma}) and  (\ref{eq: Cdottau-a}),
\begin{equation}\dot\tau_\sigma(t)=\dot\tau(f(\tau_\sigma(t)))^{(A+2q-2\delta  )/4(q-\delta  )}\label{eq: Cdottaus-dottautaus}.\end{equation}


\begin{thm}\label{thm: EFE-NLS}
Suppose you are given a twice differentiable function $a(t)>0$, a once differentiable function $\phi(t)$, and also functions $G(t), G_1(t), \dots , G_N(t)$ which satisfy the scale factor equation (\ref{eq: CEFE}) for some $N\in\mathds{N}$ and $\delta  , \varepsilon , A, A_1, \dots, A_N$ $ \in\mathds{R}$ where $A\neq 2\delta  $.  In order to construct a set of functions which solve the non-linear Schr\"odinger-type equation (\ref{eq: CNLS}), begin by solving for $\tau(t)$ in the differential equation (\ref{eq: Cdottau-a}) for any constants $\theta>0$ and  $q\in\mathds{R}\backslash\{\delta  \} $.  Let $f(\tau)$ denote the inverse of $\tau(t)$ (which exists since $\dot\tau(t)>0$ for all $t$), find the solution $\tau_\sigma(t)$ to (\ref{eq: Cdottaus-dottautaus}) and then find $\sigma(t)$ which solves (\ref{eq: Cdottaus-dotsigma}).  Let $g(\sigma)$ to denote the inverse function of $\sigma(t)$ (which exists since $\dot\tau(t)>0$ for all $t$).  Then by (\ref{eq: Cataus-usigma})-(\ref{eq: CGitaus-Fisigma}) the following functions solve the Schr\"odinger-type equation (\ref{eq: CNLS}):
\begin{eqnarray}u(\sigma)&=&a(f (\tau_\sigma(g  (\sigma)))^{(2\delta   - A )/2}\label{eq: Cu-ataus}\\
P(\sigma)&=&\frac{1}{2}\varepsilon (A-2\delta   )\psi'(\sigma)^2\label{eq: CP-psi-intheorem}\\
E(\sigma)&=&\frac{1}{2}( A-2\delta  )\theta  ^{A/(q-\delta   )}G(f (\tau_\sigma(g  (\sigma))))\label{eq: CE-Gtaus}\\
F_i(\sigma)&=&\frac{1}{2}(2\delta   -A )\theta  ^{A/(q-\delta  )}G_i(f (\tau_\sigma(g  (\sigma))))\label{eq: CFi-Gitaus}
\end{eqnarray}
for $C_i$ as in (\ref{eq: CCi-Ai}) and for 
\begin{equation}\psi(\sigma)=\phi(f(\tau_\sigma(g(\sigma))))\label{eq: Cpsi-phitaus}\end{equation}
(by (\ref{eq: Cphitaus-psisigma})).

Conversely, suppose you are given a twice differentiable function $u(\sigma)>0$, a continuous function $P(\sigma)$ and also functions $E(\sigma), F_1(\sigma), \dots, F_N(\sigma)$ which satisfy the Schr\"odinger-type equation (\ref{eq: CNLS}) for some $C_i \in\mathds{R}$, and $N\in\mathds{N}$.  In order to construct functions which solve the  scale factor  equation (\ref{eq: CEFE}), begin by solving for $\sigma(t)$  in the differential equation (\ref{eq: Cdotsigma-u}) for some $\vartheta>0$ and $n\in\mathds{R}\backslash\{0\}$.  Then find a solution $\tau_\sigma(t)$ to (\ref{eq: Cdottaus-dotsigma}) and let $f_\sigma(\tau)$ denote the inverse of $\tau_\sigma(t)$ (which exists since $\dot\sigma(t)>0$ for all $t$).  Next find functions $\tau(t)$ and $\psi(\sigma)$ which solve the differential equations 
\begin{equation}\dot\tau (t)=\dot\tau_\sigma(f_\sigma  (\tau (t)))^{4/(m+2)}\label{eq: Cdottau-dottaustau}\end{equation}
and 
\begin{equation}\psi'(\sigma)^2=\frac{6}{\varepsilon  qn}P(\sigma)\label{eq: Cpsi-P-intheorem}\end{equation}
respectively, for any $m\in\mathds{R}\backslash\{-2\}$ and $\varepsilon ,q\in\mathds{R}\backslash\{0\}$ (these equations are obtained by writing (\ref{eq: Cdottaus-dottautaus}) and (\ref{eq: Cpsi-P}) in the converse notation).  Then by (\ref{eq: Cataus-usigma})-(\ref{eq: CGitaus-Fisigma}), the following functions solve the  scale factor  equation (\ref{eq: CEFE}):
\begin{eqnarray}a(t)&=&u(\sigma (f_\sigma  (\tau (t))))^{-6/qn} \label{eq: Ca-usigmatau}\\
\phi(t)&=&\psi(\sigma (f_\sigma  (\tau (t))))\label{eq: Cphi-psisigmatau})\\
G(t)&=&\frac{6}{qn}\vartheta  ^{-2m/(m+2)}E(\sigma (f_\sigma  (\tau (t))))\label{eq: CG-Esigmatau}\\
G_i(t)&=&-\frac{6}{qn}\vartheta  ^{-2m/(m+2)}F_i(\sigma (f_\sigma  (\tau (t))))\label{eq: CGi-Fisigmatau}
\end{eqnarray}
for coefficient
\begin{equation}\delta   =\frac{q(3m-n)}{3(m+2)}\label{eq: Cgamma-mn}\end{equation}
and for powers
\begin{equation}A=\frac{qm(n+6)}{3(2+m)}, \qquad A_i=\frac{q}{6}\left(\frac{(nm-2n+12m)}{(m+2)}-nC_i\right)            \label{eq: CAi-Ci}\end{equation}
(by (\ref{eq: CCi-Ai}) in the converse notation).

\end{thm}
\begin{proof}
To prove the forward implication, we first compute $f '(\tau)$ and $g'(\sigma)$.  Differentiating the relation $f (\tau (t))=t$ with respect to $t$ gives $ f '(\tau (t))\dot\tau (t)=1$ so that $ f '(\tau)=1/\dot\tau (f (\tau))$.  Therefore by (\ref{eq: Cdottau-a}), we have 
\begin{equation}f'(\tau)=\frac{1}{\theta  }a(f (\tau))^{\delta   -q}.\label{eq: Cfprime-a}\end{equation}
Similarly $g(\sigma(t))=t$ implies $  g'(\sigma(t))\dot\sigma(t)=1$ so that $g  '(\sigma)=1/\dot\sigma  (g  (\sigma))=\dot\tau_\sigma(g  (\sigma))$, and then by (\ref{eq: Cdottaus-dottautaus}) and (\ref{eq: Cdottau-a}) we obtain
\begin{eqnarray}
g'(\sigma)&=&\dot\tau (f (\tau_\sigma(g  (\sigma))))^{(A+2q-2\delta   )/4(q-\delta   )}\notag\\
&=&\theta  ^{(A+2q-2\delta   )/4(q-\delta   )} a(f (\tau_\sigma(g  (\sigma))))^{(A+2q-2\delta   )/4}\label{eq: Cgprime-ataus}.\end{eqnarray}
By (\ref{eq: Cfprime-a}), (\ref{eq: Cgprime-ataus}) and (\ref{eq: Cdottaus-dottautaus}), we find that 
\begin{eqnarray}f'(\tau_\sigma(g(\sigma)))\dot\tau_\sigma(g(\sigma))g'(\sigma)&=& \theta^{(A+2q-2\delta   )/2(q-\delta   )-1} a(f (\tau_\sigma(g  (\sigma))))^{(A+2q-2\delta   )/2+\delta   -q}\notag\\
&=&\theta^{A/2(q-\delta  )}a(f(\tau_\sigma(g(\sigma))))^{A/2}.\label{eq: Cfprimedottausgprime}\end{eqnarray}
Differentiating the definition (\ref{eq: Cu-ataus}) of $u(\sigma)$ and using (\ref{eq: Cfprimedottausgprime}), we have that 
\begin{eqnarray}u'(\sigma)&=&\frac{1}{2}(2\delta   -A)a(f (\tau_\sigma(g  (\sigma))))^{(2\delta  -A -2)/2}\dot{a}(f (\tau_\sigma(g  (\sigma))))\notag\\
&&\qquad\qquad \cdot f '(\tau_\sigma(g  (\sigma)))\dot\tau_\sigma(g  (\sigma))g  '(\sigma)\notag\\
&=&\frac{1}{2}(2\delta   -A)\theta  ^{A/2(q-\delta   )}\dot{a}(f (\tau_\sigma(g  (\sigma))))a(f (\tau_\sigma(g  (\sigma))))^{\delta   -1}\notag\\
&=&\frac{1}{2}(2\delta   -A )\theta  ^{A/2(q-\delta  )}H(f (\tau_\sigma(g  (\sigma))))a(f (\tau_\sigma(g  (\sigma))))^{\delta   }\label{eq: Cuprime-a}
\end{eqnarray}
for $H(t)\stackrel{def.}{=}\dot{a}(t)/a(t)$.  Differentiating $u'(\sigma)$ and using (\ref{eq: Cfprimedottausgprime}) and the assumed  scale factor  equation (\ref{eq: CEFE}), the second derivative of $u$ is 
\begin{eqnarray}u''(\sigma)
&=&\frac{1}{2}(2\delta   -A )\theta  ^{A/2(q-\delta   )}f '(\tau_\sigma(g  (\sigma)))\dot\tau_\sigma(g  (\sigma))g  '(\sigma)\notag\\
&& \cdot\left[ \dot{H}(f (\tau_\sigma(g  (\sigma))))a(f (\tau_\sigma(g  (\sigma))))^{\delta   } +\delta   H(f (\tau_\sigma(g  (\sigma))))^2a(f (\tau_\sigma(g  (\sigma))))^{\delta   }\right]\notag\\
 &=&\frac{1}{2}(2\delta   -A)\theta  ^{A/(q-\delta   )}a(f (\tau_\sigma(g  (\sigma))))^{{A}/{2}+\delta   }\notag\\
 &&\cdot
 \left[-
 \varepsilon \dot{\phi}(f (\tau_\sigma(g  (\sigma))))^2+\frac{G(f (\tau_\sigma(g  (\sigma))))}{a(f (\tau_\sigma(g  (\sigma))))^A}+\displaystyle\sum_{i=1}^N \frac{G_i(t)}{a(f (\tau_\sigma(g  (\sigma))))^{A_i}}
 \right]\notag\\
 &=&\frac{1}{2}(A-2\delta   )\varepsilon  \theta  ^{A/(q-\delta   )}a(f (\tau_\sigma(g  (\sigma))))^{{A}/{2}+\delta   }\dot{\phi}(f (\tau_\sigma(g  (\sigma))))^2\notag\\
 &&\quad +\frac{1}{2}(2\delta   -A)\theta  ^{{A}/{(q-\delta   )}} G(f (\tau_\sigma(g  (\sigma))))a(f (\tau_\sigma(g  (\sigma))))^{\delta   -{A}/{2} }\notag\\
 &&\qquad+\frac{1}{2}(2\delta  -A )\theta  ^{{A}/{(q-\delta   )}}\displaystyle\sum_{i=1}^N\frac{G_i(f (\tau_\sigma(g  (\sigma))))}{a(f (\tau_\sigma(g  (\sigma))))^{(2A_i - A-2\delta   )/2}}.\label{eq: Cuprimeprime-a}\end{eqnarray}
 Differentiating the definition (\ref{eq: Cpsi-phitaus}) of $\psi(\sigma)$ and using  (\ref{eq: Cfprimedottausgprime}) we see that 
 \begin{equation}\psi'(\sigma)=\theta^{A/2(q-\delta  )}\dot\phi(f(\tau_\sigma(g(\sigma))))a(f(\tau_\sigma(g(\sigma))))^{A/2}\end{equation}
 so that by (\ref{eq: CP-psi-intheorem}), we can write $P(\sigma)$ as
 \begin{equation} P(\sigma)=\frac{1}{2}\varepsilon (A-2\delta   ) \theta^{A/(q-\delta  )}\dot\phi(f(\tau_\sigma(g(\sigma))))^2a(f(\tau_\sigma(g(\sigma))))^{A}. \label{eq: CP-a}\end{equation}
 By (\ref{eq: CP-a}) and the definitions (\ref{eq: Cu-ataus}), (\ref{eq: CE-Gtaus}) and (\ref{eq: CFi-Gitaus}) of $u(\sigma), E(\sigma)$ and $F_i(\sigma)$ respectively, (\ref{eq: Cuprimeprime-a}) becomes
 \begin{eqnarray}
u''(\sigma)  &=&P(\sigma)u(\sigma)-E(\sigma)u(\sigma)+\displaystyle\sum_{i=1}^N\frac{F_i(\sigma)}{u(\sigma)^{(2A_i - A-2\delta   )/(2\delta   -A)}}\notag\\
&=&P(\sigma)u(\sigma)-E(\sigma)u(\sigma)+\displaystyle\sum_{i=1}^N\frac{F_i(\sigma)}{u(\sigma)^{C_i}}
\end{eqnarray}
for $C_i$ as in (\ref{eq: CCi-Ai}).  This proves the forward implication.

To prove the converse statement, we will need the function $f_\sigma  '(\tau)$.  Differentiating the relation $f_\sigma  (\tau_\sigma(t))=t$ with respect to $t $ implies that $f_\sigma  '(\tau_\sigma(t))\dot\tau_\sigma(t)=1$ and so we have $f_\sigma'(\tau)=1/\dot\tau_\sigma(f_\sigma(\tau))=\dot\sigma(f_\sigma(\tau))$.  Therefore by (\ref{eq: Cdottau-dottaustau}) and (\ref{eq: Cdottaus-dotsigma}), we obtain the useful quantity
\begin{eqnarray}\dot\sigma(f_\sigma(\tau(t)))f'_\sigma(\tau(t))\dot\tau(t)&=&\dot\sigma(f_\sigma(\tau(t)))^2\dot\tau_\sigma(f_\sigma(\tau(t)))^{4/(m+2)}\notag\\
&=&\dot\sigma(f_\sigma(\tau(t)))^2\dot\sigma(f_\sigma(\tau(t)))^{-4/(m+2)}\notag\\
&=&\dot\sigma(f_\sigma(\tau(t)))^{2m/(m+2)}.\label{eq: dotsigmafsigmatauderiv=1}\end{eqnarray}
Using (\ref{eq: Cdotsigma-u}) to write this in terms of $u(\sigma)$, we have 
\begin{equation}\dot\sigma(f_\sigma(\tau(t)))f'_\sigma(\tau(t))\dot\tau(t)=
\vartheta^{-m/(m+2)}u(\sigma(f_\sigma(\tau(t))))^{\frac{m(n+6)}{n(m+2)}}.\label{eq: Cdotsigmafsprimedottau}
\end{equation}
Differentiating definition (\ref{eq: Ca-usigmatau}) of $a(t)$ gives
\begin{eqnarray}\dot{a}(t)&=&-\frac{6}{qn}u(\sigma (f_\sigma  (\tau (t))))^{-(6+qn)/qn}u'(\sigma (f_\sigma  (\tau (t))))\dot\sigma  (f_\sigma  (\tau (t))) f_\sigma  '(\tau (t))\dot\tau (t).\notag\\ \end{eqnarray}

\noindent Dividing by $a(t)=u(\sigma (f_\sigma  (\tau (t))))^{-6/qn}$ and  using (\ref{eq: Cdotsigmafsprimedottau}), we obtain
\begin{eqnarray}H(t)
&=&-\frac{6}{qn}\vartheta  ^{-m/(m+2)}u'(\sigma (f_\sigma  (\tau (t))))u(\sigma (f_\sigma  (\tau (t))))^{  \frac{m(n+6)}{n(m+2)} -1 }\notag\\
&=&-\frac{6}{qn}\vartheta  ^{-m/(m+2)}u'(\sigma (f_\sigma  (\tau (t))))u(\sigma (f_\sigma  (\tau (t))))^{ \frac{ 2(3m-n)}{n(m+2)}}\label{eq: CH-u}\end{eqnarray}
for $H(t)\stackrel{def.}{=}\dot{a}(t)/a(t)$ as usual.  Differentiating $H(t)$ and again using (\ref{eq: Cdotsigmafsprimedottau}), we get that
\begin{eqnarray}\dot{H}(t)
&=&-\frac{6}{qn}\vartheta  ^{-m/(m+2)} \dot\sigma  (f_\sigma  (\tau (t)))f_\sigma  '(\tau (t))\dot\tau (t)\cdot\notag\\
&&\qquad\left[u''(\sigma (f_\sigma  (\tau (t))))u(\sigma (f_\sigma  (\tau (t))))^{\frac{2(3m-n)}{n(m+2)}}\right.\notag\\
&&  \qquad \left.  + \frac{2(3m-n)}{n(m+2)} u'(\sigma (f_\sigma  (\tau (t))))^2u(\sigma (f_\sigma  (\tau (t))))^{\frac{2(3m-n)}{n(m+2)}-1}\right]\notag\\
&=&-\frac{6}{qn}\vartheta^{-2m/(m+2)} u''(\sigma (f_\sigma  (\tau (t)))) u(\sigma (f_\sigma  (\tau (t))))^{(12m-2n+mn)/n(m+2)}\notag\\
&& - \frac{12(3m-n)}{qn^2(m+2)}\vartheta^{-2m/(m+2)}u'(\sigma(f_\sigma(\tau(t))))^2u(\sigma(f_\sigma(\tau(t))))^{\frac{4(3m-n)}{n(m+2)}},\label{eq: CdotH-ubeforeNLS}
\end{eqnarray}

\noindent where we have simplified the powers of $u(\sigma(f_\sigma(\tau(t))))$ for the first and second terms by adding
\begin{equation}\frac{m(n+6)}{n(m+2)}+\frac{2(3m-n)}{n(m+2)}=\frac{12m-2n+mn}{n(m+2)}\end{equation}
and 
\begin{equation}\frac{m(n+6)}{n(m+2)}+\frac{2(3m-n)}{n(m+2)}-1=\frac{4(3m-n)}{n(m+2)},\end{equation}
respectively.  Since $u(\sigma)$ is assumed to satisfy the Schr\"odinger-type equation (\ref{eq: CNLS}), $\dot{H}$ in (\ref{eq: CdotH-ubeforeNLS}) now becomes
\begin{eqnarray}
\dot{H}(t)&=&-\frac{6}{qn}\vartheta  ^{-2m/(m+2)}u(\sigma (f_\sigma  (\tau (t))))^{(12m-2n+mn)/n(m+2)}\cdot\notag\\
&&\left[  P(\sigma (f_\sigma  (\tau (t))))u(\sigma (f_\sigma  (\tau (t))))  -E(\sigma (f_\sigma  (\tau (t))))u(\sigma (f_\sigma  (\tau (t)))) \right.\notag\\
&&\left.\qquad\qquad+\displaystyle\sum_{n=1}^N\frac{F_i(\sigma (f_\sigma  (\tau (t))))}{u(\sigma (f_\sigma  (\tau (t))))^{C_i }}  \right]\notag\\
&&-\frac{12(3m-n)}{qn^2(m+2)}\vartheta^{-2m/(m+2)}u'(\sigma(f_\sigma(\tau(t))))^2u(\sigma(f_\sigma(\tau(t))))^{\frac{4(3m-n)}{n(m+2)}}.\notag\\
\end{eqnarray}
Multiplying out the terms, we obtain
\begin{eqnarray}
\dot{H}(t)&=&-\frac{6}{qn}\vartheta  ^{-2m/(m+2)}u(\sigma (f_\sigma  (\tau (t))))^{\frac{2m(n+6)}{n(m+2)}}P(\sigma (f_\sigma  (\tau (t))))\\
&& \  \  \ +\frac{6}{qn}\vartheta  ^{-2m/(m+2)}u(\sigma (f_\sigma  (\tau (t))))^{\frac{2m(n+6)}{n(m+2)}}E(\sigma (f_\sigma  (\tau (t))))\notag\\
&& \  \  \ -\frac{6}{qn}\vartheta  ^{-2m/(m+2)}\displaystyle\sum_{n=1}^N\frac{F_i(\sigma (f_\sigma  (\tau (t))))}{u(\sigma (f_\sigma  (\tau (t))))^{C_i +(2n-12m-mn)/n(m+2)}}\notag\\
&& \ \ \ -\frac{12(3m-n)}{qn^2(m+2)}\vartheta  ^{-2m/(m+2)} u'(\sigma (f_\sigma  (\tau_\sigma(t))))^2u(\sigma (f_\sigma  (\tau (t))))^{\frac{4(3m-n)}{n(m+2)}}\notag
\label{eq: CdotH-u}\end{eqnarray}

\noindent where for the first two terms, we have simplified the powers of $u$ by adding
\begin{equation}\frac{12m-2n+mn}{n(m+2)}+1=\frac{2m(n+6)}{n(m+2)}.\end{equation}
By the definition (\ref{eq: Cgamma-mn}) of the constant $\delta  $, and by the computation (\ref{eq: CH-u}) of $H$ in terms of $u$, we have that 
\begin{equation}\delta   H(t)^2=\frac{12(3m-n)}{qn^2(m+2)}\vartheta^{-2m/(m+2)}u'(\sigma(f_\sigma(\tau(t))))^2 u(\sigma(f_\sigma(\tau(t))))^{\frac{4(3m-n)}{n(m+2)}}.\end{equation}
This shows that the last term in (\ref{eq: CdotH-u}) is equal to $-\delta   H(t)^2$ so that we have 
\begin{eqnarray}\dot{H}(t)+\delta   H(t)^2&=&-\frac{6}{qn}\vartheta  ^{-2m/(m+2)}u(\sigma (f_\sigma  (\tau (t))))^{\frac{2m(n+6)}{n(m+2)}}P(\sigma (f_\sigma  (\tau (t))))\notag\\
&& \  \  \ +\frac{6}{qn}\vartheta  ^{-2m/(m+2)}u(\sigma (f_\sigma  (\tau (t))))^{\frac{2m(n+6)}{n(m+2)}}E(\sigma (f_\sigma  (\tau (t))))\notag\\
&& \  \  \ -\frac{6}{qn}\vartheta  ^{-2m/(m+2)}\displaystyle\sum_{n=1}^N\frac{F_i(\sigma (f_\sigma  (\tau (t))))}{u(\sigma (f_\sigma  (\tau (t))))^{C_i +(2n-12m-mn)/n(m+2)}}.\notag\\
\label{eq: CdotHminusgamma}\end{eqnarray}
Differentiating the definition (\ref{eq: Cphi-psisigmatau}) of $\phi(t)$ and using (\ref{eq: Cdotsigmafsprimedottau}), we have that
\begin{eqnarray}\dot\phi(t)&=&\psi'(\sigma(f_\sigma(\tau(t))))\dot\sigma(f_\sigma(\tau(t)))f'_\sigma(\tau(t))\dot\tau(t)\notag\\
&=&\vartheta^{-m/(m+2)}u(\sigma(f_\sigma(\tau(t))))^{\frac{m(n+6)}{n(m+2)}}\psi'(\sigma(f_\sigma(\tau(t))))\label{eq: Cdotphi-psiprimeu}.
\end{eqnarray}
Using the definition (\ref{eq: Cpsi-P-intheorem}) of $\psi(\sigma)$ in terms of $P(\sigma)$ and squaring (\ref{eq: Cdotphi-psiprimeu}), we obtain
\begin{eqnarray}\dot\phi(t)^2&=&\frac{6}{\varepsilon  qn}\vartheta^{-2m/(m+2)}P(\sigma(f_\sigma(\tau(t))))u(\sigma(f_\sigma(\tau(t))))^{\frac{2m(n+6)}{n(m+2)}}.\end{eqnarray}
This shows that the first term in (\ref{eq: CdotHminusgamma}) is equal to $-\varepsilon \dot\phi(t)^2$ so that 
\begin{eqnarray}\dot{H}(t)+\delta   H(t)^2&=&-\varepsilon  \dot\phi(t)^2 +\frac{6}{qn}\vartheta  ^{-2m/(m+2)}u(\sigma (f_\sigma  (\tau (t))))^{\frac{2m(n+6)}{n(m+2)}}E(\sigma (f_\sigma  (\tau (t))))\notag\\
&& \  \  \ -\frac{6}{qn}\vartheta  ^{-2m/(m+2)}\displaystyle\sum_{n=1}^N\frac{F_i(\sigma (f_\sigma  (\tau (t))))}{u(\sigma (f_\sigma  (\tau (t))))^{C_i +{(2n-12m-mn)}/{n(m+2)}}}.\notag\\
\label{eq: CdotHgammadelta-u}\end{eqnarray}
Now utilizing the definitions (\ref{eq: Ca-usigmatau}), (\ref{eq: CG-Esigmatau}) and (\ref{eq: CGi-Fisigmatau})  of $a(t), G(t)$ and $G_i(t)$ respectively, (\ref{eq: CdotHgammadelta-u}) becomes
\begin{eqnarray}\dot{H}(t)+\delta   H(t)^2&=&-\varepsilon \dot{\phi}^2+\frac{G(t)}{a(t)^{\frac{qm(n+6)}{3(m+2)}}} +\displaystyle\sum_{i=1}^N \frac{G_i(t)}{a(t)^{- \frac{qn}{6}C_i-{q(2n-12m-mn)}/{6(m+2)} }} \notag\\
&=&-\varepsilon \dot{\phi}^2+\frac{G(t)}{a(t)^A}+\displaystyle\sum_{i=1}^N\frac{G_i(t)}{a(t)^{A_i}},\end{eqnarray}
which proves the theorem for $A$ and $A_i$ as in (\ref{eq: CAi-Ci}).\end{proof}

The statement of this theorem would have been much simpler if we were able to choose $\tau(t)=\tau_\sigma(t)$, since then  many compositions of functions would cancel.  In fact, such cancellations would make the use of $\tau(t)$ obsolete, leaving only $\sigma(t)$ and its inverse $g(\sigma)$ to be required for the reparameterization which takes place in the translation between the functions $a(t)$ and $u(\sigma)$.  By (\ref{eq: Cdottaus-dottautaus}) in the forward direction or equivalently (\ref{eq: Cdottau-dottaustau}) in the converse, $\dot\tau(t)=\dot\tau_\sigma(t)$ for $(A+2q-2\delta  )/4(q-\delta  )=1$ and $4/(m+2)=1$ respectively.  This corresponds to the choice of parameter $q=A/2+\delta  $, or equivalently $m=2$ and $n=6$ in the converse notation.  Also for this choice $\delta  =0$ by (\ref{eq: Cgamma-mn}).  As stated above, the theorem only holds for $q\neq \delta  =0$, so the choice $q=A/2+\delta  =A/2$ is only possible for $A\neq 0$.   In summary, when $A\neq 0$ and $\delta  =0$, we apply the theorem with $q=A/2$ or equivalently $m=2, n=6$ in the converse notation.  This will allow us to take the integration constant zero when integrating the relation $\dot\tau(t)=\dot\tau_\sigma(t)$ so that $\tau(t)=\tau_\sigma(t)$ and the statement of the theorem will become simpler.

Therefore in application, the following two versions of the theorem will be useful.  If $A\neq 0$ and $\delta  =0$ then the theorem will be implemented with $q=A/2, m=2, n=6$ and $\tau(t)=\tau_\sigma(t)$.  If $A=0$ and $\delta  \neq 0$  then in the converse notation $m=0$ and by the comments on notation at the beginning of this section, $q=-6\delta  /n$ and the theorem also takes a simpler form in this case.  In particular by integrating (\ref{eq: Cfprimedottausgprime}) and (\ref{eq: dotsigmafsigmatauderiv=1}) with $A=m=0$, we have that $f(\tau_\sigma(g(\sigma))) = \sigma + t_0$ and $\sigma(f_\sigma(\tau(t))) = t-t_0$ for some constant $t_0\in\mathds{R}$ (If $A, \delta  $ are both nonzero then the theorem would be implemented as-is).

Also in the statement of the above theorem, one can show that the constants $\theta$ and $\vartheta$ are related by $\vartheta  =\theta  ^{(A+2q-2\delta   )/2(q-\delta   )}$ or equivalently $\theta  =\vartheta  ^{2/(m+2)}$ in the converse notation.  This was, in fact, why we chose to state the theorem more simply by using two separate constants $\theta  $ and $\vartheta $ for the forward and converse implications respectively - but this is not necessary in the case of the first $\delta  = 0$ corollary since $q=A/2$ implies $\vartheta  =\theta^2$.

\break
\begin{cor}{\bf ($A\neq0, \delta  =0 \stackrel{choose}{\rightarrow} q=A/2, \tau_\sigma(t)=\tau(t) \Rightarrow m=2, n=6, \vartheta=\theta^2 $)}\label{cor: EFE-NLSAnonzeroEzero}
Suppose you are given a twice differentiable function $a(t)>0$, a once differentiable function $\phi(t)$, and also functions $G(t), G_1(t), \dots , G_N(t)$ which satisfy the  scale factor  equation
\begin{equation}\dot{H}(t)+\varepsilon  \dot\phi(t)^2=\frac{G(t)}{a(t)^A}+\displaystyle\sum_{i=1}^N\frac{G_i(t)}{a(t)^{A_i}}\label{eq: CEFEANONZERO}\end{equation}
for some $N\in\mathds{N}$ and $A\neq 0, \varepsilon , A_1, \dots, A_N \in\mathds{R}$ ($\dot{H}(t)\stackrel{def.}{=}\dot{a}(t)/a(t)$).  In order to construct a set of functions which solve the Schr\"odinger-type equation
\begin{equation}u''(\sigma)+[E(\sigma)-P(\sigma)]u(\sigma)=\displaystyle\sum_{i=1}^N\frac{F_i(\sigma)}{u(\sigma)^{C_i}},\label{eq: CNLSANONZERO}\end{equation}
begin by solving for $\sigma(t)$ in the differential equation 
\begin{equation}\dot\sigma(t)=\frac{1}{\theta}a(t)^{-A/2}\end{equation}
for some $\theta>0$.  Now allow $g(\sigma)$ to denote the inverse function of $\sigma(t)$ (which exists since $\dot\sigma(t)>0$ for all $t$).  Then the following functions solve the Schr\"odinger-type equation (\ref{eq: CNLSANONZERO}):
\begin{eqnarray}u(\sigma)&=&a(g  (\sigma))^{-A/2}\label{eq: Cu-atausANONZERO}\\
P(\sigma)&=&\frac{\varepsilon  A}{2}\psi'(\sigma)^2\label{eq: CP-psi-intheoremANONZERO}\\
E(\sigma)&=&\frac{\theta  ^{2} A}{2}G(g  (\sigma))\label{eq: CE-GtausANONZERO}\\
F_i(\sigma)&=&-\frac{\theta  ^{2} A}{2}G_i(g  (\sigma))\label{eq: CFi-GitausANONZERO}
\end{eqnarray}
for 
\begin{equation}C_i=1-2\frac{A_i}{A}, 1\leq i\leq N\label{eq: CCi-AiANONZERO}\end{equation}
and
\begin{equation}\psi(\sigma)=\phi(g(\sigma))\label{eq: Cpsi-phitausANONZERO}.\end{equation}

Conversely, suppose you are given a twice differentiable function $u(\sigma)>0$ and also functions $E(\sigma), P(\sigma), F_1(\sigma), \dots, F_N(\sigma)$ which satisfy the Schr\"odinger-type equation (\ref{eq: CNLSANONZERO}) for some $C_i \in\mathds{R}$, and $N\in\mathds{N}$.  In order to construct functions which solve the  scale factor  equation (\ref{eq: CEFEANONZERO}), begin by solving for $\sigma(t)$  in the differential equation
\begin{equation}\dot\sigma(t)=\frac{1}{\theta}u(\sigma(t))\label{eq: Cdotsigma-uANONZERO}\end{equation}
for some $\theta>0$.   Next find a function $\psi(\sigma)$ which solves the differential equation 
\begin{equation}\psi'(\sigma)^2=\frac{2}{\varepsilon  A}P(\sigma)\label{eq: Cpsi-P-intheoremANONZERO}\end{equation}
 for any  $\varepsilon , A \in\mathds{R}\backslash\{0\}$.  Then the following functions solve the  scale factor  equation (\ref{eq: CEFEANONZERO}):
\begin{eqnarray}a(t)&=&u(\sigma (t))^{-2/A} \label{eq: Ca-usigmatauANONZERO}\\
\phi(t)&=&\psi(\sigma (t))\label{eq: Cphi-psisigmatauANONZERO})\\
G(t)&=&\frac{2}{A\theta^2}E(\sigma (t))\label{eq: CG-EsigmatauANONZERO}\\
G_i(t)&=&-\frac{2}{A\theta^2}  F_i(\sigma (t))\label{eq: CGi-FisigmatauANONZERO}
\end{eqnarray}
for
\begin{equation} A_i=\frac{A}{2}\left(1-C_i\right), 1\leq i\leq N .           \label{eq: CAi-Ci}\end{equation}

\end{cor}

\begin{cor}{\bf ($A=0, \delta  \neq 0 \Rightarrow -\frac{qn}{6}=\delta  , m=0  \Rightarrow  f(\tau_\sigma(g(\sigma)))=\sigma+t_0, \sigma(f_\sigma(\tau(t)))=t-t_0   $)}\label{cor: EFE-NLSAzeroEnonzero}
Suppose you are given a twice differentiable function $a(t)>0$, a once differentiable function $\phi(t)$, and also functions $G(t), G_1(t), \dots , G_N(t)$ which satisfy the  scale factor  equation
\begin{equation}
\dot{H}(t)+\delta   H(t)^2+\varepsilon  \dot\phi(t)^2=G(t)+\displaystyle\sum_{i=1}^N \frac{G_i(t)}{a(t)^{A_i}}
\label{eq: CEFEAZERO}\end{equation}
 for some $N\in\mathds{N}$ and $\delta  \neq 0, \varepsilon , A_1, \dots, A_N \in\mathds{R}$ (where as usual, $H(t)\stackrel{def.}{=}\dot{a}(t)/a(t)$).  
 Then the functions 
\begin{eqnarray}u(\sigma)&=&a(\sigma+t_0)^{\delta   }\label{eq: Cu-atausAZERO}\\
P(\sigma)&=&-\varepsilon \delta   \psi'(\sigma)^2\label{eq: CP-psi-intheoremAZERO}\\
E(\sigma)&=&-\delta   G(\sigma+t_0)\label{eq: CE-GtausAZERO}\\
F_i(\sigma)&=&\delta   G_i(\sigma+t_0)\label{eq: CFi-GitausAZERO}
\end{eqnarray}
solve the Schr\"odinger-type equation 
 \begin{equation}u''(\sigma)+[E(\sigma)-P(\sigma)]u(\sigma)=\displaystyle\sum_{i=1}^N\frac{F_i(\sigma)}{u(\sigma)^{C_i}}\label{eq: CNLSAZERO}\end{equation}
\begin{equation}\mbox{{\it for \ } }C_i=\frac{A_i}{\delta  }-1, 1\leq i\leq N\label{eq: CCi-AiAZERO}\end{equation}
\begin{equation}\mbox{{\it and \ } }\psi(\sigma)=\phi(\sigma+t_0)\label{eq: Cpsi-phitausAZERO}.\end{equation}

Conversely, suppose you are given a twice differentiable function $u(\sigma)>0$ and also functions $E(\sigma), P(\sigma), F_1(\sigma), \dots, F_N(\sigma)$ which satisfy the Schr\"odinger-type equation (\ref{eq: CNLSAZERO}) for some $C_i \in\mathds{R}$ and $N\in\mathds{N}$.  
For $\psi(\sigma)$ such that
\begin{equation}\psi'(\sigma)^2=-\frac{1}{\varepsilon  \delta  }P(\sigma)\label{eq: Cpsi-P-intheoremAZERO}\end{equation}
 for any $\varepsilon ,\delta\in\mathds{R}\backslash\{0\}$, the following functions solve the  scale factor equation (\ref{eq: CEFEAZERO}):
\begin{eqnarray}a(t)&=&u(t-t_0)^{1/\delta  } \label{eq: Ca-usigmatauAZERO}\\
\phi(t)&=&\psi(t-t_0)\label{eq: Cphi-psisigmatauAZERO}\\
G(t)&=&-\frac{1}{\delta  }E(t-t_0)\label{eq: CG-EsigmatauAZERO}\\
G_i(t)&=&\frac{1}{\delta  }F_i(t-t_0)\label{eq: CGi-FisigmatauAZERO}
\end{eqnarray}

for
\begin{equation}\qquad A_i=\delta  \left(1+C_i\right), 1\leq i\leq N. \label{eq: CAi-CiAZERO}\end{equation}
\end{cor}

\chapter{Reformulations of the Friedmann-Robertson-Lema\^itre-Walker model}
\label{ch: FRLW}

This cosmological model assumes that the $d+1$-dimensional spacetime is both homogeneous and isotropic, resulting in a metric of the form
\begin{equation} ds^2=-dt^2+a(t)^2\left( \frac{dr^2}{1-kr^2}+r^2d\Omega_{d-1}^2\right)\label{eq: FRLWmetric}\end{equation}
where $a(t)$ is the scale factor, $k\in\{-1,0,1\}$ is the curvature parameter and
\begin{equation} d\Omega^2_{d-1}=d\theta_1^2+\sin^2\theta_1d\theta_2^2+\cdots +\sin^2\theta_1\cdots\sin^2\theta_{d-2}d\theta_{d-1}^2.\end{equation}
In this section, we take the energy density and pressure in equation (\ref{eq: T2matrix})  to be 
\begin{equation}\rho(t)=\displaystyle\sum_{i=1}^{M}\frac{D_i(t)}{a(t)^{n_i}}+\rho'(t)\label{eq: rhototal=sum}\end{equation}
and
\begin{equation}p(t)=\displaystyle\sum_{i=1}^{M}\frac{(n_i-d)D_i(t)}{da(t)^{n_i}}+p'(t)\label{eq: ptotal=sum}\end{equation}
respectively, for some $n_i\in\mathds{R}$ and $1\leq i \leq M$.  The nontrivial and distinct Einstein's equations $g^{ij}G_{ij}=-\kappa g^{ij}T_{ij}+\Lambda$ are  the $(i,j)=(0,0)$ and $(i,j)=(1,1)$ equations.  
\noindent Dividing by $(d-1)$, these equations are
\begin{equation}\frac{d}{2}H^2(t)+\frac{dk}{2a(t)^2}\stackrel{(i)}{=}\frac{\kappa}{(d-1)} \left[\frac{1}{2}\dot\phi(t)^2+V(\phi(t))+\displaystyle\sum_{i=1}^{M}\frac{D_i(t)}{a(t)^{n_i}}+\rho'(t)\right]+\frac{\Lambda}{(d-1)}\notag\end{equation}
\break
\begin{eqnarray}&&\label{eq: FRLWEFEiii}\\
\dot{H}(t)+\frac{d}{2}H(t)^2+\frac{(d-2)k}{2a(t)^2}&\stackrel{(ii)}{=}&-\frac{\kappa}{(d-1)} \left[\frac{1}{2}\dot\phi(t)^2-V(\phi(t))+\displaystyle\sum_{i=1}^{M}\frac{(n_i-d)D_i(t)}{da(t)^{n_i}}\right.\notag\\
&&\qquad\qquad\qquad\qquad\left.+p'(t)\right]+\frac{\Lambda}{(d-1)}\qquad\qquad \notag\end{eqnarray}
where $H(t)\stackrel{def.}{=}\frac{\dot{a}(t)}{a(t)}$.

\section{In terms of a Generalized EMP}
\begin{thm}\label{thm: FRLWEFE-EMP}
Suppose you are given a twice differentiable function $a(t)>0$, a once differentiable function $\phi(t)$, and also functions $D_i(t), \rho'(t), p'(t), V(x)$ which satisfy the Einstein equations $(i),(ii)$ in (\ref{eq:  FRLWEFEiii}) for some $k, n_1, \dots, n_{M}, \Lambda \in\mathds{R}, d\in\mathds{R}\backslash\{0,1\},$ $ \kappa\in\mathds{R}\backslash\{0\}$ and $M,M_2\in\mathds{N}$.  If $f(\tau)$ is the inverse of a function $\tau(t)$ which satisfies
\begin{equation}\dot{\tau}(t)=\theta a(t)^{q}\label{eq: FRLWEMPdottau-a}\end{equation}
for some $\theta>0$ and $q\in\mathds{R}\backslash\{0\}$, then 
\begin{eqnarray}Y(\tau)=a(f(\tau))^q\qquad\mbox{ and }\qquad Q(\tau)=\frac{q\kappa}{(d-1)}\varphi '(\tau)^2\label{eq: FRLWEMPY-aQ-varphi}\end{eqnarray}
solve the generalized EMP equation
\begin{eqnarray}&&Y''(\tau)+Q(\tau)Y(\tau)=\label{eq: FRLWEMP}\\
&& \ \ \ \  \ \ \ \ \ \ \frac{qk}{\theta^2 Y(\tau)^{(2+q)/q}}-\displaystyle\sum_{i=1}^{M}\frac{qn_i \kappa\mbox{\hookD}_i(\tau)}{\theta^2 d(d-1) Y(\tau)^{(n_i+q)/q}}-\frac{q\kappa(\varrho(\tau)+\textup{ \textlhookp}(\tau))}{\theta^2(d-1) Y(\tau)}\notag\end{eqnarray}
for 
\begin{equation}\varphi(\tau)=\phi(f(\tau))\label{eq: FRLWEMPvarphi-phi}\end{equation}
\begin{equation}\mbox{\hookD}_i(\tau)=D_i(f(\tau)), \ \  1\leq i\leq M \label{eq: FRLWEMPhookD-D}\end{equation}
and
\begin{equation}\varrho(\tau)=\rho'(f(\tau)), \ \textup{\textlhookp}(\tau)=p'(f(\tau)). \label{eq: FRLWEMPvarrho-rhohookp-p}\end{equation}

Conversely, suppose you are given a twice differentiable function $Y(\tau)>0$, a continuous function $Q(\tau)$, and also functions $\mbox{\hookD}_i(\tau)$ for $1\leq i\leq M, M\in\mathds{N}$ and $\varrho(\tau), \textup{\textlhookp}(\tau)$  which solve (\ref{eq: FRLWEMP}) for some constants $\theta>0, q, \kappa \in\mathds{R}\backslash\{0\}, k\in\mathds{R}, d\in\mathds{R}\backslash\{0,1\}$ and $n_i\in\mathds{R}$ for $1\leq i \leq M$.  In order to construct functions which solve $(i),(ii)$, first find $\tau(t), \varphi(\tau)$ which solve the differential equations
\begin{equation}\dot{\tau}(t)=\theta Y(\tau(t))\qquad\mbox{ and }\qquad \varphi '(\tau)^2= \frac{(d-1)}{q\kappa} Q(\tau)\label{eq: FRLWEMPdottau-Yvarphi-Q}.\end{equation}
Then the functions
\begin{eqnarray}a(t)&=&Y(\tau(t))^{1/q}\label{eq: FRLWEMPa-Y}\\
\phi(t)&=&\varphi(\tau(t))\label{eq:   FRLWEMPphi-varphi}\end{eqnarray}
\begin{equation}D_i(t)=\mbox{\hookD}_i(\tau(t)), \ \  1\leq i \leq M\label{eq: FRLWEMPD-hookD}\end{equation}
\begin{equation}\rho'(t)=\varrho(\tau(t)), \ p'(t)=\textup{\textlhookp}(\tau(t)), \label{eq: FRLWEMPrho-varrhop-hookp}\end{equation}
and\\
$V(\phi(t))$
\begin{equation}
=\left[\frac{d(d-1)}{2\kappa}\left(\frac{\theta^2(Y')^2}{q^2}+\frac{k}{Y^{2/q}}\right)-\frac{\theta^2}{2}Y^2(\varphi ')^2-\displaystyle\sum_{i=1}^{M}\frac{\mbox{\hookD}_i}{Y^{n_i/q}}-\varrho-\frac{\Lambda}{\kappa}\right]\circ\tau(t)\label{eq: FRLWEMPVphi-Y}
\end{equation}
satisfy equations $(i),(ii)$.

\end{thm}
\begin{proof}
This proof will implement Theorem \ref{thm: EFE-EMP} with constants and functions as indicated in the following table.

   \begin{table}[ht]
\centering
\caption{{ Theorem \ref{thm: EFE-EMP} applied to FRLW}}\label{tb: FRLWEMP}
\vspace{.2in}
\begin{tabular}{r | l c r | l}
In Theorem & substitute & & In Theorem& substitute \\[4pt]
\hline
\raisebox{-5pt}{$\delta  $} & \raisebox{-5pt}{$0$}      &&    \raisebox{-5pt}{$\varepsilon $} & \raisebox{-5pt}{${\kappa}/{(d-1)}$}\\[8pt]
                                       $G_0(t)$ & $\mbox{constant } k$    &&     $A_0$ & $2$ \\[8pt]
$G_i(t), 1\leq i\leq M$ &$ \frac{-n_i \kappa}{d(d-1)}D_i(t)   $  &&    $A_i$ &$n_i$ \\[8pt]
 $G_{M+1}(t)$ & $ \frac{-\kappa}{(d-1)}(\rho'(t)+p'(t))$  &&      $A_{M+1}$   &$0$\\[8pt]
 $ \lambda_0(\tau)$ &constant ${qk}/{\theta^2}$       &&      $B_0$ &$(2+q)/q$ \\[8pt] 
   $\lambda_i(\tau), 1\leq i\leq M$&$\frac{-qn_i\kappa}{\theta^2 d(d-1)}\mbox{\hookD}_i(\tau)   $ &&  $B_i$&$(n_i+q)/q$\\[8pt]
$\lambda_{M+1}(\tau)$&$\frac{-q\kappa}{\theta^2(d-1)}(\varrho(\tau)+\textup{\textlhookp}(\tau))$ &&        $B_{M+1}$&  $1$ \\[6pt]
\hline
\end{tabular}
\end{table}

To prove the forward implication, we assume to be given functions which solve the Einstein field equations $(i)$ and $(ii)$.  Subtracting equations $(ii)-(i)$ we obtain
\begin{equation}\dot{H}(t)-\frac{k}{a(t)^2}=-\frac{\kappa}{(d-1)}\left[\dot{\phi}(t)^2+\displaystyle\sum_{i=1}^{M}\frac{n_i D_i(t)}{da(t)^{n_i}}+(\rho'(t)+p'(t))
\right].\label{eq: FRLWEMPiiminusi}\end{equation}
This shows that $a(t), \phi(t), D_i(t), \rho'(t)$ and $p'(t)$ satisfy the hypothesis of Theorem \ref{thm: EFE-NLS}, applied with  constants $\epsilon, \varepsilon , N, A_0, \dots, A_{N}$ and functions $G_0(t),\dots,G_{N}(t)$ according to Table \ref{tb: FRLWEMP}.  Since $\tau(t), Y(\tau), Q(\tau)$ and $\varphi(\tau)$ defined in (\ref{eq: FRLWEMPdottau-a}), (\ref{eq: FRLWEMPY-aQ-varphi}) and (\ref{eq: FRLWEMPvarphi-phi})  are equivalent to that in the forward implication of Theorem \ref{thm: EFE-NLS}, by this theorem and by definitions (\ref{eq: FRLWEMPhookD-D}) and (\ref{eq: FRLWEMPvarrho-rhohookp-p}) of $\mbox{\hookD}(\tau)$ and $\varrho(\tau), \textup{\textlhookp}(\tau)$, the generalized EMP equation (\ref{eq: AEMP}) holds for constants $B_0, \dots, B_{N}$ and functions $\lambda_0(\tau), \dots, \lambda_{N}(\tau)$ as indicated in Table \ref{tb: FRLWEMP}.  This proves the forward implication.

To prove the converse implication, we assume to be given functions which solve the generalized EMP equation (\ref{eq: FRLWEMP}) and we begin by showing that $(i)$ is satisfied.   Differentiating definition (\ref{eq: FRLWEMPa-Y}) of $a(t)$ and by the definition of $\tau(t)$ in (\ref{eq: FRLWEMPdottau-Yvarphi-Q}), we have that
\begin{eqnarray}\dot{a}(t)&=&\frac{1}{q}Y(\tau(t))^{\frac{1}{q}-1}Y'(\tau(t))\dot{\tau}(t)\notag\\
&=&\frac{\theta}{q}Y(\tau(t))^{1/q}Y'(\tau(t)).\end{eqnarray}
Dividing by $a(t)$ we obtain
\begin{equation}H(t)\stackrel{def.}{=}\frac{\dot{a}(t)}{a(t)}=\frac{\theta}{q}Y'(\tau(t)).\label{eq: FRLWEMPH-Ytau}\end{equation}
Differentiating the definition (\ref{eq: FRLWEMPphi-varphi}) of $\phi(t)$ and using definition of  $\tau(t)$ in (\ref{eq: FRLWEMPdottau-Yvarphi-Q}),   we find that
\begin{equation}\dot{\phi}(t)=\varphi '(\tau(t))\dot{\tau}(t)=\theta\varphi '(\tau(t))Y(\tau(t)).\label{eq: FRLWEMPdotphi-Y}\end{equation}
Using (\ref{eq: FRLWEMPH-Ytau}), (\ref{eq: FRLWEMPdotphi-Y}), and  the definitions (\ref{eq: FRLWEMPa-Y}), (\ref{eq: FRLWEMPD-hookD}) and (\ref{eq: FRLWEMPrho-varrhop-hookp}) of $a(t), D_i(t)$ and $\rho'(t)$ respectively, the definition (\ref{eq: FRLWEMPVphi-Y}) of $V\circ\phi$ can be written as
\begin{equation}
V(\phi(t))=\frac{d(d-1)}{2\kappa}\left(
H(t)^2+\frac{k}{a(t)^2}\right)-\frac{1}{2}\dot\phi(t)^2-\displaystyle\sum_{i=1}^{M}\frac{D_i(t)}{a(t)^{n_i}}-\rho'(t)-\frac{\Lambda}{\kappa}
\label{eq: FRLWEMPVphi-a}.\end{equation}
This shows that $(i)$ holds  (that is, the definition of $V(\phi(t))$ was designed to be such that $(i)$ holds).

To conclude the proof we must also show that $(ii)$ holds.   In the converse direction the hypothesis of the converse of Theorem \ref{thm: EFE-EMP} holds, applied with constants $N, B_0, \dots, B_{N}$ and functions $\lambda_0(\tau), \dots, \lambda_{N}(\tau)$ as indicated in Table \ref{tb: FRLWEMP}.  Since  $\tau(t), \varphi(\tau), a(t)$ and $\phi(t)$ defined in (\ref{eq: FRLWEMPdottau-Yvarphi-Q}), (\ref{eq: FRLWEMPa-Y}) and (\ref{eq: FRLWEMPphi-varphi}) are consistent with the converse implication of Theorem A.1, applied with $\delta  $ and $\varepsilon $ as in Table \ref{tb: FRLWEMP}, by this theorem and by the definitions (\ref{eq: FRLWEMPD-hookD}) and (\ref{eq: FRLWEMPrho-varrhop-hookp}) of $D_i(t)$ and $\rho'(t), p'(t)$ the  scale factor  equation (\ref{eq: AEFE}) holds for constants $\delta  , \varepsilon , A_0, \dots, A_{N}$ and functions $G_0(t),\dots,G_{N}(t)$ according to Table \ref{tb: FRLWEMP}.  That is, we have regained (\ref{eq: FRLWEMPiiminusi}) which shows that the subtraction of equations (ii)-(i) holds in the converse direction.  Now solving (\ref{eq: FRLWEMPVphi-a}) for $\rho'(t)$ and substituting this into (\ref{eq: FRLWEMPiiminusi}), we obtain (ii).  This proves the theorem.
\end{proof}

\subsection{First reduction to classical EMP: pure scalar field}
As a special case we take $\rho'=p'=D_i=0$ and we choose the parameter $q=1$.  Then Theorem  \ref{thm: FRLWEFE-EMP} 
shows that solving the Einstein equations
\begin{equation}\frac{d}{2}H^2(t)+\frac{dk}{2a(t)^2}\stackrel{(i)'}{=}\frac{\kappa}{(d-1)} \left[\frac{1}{2}\dot\phi(t)^2+V(\phi(t))\right]+\frac{\Lambda}{(d-1)}\end{equation}
\begin{eqnarray*}
\dot{H}(t)+\frac{d}{2}H(t)^2+\frac{(d-2)k}{2a(t)^2}&\stackrel{(ii)'}{=}&-\frac{\kappa}{(d-1)} \left[\frac{1}{2}\dot\phi(t)^2-V(\phi(t))\right]+\frac{\Lambda}{(d-1)}\qquad\qquad \end{eqnarray*}
is equivalent to solving the classical EMP equation
\begin{equation}Y''(\tau)+Q(\tau)Y(\tau)=\frac{k}{\theta^2 Y(\tau)^3}\label{eq: nomatterFRLWEMP}\end{equation}
for any constant $\theta>0$.  The solutions of $(i)',(ii)'$ and (\ref{eq: nomatterFRLWEMP}) are related by
\begin{equation}a(t)=Y(\tau(t))\qquad\mbox{ and }\qquad \varphi '(\tau)^2= \frac{(d-1)}{\kappa} Q(\tau)\label{eq: nomatterFRLWclassEMPa-Yvarphi-Q}\end{equation}
for $\phi(t)=\varphi(\tau(t))$ and 
\begin{equation}\dot\tau(t)=\theta a(t)=\theta Y(\tau(t)).\label{eq: nomatterFRLWclassEMPdottau-a-Y}\end{equation}
Also in the converse direction, $V$ is taken to be
\begin{equation}
V(\phi(t))=\left[\frac{d(d-1)}{2\kappa}\left(\theta^2(Y')^2+\frac{k}{Y^{2}}\right)-\frac{\theta^2}{2}Y^2(\varphi ')^2-\frac{\Lambda}{\kappa}\right]\circ\tau(t).\label{eq: nomatterFRLWEMPVphi-Y}
\end{equation}

Referring to Appendix D for solutions of the classical and corresponding  homogeneous EMP equation (\ref{eq: nomatterFRLWEMP}), we will use the theorem to compute some exact solutions of Einstein's equations.  By comparing (\ref{eq: nomatterFRLWclassEMPdottau-a-Y}) and (\ref{eq: dottau-Ytau^s0}), we note to only consider solutions of (\ref{eq: dottau-Ytau^s0}) in Appendix D corresponding to $r_0=1$.  Also note that $\sigma(t)$ in (\ref{eq: dotsigma=1/dottau^s0}) is not relevant in the FRLW model.

\begin{example} For zero curvature $k=0$ and for $\theta=1$, we take solution 1 in Table \ref{tb: exactEMP} of the homogeneous equation $Y''(\tau)+Q(\tau)Y(\tau)=0$ with $Q(\tau)=Q_0>0$.  That is, $Y(\tau) = cos(\sqrt{Q_0}\tau)$ and by (\ref{eq: taueqnforY=cosr=1}) - (\ref{eq: dottauforY=cosr=1}) we obtain $\tau(t)=\frac{2}{\sqrt{Q_0}}Arctan\left(tanh\left(\frac{\sqrt{Q_0}}{2}(t-t_0)\right)\right)$ and 
\begin{equation}a(t)=Y(\tau(t))=sech\left(\sqrt{Q_0}(t-t_0)\right)\end{equation}
for $t_0\in\mathds{R}$.  Then by (\ref{eq: phiforY=cosr=1}) with $\alpha_0=(d-1)/\kappa$, we obtain the scalar field 
\begin{equation}\phi(t)\stackrel{def.}{=}\varphi(\tau(t))= 
2\sqrt{\frac{(d-1)}{\kappa}}Arctan\left(tanh\left(\frac{\sqrt{Q_0}}{2}(t-t_0)\right)
\right) +\beta_0\end{equation}
for $\beta_0\in\mathds{R}$.  Finally, by (\ref{eq: nomatterFRLWEMPVphi-Y}), (\ref{eq: YprimetauforY=cos}) and (\ref{eq: dottauforY=cosr=1}) we get
\begin{eqnarray}
V(\phi(t))
&=&\left[\frac{d(d-1)}{2\kappa}(Y')^2-\frac{1}{2}Y^2(\varphi ')^2-\frac{\Lambda}{\kappa}\right]\circ\tau(t)\label{eq: VfornomatterFRLWwithk=0Y=cos}\\
&=&\frac{(d-1)Q_0}{2\kappa}\left[ d \ tanh^2\left(\sqrt{Q_0}(t-t_0) \right)-sech^2\left(\sqrt{Q_0}(t-t_0)\right)\right]-\frac{\Lambda}{\kappa}
.\notag
\end{eqnarray}
so that
\begin{eqnarray}
V(w)&=&\frac{(d-1)Q_0}{2\kappa}\left[ d \ tanh^2\left(2 Arctanh\left(\tan\left(\frac{1}{2}\sqrt{\frac{\kappa}{(d-1)}}(w-\beta_0)\right)\right) \right)\right.\notag\\
&&\quad \left.-sech^2\left(2 Arctanh\left(\tan\left(\frac{1}{2}\sqrt{\frac{\kappa}{(d-1)}}(w-\beta_0)\right)\right)\right)\right]-\frac{\Lambda}{\kappa}
\end{eqnarray}
since
\begin{equation}\phi^{-1}(w)=  \frac{2}{\sqrt{Q_0}} Arctanh\left(\tan\left(\frac{1}{2}\sqrt{\frac{\kappa}{(d-1)}}(w-\beta_0)\right)\right)+t_0.\end{equation}

For $Q_0=1$ and $t_0=0$, the solver was run with $Y$ and $Y'$ both perturbed by $.05$.  The graphs of $a(t)$ below show that this solution is stable.

\begin{figure}[htbp]
\centering
\caption{Stability of FRLW Example \theexample}
\vspace{-.3in}
\includegraphics[width=4in]{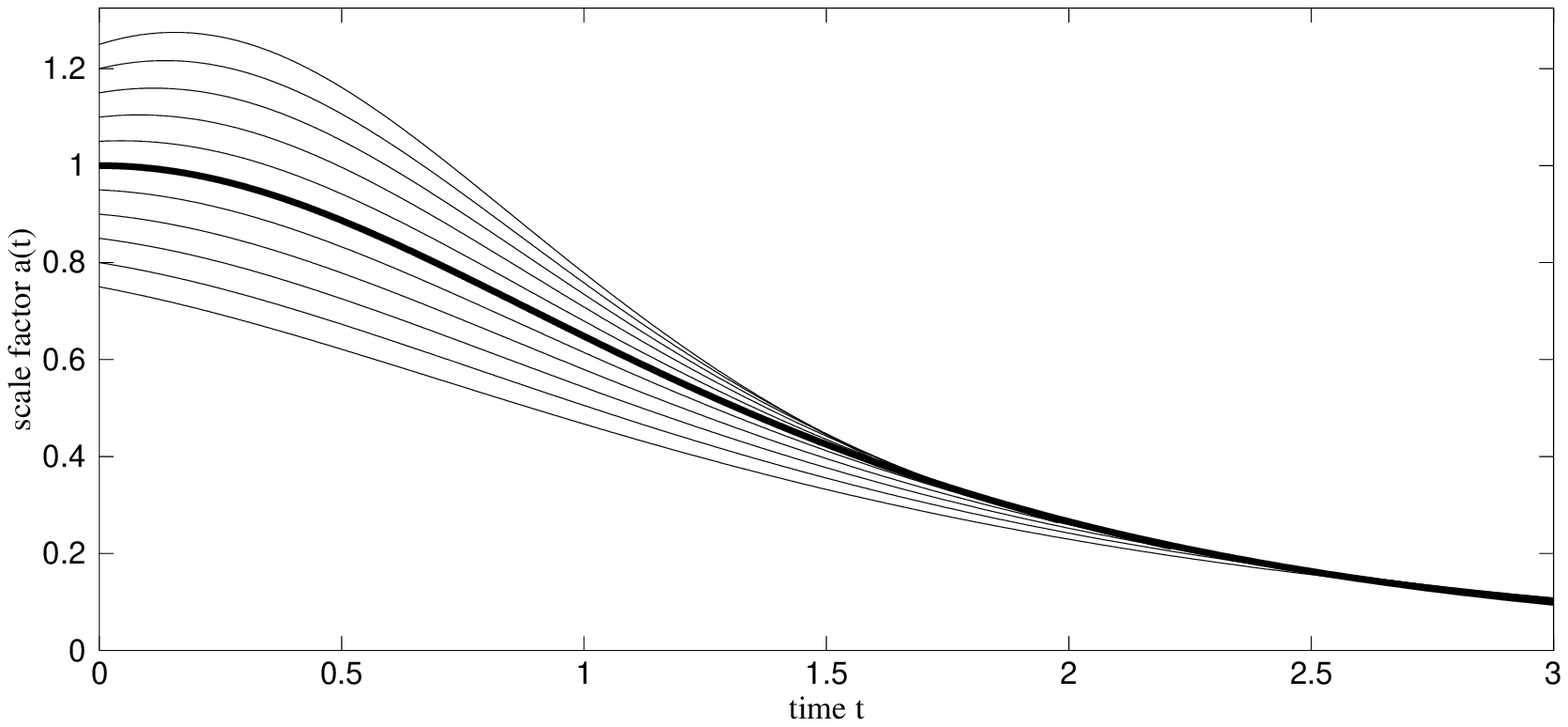}
\vspace{-.3in}
\end{figure}

\end{example}

\break
\begin{example} For zero curvature $k=0$ and for $\theta=1$, we take solution 2 in Table \ref{tb: exactEMP} of the homogeneous equation $Y''(\tau)+Q(\tau)Y(\tau)=0$ with $Q(\tau)=Q_0>0$.  That is, $Y(\tau) = sin(\sqrt{Q_0}\tau)$ and by (\ref{eq: taueqnforY=sinr=1}) - (\ref{eq: dottauforY=sinr=1}) we obtain $\tau(t)=\frac{2}{\sqrt{Q_0}}Arctan\left(e^{\sqrt{Q_0}(t-t_0)}\right)$ and 
\begin{equation}a(t)=Y(\tau(t))=sech\left(\sqrt{Q_0}(t-t_0)\right)\end{equation}
for $t_0\in\mathds{R}$.  Then by (\ref{eq: phiforY=sinr=1}) with $\alpha_0=(d-1)/\kappa$, the scalar field is
\begin{equation}\phi(t)\stackrel{def.}{=}\varphi(\tau(t))= 
2\sqrt{\frac{(d-1)}{\kappa}}Arctan\left(e^{\sqrt{Q_0}(t-t_0)}\right) + \beta_0\end{equation}
for $\beta_0\in\mathds{R}$.  Finally, by (\ref{eq: nomatterFRLWEMPVphi-Y}), (\ref{eq: YprimetauforY=sin}) and (\ref{eq: dottauforY=sinr=1}) we have
\begin{eqnarray}
V(\phi(t))
&=&\left[\frac{d(d-1)}{2\kappa}(Y')^2-\frac{1}{2}Y^2(\varphi ')^2-\frac{\Lambda}{\kappa}\right]\circ\tau(t)\label{eq: VfornomatterFRLWwithk=0Y=cos}\\
&=&\frac{(d-1)Q_0}{2\kappa}\left[ d \ tanh^2\left(\sqrt{Q_0}(t-t_0) \right)-sech^2\left(\sqrt{Q_0}(t-t_0)\right)\right]-\frac{\Lambda}{\kappa}
.\notag
\end{eqnarray}
so that
\begin{eqnarray}
V(w)&=&\frac{(d-1)Q_0}{2\kappa}\left[ d \ tanh^2\left( \ln\left(\tan\left(\frac{1}{2}\sqrt{\frac{\kappa}{(d-1)}}(w-\beta_0)\right)\right) \right)\right.\notag\\
&&\qquad \left. -sech^2\left(  \ln\left(\tan\left(\frac{1}{2}\sqrt{\frac{\kappa}{(d-1)}}(w-\beta_0)\right)\right)  \right)\right]-\frac{\Lambda}{\kappa}
\end{eqnarray}
since
\begin{equation}\phi^{-1}(w)=  \frac{1}{\sqrt{Q_0}}\ln\left(\tan\left(\frac{1}{2}\sqrt{\frac{\kappa}{(d-1)}}(w-\beta_0)\right)\right)+t_0.\end{equation}
This differs from Example 1 only in the form of the potential $V$.  

For $Q_0=1$ and $t_0=0$, the solver was run with $Y$ and $Y'$ both perturbed by $.05$.  The graphs of $a(t)$ below show that this solution is stable.
\begin{figure}[htbp]
\centering
\caption{Stability of FRLW Example \theexample}
\vspace{-.3in}
\includegraphics[width=4in]{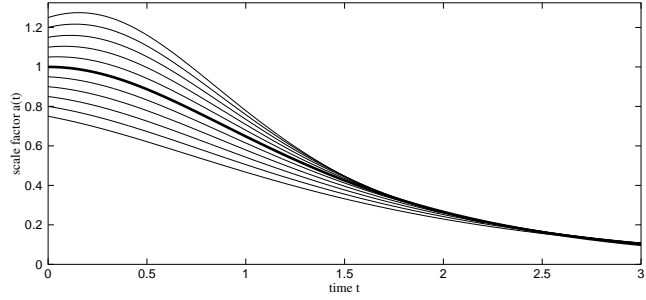}
\vspace{-.3in}
\end{figure}
\end{example}

\begin{example}
For zero curvature $k=0$ and for $\theta=1$, we take solution 4 in Table \ref{tb: exactEMP} of the homogeneous equation $Y''(\tau)+Q(\tau)Y(\tau)=0$ with $Q(\tau) = d_0(1-d_0)/\tau^2$ for an arbitrary constant $0\leq d_0 < 1$.  That is, $Y(\tau)=a_0\tau^{d_0}$ and by setting $r_0=1$ in (\ref{eq: taueqnforY=powerr=general}) - (\ref{eq: dottauforY=powerr=general}) we obtain $\tau(t)=\left((1-d_0)a_0(t-t_0)\right)^{\frac{1}{1-d_0}}$  and 
\begin{equation}a(t)=Y(\tau(t))=A_0(t-t_0)^{\frac{d_0}{1-d_0}}\label{eq: FRLWex3d0}\end{equation}
for $t>t_0\in\mathds{R}$ and $A_0\stackrel{def.}{=}\left( a_0(1-d_0)^{d_0} \right)^{1/(1-d_0)}$.  Then by (\ref{eq: phiforY=powerr=general}) with $\alpha_0=(d-1)/\kappa$,  the scalar field is
\begin{equation}\phi(t)\stackrel{def.}{=}\varphi(\tau(t)) = B \ln(t-t_0) + \beta_0\end{equation}
for $\beta_0\in\mathds{R}$ and $B\stackrel{def.}{=}\sqrt{\frac{d_0(d-1)}{\kappa (1-d_0)}}$.  Finally, by (\ref{eq: nomatterFRLWEMPVphi-Y}), (\ref{eq: YprimetauforY=powerr=general}) and (\ref{eq: dottauforY=powerr=general}), we get
\begin{eqnarray}
V(\phi(t))
&=&\left[\frac{d(d-1)}{2\kappa}(Y')^2-\frac{1}{2}Y^2(\varphi ')^2-\frac{\Lambda}{\kappa}\right]\circ\tau(t)\label{eq: VfornomatterFRLWwithk=0Y=tau^d0}\notag\\
&=&\frac{B^2}{2}\left(\frac{ d_0(d+1) - 1}{(1-d_0)} \right)\frac{1}{(t-t_0)^2}-\frac{\Lambda}{\kappa}
.\label{eq: VfornomatterFRLWk=0}
\end{eqnarray}
and
\begin{eqnarray}
V(w)&=&\frac{B^2}{2}\left(\frac{ d_0(d+1) - 1}{(1-d_0)} \right)e^{-\frac{2}{B}(w-\beta_0)}-\frac{\Lambda}{\kappa}
\end{eqnarray}
since 
\begin{equation}\phi^{-1}(w)=e^{\frac{1}{B}(w-\beta_0)}+t_0.\end{equation}
By setting $d=3$, $\Lambda=t_0=0$ and identifying $d_0/(1-d_0)$ here with $n$ in \cite{EM}, we obtain the zero curvature solution in example 4.4 of Ellis and Madsen \cite{EM}.  Also, by setting $d=2$, $\Lambda=0$, $a_0=(\sqrt{\kappa}/d_0)^{d_0}$ and identifying $(1-d_0)/d_0$ and $\kappa$ here with $\gamma_2$ and $\kappa_2$ in \cite{GarcCatCamp}, we obtain the Cruz-Martinez solution with $\epsilon_a=1$, where one must note that our integration constant $\beta_0$ corresponds to $ln(\gamma_2\sqrt{\kappa_2})/\sqrt{\kappa_2\gamma_2}+\phi_0$ in \cite{GarcCatCamp}.  Similarly with $d=3$, $\Lambda=0$, $a_0=(\sqrt{\kappa/3}/d_0)^{d_0}$ and identifying $(1-d_0)/d_0$ and $\kappa$ here with $3\gamma_3/2$ and $\kappa_3$ in \cite{GarcCatCamp}, we obtain the $(3+1)$ counterpart of the Cruz-Martinez solution with $\epsilon_a=1$, where again our integration constant $\beta_0$ corresponds to $2ln\left(\gamma_3 \sqrt{3\kappa_3}/2\right)/\sqrt{3\kappa_3\gamma_3}+\phi_0$ in \cite{GarcCatCamp}.  One can also compare this example with solutions in \cite{1, Gumjud}

For $a_0=1$ and $t_0=0$, the solver was run with $Y, Y'$ and $\tau$ perturbed by $.01$.  The graphs of $a(t)$ below show that the solution is unstable.  In both cases, the absolute error grows by two orders of magnitude over the graphed time intervals.
\begin{figure}[htbp]
\centering
\caption{Instability of FRLW Example \theexample, $d_0=1/2$}
\vspace{-.3in}
\includegraphics[width=4in]{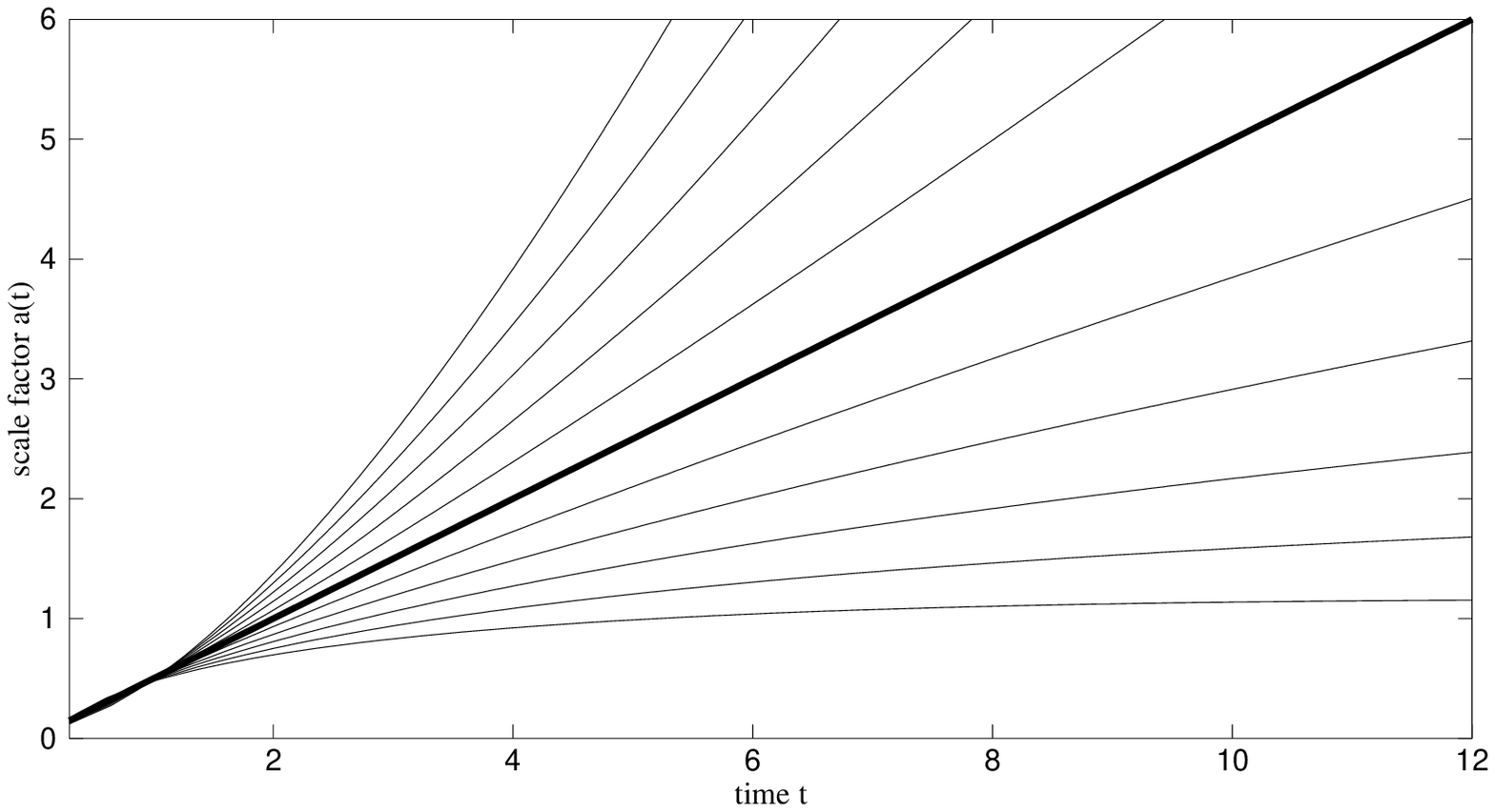}
\vspace{-.3in}

\end{figure}

\begin{figure}[htbp]
\centering
\caption{Instability of FRLW Example \theexample, $d_0=1/3$}
\includegraphics[width=4in]{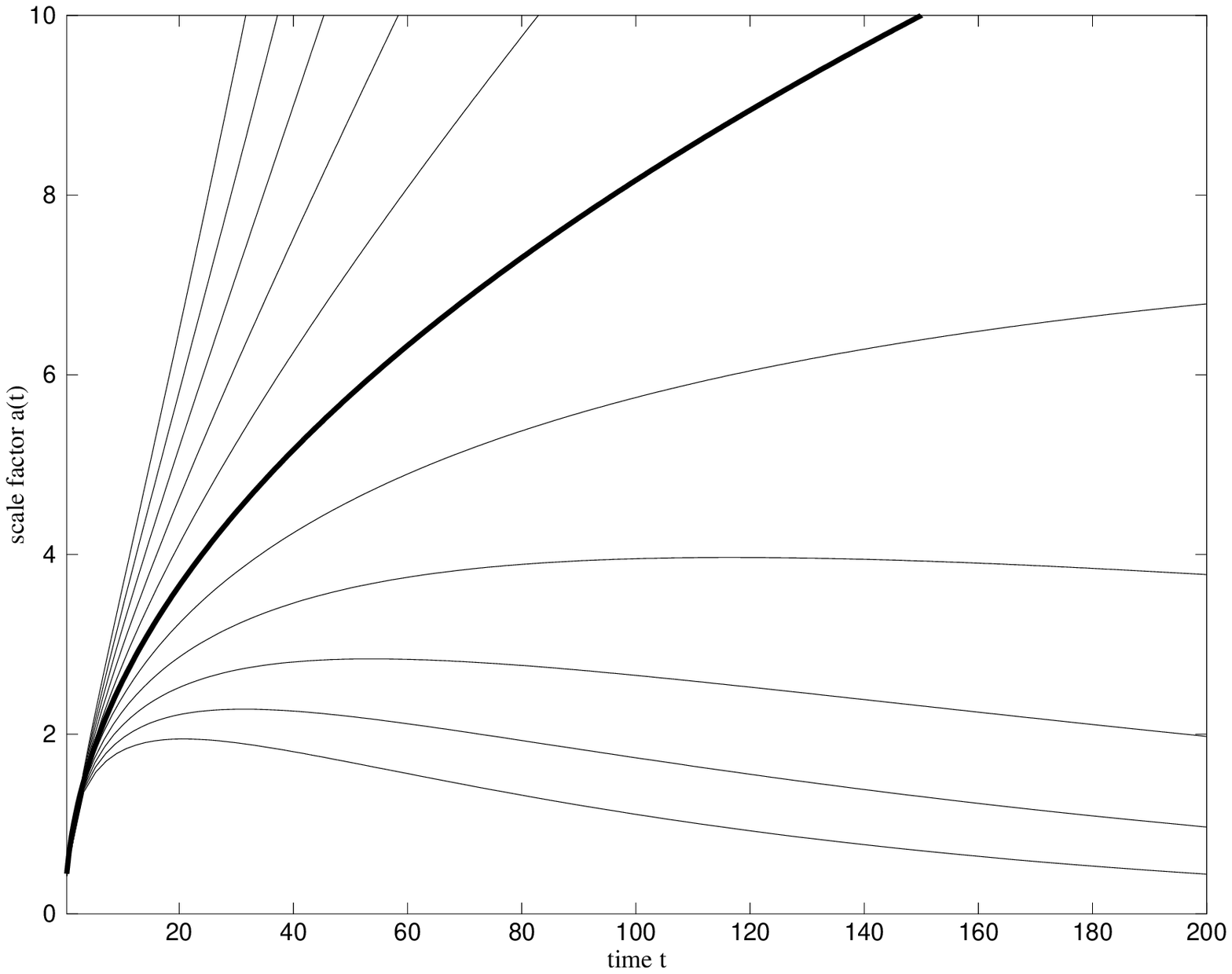}
\vspace{-.3in}
\end{figure}

\end{example}

\break
\begin{example}
For zero curvature $k=0$ and $\theta=1$, we take solution $Y(\tau)=\tau, Q(\tau)=0$ from line 4 in Table \ref{tb: exactEMP}  with $d_0=1$ for the homogeneous equation $Y''(\tau)+Q(\tau)Y(\tau)=0$.  By following   (\ref{eq:  taueqnforY=powerr0d0=1}) - (\ref{eq: dottauforY=powerr0d0=1}) we obtain $\tau(t)=a_0e^{t-t_0}$  and 
\begin{equation}a(t)=Y(\tau(t))=a_0e^{t-t_0}\end{equation}
 for $a_0,t_0\in\mathds{R}$ and $t>t_0$.  Since $Q(\tau)=0=\varphi'(\tau)$ we have constant scalar field
$\phi(t)\stackrel{def.}{=}\varphi(\tau(t)) = \phi_0\in\mathds{R}$.  Finally, by equation (\ref{eq: nomatterFRLWEMPVphi-Y}) for $V(\phi(t))$ and using that $Y'(\tau)=1$, we obtain constant potential
$
V=\frac{1}{\kappa}\left(\frac{d(d-1)}{2}-\Lambda\right).
$

For $a_0=1$ and $t_0=0$, the solver was run with $Y, Y'$ and $\tau$ perturbed by $.1$.  The graphs of $a(t)$ below show that the solution is unstable.  The absolute error grows by up to four orders of magnitude over the graphed time interval.

\begin{figure}[htbp]
\centering
\caption{Instability of FRLW Example \theexample}
\includegraphics[width=4in]{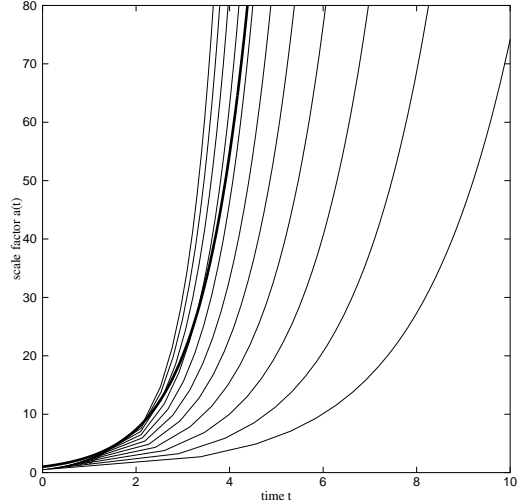}
\end{figure}
\end{example}

\begin{example}
For negative curvature $k=-1$ and for $\theta=1$, we consider the classical EMP equation $Y''(\tau)+Q(\tau)Y(\tau) = -1/Y(\tau)^3$.  For solution 5 in Table \ref{tb: exactEMP} with $b_0=d_0=0$ and $c_0=1$, we have that $Q(\tau)=0$ 
and $Y(\tau)=(a_0+2\tau)^{1/2}$ for some $a_0\in\mathds{R}$.  Following (\ref{eq: taueqnforY=superpowerb=0r=1}) - (\ref{eq: dottauforY=superpowerb=0r=1}) we obtain $\tau(t)=\frac{1}{2}\left((t-t_0)^2-a_0\right)$ and 
\begin{equation}a(t)=Y(\tau(t))=t-t_0\end{equation}
for $t_0\in\mathds{R}$.  Since $Q(\tau)=0=\varphi'(\tau)$ we obtain constant scalar field $\phi(t)\stackrel{def.}{=}\varphi(\tau(t)) = \phi_0\in\mathds{R}$.  Finally, by (\ref{eq: nomatterFRLWEMPVphi-Y}), (\ref{eq: YprimetauforY=superpowerb=d=0r=1}) and (\ref{eq: dottauforY=superpowerb=0r=1}) we obtain constant potential $V(\phi(t))=-\Lambda/\kappa$.

For $a_0=t_0=0$, the solver was run with $Y, Y'$ and $\tau$ perturbed by $.1$.  The graphs of $a(t)$ below show that the solution is unstable.  The absolute error grows by two orders of magnitude over the graphed time interval.

\begin{figure}[htbp]
\centering
\vspace{-.1in}
\caption{Instability of FRLW Example \theexample}
\vspace{-.1in}
\includegraphics[width=4in]{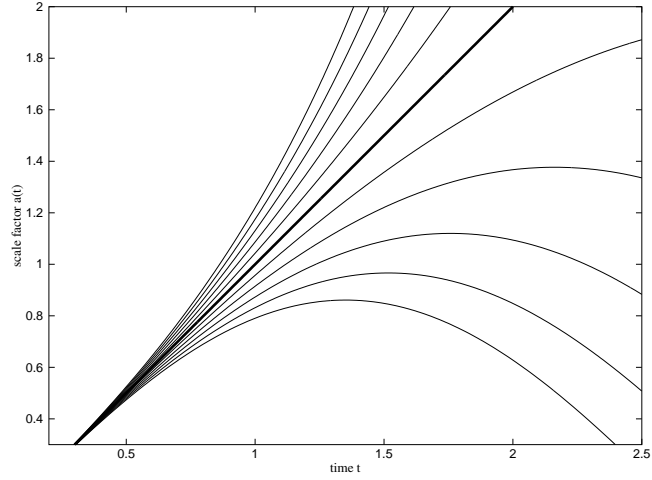}
\end{figure}
\end{example}

\break
\begin{example}
For $\theta=1$ and arbitrary curvature $k=a_0b_0-c_0^2$ for some $b_0>0$ and $a_0,c_0\in\mathds{R}$, we consider the classical EMP equation $Y''(\tau)+Q(\tau)Y(\tau) = (a_0b_0-c_0^2)/Y(\tau)^3$.  For solution 5 in Table \ref{tb: exactEMP} with $d_0=0$, we have that $Q(\tau)=0$ and $Y(\tau)=(a_0+b_0\tau^2+2c_0\tau)^{1/2}$.  Following (\ref{eq: taueqnforY=superpowerr=1}) - (\ref{eq: dottauforY=superpowerr=1}) with $k=\lambda$ we obtain $\tau(t)=\frac{1}{4b_0^{3/2}}\left(b_0e^{\sqrt{b_0}(t-t_0)}-4k e^{-\sqrt{b_0}(t-t_0)}-4\sqrt{b_0}c_0\right)$  and 
\begin{equation}a(t)=Y(\tau(t))=\frac{1}{4}e^{\sqrt{b_0}(t-t_0)}+\frac{k}{b_0}e^{-\sqrt{b_0}(t-t_0)}\end{equation}
for $t_0\in\mathds{R}$.  Since $Q(\tau)=0=\varphi'(\tau)$ we obtain constant scalar field
\begin{equation}\phi(t)\stackrel{def.}{=}\varphi(\tau(t)) = \phi_0\end{equation}
for constant $\phi_0\in\mathds{R}$.  Finally, by (\ref{eq: nomatterFRLWEMPVphi-Y}), (\ref{eq: YprimetauforY=superpowerd=0b>0r=1}) and (\ref{eq: dottauforY=superpowerr=1}) we obtain 
\begin{eqnarray}
V(\phi(t))
&=&\frac{d(d-1)}{2\kappa}\left(
\frac{
b_0\left(b_0e^{\sqrt{b_0}(t-t_0)}-4k e^{-\sqrt{b_0}(t-t_0)}\right)^2 + 16b_0^2k
}{
\left(b_0e^{\sqrt{b_0}(t-t_0)}+4k e^{-\sqrt{b_0}(t-t_0)}\right)^2
}
\right)-\frac{\Lambda}{\kappa}\notag\\
&=&  \frac{1}{\kappa}\left(\frac{d(d-1)b_0}{2}-\Lambda\right).
\end{eqnarray}

For $a_0=b_0=1$ and $t_0=c_0=0$, the solver was run with $Y, Y'$ and $\tau$ perturbed by $.1$.  The graphs of $a(t)$ show that the solution is unstable.  In both cases the absolute error grows by up to two orders of magnitude over the graphed time intervals.
\vspace{1in}
\begin{figure}[htbp]
\centering
\caption{Instability of FRLW Example \theexample, $k=1$}
\includegraphics[width=4in]{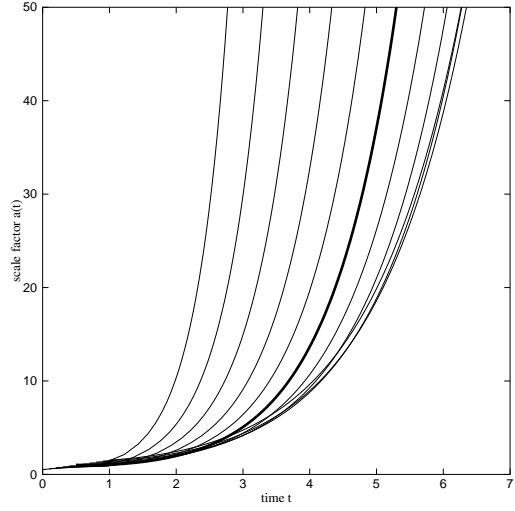}
\end{figure}
\vspace{-1.4in}
\begin{figure}[htbp]
\centering
\caption{Instability of FRLW Example \theexample, $k=-1$}
\includegraphics[width=4in]{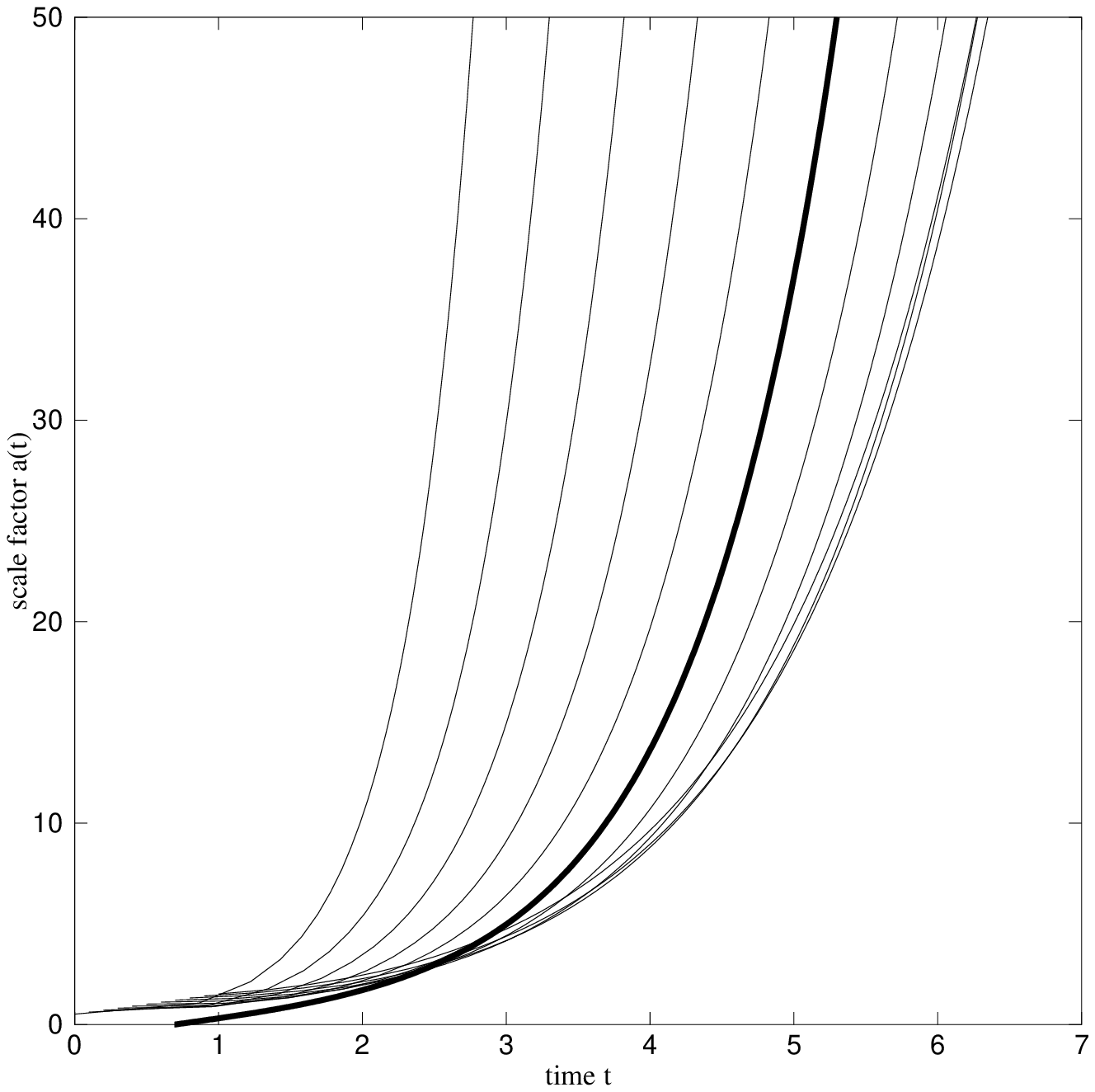}
\end{figure}
\vspace{-.3in}
\end{example}

\break
\begin{example}
For curvature $k$ and $\theta>0$, we consider the EMP equation $Y''(\tau)+Q(\tau)Y(\tau)=k/\theta^2Y(\tau)^3$.  For solution 5 in Table \ref{tb: exactEMP} with $a_0=b_0=0$ and $c_0=1/2$, we have that $Q(\tau)=d_0(1-d_0)/\tau^2$ for the choice $d_0=(1/2)-(\sqrt{-k}/\theta)$, and $Y(\tau)=\sqrt{\tau}$.  Following (\ref{eq:    taueqnforY=sqrttau})-(\ref{eq: tauforY=sqrttau}) we obtain $\tau(t)=\frac{\theta^2}{4}(t-t_0)^2$ and 
\begin{equation}a(t)=Y(\tau(t))=\frac{\theta}{2}(t-t_0)\end{equation}
for $t>t_0\in\mathds{R}$.  Then by (\ref{eq:  phiforY=sqrttau}) with $\alpha_0=(d-1)/\kappa$ we obtain scalar field
\begin{equation}\phi(t)=\varphi(\tau(t))=\sqrt{ \frac{(d-1)(\theta^2+4k) }{\kappa\theta^2}  }\ln(t-t_0) + \beta_0.\end{equation}
Also by (\ref{eq:  nomatterFRLWEMPVphi-Y}) we obtain
\begin{eqnarray}
V(\phi(t))
&=&\left[\frac{d(d-1)}{2\kappa}\left(\theta^2(Y')^2+\frac{k}{Y^{2}}\right)-\frac{\theta^2}{2}Y^2(\varphi ')^2-\frac{\Lambda}{\kappa}\right]\circ\tau(t)\notag\\
&=&\frac{(d-1)^2}{2\kappa\theta^2}\frac{\left( \theta^2+4k \right)}{(t-t_0)^2}-\frac{\Lambda}{\kappa}
\end{eqnarray}
so that
\begin{equation}V(w)=\frac{(d-1)^2}{2\kappa\theta^2}\left( \theta^2+4k \right) e^{-2\sqrt{\frac{\kappa\theta^2}{(d-1)(\theta^2+4k)}}(w-\beta_0)} -\frac{\Lambda}{\kappa}\end{equation}
since 
\begin{equation}\phi^{-1}(w)=e^{\sqrt{\frac{\kappa\theta^2}{(d-1)(\theta^2+4k)}}(w-\beta_0)}+t_0.\end{equation}
Note that although the number $d_0$ may be complex, the above solution is real for each $k\in\{-1,0,1\}$ by a proper choice of $\theta>0$.
By taking $d=3, \Lambda=t_0=0$ and also identifying $\theta/2$ here with $A$ in \cite{EM},  we obtain the solutions in example 4.5 of Ellis and Madsen \cite{EM}.  One can also compare with solutions in \cite{1}.    

For $t_0=0$, the solver was run with $Y, Y'$ and $\tau$ perturbed by $.01$.  The graphs of $a(t)$ below show that the solution is unstable.  In all three cases below the absolute error grows by up to two orders of magnitude over the graphed time interval.
\vspace{.6in}
\begin{figure}[htbp]
\centering
\vspace{.2in}
\caption{Instability of FRLW Example \theexample, $k=0, \theta=1$}
\includegraphics[width=4in]{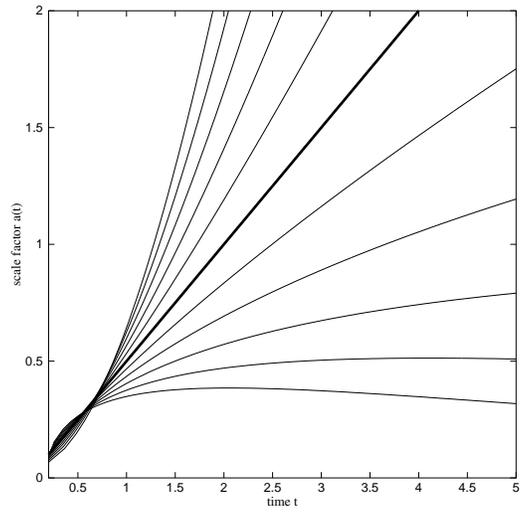}
\end{figure}
\begin{figure}[htbp]
\centering
\vspace{-.5in}
\caption{Instability of FRLW Example \theexample, $k=1, \theta=1$}
\includegraphics[width=4in]{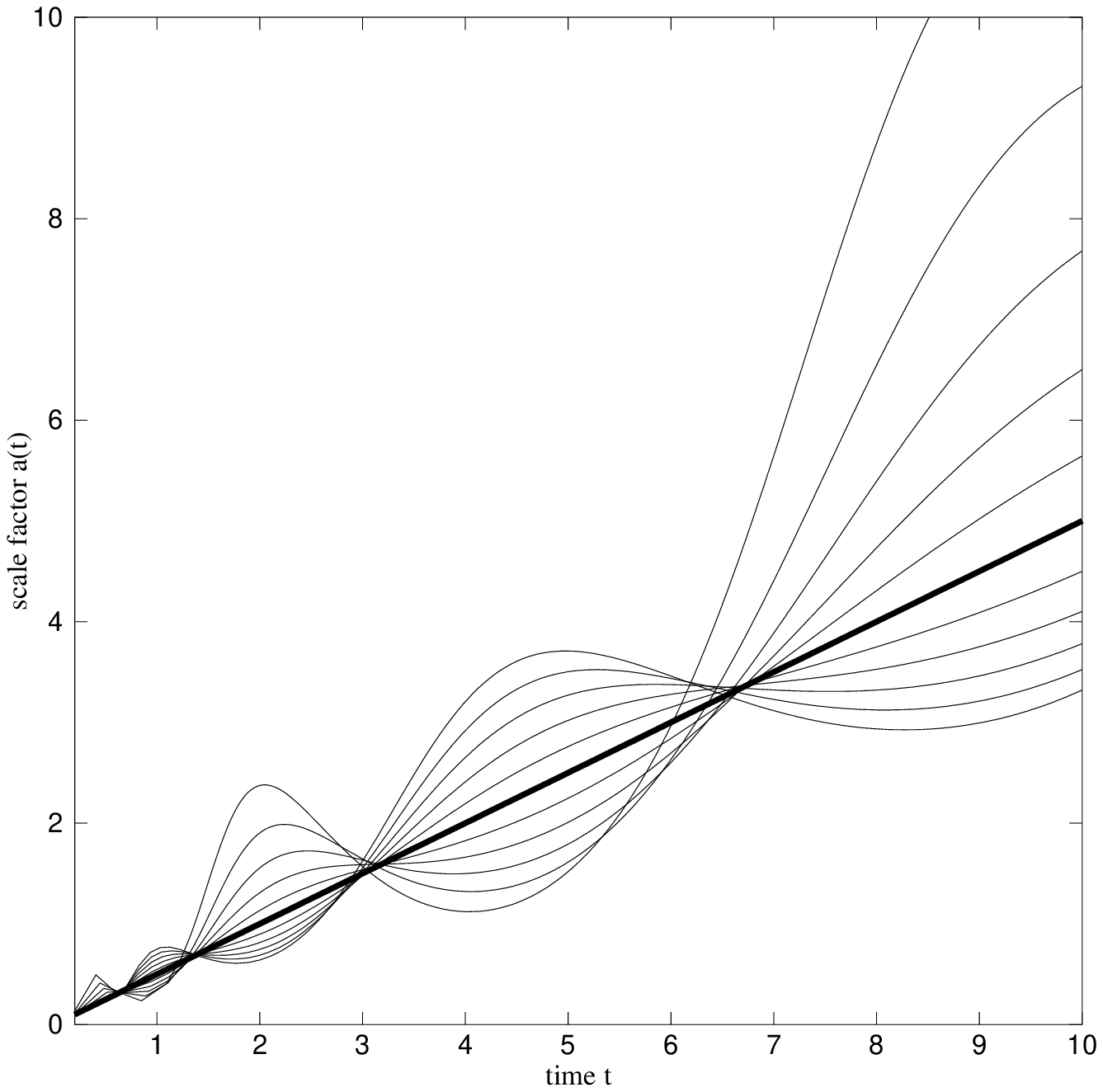}
\end{figure}
\begin{figure}[htbp]
\centering
\vspace{-.4in}
\caption{Instability of FRLW Example \theexample, $k=-1, \theta=4$}
\includegraphics[width=4in]{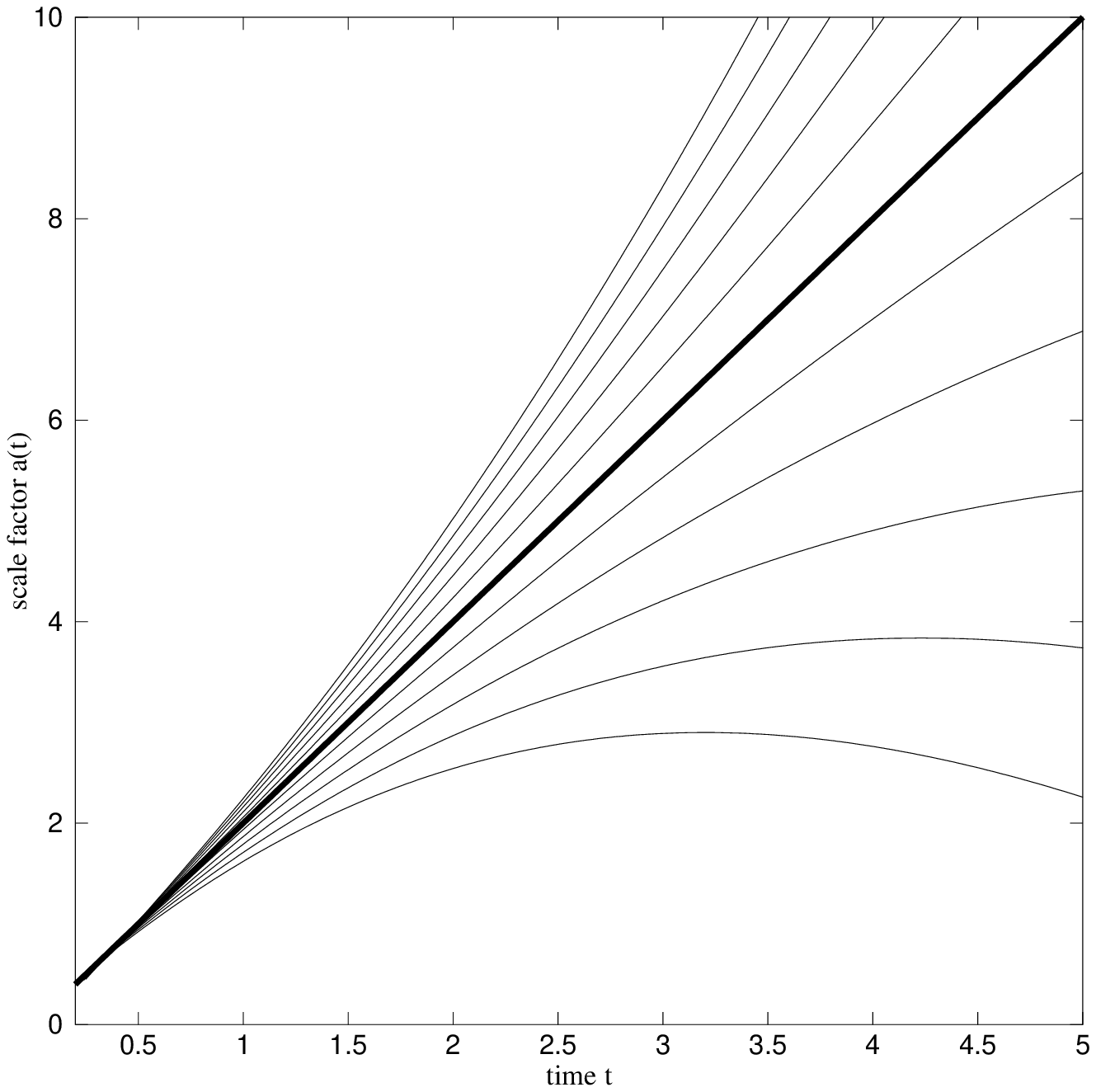}
\end{figure}
\end{example}

\break
\begin{example}
For $k\in \{0,1\}$ and $\theta>0$, we consider the EMP equation $Y''(\tau)+Q(\tau)Y(\tau) = k/\theta^2 Y(\tau)^{3}$.  For solution 6 in Table \ref{tb: exactEMP} with $\lambda_1=k/\theta^2$ and $B_1=3$, we have that $Q(\tau)=k/\theta^2\tau^{4}$ and $Y(\tau)=\tau$.  Following (\ref{eq: taueqnforY=tauwithB}) - (\ref{eq: tauforY=tauwithB}) we obtain $\tau(t)=a_0e^{\theta(t-t_0)}$  and 
\begin{equation}a(t)=Y(\tau(t))=a_0  e^{\theta (t-t_0)}\end{equation}
for $a_0>0$ and $t>t_0$.  Then by (\ref{eq: phiforY=tauwithB}) with $\alpha_0=(d-1)/\kappa$, the scalar field is
 \begin{equation}\phi(t)\stackrel{def.}{=}\varphi(\tau(t)) = - \frac{1}{a_0}\sqrt{\frac{(d-1)k}{\theta^2\kappa }}e^{-\theta(t-t_0)}+\beta_0.\end{equation}
  Finally, by (\ref{eq:  nomatterFRLWEMPVphi-Y})  we obtain 
\begin{eqnarray}
V(\phi(t))
&=&\left[\frac{d(d-1)}{2\kappa}\left(\theta^2(Y')^2+\frac{k}{Y^{2}}\right)-\frac{\theta^2}{2}Y^2(\varphi ')^2-\frac{\Lambda}{\kappa}\right]\circ\tau(t)\notag\\
&=&\frac{d(d-1)}{2\kappa}\theta^2+\frac{(d-1)^2k}{2\kappa a_0^2}e^{-2\theta(t-t_0)}-\frac{\Lambda}{\kappa}\notag
\end{eqnarray}
so that
\begin{equation}V(w)=\frac{(d-1)\theta^2}{2}\left(\frac{d}{\kappa }+(w-\beta_0)^2\right)-\frac{\Lambda}{\kappa}\end{equation}
by composition with $\phi^{-1}$.  By taking $d=3$, $\Lambda=t_0=0$ and identifying $\theta$ with $\omega$, $a_0$ with $A$ and $\beta_0$ with $\phi_0$ in \cite{EM}, this is example 4.1 of Ellis and Madsen \cite{EM}.  One can also compare this example with the (non-phantom) exponential expansion solution in \cite{Gumjudexact}, and also other solutions in \cite{1}.  

For $a_0=\theta=k=1$ and $t_0=0$, the solver was run with $Y, Y'$ and $\tau$ perturbed by $.1$.  The graphs of $a(t)$ below show that the solution is unstable.  The absolute error grows by two orders of magnitude over the graphed time interval.
\begin{figure}[htbp]
\centering
\caption{Instability of FRLW Example \theexample}
\includegraphics[width=4in]{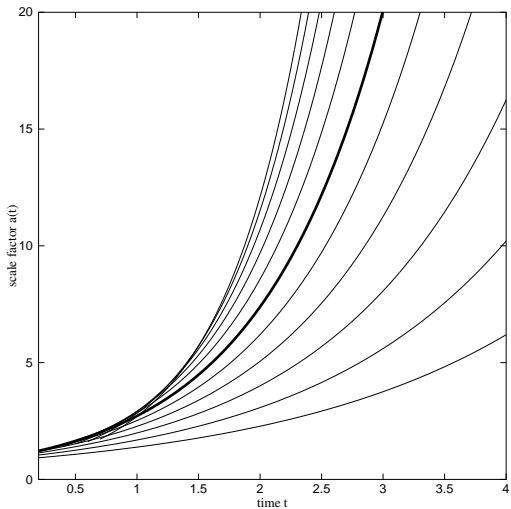}
\end{figure}
\end{example}

\break
\begin{example}
For $\theta=1$ and positive curvature $k=1$, we consider the classical EMP equation $Y''(\tau)+Q(\tau)Y(\tau)=1/Y(\tau)^3$.  For solution 7 in Table \ref{tb: exactEMP} we have $Y(\tau)=(a_0^2\tau^2+b_0^2)^{1/2}$  with $\lambda_1=1$ and we take $Q(\tau)=(1-a_0^2b_0^2)/\left(a_0^2\tau^2+b_0^2\right)^2$ for $a_0,b_0>0$ and $(1-a_0^2b_0^2)>0$.  Following (\ref{eq: taueqnforY=sqrttau^2+b^2})-(\ref{eq: dottauforY=sqrttau^2+b^2}) we obtain $\tau(t)=\frac{b_0}{a_0}\sinh(a_0 (t-t_0))$ and 
\begin{equation}a(t)=Y(\tau(t))=b_0 \cosh(a_0 (t-t_0))\end{equation}
for $t_0\in\mathds{R}$.  Then by (\ref{eq: phiforY=sqrttau^2+b^2r0=1}) with $\alpha_0=(d-1)/\kappa$, we  have scalar field 
\begin{equation}\phi(t)=\frac{B_d}{a_0}\sqrt{\frac{2}{(d-1)}}Arctan\left(\sinh(a_0 (t-t_0))\right)+\beta_0\label{eq: phiforex8EM43}\end{equation}
for $B_d\stackrel{def.}{=}\sqrt{\frac{(d-1)^2(1-a_0^2b_0^2)}{2\kappa b_0^2}}$.  By (\ref{eq: FRLWEMPVphi-Y}) we obtain
\begin{eqnarray}
V(\phi(t))
&=&\left[\frac{d(d-1)}{2\kappa}\left((Y')^2+\frac{1}{Y^{2}}\right)-\frac{1}{2}Y^2(\varphi ')^2-\frac{\Lambda}{\kappa}\right]\circ\tau(t)\notag\\
&=&\frac{d(d-1)a_0^2}{2\kappa}+B_d^2sech^2(a_0 (t-t_0))  -\frac{\Lambda}{\kappa}
\end{eqnarray}
so that 
\begin{equation}V(w)=\frac{d(d-1)a_0^2}{2\kappa}+B_d^2 \cos^2\left(\frac{a_0}{B_d}\sqrt{\frac{(d-1)}{2}}(w-\beta_0)\right)-\frac{\Lambda}{\kappa}\end{equation}
by composition with $\phi^{-1}$.  By taking $d=3, \Lambda=t_0=0$ and identifying $a_0$ with $\omega$, $b_0$ with $A$ and $B_d=B_3$ with $B$ in \cite{EM}, this is comparable to example 4.3 of Ellis and Madsen \cite{EM}.  One can verify via differentiation that $\phi(t)$ in (\ref{eq: phiforex8EM43}) agrees up to a constant with $\phi(t)$ in example 4.3 of \cite{EM}.  

For $a_0=\theta=1$ and $t_0=0$, the solver was run with $Y, Y'$ and $\tau$ perturbed by $.01$.  The graphs of $a(t)$ below show that the solution is unstable.  The absolute error grows by up to two orders of magnitude over the graphed time interval.

\begin{figure}[htbp]
\centering
\vspace{-.1in}
\caption{Instability of FRLW Example \theexample}
\includegraphics[width=4in]{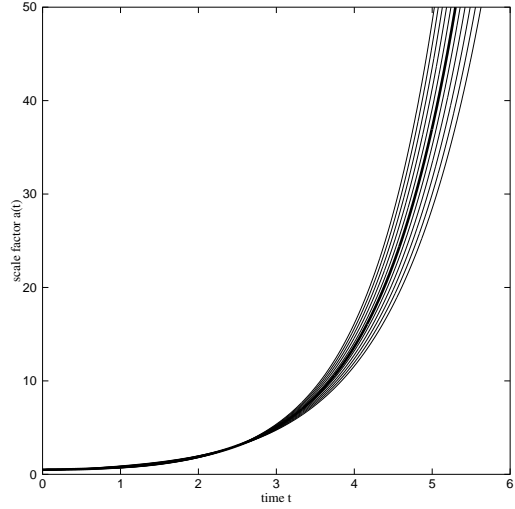}
\end{figure}

\end{example}

\begin{example}
For $\theta=1$ and arbitrary curvature $k$, we consider the classical EMP equation $Y''(\tau)+Q(\tau)Y(\tau)=k/Y(\tau)^3$.  For solution 7  in Table \ref{tb: exactEMP}, we have $Y(\tau)=(a_0^2\tau^2-b_0^2)^{1/2}$ with $\lambda_1=k$ and we take $Q(\tau)=(k+a_0^2b_0^2)/\left(a_0^2\tau^2-b_0^2\right)^2$ for $a_0,b_0>0$ such that $(k+a_0^2b_0^2)>0$.  Following (\ref{eq: taueqnforY=sqrttau^2-b^2})-(\ref{eq: dottauforY=sqrttau^2-b^2}) we obtain $\tau(t)=\frac{b_0}{a_0}\cosh(a_0 (t-t_0))$ and 
\begin{equation}a(t)=Y(\tau(t))=b_0 \sinh(a_0 (t-t_0))\end{equation}
for $t_0\in\mathds{R}$.  Then by (\ref{eq: phiforY=sqrttau^2-b^2r0=1}) with $\alpha_0=(d-1)/\kappa$, the scalar field is
\begin{equation}\phi(t)=- \frac{B_d}{a_0}\sqrt{\frac{2}{(d-1)}}Arctanh\left(\cosh(a_0 (t-t_0))\right)+\beta_0\label{eq: phiforex8EM42}\end{equation}
for $B_d\stackrel{def.}{=}\sqrt{\frac{(d-1)^2(k+a_0^2b_0^2)}{2\kappa b_0^2}}$.  By (\ref{eq: FRLWEMPVphi-Y}) we obtain
\begin{eqnarray}
V(\phi(t))
&=&\left[\frac{d(d-1)}{2\kappa}\left((Y')^2+\frac{k}{Y^{2}}\right)-\frac{1}{2}Y^2(\varphi ')^2-\frac{\Lambda}{\kappa}\right]\circ\tau(t)\notag\\
&=&\frac{d(d-1)a_0^2}{2\kappa }  + B_d^2  csch^2(a_0 (t-t_0)) -\frac{\Lambda}{\kappa}
\end{eqnarray}
 so that 
\begin{equation}V(w)=\frac{d(d-1)a_0^2}{2\kappa}+B_d^2 \sinh^2\left(\frac{a_0}{B_d}\sqrt{\frac{(d-1)}{2}}(w-\beta_0)\right)-\frac{\Lambda}{\kappa}\end{equation}
by composition with $\phi^{-1}$.  By taking $d=3, \Lambda=t_0=0$ and identifying $a_0$ with $\omega$, $b_0$ with $A$ and $B_d=B_3$ with $B$ in \cite{EM}, this is comparable to example 4.2 of Ellis and Madsen \cite{EM}.  One can verify via differentiation that $\phi(t)$ in (\ref{eq: phiforex8EM43}) agrees up to a constant with $\phi(t)$ in example 4.2 of \cite{EM}.  

\end{example}

\subsection{Second reduction to classical EMP: zero curvature}
\label{sec: FRLWcurvzero}
For a second set of examples, we take special case with curvature $k=0$ and $\rho'=p'=D_i=0$ for all $i\neq 1$.  For non-zero matter density $D_1/a(t)^{n_1}\equiv D/a(t)^n$ with $D,n\neq 0$ constants, and we choose parameter $q=n/2$.  Then Theorem \ref{thm: FRLWEFE-EMP} shows that solving Einstein's equations
\begin{equation}\frac{d}{2}H^2(t)\stackrel{(i)''}{=}\frac{\kappa}{(d-1)} \left[\frac{1}{2}\dot\phi(t)^2+V(\phi(t))+\frac{D}{a(t)^{n}}\right]+\frac{\Lambda}{(d-1)}\end{equation}
\begin{eqnarray*}
\dot{H}(t)+\frac{d}{2}H(t)^2&\stackrel{(ii)''}{=}&-\frac{\kappa}{(d-1)} \left[\frac{1}{2}\dot\phi(t)^2-V(\phi(t))+\frac{(n-d)D}{da(t)^{n}}\right]+\frac{\Lambda}{(d-1)} \end{eqnarray*} 
is equivalent to solving the classical EMP equation
\begin{equation}Y''(\tau)+Q(\tau)Y(\tau)=\frac{ - n^2 \kappa D}{2\theta^2 d(d-1) Y(\tau)^3}\label{eq: zerocurvFRLWEMP}\end{equation}
for any constant $\theta>0$.  The solutions of $(i)'',(ii)''$ and (\ref{eq: zerocurvFRLWEMP}) are related by
\begin{equation}a(t)=Y(\tau(t))^{2/n}\qquad\mbox{ and }\qquad \varphi '(\tau)^2= \frac{2(d-1)}{n\kappa} Q(\tau)\label{eq: zerocurvFRLWclassEMPa-Yvarphi-Q}\end{equation}
for $\phi(t)=\varphi(\tau(t))$ and 
\begin{equation}\dot\tau(t)=\theta a(t)^{n/2}=\theta Y(\tau(t)).\label{eq: zerocurvFRLWclassEMPdottau-a-Y}\end{equation}
Also in the converse direction, $V$ is taken to be
\begin{equation}
V(\phi(t))=\left[\frac{2\theta^2 d(d-1)}{\kappa n^2}(Y')^2-\frac{\theta^2}{2}Y^2(\varphi ')^2-\frac{D}{Y^2}-\frac{\Lambda}{\kappa}\right]\circ\tau(t).\label{eq: zerocurvFRLWEMPVphi-Y}
\end{equation}

We now refer to Appendix D for solutions of the classical and corresponding  homogeneous EMP equation (\ref{eq: zerocurvFRLWEMP}) that we will map over to solutions of Einstein's equations.  By comparing (\ref{eq: zerocurvFRLWclassEMPdottau-a-Y}) and (\ref{eq: dottau-Ytau^s0}), we note to only consider solutions of (\ref{eq: dottau-Ytau^s0}) in Appendix D corresponding to $r_0=1$.  Also note that $\sigma(t)$ in (\ref{eq: dotsigma=1/dottau^s0}) is not relevant in the FRLW model.

\begin{example} For zero matter density ($D=0$) and for $\theta=1$, we take solution 1 in Table \ref{tb: exactEMP} of the homogeneous equation $Y''(\tau)+Q(\tau)Y(\tau)=0$ with $Q(\tau)=Q_0>0$.  That is, $Y(\tau) = cos(\sqrt{Q_0}\tau)$ and by (\ref{eq: taueqnforY=cosr=1}) - (\ref{eq: dottauforY=cosr=1}) we obtain $\tau(t)=\frac{2}{\sqrt{Q_0}}Arctan\left(tanh\left(\frac{\sqrt{Q_0}}{2}(t-t_0)\right)\right)$ and 
\begin{equation}a(t)=Y(\tau(t))^{2/n}=sech^{2/n}\left(\sqrt{Q_0}(t-t_0)\right)\end{equation}
for $t_0\in\mathds{R}$.  Then by (\ref{eq: phiforY=cosr=1}) with $\alpha_0=2(d-1)/n\kappa$, we have scalar field
\begin{equation}\phi(t)\stackrel{def.}{=}\varphi(\tau(t))= 
2^{3/2}\sqrt{\frac{(d-1)}{n\kappa}}Arctan\left(tanh\left(\frac{\sqrt{Q_0}}{2}(t-t_0)\right)
\right) +\beta_0\end{equation}
for $\beta_0\in\mathds{R}$.  Finally, by (\ref{eq: zerocurvFRLWEMPVphi-Y}), (\ref{eq: YprimetauforY=cos}) and (\ref{eq: dottauforY=cosr=1}) we obtain
\begin{eqnarray}
V(\phi(t))
&=&\left[\frac{2d(d-1)}{n^2 \kappa}(Y')^2 -\frac{1}{2}Y^2(\varphi ')^2-\frac{\Lambda}{\kappa}\right]\circ\tau(t)\\
&=&\frac{(d-1)Q_0}{n\kappa}\left(    \frac{2d}{n } tanh^2\left(\sqrt{Q_0}(t-t_0) \right) - sech^2\left(\sqrt{Q_0}(t-t_0)\right)\right)  -\frac{\Lambda}{\kappa}.\notag
\end{eqnarray}
so that
\begin{eqnarray}
V(w)
&=&\frac{(d-1)Q_0}{n\kappa}\left(    \frac{2d}{n } tanh^2\left(2 Arctanh\left(\tan\left(\frac{1}{2^{3/2}}\sqrt{\frac{n\kappa}{(d-1)}}(w-\beta_0)\right)\right) \right)\right. \notag\\
&&\quad \left. - sech^2\left(2 Arctanh\left(\tan\left(\frac{1}{2^{3/2}}\sqrt{\frac{n\kappa}{(d-1)}}(w-\beta_0)\right)\right)\right)\right)  -\frac{\Lambda}{\kappa}.
\end{eqnarray}
since
\begin{equation}\phi^{-1}(w)=  \frac{2}{\sqrt{Q_0}}Arctanh\left(\tan\left(\frac{1}{2^{3/2}}\sqrt{\frac{n\kappa}{(d-1)}}(w-\beta_0)\right)\right)+t_0.\end{equation}
This solution is the same as Example 1 when $n=2$.  One reasonably expects the convergence properties to be the same as Example 1, since $Y(\tau)$ and the EMP equation do not depend on $n$. \end{example}

\begin{example} For zero matter density ($D=0$) and for $\theta=1$, we take solution 2 in Table \ref{tb: exactEMP} of the homogeneous equation $Y''(\tau)+Q(\tau)Y(\tau)=0$ with $Q(\tau)=Q_0>0$.  That is, $Y(\tau) = sin(\sqrt{Q_0}\tau)$ and by (\ref{eq: taueqnforY=sinr=1}) - (\ref{eq: dottauforY=sinr=1}) we obtain $\tau(t)=\frac{2}{\sqrt{Q_0}}Arctan\left(e^{\sqrt{Q_0}(t-t_0)}\right)$ and 
\begin{equation}a(t)=Y(\tau(t))^{2/n}=sech^{2/n}\left(\sqrt{Q_0}(t-t_0)\right)\end{equation}
for $t_0\in\mathds{R}$.  Then by (\ref{eq: phiforY=sinr=1}) with $\alpha_0=2(d-1)/n\kappa$, the scalar field is
\begin{equation}\phi(t)\stackrel{def.}{=}\varphi(\tau(t))= 
2^{3/2}\sqrt{\frac{(d-1)}{n\kappa}}Arctan\left(e^{\sqrt{Q_0}(t-t_0)}\right) + \beta_0\end{equation}
for $\beta_0\in\mathds{R}$.  Finally, by (\ref{eq: zerocurvFRLWEMPVphi-Y}), (\ref{eq: YprimetauforY=sin}) and (\ref{eq: dottauforY=sinr=1}) we have
\begin{eqnarray}
V(\phi(t))
&=&\left[\frac{2d(d-1)}{n^2 \kappa}(Y')^2 -\frac{1}{2}Y^2(\varphi ')^2-\frac{\Lambda}{\kappa}\right]\circ\tau(t)\\
&=&\frac{(d-1)Q_0}{n\kappa}\left(    \frac{2d}{n } tanh^2\left(\sqrt{Q_0}(t-t_0) \right) - sech^2\left(\sqrt{Q_0}(t-t_0)\right)\right)  -\frac{\Lambda}{\kappa}
.\notag
\end{eqnarray}
so that
\begin{eqnarray}
V(w)&=&\frac{(d-1)Q_0}{n\kappa}\left[    \frac{2d}{n } tanh^2\left(\ln\left(\tan\left(\frac{1}{2^{2/3}}\sqrt{\frac{n\kappa}{(d-1)}}(w-\beta_0)\right)\right) \right)\right. \notag\\
&&\quad \left. - sech^2\left(\ln\left(\tan\left(\frac{1}{2^{2/3}}\sqrt{\frac{n\kappa}{(d-1)}}(w-\beta_0)\right)\right)\right)\right]  -\frac{\Lambda}{\kappa}
\end{eqnarray}
since
\begin{equation}\phi^{-1}(w)=  \frac{1}{\sqrt{Q_0}}\ln\left(\tan\left(\frac{1}{2^{2/3}}\sqrt{\frac{n\kappa}{(d-1)}}(w-\beta_0)\right)\right)+t_0.\end{equation}
This differs from Example 7 only in the form of the potential $V$.  This solution is the same as Example 2 when $n=2$.  One reasonably expects the convergence properties to be the same as Example 2, since $Y(\tau)$ and the EMP equation do not depend on $n$.
\end{example}

\begin{example}
For $D=2d(d-1)/n^2\kappa$ and $\theta=1$, we consider the classical EMP equation $Y''(\tau)+Q(\tau)Y(\tau) = -1/ Y(\tau)^3$.  For solution 5 in Table \ref{tb: exactEMP} with $b_0=d_0=0$ and $c_0=1$, we have that $Q(\tau)=0$ 
and $Y(\tau)=(a_0+2\tau)^{1/2}$ for some $a_0\in\mathds{R}$.  Following (\ref{eq: taueqnforY=superpowerb=0r=1}) - (\ref{eq: dottauforY=superpowerb=0r=1}) we obtain $\tau(t)=\frac{1}{2}\left((t-t_0)^2-a_0\right)$ and 
\begin{equation}a(t)=Y(\tau(t))^{2/n}=(t-t_0)^{2/n}\end{equation}
for $t_0\in\mathds{R}$ and $t>t_0$.  Since $Q(\tau)=0=\varphi'(\tau)$, the scalar field is constant
\begin{equation}\phi(t)\stackrel{def.}{=}\varphi(\tau(t)) = \phi_0\end{equation}
for $\phi_0\in\mathds{R}$.  Finally, by (\ref{eq: zerocurvFRLWEMPVphi-Y}), (\ref{eq: YprimetauforY=superpowerb=d=0r=1}) and (\ref{eq: dottauforY=superpowerb=0r=1}) we obtain constant potential
\begin{equation}
V(\phi(t))=-\frac{\Lambda}{\kappa}.
\end{equation}
This solution is the same as Example 5 when $n=2$.  One reasonably expects the convergence properties to be the same as Example 5, since $Y(\tau)$ and the EMP equation do not depend on $n$.
\end{example}

\begin{example}
For $\theta=1$ and $D=2d(d-1)/n^2\kappa$ we consider the classical EMP equation $Y''(\tau)+Q(\tau)Y(\tau) = -1/Y(\tau)^3$.  For solution 5 in Table \ref{tb: exactEMP} with $d_0=0$ and $a_0=(c_0^2-1)/b_0$ for some $b_0>0$ and $c_0\in\mathds{R}$, we have that $Q(\tau)=0$ and $Y(\tau)=(a_0+b_0\tau^2+2c_0\tau)^{1/2}$.  Following (\ref{eq: taueqnforY=superpowerr=1}) - (\ref{eq: dottauforY=superpowerr=1}) with $\lambda=-1$ we obtain $\tau(t)=\frac{1}{4b_0^{3/2}}\left(b_0e^{\sqrt{b_0}(t-t_0)}+4 e^{-\sqrt{b_0}(t-t_0)}-4\sqrt{b_0}c_0\right)$  and 
\begin{equation}a(t)=Y(\tau(t))^{2/n}=\left(\frac{1}{4}e^{\sqrt{b_0}(t-t_0)}-\frac{1}{b_0}e^{-\sqrt{b_0}(t-t_0)}\right)^{2/n}\end{equation}
for $t_0\in\mathds{R}$.  Since $Q(\tau)=0=\varphi'(\tau)$, the scalar field is any constant
\begin{equation}\phi(t)\stackrel{def.}{=}\varphi(\tau(t)) = \phi_0\in\mathds{R}.\end{equation}
 Finally, by (\ref{eq: zerocurvFRLWEMPVphi-Y}), (\ref{eq: YprimetauforY=superpowerd=0b>0r=1}) and (\ref{eq: dottauforY=superpowerr=1}) we obtain constant potential
\begin{eqnarray}
V(\phi(t))
&=&\frac{2 d(d-1)}{\kappa n^2}
\left(
\frac{
b_0\left(b_0e^{\sqrt{b_0}(t-t_0)}+4 e^{-\sqrt{b_0}(t-t_0)}\right)^2 - 16b_0^2
}{
\left(b_0e^{\sqrt{b_0}(t-t_0)} - 4 e^{-\sqrt{b_0}(t-t_0)}\right)^2
}
\right)
-\frac{\Lambda}{\kappa}\notag\\
&=&  \frac{1}{\kappa}\left(\frac{2d(d-1)b_0}{n^2}-\Lambda\right).
\end{eqnarray}
\end{example}
This solution is the same as Example 6 when $n=2$.  One reasonably expects the convergence properties to be the same as Example 6, since $Y(\tau)$ and the EMP equation do not depend on $n$.

\begin{example}
For constant $C\stackrel{def.}{=}\frac{  n^2\kappa D}{2\theta^2 d(d-1)} > 0$, we consider the EMP equation $Y''(\tau)+Q(\tau)Y(\tau)= - C/\theta^2Y(\tau)^3$.  For solution 5 in Table \ref{tb: exactEMP} with $a_0=b_0=0$ and $c_0=1/2$, we have that $Q(\tau)=d_0(1-d_0)/\tau^2$ for the choice $d_0=(1/2)-(\sqrt{C}/\theta)$, and $Y(\tau)=\sqrt{\tau}$.  Following (\ref{eq:    taueqnforY=sqrttau})-(\ref{eq: tauforY=sqrttau}) we obtain $\tau(t)=\frac{\theta^2}{4}(t-t_0)^2$ and 
\begin{equation}a(t)=Y(\tau(t))^{2/n}=\left(\frac{\theta}{2}(t-t_0)\right)^{2/n}\end{equation}
for $t>t_0\in\mathds{R}$.  Then by (\ref{eq:  phiforY=sqrttau}) with $\alpha_0=2(d-1)/n\kappa$, we get scalar field
\begin{equation}\phi(t)=\varphi(\tau(t))=\sqrt{ \frac{2(d-1)(\theta^2-4C) }{n\kappa\theta^2}  }\ln(t-t_0) + \beta_0.\end{equation}
Also by (\ref{eq:  nomatterFRLWEMPVphi-Y}) we obtain
\begin{eqnarray}
V(\phi(t))
&=&\frac{(d-1)(\theta^2-4C)(2d-n)}{\kappa n^2\theta^2 (t-t_0)^2}-\frac{\Lambda}{\kappa}
\end{eqnarray}
so that 
\begin{equation}V(w)=\frac{(d-1)(\theta^2-4C)(2d-n)}{\kappa n^2\theta^2 }e^{-2\sqrt{\frac{n\kappa\theta^2}{2(d-1)(\theta^2-4C)}}(w-\beta_0)}-\frac{\Lambda}{\kappa}\end{equation}
since 
\begin{equation}\phi^{-1}(w)=e^{\sqrt{\frac{n\kappa\theta^2}{2(d-1)(\theta^2-4C)}}(w-\beta_0)}+t_0.\end{equation}

For $C=1, \theta=4, n=2$ and $t_0=0$, the solver was run with $Y, Y'$ and $\tau$ perturbed by $.001$.  The graphs of $a(t)$ below show that the solution is unstable.  The absolute error grows three orders of magnitude over the graphed time interval.  Since $Y(\tau)$ does not depend on $n$, other choices of $n$ will also be unstable.

\begin{figure}[htbp]
\centering
\vspace{-.1in}
\caption{Instability of FRLW Example \theexample}
\includegraphics[width=4in]{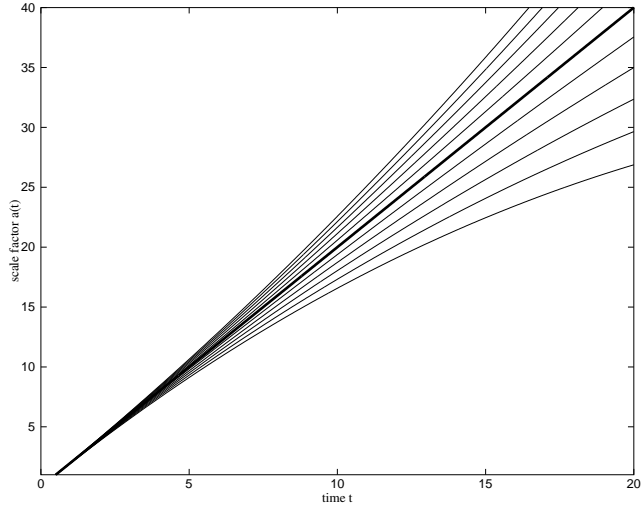}
\end{figure}
\end{example}

\begin{example}
For $\theta=1$ and constant $C\stackrel{def.}{=}\frac{  n^2\kappa D}{2d(d-1)} > 0$, we consider the EMP equation $Y''(\tau)+Q(\tau)Y(\tau)= - C/Y(\tau)^3$.    For solution 7  in Table \ref{tb: exactEMP} we have  $Y(\tau)=(a_0^2\tau^2-b_0^2)^{1/2}$ with $\lambda_1= - C$ and we let $Q(\tau)=(a_0^2b_0^2-C)/\left(a_0^2\tau^2-b_0^2\right)^2$ for $a_0,b_0>0$ chosen such that $(a_0^2b_0^2-C)>0$.  Following (\ref{eq: taueqnforY=sqrttau^2-b^2})-(\ref{eq: dottauforY=sqrttau^2-b^2}) we obtain $\tau(t)=\frac{b_0}{a_0}\cosh(a_0 (t-t_0))$ and 
\begin{equation}a(t)=Y(\tau(t))^{2/n}=\left(b_0 \sinh(a_0 (t-t_0))\right)^{2/n}\end{equation}
for $t_0\in\mathds{R}$.  Then by (\ref{eq: phiforY=sqrttau^2-b^2r0=1}) with $\alpha_0=2(d-1)/n\kappa$, the scalar field is
\begin{equation}\phi(t)=- \frac{B_d}{a_0}\sqrt{\frac{2}{(d-1)}}Arctanh\left(\cosh(a_0 (t-t_0))\right)+\beta_0\label{eq: phiforex8EM42}\end{equation}
for $B_d\stackrel{def.}{=}\sqrt{\frac{(d-1)^2(a_0^2b_0^2-C)}{n\kappa b_0^2}}$.  By (\ref{eq: FRLWEMPVphi-Y}) we obtain
\begin{eqnarray}
V(\phi(t))
&=&\left[\frac{2 d(d-1)}{\kappa n^2}(Y')^2-\frac{1}{2}Y^2(\varphi ')^2-\frac{D}{Y^2}-\frac{\Lambda}{\kappa}\right]\circ\tau(t)\notag\\
&=&\frac{2d(d-1)a_0^2}{n^2\kappa}+ \frac{B_d^2(2d-n)}{n(d-1)} csch^2(a_0(t-t_0)) -\frac{\Lambda}{\kappa}\notag\\
\end{eqnarray}
 so that 
\begin{equation}V(w)=\frac{2d(d-1)a_0^2}{n^2\kappa}+ \frac{B_d^2(2d-n)}{n(d-1)} \cosh^2\left(\frac{a_0}{B_d}\sqrt{\frac{(d-1)}{2}}(w-\beta_0)\right)-\frac{\Lambda}{\kappa}\end{equation}
by composition with $\phi^{-1}$.  

For $C=b_0=1, a_0=2, n=3$ and $t_0=0$, the solver was run with $Y, Y'$ and $\tau$ perturbed by $.01$.  The graphs of $a(t)$ below show that the solution is unstable.  The absolute error grows 16 orders of magnitude over the graphed time interval.  Since $Y(\tau)$ does not depend on $n$, other choices of $n$ will also be unstable.
\vspace{.3in}
\begin{figure}[htbp]
\centering
\vspace{-.1in}
\caption{Instability of FRLW Example \theexample}
\vspace{-.1in}
\includegraphics[width=4in]{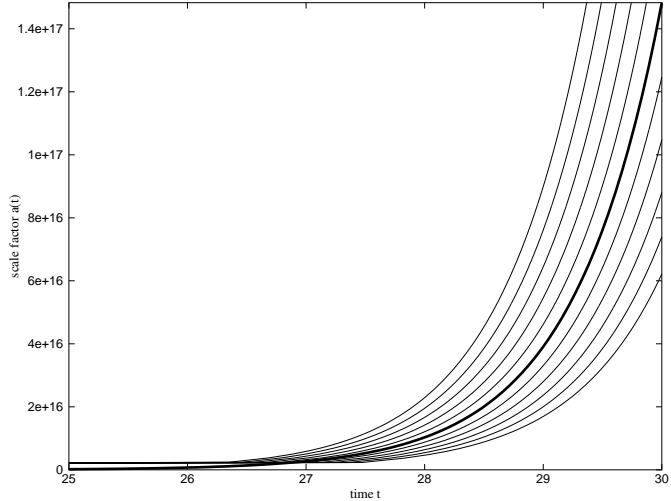}
\end{figure}

\end{example}

\break
\section{In terms of a Schr\"odinger-Type Equation}

To reformulate the Einstein field equations (i),(ii) in (\ref{eq:  FRLWEFEiii}) in terms of a Schr\"odinger-type equation (with one less non-linear term than that which is provided by the generalized EMP formulation), one can apply Corollary \ref{cor: EFE-NLSAnonzeroEzero} to the difference (ii)-(i).  In doing so, Corollary \ref{cor: EFE-NLSAnonzeroEzero} may be applied in a few different ways:  either by taking $\frac{G(t)}{a(t)^A}=\frac{k}{a(t)^2}$ in which case this term transforms into the linear term $Eu(\sigma)=\theta^2 k u(\sigma)$ in the corresponding Schr\"odiner-type equation, or by taking $\frac{G(t)}{a(t)^A}=-\frac{\kappa n_j D_j(t)}{d(d-1)a(t)^{n_j}}$ where $j$ is some index with $n_j\neq 0$ in which case this nonlinear term transforms into the linear term $E(\sigma)u(\sigma)=-\frac{\theta^2 \kappa  n_j^2}{2d(d-1)}\mathrm{D}_j(\sigma)u(\sigma)$ in the corresponding Schr\"odinger-type equation.  We will state both applications.

\begin{thm} $\left(\mbox{Apply Corollary \ref{cor: EFE-NLSAnonzeroEzero} with }\frac{G(t)}{a(t)^A}=\frac{k}{a(t)^2}\right)$\\
\label{thm: FRLWEFE-NLS}
Suppose you are given a twice differentiable function $a(t)>0$, a once differentiable function $\phi(t)$, and also functions $D_i(t), \rho'(t), p'(t), V(x)$ which satisfy the Einstein equations $(i),(ii)$ in (\ref{eq:  FRLWEFEiii}) for some $k, n_1, \dots, n_{M}, \Lambda \in\mathds{R}, d\in\mathds{R}\backslash\{0,1\}, \kappa\in\mathds{R}\backslash\{0\}$ and $M\in\mathds{N}$.  Let $g(\sigma)$ denote the inverse of a function $\sigma(t)$ which satisfies
\begin{equation}\dot{\sigma}(t)=\frac{1}{\theta a(t)}\label{eq: FRLWNLSdotsigma-a}\end{equation}
for some $\theta>0$.  Then the following functions 
\begin{eqnarray}u(\sigma)&=&\frac{1}{a(g(\sigma))}\label{eq: FRLWNLSu-a}\\
P(\sigma)&=&\frac{\kappa}{(d-1)}\psi '(\sigma)^2\label{eq: FRLWNLSP-psi}\end{eqnarray}
solve the Schr\"odinger-type equation
\begin{equation}u''(\sigma)+[\theta^2 k -P(\sigma)]u(\sigma)=\displaystyle\sum_{i=1}^{M}\frac{\theta^2  \kappa n_i \mathrm{D}_i(\sigma)}{d(d-1) u(\sigma)^{1-n_i}}+\frac{\theta^2 \kappa(\uprho(\sigma)+\mathrm{p}(\sigma))}{(d-1) u(\sigma)}
\label{eq: FRLWNLS}\end{equation}
for 
\begin{equation}\psi(\sigma)=\phi(g(\sigma))\label{eq: FRLWNLSpsi-phi}\end{equation}
\begin{equation}\mathrm{D}_i(\sigma)=D_i(g(\sigma)), \ \  1\leq i\leq M \label{eq: FRLWNLSrmD-D}\end{equation}
and
\begin{equation}\uprho(\sigma)=\rho'(g(\sigma)), \ \mathrm{p}(\sigma)=p'(g(\sigma)). \label{eq: FRLWNLSuprho-rhormp-p}\end{equation}

Conversely, suppose you are given a twice differentiable function $u(\sigma)>0$, and also functions $P(\sigma), \mathrm{D}_i(\sigma)$ for $1\leq i\leq M, M\in\mathds{N}$ and $\uprho(\sigma), \mathrm{p}(\sigma)$ which solve (\ref{eq: FRLWNLS}) for some constants $\theta>0, k\in\mathds{R}, \kappa\in\mathds{R}\backslash\{0\}, d\in\mathds{R}\backslash\{0,1\}$ and $n_i\in\mathds{R}$ for $1\leq i \leq M$.  In order to construct functions which solve $(i),(ii)$, first find $\sigma(t), \psi(\sigma)$ which solve the differential equations
\begin{equation}\dot{\sigma}(t)=\frac{1}{\theta} u(\sigma(t))\qquad\mbox{ and }\qquad \psi '(\sigma)^2= \frac{(d-1)}{\kappa} P(\sigma).\label{eq: FRLWNLSdotsigma-upsi-P}\end{equation}
Then the functions
\begin{eqnarray}a(t)&=&\frac{1}{u(\sigma(t))}\label{eq: FRLWNLSa-u}\\
\phi(t)&=&\psi(\sigma(t))\label{eq: FRLWNLSphi-psi}\end{eqnarray}
\begin{equation}D_i(t)=\mathrm{D}_i(\sigma(t)), \ \  1\leq i \leq M\label{eq: FRLWNLSD-rmD}\end{equation}
\begin{equation}\rho'(t)=\uprho(\sigma(t)), \ p'(t)=\mathrm{p}(\sigma(t)),\label{eq: FRLWNLSrho-uprhop-rmp}\end{equation}
and
\begin{equation}
V(\phi(t))=\left[\frac{d(d-1)}{2\kappa}\left(\frac{1}{\theta^2}(u')^2+ ku^2 \right)-\frac{u^2(\psi ')^2}{2\theta^2}-\displaystyle\sum_{i=1}^{M}\mathrm{D}_i u^{n_i} -\uprho-\frac{\Lambda}{\kappa}\right]\circ\sigma(t)\label{eq: FRLWNLSVphi-u}
\end{equation}
satisfy the equations $(i),(ii)$.\end{thm}

\begin{proof}
This proof will implement Corollary \ref{cor: EFE-NLSAnonzeroEzero} with constants and functions as indicated in the following table.

\begin{table}[ht]
\centering
\caption{{ Corollary \ref{cor: EFE-NLSAnonzeroEzero} applied to FRLW}}\label{tb: FRLWNLS}
\vspace{.2in}
\begin{tabular}{r | l c r | l}
In Corollary & substitute & & In Corollary & substitute \\[4pt]
\hline
\raisebox{-5pt}{$\varepsilon $} & \raisebox{-5pt}{${\kappa}/{(d-1)}$}      &&    \raisebox{-5pt}{$E(\sigma)$} & \raisebox{-5pt}{constant $\theta^2 k$}\\[8pt]
                                       $G(t)$ & $\mbox{constant } k$    &&     $A$ & $2$ \\[8pt]
$G_i(t), 1\leq i\leq M$ &$ -\frac{\kappa n_i }{d(d-1)}D_i(t)  $     &&    $A_i$ &$n_i$ \\[8pt]
 $G_{M+1}(t)$ & $ \frac{-\kappa}{(d-1)}(\rho'(t)+p'(t))$  &&      $A_{M+1}$   &$0$\\[8pt]
   $F_i(\sigma), 1\leq i\leq M$& $\frac{\theta^2 \kappa n_i }{d(d-1)}\mathrm{D}_i(\sigma)  $&&  $C_i$&$1-n_i$\\[8pt]
$F_{M+1}(\sigma)$&$ \frac{\theta^2 \kappa}{(d-1)} (\uprho(\sigma)+\mathrm{p}(\sigma)) $ &&        $C_{M+1}$&  $1$ \\[6pt]
\hline
\end{tabular}
\end{table}

To prove the forward implication, we assume to be given functions which solve the Einstein field equations $(i)$ and $(ii)$.  Subtracting equations $(ii)-(i)$, 
\begin{equation}\dot{H}(t)-\frac{k}{a(t)^2}=-\frac{\kappa}{(d-1)}\left[\dot{\phi}(t)^2+\displaystyle\sum_{i=1}^{M}\frac{n_i D_i(t)}{da(t)^{n_i}}+(\rho'(t)+p'(t))
\right].\label{eq: FRLWNLSiiminusi}\end{equation}
This shows that $a(t), \phi(t), D_i(t), \rho'(t)$ and $p'(t)$ satisfy the hypothesis of Corollary \ref{cor: EFE-NLSAnonzeroEzero}, applied with  constants $\varepsilon , N, A, A_1 \dots, A_{N}$ and functions $G(t), G_1(t) \dots,G_{N}(t)$ according to Table \ref{tb: FRLWNLS}.  Since $\sigma(t), u(\sigma), P(\sigma)$ and $\psi(\sigma)$ defined in (\ref{eq: FRLWNLSdotsigma-a}), (\ref{eq: FRLWNLSu-a}), (\ref{eq: FRLWNLSP-psi}) and (\ref{eq: FRLWNLSpsi-phi})  are equivalent to that in the forward implication of Corollary \ref{cor: EFE-NLSAnonzeroEzero}, by this corollary and by definitions (\ref{eq: FRLWNLSrmD-D}) and (\ref{eq: FRLWNLSuprho-rhormp-p}) of $\mathrm{D}_i(\sigma)$ and $\uprho(\sigma), \mathrm{p}(\sigma)$, the Schr\"odinger-type equation (\ref{eq: CNLSANONZERO}) holds for constants $C_1, \dots, C_{N}$ and functions $F_1(\sigma), \dots, F_{N}(\sigma)$ as indicated in Table \ref{tb: FRLWNLS}.  This proves the forward implication.

To prove the converse implication, assume we are given functions which solve the Schr\"odinger-type equation (\ref{eq: FRLWNLS}) and we begin by showing that $(i)$ is satisfied.   Differentiating the definition (\ref{eq:  FRLWNLSa-u}) of $a(t)$ and by the definition in (\ref{eq: FRLWNLSdotsigma-upsi-P}) of $\sigma(t)$, we see that
\begin{eqnarray}\dot{a}(t)&=&-\frac{u'(\sigma(t))}{u(\sigma(t))^2}\dot\sigma(t)\notag\\
&=&-\frac{u'(\sigma(t))}{\theta u(\sigma(t))}.\end{eqnarray}
Dividing by $a(t)$, we obtain
\begin{equation}H(t)\stackrel{def.}{=}\frac{\dot{a}(t)}{a(t)}= - \frac{1}{\theta}u'(\sigma(t)).\label{eq: FRLWNLSH-u}\end{equation}
Differentiating the definition (\ref{eq: FRLWNLSphi-psi}) of $\phi(t)$ and using definition in (\ref{eq: FRLWNLSdotsigma-upsi-P}) of  $\sigma(t)$, we get that
\begin{equation}\dot{\phi}(t)=\psi '(\sigma(t))\dot{\sigma}(t)=\frac{1}{\theta}\psi '(\sigma(t))u(\sigma(t)).\label{eq: FRLWNLSdotphi-u}\end{equation}
Using (\ref{eq: FRLWNLSH-u}) and (\ref{eq: FRLWNLSdotphi-u}), and also the definitions (\ref{eq: FRLWNLSa-u}), (\ref{eq: FRLWNLSD-rmD}) and (\ref{eq: FRLWNLSrho-uprhop-rmp}) of $a(t), D_i(t)$ and $\rho'(t), p'(t)$ respectively, the definition (\ref{eq: FRLWNLSVphi-u}) of $V\circ\phi$ can be written
\begin{equation}
V(\phi(t))=\frac{d(d-1)}{2\kappa}\left(
H(t)^2+\frac{k}{a(t)^2}\right)-\frac{1}{2}\dot\phi(t)^2-\displaystyle\sum_{i=1}^{M}\frac{D_i(t)}{a(t)^{n_i}}-\rho'(t)-\frac{\Lambda}{\kappa}
\label{eq: FRLWNLSVphi-a}.\end{equation}
This shows that $(i)$ holds  (That is, the definition of $V(\phi(t))$ was designed to be such that $(i)$ holds).

To conclude the proof we must also show that $(ii)$ holds.   In the converse direction the hypothesis of the converse of Corollary \ref{cor: EFE-NLSAnonzeroEzero} holds, applied with constants $N, C_1, \dots, C_{N}$ and functions $E(\sigma), F_1(\sigma), \dots, F_{N}(\sigma)$ as indicated in Table \ref{tb: FRLWNLS}.  Since  $\sigma(t), \psi(\sigma), a(t)$ and $\phi(t)$ defined in (\ref{eq: FRLWNLSdotsigma-upsi-P}), (\ref{eq: FRLWNLSa-u}) and (\ref{eq: FRLWNLSphi-psi}) are consistent with the converse implication of Corollary \ref{cor: EFE-NLSAnonzeroEzero}, applied with $\varepsilon $ and $A$ as in Table \ref{tb: FRLWNLS}, by this corollary and by definitions (\ref{eq: FRLWNLSD-rmD}) and (\ref{eq: FRLWNLSrho-uprhop-rmp}) of $D_i(t)$ and $\rho'(t), p'(t)$ the  scale factor  equation (\ref{eq: CEFEANONZERO}) holds for constants $\varepsilon , A, A_1, \dots, A_{N}$ and functions $G(t), G_1(t),\dots,G_{N}(t)$ according to Table \ref{tb: FRLWNLS}.  That is, we have regained (\ref{eq: FRLWNLSiiminusi}) which shows that the subtraction of equations (ii)-(i) holds in the converse direction.  Now solving (\ref{eq: FRLWNLSVphi-a}) for $\rho'(t)$ and substituting this into (\ref{eq: FRLWNLSiiminusi}), we obtain (ii).  This proves the theorem.
\end{proof}

\subsection{Reduction to linear Schr\"odinger: pure scalar field}
To compute some exact solutions, we take special case $\rho'=p'=D_i=0$ so that Theorem \ref{thm: FRLWEFE-NLS} shows that solving the Einstein equations 
\begin{equation}\frac{d}{2}H^2(t)+\frac{dk}{2a(t)^2}\stackrel{(i)'''}{=}\frac{\kappa}{(d-1)} \left[\frac{1}{2}\dot\phi(t)^2+V(\phi(t))\right]+\frac{\Lambda}{(d-1)}\end{equation}
\begin{eqnarray}&&\\
\dot{H}(t)+\frac{d}{2}H(t)^2+\frac{(d-2)k}{2a(t)^2}&\stackrel{(ii)'''}{=}&-\frac{\kappa}{(d-1)} \left[\frac{1}{2}\dot\phi(t)^2-V(\phi(t))\right]+\frac{\Lambda}{(d-1)}\qquad\qquad\notag \end{eqnarray}
is equivalent to solving the linear Schr\"odinger equation
\begin{equation}u''(\sigma)+[\theta^2 k -P(\sigma)]u(\sigma)=0
\label{eq: nomatterFRLWNLS}\end{equation}
for any constant $\theta>0$.  The solutions of $(i)''',(ii)'''$ and (\ref{eq: nomatterFRLWNLS}) are related by 
\begin{equation}a(t)=\frac{1}{u(\sigma(t))}\quad\mbox{ and }\quad
\psi '(\sigma)^2=\frac{(d-1)}{\kappa}P(\sigma)\end{equation}
for $\phi(t)=\psi(\sigma(t))$ and 
\begin{equation}\dot{\sigma}(t)=\frac{1}{\theta a(t)}=\frac{1}{\theta} u(\sigma(t)).\label{eq: nomatterFRLWLSdotsigmau}\end{equation}
Also in the converse direction, $V$ is taken to be
\begin{equation}
V(\phi(t))=\left[\frac{d(d-1)}{2\kappa}\left(\frac{1}{\theta^2}(u')^2+ ku^2 \right)-\frac{u^2(\psi ')^2}{2\theta^2}-\frac{\Lambda}{\kappa}\right]\circ\sigma(t).\label{eq: nomatterFRLWNLSVphi-u}
\end{equation}

We now refer to Appendix E for solutions of the linear Schr\"odinger equation (\ref{eq: nomatterFRLWNLS}), which we will map to solutions of Einstein's equations using the theorem. 

\begin{example}
For zero curvature $k=0$ and $\theta=1$, we take solution 1 in Table \ref{tb: exactNLS} with $a_0=d_0=0$ so that we have $u(\sigma)=b_0\sigma+c_0$ and $P(\sigma)=0$ for $b_0>0$ and $c_0\in\mathds{R}$.  By (\ref{eq: sigmaeqnforu=linearr=1}) - (\ref{eq: dotsigmaanduforu=linearr=1}) we obtain $\sigma(t)=e^{b_0(t-t_0)}-\frac{c_0}{b_0}$ and
\begin{equation}a(t)=\frac{1}{u(\sigma(t))}=\frac{1}{b_0}e^{-b_0(t-t_0)}\end{equation}
for $t_0\in\mathds{R}$.  Since $P=0=\psi'(\sigma)$, the scalar field is constant
\begin{equation}\psi(\sigma)=\psi_0\in\mathds{R}.\end{equation}
Finally, by (\ref{eq: nomatterFRLWNLSVphi-u}) and (\ref{eq: uprimesigmaforu=linearr=1}) we obtain constant potential
\begin{eqnarray}
V(\phi(t))&=&\left[\frac{d(d-1)}{2\kappa}(u')^2-\frac{\Lambda}{\kappa}\right]\circ\sigma(t)\notag\\
&=&\frac{d(d-1)}{2\kappa}b_0^2-\frac{\Lambda}{\kappa}.
\label{eq: Vforu=linearr=1}
\end{eqnarray}

\end{example}
\begin{example}
For zero curvature $k=0$ and $\theta=1$, we take solution 1 in Table \ref{tb: exactNLS} with $a_0>0$ and $d_0=0$ so that we have  $u(\sigma)=a_0\sigma^2+b_0\sigma+c_0$ and $P(\sigma)=2a_0/(a_0\sigma^2+b_0\sigma+c_0)$.  By (\ref{eq:  sigmaeqnforu=quadraticr=1}) - (\ref{eq:  dotsigmaandusigmaforu=quadraticr=1}) we obtain $\sigma(t)=\frac{1}{2a_0}\left(\sqrt{-\Delta} \ tan\left[\frac{\sqrt{-\Delta}}{2}(t-t_0)\right]-b_0\right)$ and 
\begin{equation}a(t)=\frac{1}{u(\sigma(t))}=\frac{4a_0}{-\Delta}cos^2\left[\frac{\sqrt{-\Delta}}{2}(t-t_0)\right]\end{equation}
 for negative discriminant $\Delta=b_0^2-4a_0c_0<0$ and $t_0\in\mathds{R}$.   Then by (\ref{eq: phiforu=quadraticr=1}) with $\alpha_0=(d-1)/\kappa$, the scalar field is
 \begin{eqnarray}\phi(t)&\stackrel{def.}{=}&\psi(\sigma(t))=
 \sqrt{\frac{2(d-1)}{\kappa}} \ \ln\left[\frac{\sqrt{-\Delta}}{a_0}\left( \ tan\left[\frac{\sqrt{-\Delta}}{2}(t-t_0)\right] \right.\right.\notag\\
 &&\qquad\qquad\qquad\left.\left.+  sec\left[\frac{\sqrt{-\Delta}}{2}(t-t_0)\right] \right)\right] + \beta_0\end{eqnarray}
for $\beta_0\in\mathds{R}$.  Finally, by (\ref{eq:  nomatterFRLWNLSVphi-u}), (\ref{eq: uprimesigmaforu=quadraticr=1}) and (\ref{eq: dotsigmaandusigmaforu=quadraticr=1}) we obtain
\begin{eqnarray}
V(\phi(t))&=&\left[\frac{d(d-1)}{2\kappa}(u')^2-\frac{1}{2}u^2(\psi ')^2-\frac{\Lambda}{\kappa}\right]\circ\sigma(t)\notag\\
&=& \frac{-\Delta(d-1)}{2\kappa}  \left(   d \ tan^2\left[\frac{\sqrt{-\Delta}}{2}(t-t_0)\right] - \frac{1}{2}sec^2\left[\frac{\sqrt{-\Delta}}{2}(t-t_0)\right] \right) -\frac{\Lambda}{\kappa}.
\notag\\
\end{eqnarray}\end{example}

\begin{example}
For positive curvature $k=1$ and $\theta=1$, we take solution 2 in Table \ref{tb: exactNLS} with $b_0=1/\sqrt{2}$ and $a_0>0$ so that we have $u(\sigma) = a_0\cos^2(\sigma/\sqrt{2})$, $P(\sigma)=\tan^2(\sigma/\sqrt{2})$ and $E=1$.  By (\ref{eq: sigmaeqnforu=cos^2r=1}) - (\ref{eq: dotsigmaforu=cos^2r=1}) we obtain $\sigma(t)=\sqrt{2}Arctan\left(\frac{a_0}{\sqrt{2}}(t-t_0)\right)$ and 
\begin{equation}a(t)=\frac{1}{u(\sigma(t))}=\frac{1}{a_0}\left(1+\frac{a_0^2}{2}(t-t_0)^2\right)
\end{equation}
for $t_0\in\mathds{R}$.  Then by (\ref{eq: phiforu=cos^2r=1withc=0}) with $\alpha_0=(d-1)/\kappa$, the scalar field becomes
\begin{eqnarray}
\phi(t)
&=&\sqrt{\frac{(d-1)}{2\kappa}} \ln\left[  \frac{a_0^2}{2}(t-t_0)^2 +1  \right] + \beta_0
\label{eq: phiforFRLWNLSposcurvu=cos^2r=1andc=0}\end{eqnarray}
for $\beta_0\in\mathds{R}$.  Finally, by (\ref{eq:  nomatterFRLWNLSVphi-u}), (\ref{eq: uprimesigmaforu=cos^2r=1}) and (\ref{eq: dotsigmaforu=cos^2r=1}), we obtain
\begin{eqnarray}
V(\phi(t))
&=&\frac{(d-1)}{\kappa} \left( \frac{(2d-1)a_0^4(t-t_0)^2+2da_0^2}{\left(  2+a_0^2(t-t_0)^2  \right)^2  }  \right)-\frac{\Lambda}{\kappa}.
\end{eqnarray}
That is, we have 
\begin{equation}V(w)
=C_1e^{-\sqrt{\frac{2\kappa}{(d-1)}}w}-C_2e^{-2\sqrt{\frac{2\kappa}{(d-1)}}w} -\frac{\Lambda}{\kappa}
\end{equation}
for constants 
\begin{equation}C_1=\frac{(d-1)(2d-1)a_0^2}{2\kappa}e^{\sqrt{\frac{2\kappa}{(d-1)}}\beta_0}, \qquad C_2=\frac{(d-1)^2a_0^2}{2\kappa}e^{2\sqrt{\frac{2\kappa}{(d-1)}}\beta_0}\end{equation}
and $w\geq\sqrt{\frac{(d-1)}{2\kappa}}ln(a_0^2/2)+\beta_0$, since
\begin{equation}\phi^{-1}(w)=\frac{\sqrt{2}}{a_0}\sqrt{e^{\sqrt{\frac{2\kappa}{(d-1)}}(w-\beta_0)}-1}+t_0.\end{equation}
  By taking $d=3$ and $t_0=0$, replacing $a_0$ with $1/a_0$, and identifying $\kappa$ and $\beta_0$ here with $K^2$ and $ \phi_0$  in \cite{OT} respectively, we obtain the string-inspired solution II of \cite{OT}.  One can check that the conditions on the constants $C_1, C_2$ in \cite{OT} (with $d=3$) agree with the example here since
\begin{equation} a_0^2=\frac{2\kappa C_1^2}{(2d-1)^2 C_2} \ \ \mbox{ and } \ \ \beta_0=\sqrt{\frac{(d-1)}{2\kappa}}\ln\left(\frac{(2d-1)C_2}{(d-1)C_1}\right).\end{equation}

For $a_0=1$ and $t_0=0$, the solver was run with $u, u'$ and $\sigma$ perturbed by $.001$.  The graphs of $a(t)$ below show that the solution is unstable.  The absolute error grows by up to five orders of magnitude over the graphed time interval.

\begin{figure}[htbp]
\centering
\caption{Instability of FRLW Example \theexample}
\vspace{-.1in}
\includegraphics[width=4in]{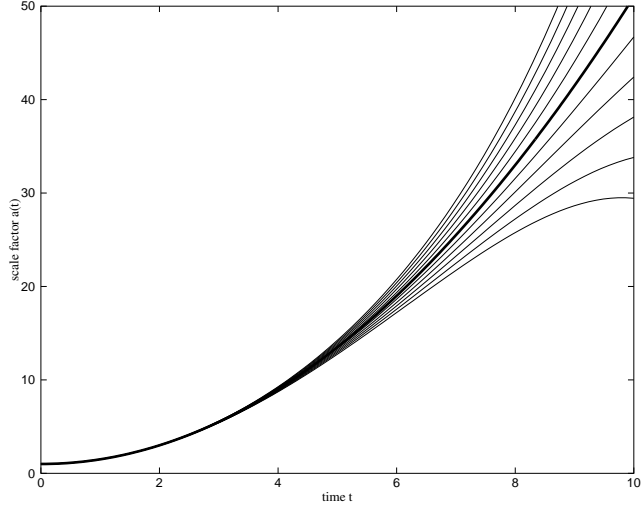}
\end{figure}

\end{example}

\break
\begin{example}
For negative curvature $k=-1$ and $\theta=1$, we take solution 4 in Table \ref{tb: exactNLS} with $c_0=-1$ and $b_0=0$ so that we have $u(\sigma)=a_0e^{-\sigma}$, $P(\sigma)=0$ and $E=-1$.  By (\ref{eq: sigmaeqnforu=-expr=general}) - (\ref{eq: dotsigmausigmaforu=-expr=general}) with $r_0=1$ we obtain $\sigma(t)=\ln\left(a_0(t-t_0)\right)$ and 
\begin{equation}a(t)=\frac{1}{u(\sigma(t))}=(t-t_0)\end{equation}
for $t_0\in\mathds{R}$.  Since $P=0=\psi'(\sigma)$, we get
$\psi(\sigma)=\psi_0$
for constant $\psi_0\in\mathds{R}$.  Finally, by (\ref{eq: nomatterFRLWNLSVphi-u}), (\ref{eq: uprimesigmaforu=-expr=general}) and (\ref{eq: dotsigmausigmaforu=-expr=general}), we obtain constant potential $V(\phi(t))=-\Lambda/\kappa$.

For $a_0=1$ and $t_0=0$, the solver was run with $u, u'$ and $\sigma$ perturbed by $.001$.  The graphs of $a(t)$ below show that the solution is unstable.  The absolute error grows four orders of magnitude over the graphed time interval.

\begin{figure}[htbp]
\centering
\vspace{-.1in}
\caption{Instability of FRLW Example \theexample}
\vspace{-.1in}
\includegraphics[width=4in]{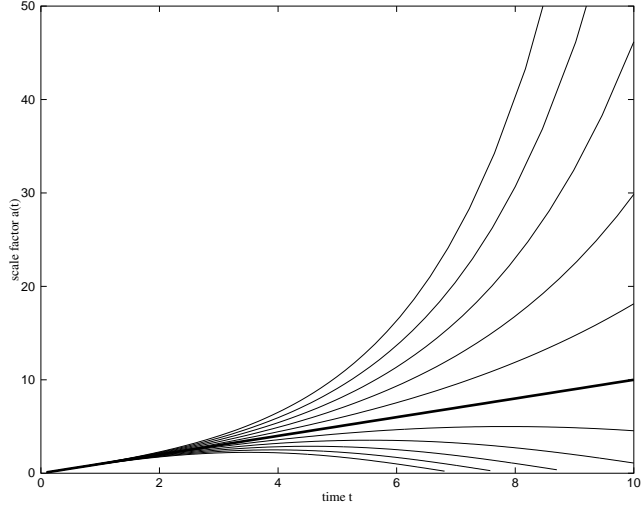}
\end{figure}
\end{example}

\break

\begin{example}
For negative curvature $k=-1$ and $\theta=1$, we take solution 4 in Table \ref{tb: exactNLS} with $c_0=-1$ and $a_0,b_0>0$ so that we have $u(\sigma)=a_0e^{-\sigma}-b_0e^{\sigma}$, $P(\sigma)=0$ and $E=-1$.  By (\ref{eq: sigmaeqnforu=linearcomboexpr=1}) - (\ref{eq: dotsigmausigmaforu=linearcomboexpr=1}) we obtain $\sigma(t)=\ln\left(\sqrt{\frac{a_0}{b_0}}tanh(\sqrt{a_0b_0}(t-t_0))\right)$ and
\begin{equation}a(t)=\frac{1}{u(\sigma(t))}=\frac{1}{2\sqrt{a_0b_0}} \ sinh(2\sqrt{a_0b_0}(t-t_0))
\end{equation}
for $t_0\in\mathds{R}$.  Since $P=0=\psi'(\sigma)$, we have that
$\psi(\sigma)=\psi_0$
for constant $\psi_0\in\mathds{R}$.  Finally, by (\ref{eq: nomatterFRLWNLSVphi-u}), (\ref{eq: uprimesigmaforu=linearcomboexpr=1}) and (\ref{eq: dotsigmausigmaforu=linearcomboexpr=1}), we obtain constant potential
\begin{eqnarray}
V(\phi(t))&=&\left[\frac{d(d-1)}{2\kappa}\left((u')^2 - u^2 \right)-\frac{\Lambda}{\kappa}\right]\circ\sigma(t)\notag\\
&=&\frac{2d(d-1)}{\kappa} a_0b_0-\frac{\Lambda}{\kappa}
\end{eqnarray}
since $coth^2(x)-csch^2(x)=1$.

\end{example}

\begin{example}
For arbitrary curvature $k$ and $\theta=1$, we take solution 5 in Table \ref{tb: exactNLS} with $c_0=-1$ and $b_0=k+1$ so that we have $u(\sigma)=(a_0/ \sigma)e^{-\sigma^2/2}$, $P(\sigma)=\sigma^2+2/\sigma^2+(k+1)$ and $E=k$ for $a_0>0$.  By (\ref{eq: sigmaeqnforu=e^x^2/xr=1cnegative}) - (\ref{eq: dotsigmausigmaforu=e^x^2/xr=1cnegative}) we obtain $\sigma(t)=\sqrt{2\ln(a_0(t-t_0))}$
and 
\begin{equation}a(t)=\frac{1}{u(\sigma(t))}=\sqrt{2}(t-t_0)\sqrt{\ln(a_0(t-t_0))}\end{equation}
for $t>1/a_0+t_0\in\mathds{R}$.  Then by (\ref{eq:  phiforu=xe^x^2r=1cnegative}) with $\alpha_0=(d-1)/\kappa$, the scalar field is
\begin{eqnarray}
\phi(t)&=&\psi(\sigma(t))\notag\\
&=&\sqrt{ \frac{ (d-1)}{2\kappa} } \left(\sqrt{2 \ln^2( a_0(t-t_0)) + (k+1)\ln(a_0(t-t_0)) + 1} \right.\notag\\
&&\quad + \ln\left[2\ln( a_0(t-t_0))\right] - \ln\left[2(1+k)\ln(a_0(t-t_0))+4\right. \notag\\
&&\quad  + \left. 4\sqrt{2 \ln^2(a_0(t-t_0)) + (1+k)\ln( a_0(t-t_0)) + 1} \right] \notag\\
&&\quad - \frac{(1+k)}{\sqrt{2}}\ln\left[2\sqrt{2}\sqrt{  2 \ln^2(a_0(t-t_0)) +(1+k)\ln( a_0(t-t_0)) + 1} \right.\notag\\
&&\left.\left.- (1+k)-4\ln( a_0(t-t_0))\right]\right)+\beta_0
\end{eqnarray}
for $\beta_0\in\mathds{R}$.  Finally, by (\ref{eq:  nomatterFRLWNLSVphi-u}), (\ref{eq: uprimesigmaforu=e^x^2/xr=1cnegative}) and (\ref{eq: dotsigmausigmaforu=e^x^2/xr=1cnegative}), we obtain
\begin{eqnarray}
V(\phi(t))&=&\left[\frac{d(d-1)}{2\kappa}\left((u')^2+ ku^2 \right)-\frac{u^2(\psi ')^2}{2}-\frac{\Lambda}{\kappa}\right]\circ\sigma(t)\label{eq: VfornomatterFRLWu=e^x^2/xr=1cnegative}
\\
&=&\frac{(d-1)}{2\kappa(t-t_0)^2}\left(  d\left( 1 + \frac{1}{2\ln( a_0(t-t_0))}  \right)^2+ \frac{dk}{2\ln(a_0(t-t_0))} \right.\notag\\
&&\quad \left. -\frac{\left[ 2\ln^2( a_0(t-t_0)) +1+(k+1)\ln( a_0(t-t_0)) \right]}{2\ln^2(a_0(t-t_0))} \right)
-\frac{\Lambda}{\kappa}
\notag\end{eqnarray}

For $a_0=k=1$ and $t_0=0$, the solver was run with $u, u'$ and $\sigma$ perturbed by $.001$.  The graphs of $a(t)$ below show that the solution is unstable.  The absolute error grows up to four orders of magnitude over the graphed time interval.  Since $E-P(\sigma)$ is independent of $k$, the below graph is also applicable to $k=0, -1$.

\begin{figure}[htbp]
\centering
\vspace{-.1in}
\caption{Instability of FRLW Example \theexample}
\vspace{-.1in}
\includegraphics[width=4in]{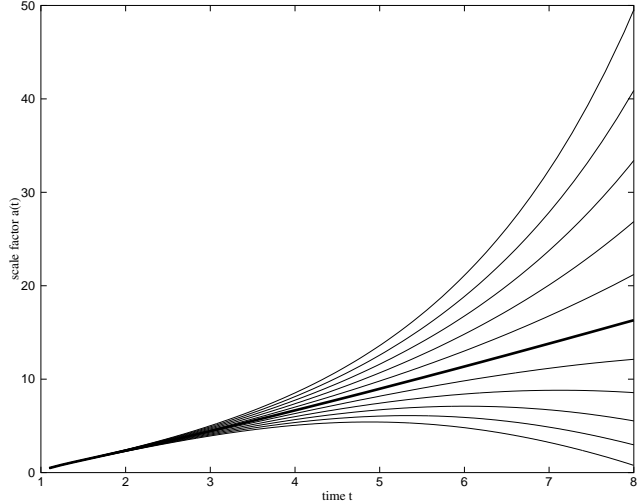}
\end{figure}

\end{example}

\break
\begin{example}
For positive curvature $k=1$ and $\theta=1$, we take solution 5 in Table \ref{tb: exactNLS} with  $c_0=0$ and $a_0=b_0=1$ so that we have $u(\sigma)=1/\sigma$, $P(\sigma)=2/\sigma^2+1$ and $E=1$.  By (\ref{eq: sigmaeqnforu=e^x^2/x}) - (\ref{eq: dotsigmausigmaforu=e^x^2/x}) with $r_0=1$ we obtain $\sigma(t)=\sqrt{ 2(t-t_0) }$
and 
\begin{equation}a(t)=\frac{1}{u(\sigma(t))}=  \sqrt{ 2(t-t_0) }\end{equation}
for $t>t_0\in\mathds{R}$.  Then by (\ref{eq:  phiforu=1/sigmargeneral}) with $\alpha_0=(d-1)/\kappa$, we have
\begin{eqnarray}
\phi(t)
&=&\sqrt{\frac{(d-1)}{\kappa}} \left(\sqrt{g(t)}+\frac{1}{\sqrt{2}}\ln\left[ 2(t-t_0) \right]-\sqrt{2}\ln\left[\sqrt{2}+\sqrt{g(t)})\right]\right)+\beta_0\notag
\end{eqnarray}
for $g(t)= 2(t-t_0)+2$ and $\beta_0\in\mathds{R}$.  Finally, by (\ref{eq:  nomatterFRLWNLSVphi-u}), (\ref{eq: uprimesigmaforu=e^x^2/x}) and (\ref{eq: dotsigmausigmaforu=e^x^2/x}), we get
\begin{eqnarray}
V(\phi(t))
&=&\frac{(d-1)}{4\kappa} \left(  \frac{(d-2)}{2(t-t_0)^2}+\frac{(d-1)}{(t-t_0)}  \right) -\frac{\Lambda}{\kappa}.
\label{eq: VfornomatterposcurvFRLWu=1/xr=1czero}
\end{eqnarray}

For $t_0=0$, the solver was run with $u, u'$ and $\sigma$ perturbed by $.01$.  The graphs of $a(t)$ below show that the solution is unstable.  The absolute error grows two orders of magnitude over the graphed time interval.  
\begin{figure}[htbp]
\centering
\vspace{-.2in}
\caption{Instability of FRLW Example \theexample}
\vspace{-.1in}
\includegraphics[width=4in]{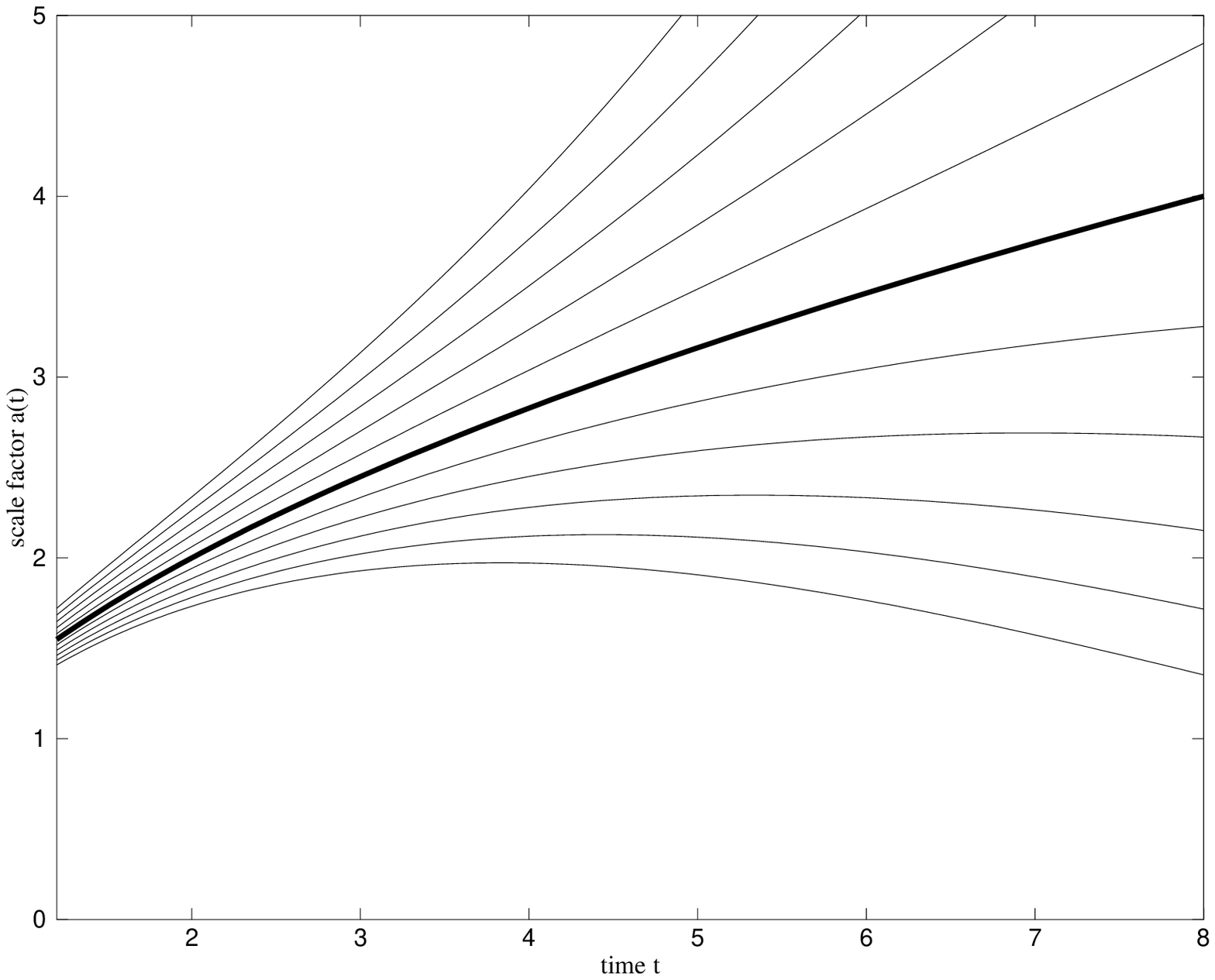}
\end{figure}
\end{example}

\break
\begin{example}
For zero curvature $k=0$ and $\theta=1$, we take solution 5 in Table \ref{tb: exactNLS} with  $b_0=c_0=0$ and $a_0=0$ so that we have $u(\sigma)=1/\sigma$, $P(\sigma)=2/\sigma^2$ and $E=0$.  By (\ref{eq: sigmaeqnforu=e^x^2/x}) - (\ref{eq: dotsigmausigmaforu=e^x^2/x}) with $r_0=1$ we obtain $\sigma(t)=\sqrt{ 2(t-t_0) }$
and 
\begin{equation}a(t)=\frac{1}{u(\sigma(t))}=  \sqrt{ 2(t-t_0) }\end{equation}
for $t>t_0\in\mathds{R}$.  Then by (\ref{eq:  phiforu=1/sigmargeneral}) with $\alpha_0=(d-1)/\kappa$ and $b_0=0$, we have
\begin{eqnarray}
\phi(t)
&=&\sqrt{\frac{(d-1)}{2\kappa}} \ln\left[ 2(t-t_0) \right]+\beta_0\notag
\end{eqnarray}
for  $\beta_0\in\mathds{R}$.  Finally, by (\ref{eq:  nomatterFRLWNLSVphi-u}), (\ref{eq: uprimesigmaforu=e^x^2/x}) and (\ref{eq: dotsigmausigmaforu=e^x^2/x}), we obtain
\begin{eqnarray}
V(\phi(t))
&=&\frac{(d-1)(d-2)}{8\kappa(t-t_0)^2}     -\frac{\Lambda}{\kappa}.
\label{eq: VfornomatterposcurvFRLWu=1/xr=1czero}
\end{eqnarray}

For $t_0=0$, the solver was run with $u, u'$ and $\sigma$ perturbed by $.001$.  The graphs of $a(t)$ below show that the solution is unstable.  The absolute error grows up to three orders of magnitude over the graphed time interval.  

\begin{figure}[htbp]
\centering
\vspace{-.2in}
\caption{Instability of FRLW Example \theexample}
\vspace{-.1in}
\includegraphics[width=4in]{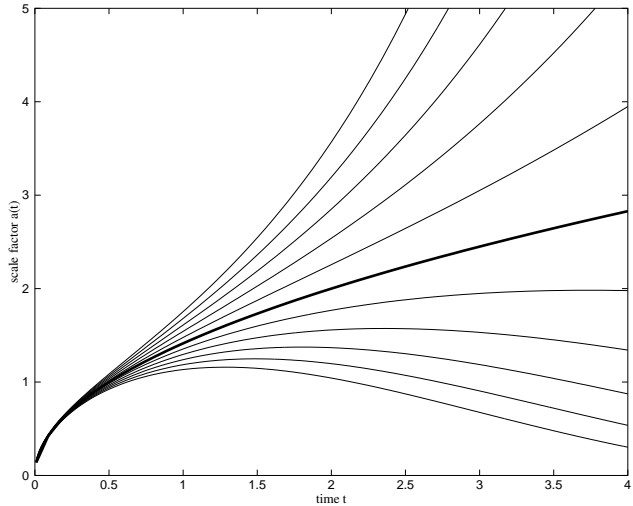}
\end{figure}
\end{example}


\break
\begin{example}
We take solution 7 in Table \ref{tb: exactNLS} with $c_0=k\theta^2$ and $b_0>0$ and $a_0>0$ so that we have $u(\sigma)=a_0 / \sigma^{b_0}$, $P(\sigma)=b_0(b_0+1)/\sigma^2+k\theta^2$ and $E=k\theta^2$.  By (\ref{eq: sigmaeqnforu=x^b}) - (\ref{eq: dotsigmaforu=x^b}) we obtain 
\begin{equation}\label{eq: sigmaneededinFRLWex}\sigma(t)=\frac{A^{1/b_0}}{a_0} (t-t_0)^{1/(1+b_0)}\end{equation} and 
\begin{equation}a(t)=\frac{1}{u(\sigma(t))}=A(t-t_0)^{b_0/(b_0+1)}\end{equation}
for $t>t_0\in\mathds{R}$ and $A\stackrel{def.}{=}\frac{1}{a_0}\left(\frac{(1+b_0)a_0}{\theta}\right)^{b_0/(1+b_0)}$.  Then by (\ref{eq: phiforu=x^b}) with $\alpha_0=(d-1)/\kappa$, the scalar field is
\begin{eqnarray}\phi(t)&=&  \beta_0+\sqrt{\frac{(d-1)}{\kappa}}\left( \sqrt{b_0(b_0+1)+k\theta^2\sigma(t)^2} \right.\label{eq: phiforFRLWexEM44kneq0}\\
&&\left.- \sqrt{b_0(b_0+1)} \log\left( \frac{\sqrt{b_0(b_0+1)}+\sqrt{b_0(b_0+1)+k\theta^2\sigma(t)^2}}{\sigma(t)} \right) \right) \notag
\end{eqnarray}
for $\sigma(t)$ in (\ref{eq: sigmaneededinFRLWex}).  Also by (\ref{eq: nomatterFRLWNLSVphi-u}), (\ref{eq: dotsigmaforu=x^b}) and (\ref{eq: uprimesigmaforu=x^b}), we obtain potential
\begin{eqnarray}
V(\phi(t))
&=&-\frac{\Lambda}{\kappa}+\frac{B^2}{2(t-t_0)^2}\left(\frac{((d-1)b_0-1)}{(1+b_0)}+ \frac{k(d-1)(1+b_0)}{b_0 A^{2} (t-t_0)^{-2/(1+b_0)}}  \right)
 \ \ \  \ \ \ 
\end{eqnarray}
for $B^2\stackrel{def.}{=}\frac{(d-1)b_0}{(1+b_0)\kappa}$.

\noindent By taking $d=3$, $\Lambda=t_0=0$, $\theta=a_0(1+b_0)$ and identifying $b_0/(b_0+1)$ here with $n$ in \cite{EM}, we obtain example 4.4 of Ellis and Madsen for $0<n<1$.  Note that the above form of $\phi(t)$ is the integrated version of $\phi(t)$ in \cite{EM}.  Since $\phi(t)$ in (\ref{eq: phiforFRLWexEM44kneq0}) and has the property that 
\begin{equation}\dot\phi(t)^2=\frac{B^2}{(t-t_0)^2}\left( 1+\frac{k\theta^2(a_0A)^{2/b_0}}{b_0(b_0+1)}(t-t_0)^{2/(1+b_0)} \right),\end{equation}
we are in agreement with \cite{EM} for the above choice of $\theta$.  Also note that there is a typo in equation (45) of \cite{EM}, where one should multiply $V$ by $1/4$ to obtain the correct $V$ with a two appearing in the denominator instead of the numerator.  In contrast to example 3 in this thesis, this example generalizes the Ellis and Madsen example for nonzero curvature.

For $a_0=b_0=1$ and $t_0=0$, the solver was run with $u, u'$ and $\sigma$ perturbed by $.001$.  The graphs of $a(t)$ below show that the solution is unstable.  The absolute error grows three orders of magnitude over the graphed time interval.  Since $E-P(\sigma)$ is independent of $k$, the below graph is applicable to all values of the curvature $k=1, 0, -1$.

\begin{figure}[htbp]
\centering
\vspace{-.1in}
\caption{Instability of FRLW Example \theexample}
\vspace{-.1in}
\includegraphics[width=4in]{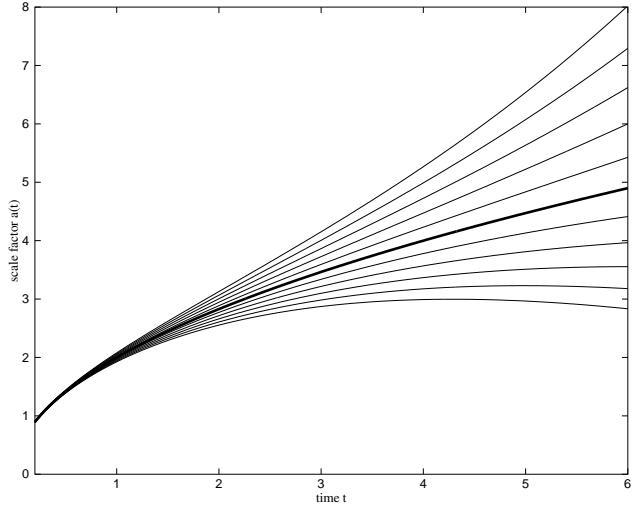}
\end{figure}

\end{example}

\break
\subsection{A nonlinear Schr\"odinger example}
\begin{example}
We consider equation (\ref{eq: FRLWNLS}) with $\uprho=\mathrm{p}=0$, $M=k=\theta=1$, $n=4$ and $D_1(\sigma)=D>0$ a constant.  For solution 3 in Table \ref{tb: exactNLS} with $a_0=\sqrt{\frac{d(d-1)}{2\kappa D}}b_0$, $c_0=1-2b_0^2$ and $b_0>0$ chosen such that $c_0>0$, we have $u(\sigma)=a_0 tanh(b_0\sigma)$ and $P(\sigma)=c_0$.  By (\ref{eq:  sigmaeqnforu=tanh}) - (\ref{eq: dotsigmausigmaforu=tanh}) we obtain $\sigma(t)=\frac{1}{b_0}Arcsinh\left(e^{a_0b_0(t-t_0)}\right)$ and 
\begin{equation}a(t)=\frac{1}{u(\sigma(t))}= \frac{ \sqrt{1+e^{2a_0b_0(t-t_0)}}   }{ a_0e^{a_0b_0(t-t_0)}  }\end{equation}
for $t_0\in\mathds{R}$.   Then by (\ref{eq: phiforu=tanhr=1}) with $\alpha_0=(d-1)/\kappa$, we get 
\begin{eqnarray}\phi(t)&=&\psi(\sigma(t))=\frac{\sqrt{(d-1)(1-2b_0^2)}}{b_0\sqrt{\kappa}}Arcsinh\left(e^{a_0b_0(t-t_0)}\right)+\beta_0\end{eqnarray}
for $\beta_0\in\mathds{R}$.  Finally, by (\ref{eq: FRLWNLSVphi-u}) and (\ref{eq: uprimesigmaforu=tanhr=1}) we obtain
\begin{eqnarray}
V(\phi(t))&=&\left[\frac{d(d-1)}{2\kappa}\left((u')^2+ u^2 \right)-\frac{1}{2}u^2(\psi ')^2-D u^{4} -\frac{\Lambda}{\kappa}\right]\circ\sigma(t).\notag\\
&=&\frac{a_0^2(d-1)}{2\kappa ( 1+e^{2a_0b_0(t-t_0)} )^2}\left[    \left(  (d-1)(1-b_0^2)+b_0^2  \right)e^{4a_0b_0(t-t_0)} \right.\notag\\
 &&\left. +    \left(d-1   + 2b_0^2    \right) 
  e^{2a_0b_0(t-t_0)}     + db_0^2 \right]   -\frac{\Lambda}{\kappa}\notag\\
\end{eqnarray}

For $a_0=1, b_0=1/2$ and $t_0=0$, the solver was run with $u, u'$ and $\sigma$ perturbed by $.001$.  The graphs of $a(t)$ below show that the solution is unstable.  The absolute error grows two orders of magnitude over the graphed time interval.

\begin{figure}[htbp]
\centering
\vspace{-.1in}
\caption{Instability of FRLW Example \theexample}
\vspace{-.1in}
\includegraphics[width=4in]{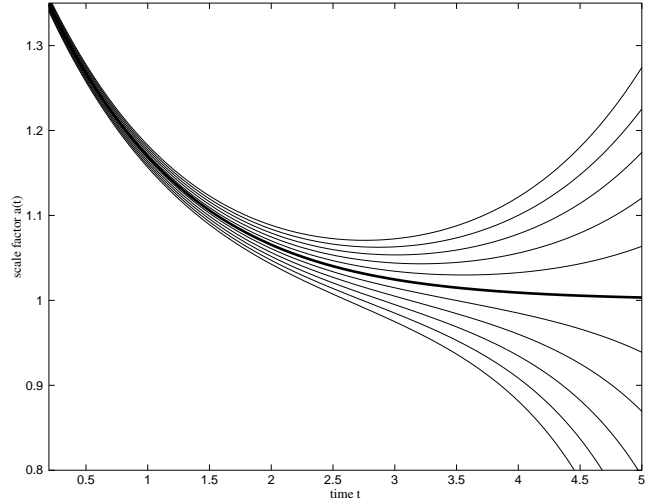}
\end{figure}

\end{example}

\break

\section{In terms of an Alternate Schr\"odinger-Type Equation}

\begin{thm}{$\left(\mbox{Apply Corollary \ref{cor: EFE-NLSAnonzeroEzero} with }\frac{G(t)}{a(t)^A}=-\frac{\kappa n_j D_j(t)}{d(d-1)a(t)^{n_j}}\right)$}\\
\label{thm: FRLWEFE-altNLS}
Suppose you are given a twice differentiable function $a(t)>0$, a once differentiable function $\phi(t)$, and also functions $D_i(t), \rho'(t), p'(t), V(x)$ which satisfy the Einstein equations $(i),(ii)$ in (\ref{eq:  FRLWEFEiii}) for some $k, n_1, \dots, n_{M}, \Lambda \in\mathds{R}, d\in\mathds{R}\backslash\{0,1\}, \kappa\in\mathds{R}\backslash\{0\}$ and $M\in\mathds{N}$.  Let $g(\sigma)$ denote the inverse of a function $\sigma(t)$ which satisfies
\begin{equation}\dot{\sigma}(t)=\frac{1}{\theta} a(t)^{-n_j/2}\label{eq: FRLWNLSaltdotsigma-a}\end{equation}
for some $\theta>0$ and where $j$ is some index for which $n_j\neq 0$.  Then the functions 
\begin{eqnarray}u(\sigma)&=&a(g(\sigma))^{-n_j/2}\label{eq: FRLWNLSaltu-a}\\
P(\sigma)&=&\frac{ n_j \kappa}{2(d-1)}\psi '(\sigma)^2\label{eq: FRLWNLSaltP-psi}\end{eqnarray}
solve the Schr\"odinger-type equation
\begin{equation}u''(\sigma)+\left[\frac{-\theta^2 n_j^2\kappa}{2d(d-1)}\mathrm{D}_j(\sigma) -P(\sigma)\right]u(\sigma)=\frac{-\theta^2 n_j k}{2u(\sigma)^{1-\frac{4}{n_j}}}+\displaystyle\sum_{\stackrel{1\leq i \leq M}{i\neq j}}\frac{\theta^2  n_jn_i \kappa \mathrm{D}_i(\sigma)}{2d(d-1) u(\sigma)^{1-2\frac{n_i}{n_j}}}\notag\end{equation}
\begin{equation} \ \ \ \ \ \  \ \ \ \  \ \ \ \ \ \  \ \ \ \ \ \ \ \ \ \ \ \ \ \ \ \ \ \ \ +\frac{\theta^2 n_j \kappa(\uprho(\sigma)+\mathrm{p}(\sigma))}{2(d-1) u(\sigma)}
\label{eq: FRLWNLSalt}\end{equation}
for 
\begin{equation}\psi(\sigma)=\phi(g(\sigma))\label{eq: FRLWNLSaltpsi-phi}\end{equation}
\begin{equation}\mathrm{D}_i(\sigma)=D_i(g(\sigma)), \ \  1\leq i\leq M \label{eq: FRLWNLSaltrmD-D}\end{equation}
and
\begin{equation}\uprho(\sigma)=\rho'(g(\sigma)), \ \mathrm{p}(\sigma)=p'(g(\sigma)). \label{eq: FRLWNLSaltuprho-rhormp-p}\end{equation}

Conversely, suppose you are given a twice differentiable function $u(\sigma)>0$, and also functions $P(\sigma), \mathrm{D}_{i}(\sigma)$ for $1\leq i\leq M, M\in\mathds{N}$ and $\uprho(\sigma), \mathrm{p}(\sigma)$ which solve (\ref{eq: FRLWNLSalt}) for some constants $\theta>0, k\in\mathds{R}, \kappa\in\mathds{R}\backslash\{0\},  d\in\mathds{R}\backslash\{0,1\}$ and $n_i\in\mathds{R}$ for $1\leq i \leq M$.  In order to construct functions which solve $(i),(ii)$, first find $\sigma(t), \psi(\sigma)$ which solve the differential equations
\begin{equation}\dot{\sigma}(t)=\frac{1}{\theta} u(\sigma(t))\qquad\mbox{ and }\qquad \psi '(\sigma)^2= \frac{2(d-1)}{ n_j \kappa} P(\sigma).\label{eq: FRLWNLSaltdotsigma-upsi-P}\end{equation}
Then the functions
\begin{eqnarray}a(t)&=&u(\sigma(t))^{-2/n_j}\label{eq: FRLWNLSalta-u}\\
\phi(t)&=&\psi(\sigma(t))\label{eq: FRLWNLSaltphi-psi}\end{eqnarray}
\begin{equation}D_{i}(t)=\mathrm{D}_i(\sigma(t)), 1\leq i \leq M, \label{eq: FRLWNLSaltD-rmD}\end{equation}
\begin{equation}\rho'(t)=\uprho(\sigma(t)), \ p'(t)=\mathrm{p}(\sigma(t))\label{eq: FRLWNLSaltrho-uprhop-rmp}\end{equation}
and\\
$V(\phi(t))$
\begin{equation}
=\left[\frac{d(d-1)}{2\kappa}\left(\frac{4}{n_j^2 \theta^2}(u')^2+\frac{k}{u^{-4/n_j}}\right)-\frac{1}{2\theta^2}u^2(\psi ')^2-\displaystyle\sum_{i=1}^{M}\frac{\mathrm{D}_i}{u^{-2n_i/n_j}}-\uprho-\frac{\Lambda}{\kappa}\right]\circ\sigma(t)
\label{eq: FRLWNLSaltVphi-u}
\end{equation}
satisfy the equations $(i),(ii)$.\end{thm}

\begin{proof}
This proof will implement Corollary \ref{cor: EFE-NLSAnonzeroEzero} with constants and functions as indicated in the following table.

      \begin{table}[ht]
\centering
\caption{{ Corollary \ref{cor: EFE-NLSAnonzeroEzero} applied to FRLW, alternate}}\label{tb: FRLWNLSalt}
\vspace{.2in}
\begin{tabular}{r | l c r | l}
In Corollary & substitute & & In Corollary & substitute \\[4pt]
\hline
\raisebox{-5pt}{$\varepsilon $} & \raisebox{-5pt}{${\kappa}/{(d-1)}$}      &&    \raisebox{-5pt}{$E(\sigma)$} & \raisebox{-5pt}{$-\frac{\theta^2  n_j^2 \kappa}{2d(d-1)} \mathrm{D}_j(\sigma)$}\\[8pt]
                                       $G(t)$ & $\frac{-n_j\kappa}{d(d-1)}D_j(t) $    &&     $A$ & $\mbox{some }n_j\neq 0$ \\[8pt]
$G_i(t), 1\leq i\leq M, i\neq j$ & $\frac{-n_i \kappa}{d(d-1)}D_i(t) $     &&    $A_{i\neq j}$ &$n_i$ \\[8pt]
$G_j(t)$ & $k $     &&    $A_j$ &$2$ \\[8pt]
 $G_{M+1}(t)$ & $ \frac{-\kappa}{(d-1)}(\rho'(t)+p'(t))$  &&      $A_{M+1}$   &$0$\\[8pt]
   $F_i(\sigma), 1\leq i\leq M, i\neq j$&$\frac{\theta^2 n_j n_i \kappa }{2d(d-1)}\mathrm{D}_i(\sigma)$ &&  $C_{i\neq j}$&$1-2\frac{n_i}{n_j}$\\[8pt]
   $F_j(\sigma)$&$-\theta^2 n_jk/2$&&$C_j$&$1-\frac{4}{n_j}$\\[8pt]
$F_{M+1}(\sigma)$&$  \frac{\theta^2 n_j\kappa}{2(d-1)}(\uprho(\sigma)+\mathrm{p}(\sigma)) $ &&        $C_{M+1}$&  $1$ \\[6pt]
\hline
\end{tabular}
\end{table}

To prove the forward implication, we assume to be given functions which solve the Einstein field equations $(i)$ and $(ii)$ from (\ref{eq:  FRLWEFEiii}).  Subtracting equations $(ii)-(i)$, we see that
\begin{equation}\dot{H}(t)-\frac{k}{a(t)^2}=-\frac{\kappa}{(d-1)}\left[\dot{\phi}(t)^2+\displaystyle\sum_{i=1}^{M}\frac{n_i D_i(t)}{da(t)^{n_i}}+(\rho'(t)+p'(t))
\right].\label{eq: FRLWNLSaltiiminusi}\end{equation}
This shows that $a(t), \phi(t), D_i(t), \rho'(t)$ and $p'(t)$ satisfy the hypothesis of Corollary \ref{cor: EFE-NLSAnonzeroEzero}, applied with  constants $\varepsilon , N, A, A_1 \dots, A_{N}$ and functions $G(t), G_1(t) \dots,G_{N}(t)$ according to Table \ref{tb: FRLWNLSalt}.  Since $\sigma(t), u(\sigma), P(\sigma)$ and $\psi(\sigma)$ defined in (\ref{eq: FRLWNLSaltdotsigma-a}), (\ref{eq: FRLWNLSaltu-a}), (\ref{eq: FRLWNLSaltP-psi}) and (\ref{eq: FRLWNLSaltpsi-phi})  are equivalent to that in the forward implication of Corollary \ref{cor: EFE-NLSAnonzeroEzero}, by this theorem and by definitions (\ref{eq: FRLWNLSaltrmD-D}) and (\ref{eq: FRLWNLSaltuprho-rhormp-p}) of $\mathrm{D}_i(\sigma)$ and $\uprho(\sigma), \mathrm{p}(\sigma)$, the Schr\"odinger-type equation (\ref{eq: CNLSANONZERO}) holds for constants $C_1, \dots, C_{N}$ and functions $F_1(\sigma), \dots, F_{N}(\sigma)$ as indicated in Table \ref{tb: FRLWNLSalt}.  This proves the forward implication.

To prove the converse implication, we assume to be given functions which solve the Schr\"odinger-type equation (\ref{eq: FRLWNLSalt}) and we begin by showing that $(i)$ is satisfied.   By differentiating the definition (\ref{eq: FRLWNLSalta-u}) of $a(t)$ and using the definition  of $\sigma(t)$ in (\ref{eq: FRLWNLSaltdotsigma-upsi-P}), we obtain
\begin{eqnarray}\dot{a}(t)&=&-\frac{2}{n_j} u(\sigma(t))^{-2/n_j-1} u'(\sigma(t)) \dot\sigma(t)\notag\\
&=&-\frac{2}{n_j\theta} u(\sigma(t))^{-2/n_j} u'(\sigma(t)).\end{eqnarray}
Dividing by $a(t)$, we have that
\begin{equation}H(t)\stackrel{def.}{=}\frac{\dot{a}(t)}{a(t)}=-\frac{2}{n_j\theta}u'(\sigma(t)).\label{eq: FRLWNLSaltH-u}\end{equation}
Differentiating the definition (\ref{eq: FRLWNLSaltphi-psi}) of $\phi(t)$ and using definition in (\ref{eq: FRLWNLSaltdotsigma-upsi-P}) of  $\sigma(t)$, we have
\begin{equation}\dot{\phi}(t)=\psi '(\sigma(t))\dot{\sigma}(t)=\frac{1}{\theta}\psi '(\sigma(t))u(\sigma(t)).\label{eq: FRLWNLSaltdotphi-u}\end{equation}
Using (\ref{eq: FRLWNLSaltH-u}) and (\ref{eq: FRLWNLSaltdotphi-u}), and also the definitions (\ref{eq: FRLWNLSalta-u}), (\ref{eq: FRLWNLSaltD-rmD}) and (\ref{eq: FRLWNLSaltrho-uprhop-rmp}) of $a(t), D_i(t)$ and $\rho'(t), p'(t)$, the definition (\ref{eq: FRLWNLSaltVphi-u}) of $V\circ\phi$ can be written as
\begin{equation}
V(\phi(t))=\frac{d(d-1)}{2\kappa}\left(
H(t)^2+\frac{k}{a(t)^2}\right)-\frac{1}{2}\dot\phi(t)^2-\displaystyle\sum_{i=1}^{M}\frac{D_i(t)}{a(t)^{n_i}}-\rho'(t)-\frac{\Lambda}{\kappa}
\label{eq: FRLWNLSaltVphi-a}.\end{equation}
This shows that $(i)$ holds (that is, the definition of $V(\phi(t))$ was designed to be such that $(i)$ holds).

To conclude the proof we must also show that $(ii)$ holds.   In the converse direction the hypothesis of the converse of Corollary \ref{cor: EFE-NLSAnonzeroEzero} holds, applied with constants $N, C_1, \dots, C_{N}$ and functions $E(\sigma), F_1(\sigma), \dots, F_{N}(\sigma)$ as indicated in Table \ref{tb: FRLWNLSalt}.  Since  $\sigma(t), \psi(\sigma), a(t)$ and $\phi(t)$ defined in (\ref{eq: FRLWNLSaltdotsigma-upsi-P}), (\ref{eq: FRLWNLSalta-u}) and (\ref{eq: FRLWNLSaltphi-psi}) are consistent with the converse implication of Corollary \ref{cor: EFE-NLSAnonzeroEzero}, applied with $\varepsilon $ and $A$ as in Table \ref{tb: FRLWNLSalt}, by this corollary and by definitions (\ref{eq: FRLWNLSaltD-rmD}) and (\ref{eq: FRLWNLSaltrho-uprhop-rmp}) of $D_i(t)$ and $\rho'(t), p'(t)$ the  scale factor  equation (\ref{eq: CEFEANONZERO}) holds for constants $\varepsilon , A, A_1, \dots, A_{N}$ and functions $G(t), G_1(t),\dots,G_{N}(t)$ according to Table \ref{tb: FRLWNLSalt}.  That is, we have regained (\ref{eq: FRLWNLSaltiiminusi}) which shows that the subtraction of equations (ii)-(i) holds in the converse direction.  Now solving (\ref{eq: FRLWNLSaltVphi-a}) for $\rho'(t)$ and substituting this into (\ref{eq: FRLWNLSaltiiminusi}), we obtain (ii).  This proves the theorem.
\end{proof}

\break

\subsection{Reduction to linear Schr\"odinger: zero curvature}
To compute some examples, we take $k=0$ and $\rho'=p'=D_i=0$ for all $i\neq j$ and also $n_j=n$, $D_j=D>0$ so that Theorem \ref{thm: FRLWEFE-altNLS} shows that solving the Einstein equations 
\begin{equation}\frac{d}{2}H^2(t)\stackrel{(i)''''}{=}\frac{\kappa}{(d-1)} \left[\frac{1}{2}\dot\phi(t)^2+V(\phi(t))+\frac{D}{a(t)^{n}}\right]+\frac{\Lambda}{(d-1)}\end{equation}
\begin{eqnarray*}
\dot{H}(t)+\frac{d}{2}H(t)^2&\stackrel{(ii)''''}{=}&-\frac{\kappa}{(d-1)} \left[\frac{1}{2}\dot\phi(t)^2-V(\phi(t))+\frac{(n-d)D}{da(t)^{n}}\right]+\frac{\Lambda}{(d-1)}\qquad\qquad \end{eqnarray*}
is equivalent to solving the linear Schr\"odinger equation
\begin{equation}u''(\sigma)+\left[\frac{-\theta^2 n^2\kappa D}{2d(d-1)} -P(\sigma)\right]u(\sigma)=0\notag\label{eq: nocurvFRLWaltNLS}\end{equation}
for any constant $\theta>0$.  The solutions of $(i)'''',(ii)''''$ and (\ref{eq: nocurvFRLWaltNLS}) are related by 
\begin{equation}a(t)=\frac{1}{u(\sigma(t))^{2/n}}\quad\mbox{ and }\quad
\psi '(\sigma)^2=\frac{2(d-1)}{n\kappa}P(\sigma)\end{equation}
for $\phi(t)=\psi(\sigma(t))$ and 
\begin{equation}\dot{\sigma}(t)=\frac{1}{\theta a(t)^{n/2}}=\frac{1}{\theta} u(\sigma(t)).\label{eq: nocurvFRLWaltNLSdotsigmau}\end{equation}
Also in the converse direction, $V$ is taken to be
\begin{equation}
V(\phi(t))=\left[\frac{2d(d-1)(u')^2}{\kappa n^2 \theta^2} -\frac{1}{2\theta^2}u^2(\psi ')^2-Du^2-\frac{\Lambda}{\kappa}\right]\circ\sigma(t).
\label{eq: nocurvFRLWNLSaltVphi-u}
\end{equation}

We now refer to Appendix E for solutions of the linear Schr\"odinger equation (\ref{eq: nocurvFRLWaltNLS}).  We will map these solutions to exact solutions of Einstein's equations.  Since $E=\frac{-\theta^2 n^2\kappa D}{2d(d-1)}<0$, we only consider entries in Table \ref{tb: exactNLS} for which $E<0$.

\begin{example}
For $\theta=1$ and choice of constant $D=\frac{2d(d-1)}{n^2\kappa}$, we take solution 4 in Table \ref{tb: exactNLS} with $c_0=-1$ and $b_0=0$ so that we have $u(\sigma)=a_0e^{-\sigma}$, $P(\sigma)=0$ and $E=-1$.  By (\ref{eq: sigmaeqnforu=-expr=general}) - (\ref{eq: dotsigmausigmaforu=-expr=general}) with $r_0=1$ we obtain $\sigma(t)=\ln\left(a_0(t-t_0)\right)$ and 
\begin{equation}a(t)=\frac{1}{u(\sigma(t))^{2/n}}=(t-t_0)^{2/n}\end{equation}
for $t_0\in\mathds{R}$.  Since $P=0=\psi'(\sigma)$, the scalar field is constant $\psi(\sigma)=\psi_0\in\mathds{R}.$
 Finally, by (\ref{eq: nocurvFRLWNLSaltVphi-u}), (\ref{eq: uprimesigmaforu=-expr=general}) and (\ref{eq: dotsigmausigmaforu=-expr=general}), we obtain constant potential
\begin{eqnarray}
V(\phi(t))&=&\left[\frac{2d(d-1)(u')^2}{n^2 \kappa } -\frac{2d(d-1)}{n^2\kappa}u^2-\frac{\Lambda}{\kappa}\right]\circ\sigma(t)\notag\\
&=&-\frac{\Lambda}{\kappa}.
\end{eqnarray}

For $a_0=1, n=3$ and $t_0=0$, the solver was run with $u, u'$ and $\sigma$ perturbed by $.001$.  The graphs of $a(t)$ below show that the solution is unstable.  The absolute error grows at least two orders of magnitude over the graphed time interval.  Since $u$ is independent of $n$, $a(t)$ is unstable for all values of $n$.

\vspace{.1in}
\begin{figure}[htbp]
\centering
\caption{Instability of FRLW Example \theexample}
\includegraphics[width=4in]{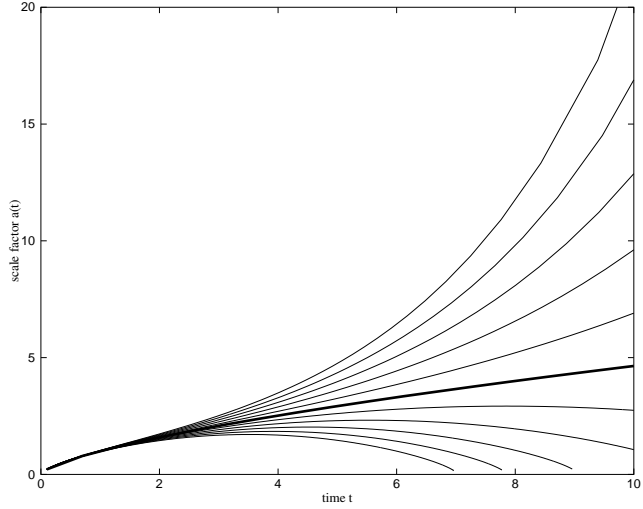}
\end{figure}
\end{example}

 \begin{example}
For $\theta=1$ and choice of constant $D=\frac{2d(d-1)}{n^2\kappa}$, we take solution 4 in Table \ref{tb: exactNLS} with $c_0=-1$ and $a_0,b_0>0$ so that we have $u(\sigma)=a_0e^{-\sigma}-b_0e^{\sigma}$, $P(\sigma) = 0$ and $E=-1$.  By (\ref{eq: sigmaeqnforu=linearcomboexpr=1}) - (\ref{eq: dotsigmausigmaforu=linearcomboexpr=1}) we obtain $\sigma(t)=\ln\left(\sqrt{\frac{a_0}{b_0}}tanh(\sqrt{a_0b_0}(t-t_0))\right)$ and 
\begin{equation}a(t)=\frac{1}{u(\sigma(t))^{2/n}}=\frac{1}{(2\sqrt{a_0b_0})^{2/n}} \ sinh^{2/n}(2\sqrt{a_0b_0}(t-t_0))\end{equation}
for $t_0\in\mathds{R}$.   Since $P=0=\psi'(\sigma)$, the scalar field is constant
\begin{equation}\psi(\sigma)=\psi_0\in\mathds{R}\end{equation}
Finally, by (\ref{eq: nocurvFRLWNLSaltVphi-u}), (\ref{eq: uprimesigmaforu=linearcomboexpr=1}) and (\ref{eq: dotsigmausigmaforu=linearcomboexpr=1}), we obtain constant potential
\begin{eqnarray}
V(\phi(t))&=&\left[\frac{2d(d-1)(u')^2}{n^2 \kappa } -\frac{2d(d-1)}{n^2\kappa}u^2-\frac{\Lambda}{\kappa}\right]\circ\sigma(t)\notag\\
&=&\frac{8d(d-1)a_0b_0}{n^2 \kappa }-\frac{\Lambda}{\kappa}
\end{eqnarray}
since $coth^2(x)-csch^2(x)=1$.

\end{example}

\begin{example}
For $\theta=1$ and choice of constant $D=2d(d-1)/n^2\kappa$, we take solution 5 in Table \ref{tb: exactNLS} with $c_0=-1$ and $b_0 =0$ so that we have $u(\sigma)=(a_0/\sigma)e^{-\sigma^2/2}$, $P(\sigma)=\sigma^2+2/\sigma^2$ and $E=-1$ for $a_0>0$.  By (\ref{eq: sigmaeqnforu=e^x^2/xr=1cnegative}) - (\ref{eq: dotsigmausigmaforu=e^x^2/xr=1cnegative}) we obtain $\sigma(t)=\sqrt{2\ln( a_0(t-t_0))}$ and 
\begin{eqnarray}
a(t)&=&\frac{1}{u(\sigma(t))^{2/n}}\notag\\
&=&\left(\sqrt{2}(t-t_0)\sqrt{\ln( a_0(t-t_0))}\right)^{2/n}\notag\\
&=&\left(2(t-t_0)^2\ln( a_0(t-t_0))\right)^{1/n}\end{eqnarray}
for $t>t_0$.  Then by (\ref{eq:  phiforu=xe^x^2r=1cnegative}) with $\alpha_0=2(d-1)/n\kappa$, we obtain scalar field
\begin{eqnarray}
\phi(t)&=&\psi(\sigma(t))\notag\\
&=&\sqrt{ \frac{ (d-1)}{n\kappa} } \left(\sqrt{2 \ln^2( a_0(t-t_0)) +1} + \ln\left[2\ln( a_0(t-t_0))\right]\right.\notag\\
&& \left.- \ln\left[4+ 4\sqrt{2 \ln^2( a_0(t-t_0))   + 1} \right] \right)+\beta_0
\end{eqnarray}
for $\beta_0\in\mathds{R}$.  Finally, by (\ref{eq:  nomatterFRLWNLSVphi-u}), (\ref{eq: uprimesigmaforu=e^x^2/xr=1cnegative}) and (\ref{eq: dotsigmausigmaforu=e^x^2/xr=1cnegative}), we have that
\begin{eqnarray}
V(\phi(t))&=&\left[\frac{2d(d-1)}{n^2\kappa}(u')^2-\frac{1}{2}u^2(\psi ')^2-Du^2-\frac{\Lambda}{\kappa}\right]\circ\sigma(t)\notag\\
&=&\frac{(d-1)}{2n\kappa(t-t_0)^2 \ln^2(a_0(t-t_0))} \cdot\notag\\
&& \quad \left(-2\ln^2(a_0(t-t_0))-1\right.\notag\\
&&\quad \left.+\frac{d}{n}\left( \left( 2\ln( a_0(t-t_0))+1\right)^2 -2\ln(a_0(t-t_0)) \right) \right)-\frac{\Lambda}{\kappa}\notag\\
\end{eqnarray}

For $a_0=1, n=3$ and $t_0=0$, the solver was run with $u, u'$ and $\sigma$ perturbed by $.001$.  The graphs of $a(t)$ below show that the solution is unstable.  The absolute error grows three orders of magnitude over the graphed time interval.  Since $u$ is independent of $n$, $a(t)$ is unstable for all values of $n$.

\begin{figure}[htbp]
\centering
\vspace{-.1in}
\caption{Instability of FRLW Example \theexample}
\vspace{-.1in}
\includegraphics[width=4in]{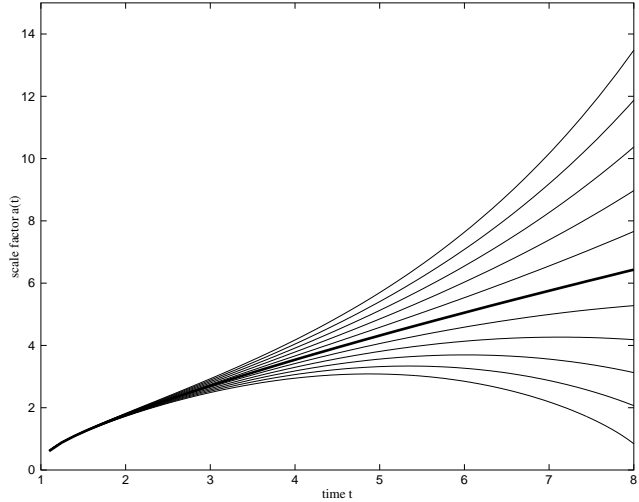}
\end{figure}

\end{example}
\break
\begin{example}
For $C\stackrel{def.}{=} \frac{n^2\kappa D}{2d(d-1)} $ and $\theta=1$, we take solution 6 in Table \ref{tb: exactNLS} with $c_0=2b_0^2-C$ so that we have $u(\sigma) = - a_0\cosh^2(b_0\sigma)$, $P(\sigma)=2b_0^2\tanh^2(b_0\sigma)+c_0=2b_0^2\tanh^2(b_0\sigma)+2b_0^2-C$ and $E=-C$.  By (\ref{eq: sigmaeqnforu=cosh^2r=1}) - (\ref{eq: dotsigmaforu=cosh^2r=1}) we obtain $\sigma(t)=\frac{-1}{b_0}Arctanh\left(a_0b_0(t-t_0)\right)$ and 
\begin{equation}a(t)=\frac{1}{u(\sigma(t))^{2/n}}=\left(\frac{1}{a_0}\left( a_0^2b_0^2(t-t_0)^2 - 1\right)\right)^{2/n}
\end{equation}
for $t_0\in\mathds{R}$.  Then by (\ref{eq: phiforu=cosh^2r=1}) with $\alpha_0=2(d-1)/n\kappa$, the scalar field is
\begin{eqnarray}
\phi(t)&=&\psi(\sigma(t)) \notag\\
&=& \pm\left( \sqrt{\frac{4(d-1)}{n\kappa}} \ln\left[ 2 \left(   \sqrt{  2 b_0^4  a_0^2(t-t_0)^2 +2b_0^2- C }  -  \sqrt{2} b_0^2a_0(t-t_0)\right)\right]\right.\notag\\
&& \quad \left.+ \frac{1}{b_0}\sqrt{\frac{2(d-1)}{n\kappa}} \sqrt{4b_0^2-C} Arctanh\left[  \frac{ \sqrt{ 4b_0^2 - C} a_0b_0(t-t_0) }{ \sqrt{ 2b_0^4a_0^2(t-t_0)^2+2b_0^2-C} } \right]\right) + \beta_0 \notag
\label{eq: phiforFRLWNLSposcurvu=cosh^2r=1andc=0}\end{eqnarray}
for $\beta_0\in\mathds{R}$.  Finally, by (\ref{eq:  nomatterFRLWNLSVphi-u}), (\ref{eq: dotsigmaforu=cosh^2r=1}) and (\ref{eq: uprimesigmaforu=cosh^2r=1}), we have that
\begin{eqnarray}
V(\phi(t))&=&\left[\frac{2d(d-1)(u')^2}{\kappa n^2 } -\frac{1}{2}u^2(\psi ')^2-Du^2-\frac{\Lambda}{\kappa}\right]\circ\sigma(t).\notag\\
&=&\frac{a_0^2(d-1)\left[  (4d - n) 2 a_0^2   b_0^4 (t-t_0)^2 -2nb_0^2+ C(n -2d)  \right]}{\kappa n^2 (a_0^2b_0^2(t-t_0)^2-1)^2}    -\frac{\Lambda}{\kappa}.\notag\\
\end{eqnarray}
One can compare this to the solutions in \cite{Lima} and in section 5 of \cite{GarcCatCamp}.
\end{example}

\subsection{A nonlinear Schr\"odinger example}

\begin{example}
For $D_j=\uprho=\mathrm{p}=0$ and $D_i=0$ for all $i$, we take  $n_j=4, \theta=1$ and positive curvature $k=1$.  We use solution 2 in Table \ref{tb: exactNLS} with $a_0=1/b_0^2$ and $b_0>0$ so that we have $u(\sigma)=\frac{1}{b_0^2}\cos^2(b_0\sigma)$.  Using the second potential for solution 2 in the table, we have $P(\sigma)=4b_0^2\tan^2(b_0\sigma)$ and $F_1= - 2$.  By (\ref{eq: sigmaeqnforu=cos^2r=1})-(\ref{eq: dotsigmaforu=cos^2r=1}) we have $\sigma(t)=\frac{1}{b_0}Arctan\left(\frac{\theta}{b_0}(t-t_0)\right)$ and 
\begin{equation}a(t)=\frac{1}{u(\sigma(t))^{1/2}}=\sqrt{b_0^2+(t-t_0)^2}\end{equation}
for $t>t_0\in\mathds{R}$.  Then by (\ref{eq: phiforu=cos^2P=tan}) with $\alpha_0=(d-1)/2\kappa$, we obtain scalar field
\begin{eqnarray}
\phi(t)&=&\sqrt{\frac{(d-1)}{2\kappa}}\ln\left(  1+ \frac{1}{b_0^2}(t-t_0)^2 \right)+\beta_0
\end{eqnarray}
for $\beta_0\in\mathds{R}$.  Finally, by (\ref{eq: FRLWNLSaltVphi-u}), (\ref{eq: dotsigmaforu=cos^2r=1}) and (\ref{eq: uprimesigmaforu=cos^2r=1}) we obtain
\begin{eqnarray}
V(\phi(t))
&=&\left[\frac{d(d-1)}{2\kappa}\left(\frac{1}{4}(u')^2+u\right)-\frac{1}{2}u^2(\psi ')^2-\frac{\Lambda}{\kappa}\right]\circ\sigma(t)\notag\\
&=&\frac{(d-1)}{2\kappa}\left( \frac{db_0^2+2(d-1)(t-t_0)^2}{(b_0^2+(t-t_0)^2)^2}\right)-\frac{\Lambda}{\kappa}.\notag\\
\end{eqnarray}
Composing $V(\phi(t))$ with the inverse
\begin{equation}\phi^{-1}(w)= b_0\sqrt{e^{\sqrt{2\kappa/(d-1)}(w-\beta_0)}-1} +t_0\end{equation} 
for $w\geq \beta_0$, we obtain the potential
\begin{equation}V(w)
=C_1e^{-\sqrt{2\kappa/(d-1)}w}-C_2e^{-2\sqrt{2\kappa/(d-1)}w}-\frac{\Lambda}{\kappa}
\end{equation}
for constants
\begin{equation}C_1=\frac{(d-1)^2}{\kappa b_0^2}e^{\sqrt{2\kappa/(d-1)}\beta_0} \ \ \mbox{ and } \ \ C_2=\frac{(d-1)(d-2)}{2\kappa b_0^2}e^{2\sqrt{2\kappa/(d-1)}\beta_0}.\end{equation}
By taking $d=3, t_0=0$, and identifying $a_0, \kappa$ and $\beta_0$ here with $b_0, K^2$ and $\phi_0$ respectively in the Ozer and Taha paper, we obtain the string-inspired solution I of \cite{OT}.  One can check that the conditions on the constants $C_1, C_2$ in \cite{OT} (with $d=3$) agree with the example here since we have
\begin{equation}b_0^2=\frac{2(d-1)^3C_2}{(d-2)\kappa C_1^2} \ \ \mbox{ and } \ \ \beta_0=\sqrt{\frac{(d-1)}{2\kappa}}\ln\left(\frac{2(d-1)C_2}{(d-2)C_1}\right).\end{equation}

For $b_0=1$ and $t_0=0$, the solver was run with $u, u'$ and $\sigma$ perturbed by $.001$.  The graphs of $a(t)$ below show that the solution is unstable.  The absolute error grows by three orders of magnitude over the graphed time interval.

\begin{figure}[htbp]
\centering
\vspace{-.1in}
\caption{Instability of FRLW Example \theexample}
\vspace{-.1in}
\includegraphics[width=4in]{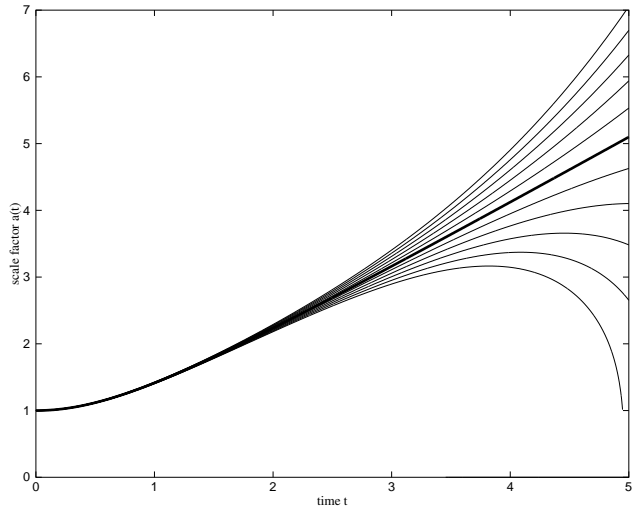}
\end{figure}

\end{example}

\chapter{Reformulations of a Bianchi I model}
\label{ch: BI}

For the homogeneous, anisotropic Bianchi I metric
\begin{equation} ds^2=-dt^2+X_1^2(t)dx_1^2+\cdots+X_{d}^2(t)dx_{d}^2\label{eq: BImetric}
\end{equation}
in a $d+1$-dimensional spacetime for $d\neq 0,1$, the nonzero Einstein equations $g^{ij}G_{ij}=-\kappa g^{ij}T_{ij}+\Lambda$ are
\begin{equation}
\displaystyle\sum_{l< k}H_lH_k
\stackrel{(I_0)}{=}
\kappa\left[\frac{1}{2}\dot\phi^2+V\circ\phi+\rho\right]+\Lambda\label{eq: BIEFEI1Id}\end{equation}
\begin{eqnarray}\displaystyle\sum_{l\neq 1}(\dot{H}_l+H_l^2)+\displaystyle\sum_{\stackrel{l < k}{l,k\neq 1}}H_lH_k&\stackrel{(I_1)}{=}&-\kappa\left[\frac{1}{2}\dot\phi^2-V\circ\phi+ p\right]+\Lambda\notag\\
&\vdots&\notag\\
\displaystyle\sum_{l\neq i}(\dot{H}_l+H_l^2)+\displaystyle\sum_{\stackrel{l < k}{l,k\neq i}}H_lH_k&\stackrel{(I_i)}{=}&-\kappa\left[\frac{1}{2}\dot\phi^2-V\circ\phi+ p\right]+\Lambda\notag\\
&\vdots&\notag\\
\displaystyle\sum_{l\neq d}(\dot{H}_l+H_l^2)+\displaystyle\sum_{\stackrel{l < k}{l,k\neq d}}H_lH_k&\stackrel{(I_d)}{=}&-\kappa\left[\frac{1}{2}\dot\phi^2-V\circ\phi+ p\right]+\Lambda\notag\end{eqnarray}
where $H_l(t)\stackrel{def.}{=}\dot{a}_l/a_l$ and $i,l,k\in\{1,\dots,d\}$.

\section{In terms of a Generalized EMP}
\begin{thm}\label{thm: BIEMP}
Suppose you are given twice differentiable functions $X_1(t), \dots, X_d(t)>0$, a once differentiable function $\phi(t)$, and also functions $\rho(t), p(t), V(x)$ which satisfy the Einstein equations $(I_0),\dots,(I_d)$ in (\ref{eq: BIEFEI1Id}) for some $\Lambda\in\mathds{R}, d\in\mathds{N}\backslash\{0,1\}, \kappa\in\mathds{R}\backslash\{0\}$.  Denote 
\begin{equation}R(t)\stackrel{def.}{=} \left(X_1(t)\cdots X_d(t)\right)^\nu\label{eq: BIEMPR-X}\end{equation}
for some $\nu\neq 0$.  If $f(\tau)$ is the inverse of a function $\tau(t)$ which satisfies
\begin{equation}\dot{\tau}(t)=\theta R(t)^q\label{eq: BIEMPdottau-X}\end{equation}
for some constants $\theta> 0$ and $q\neq 0$, then
\begin{equation}Y(\tau)= R(f(\tau))^q  \qquad\mbox{ and }\qquad Q(\tau)= \frac{q\nu d\kappa}{(d-1)} \varphi'(\tau)^2\label{eq: BIEMPY-aQ-varphi}\end{equation}
solve the generalized EMP equation
\begin{equation}Y''(\tau)+Q(\tau)Y(\tau)=\frac{-2q\nu d\kappa D}{\theta^2(d-1) Y(\tau)^{(2+q\nu)/q\nu}}-\frac{q\nu d\kappa\left(\varrho(\tau)+\textup{\textlhookp}(\tau)\right)}{\theta^2(d-1)Y(\tau)}\label{eq: BIEMP}\end{equation}
for
\begin{equation}\varphi(\tau)=\phi(f(\tau))\label{eq: BIEMPvarphi-phi}\end{equation}
\begin{equation}\varrho(\tau)=\rho(f(\tau)), \ \textup{\textlhookp}(\tau)=p(f(\tau)) \label{eq: BIEMPvarrho-rhohookp-p}\end{equation}
and where 
\begin{equation}D\stackrel{def.}{=}\frac{1}{2d\kappa}X_1^2X_2^2\cdots X_d^2\left(\displaystyle\sum_{l<k}\eta_{lk}^2\right)\label{eq: BIEMPDconstant-X}\end{equation}
is a constant for
\begin{equation}\eta_{lk}\stackrel{def.}{=}H_l-H_k, \mbox{\small \    $l\neq k, l, \ k\in\{1,\dots, d\}$.}\label{eq: BIEMPeta-X}\end{equation}

Conversely, suppose you are given a twice differentiable function $Y(\tau)>0$, a continuous function $Q(\tau)$, and also functions  $\varrho(\tau), \textup{\textlhookp}(\tau)$ which solve (\ref{eq: BIEMP}) for some constants $\theta>0$ and $q, \nu, \kappa \in\mathds{R}\backslash\{0\}, k, D\in\mathds{R}, d\in\mathds{N}\backslash\{0,1\}$.  In order to construct functions which solve $(I_0),\dots, (I_d)$, first find $\tau(t), \varphi(\tau)$ which solve the differential equations
\begin{equation}\dot{\tau}(t)=\theta Y(\tau(t))\qquad\mbox{ and }\qquad \varphi '(\tau)^2= \frac{(d-1)}{q\nu d\kappa} Q(\tau).\label{eq: BIEMPdottau-Yvarphi-Q}\end{equation}
Next find a function $\sigma(t)$ such that 
\begin{equation}\dot\sigma(t)=\frac{1}{\dot\tau(t)^{1/q\nu}}\label{eq: BIEMPdotsigma-dottau}\end{equation}
and let 
\begin{equation}R(t)=Y(\tau(t))^{1/q}\qquad\qquad\alpha_l(t)\stackrel{def.}{=}c_l\sigma(t), \ \mbox{\small $l\in\{1,\dots,d\}$}\label{eq: BIEMPR-Yalpha-sigma}\end{equation} 
where $c_l$ are any constants for which both
\begin{equation}\displaystyle\sum_{l=1}^d c_l=0\qquad\mbox{ and }\qquad\qquad\displaystyle\sum_{l<k}c_lc_k=-\theta^{2/q\nu} D\kappa.\label{eq: BIEMPccondition}\end{equation}
Then the functions
\begin{equation}X_l(t)=R(t)^{1/\nu d}e^{\alpha_l(t)}\label{eq: BIEMPX-Ralpha}\end{equation}
\begin{equation}\phi(t)=\varphi(\tau(t))\label{eq:   BIEMPphi-varphi}\end{equation}
\begin{equation}\rho(t)=\varrho(\tau(t)), \ p(t)=\textup{\textlhookp}(\tau(t))\label{eq: BIEMPrho-varrhop-hookp}\end{equation}
and
\begin{equation}
V(\phi(t))=\left[\frac{(d-1)\theta^2}{2\nu^2 d\kappa q^2}(Y')^2-\frac{D}{Y^{2/q\nu}}-\frac{\theta^2}{2}Y^2(\varphi ')^2-\varrho-\frac{\Lambda}{\kappa}\right]\circ\tau(t)\label{eq: BIEMPVphi-Y}
\end{equation}
satisfy the Einstein equations $(I_0),\dots,(I_d)$.
\end{thm}

\begin{proof}
This proof will implement Theorem \ref{thm: EFE-EMP} with constants and functions as indicated in the following table.

   \begin{table}[ht]
\centering
\caption{{ Theorem \ref{thm: EFE-EMP} applied to Bianchi I}}\label{tb: BIEMP}
\vspace{.2in}
\begin{tabular}{r | l c r | l}
In Theorem & substitute & & In Theorem& substitute \\[4pt]
\hline
\raisebox{-5pt}{$a(t)$} & \raisebox{-5pt}{$R(t)$} && \raisebox{-5pt}{$N$ }& \raisebox{-5pt}{$1$}\\[8pt]
$\delta  $ &$0$      &&    $\varepsilon $ & ${\nu d\kappa}/{(d-1)}$\\[8pt]
$G_0(t)$ & $\mbox{constant } -2\nu d\kappa D /(d-1)$    &&     $A_0$ & $2/\nu$ \\[8pt]
 $G_{1}(t)$ & $ \frac{-\nu d\kappa}{(d-1)}(\rho(t)+p(t))$  &&      $A_{1}$   &$0$\\[8pt]
 $ \lambda_0(\tau)$ &constant ${-2q\nu d \kappa D}/{\theta^2(d-1)}$       &&      $B_0$ &$(2+q\nu)/q\nu$ \\[8pt] 
$\lambda_{1}(\tau)$&$\frac{-q\nu d\kappa}{\theta^2(d-1)}(\varrho(\tau)+\textup{\textlhookp}(\tau))$ &&        $B_{1}$&  $1$ \\[6pt]
\hline
\end{tabular}
\end{table}

To prove the forward implication, we assume to be given functions which solve the Einstein field equations $(I_0),\dots,(I_d)$.  Forming the linear combination $d(I_0)-\displaystyle\sum_{i=1}^d(I_i)$ of Einstein's equations, we obtain
\begin{equation}
d\displaystyle\sum_{l < k}H_lH_k-\displaystyle\sum_{i=1}^d \displaystyle\sum_{l\neq i}(\dot{H}_l+H_l^2)-\displaystyle\sum_{i=1}^d\displaystyle\sum_{\stackrel{l<k}{l,k\neq i}}H_lH_k=d\kappa\left[\dot\phi^2+ \left(\rho+ p\right)\right]\label{eq: BIEMPdI1minusIi}
\end{equation}
where $l, k\in\{1, \dots, d\}$.  The second double sum on the left-hand side of (\ref{eq: BIEMPdI1minusIi}) contains the quantity $(\dot{H}_l+H_l^2) \ \ $  $(d-1)$-times for any fixed $l$, and the third double sum contains the quantity $H_lH_k \ \ $ $(d-2)$-times for any fixed $l,k$ pair with $l<k$ so that we have
\begin{equation}d\displaystyle\sum_{l<k}H_lH_k-(d-1)\displaystyle\sum_{l=1}^d(\dot{H}_l+H_l^2)-(d-2)\displaystyle\sum_{l<k}H_lH_k=d\kappa\left[\dot\phi^2+ \left(\rho+p\right)\right].\end{equation}
Collecting the first and third sums gives the equation
\begin{equation}2\displaystyle\sum_{l<k}H_lH_k-(d-1)\displaystyle\sum_{l=1}^d(\dot{H}_l+H_l^2)=d\kappa\left[\dot\phi^2+\left(\rho+p\right)\right].\label{eq: BIEMPI1minusIicollected}\end{equation}
Using the definition (\ref{eq: BIEMPR-X}) of $R(t)$, we define
\begin{equation}H_R\stackrel{def.}{=}\frac{\dot{R}}{R}=\frac{\nu \left(X_1\cdots X_d\right)^{\nu-1}\left(\dot{X}_1X_2\cdots X_d + \cdots + X_1X_2\cdots \dot{X}_d\right)}{\left(X_1\cdots X_d\right)^\nu}=\nu\displaystyle\sum_{l=1}^d H_l.\label{eq: BIEMPH-X}\end{equation}
Differentiating $H_R$ shows that
\begin{equation}\dot{H}_R=\nu\displaystyle\sum_{l=1}^d\dot{H}_l,\label{eq: BIEMPdotH-X}\end{equation}
therefore (\ref{eq: BIEMPI1minusIicollected}) can be written as
\begin{equation}2\displaystyle\sum_{l<k}H_lH_k-(d-1)\left(\frac{1}{\nu}\dot{H}_R+\displaystyle\sum_{l=1}^dH_l^2
\right)=d\kappa\left[\dot\phi^2+\left(\rho+p\right)\right].\end{equation}
Multiplying this by $\frac{-\nu}{(d-1)}$ and rearranging, we find that
\begin{equation}\dot{H}_R+\frac{\nu}{(d-1)}\left((d-1)\displaystyle\sum_{l=1}^dH_l^2-2\displaystyle\sum_{l<k}H_lH_k\right)=\frac{-\nu d\kappa}{(d-1)}\left[\dot{\phi}^2+\left(\rho+p\right)
\right].\label{eq: BIEMPI1minusIiwithdotH}\end{equation}
Using the definition (\ref{eq: BIEMPeta-X}) of the quantities $\eta_{lk}$, we have that 
\begin{equation}
\displaystyle\sum_{l<k}\eta_{lk}^2=\displaystyle\sum_{l<k}\left(H_l^2-2H_lH_k+H_k^2\right).\label{eq: BIEMPetalksqrddefn}\end{equation}
The first and last terms on the right-hand side of (\ref{eq: BIEMPetalksqrddefn}) sum to
\begin{eqnarray}
\displaystyle\sum_{l<k}\left(H_l^2+H_k^2\right)
&=&\displaystyle\sum_{l=1}^{d-1}\displaystyle\sum_{k=l+1}^dH_l^2
+\displaystyle\sum_{k=2}^d\displaystyle\sum_{l=1}^{k-1}H_k^2\notag\\
&=&\sum_{l=1}^{d-1}(d-l)H_l^2+\sum_{k=2}^d(k-1)H_k^2\notag\\
&=&(d-1)H_1^2+ \sum_{j=2}^{d-1}(d-j)H_j^2 +\sum_{j=2}^{d-1}(j-1)H_j^2+(d-1)H_d^2\notag\\
&=&(d-1)\sum_{j=1}^dH_j^2,\end{eqnarray}
therefore (\ref{eq: BIEMPetalksqrddefn}) becomes
\begin{equation}
\displaystyle\sum_{l<k}\eta_{lk}^2
=(d-1)\displaystyle\sum_{l=1}^dH_l^2-2\displaystyle\sum_{l<k}H_lH_k.
\end{equation}
Using this to rewrite (\ref{eq: BIEMPI1minusIiwithdotH}), we obtain
\begin{equation}\dot{H}_R+\frac{\nu}{(d-1)}\displaystyle\sum_{l<k}\eta_{lk}^2=\frac{-\nu d\kappa}{(d-1)}\left[\dot{\phi}^2+ \left(\rho+p\right)\right].\label{eq: BIEMPI1minusIiwitheta}
\end{equation}

Next we will confirm that $D$ is a constant.  Since the right-hand sides of Einstein equations $(I_i)$ are the same for all $i\in\{1, \dots, d\}$, by equating the left-hand sides of any two equations $(I_i)$ and $(I_j)$ for $i\neq j$, we get that
\begin{equation}\displaystyle\sum_{l\neq i}(\dot{H}_l+H_l^2)+\displaystyle\sum_{\stackrel{l<k}{l,k\neq i}}H_lH_k=\displaystyle\sum_{l\neq j}(\dot{H}_l+H_l^2)+\displaystyle\sum_{\stackrel{l<k}{l,k\neq j}}H_lH_k\label{eq: BIEMPlhsIi=lhsIj}\end{equation}
where we recall that the sum indices $l,k\in\{1, \dots, d\}$.  For the first sum on each side of (\ref{eq: BIEMPlhsIi=lhsIj}), the left and the right-hand sides of (\ref{eq: BIEMPlhsIi=lhsIj}) contain all the same terms, except for the $j^{th}$ indexed term which appears on the left, and the $i^{th}$ indexed term which appears on the right.  Therefore many terms cancel and we are left with
\begin{equation}\dot{H}_j+H_j^2+\displaystyle\sum_{\stackrel{l<k}{l,k\neq i}}H_lH_k=\dot{H}_i+H_i^2+\displaystyle\sum_{\stackrel{l<k}{l,k\neq j}}H_lH_k.\label{eq: BIEMPlhsIi=lhsIjfirstsumcancelled}\end{equation}
For the second (double) sum on each side of (\ref{eq: BIEMPlhsIi=lhsIj}), the left and right-hand sides of (\ref{eq: BIEMPlhsIi=lhsIjfirstsumcancelled}) contain all the same terms, except for the terms where either $l,k=j$ which appear on the left, and the terms where either $l,k=i$ which appear on the right.  Therefore many terms cancel and by adding $H_jH_i$ to both sides we obtain 
\begin{equation}\dot{H}_j+H_j^2+H_j\displaystyle\sum_{{l\neq j}}H_l=\dot{H}_i+H_i^2+H_i\displaystyle\sum_{{l\neq i}}H_l,\label{eq: BIEMPlhsIi=lhsIjsumscancelled}\end{equation}
or equivalently 
\begin{equation}\dot\eta_{ij}+\frac{1}{\nu}\eta_{ij}H_R=0\label{eq: BIEMPdotf+fH=0withetaH}\end{equation}
where we have used the expression (\ref{eq: BIEMPH-X}) for $H_R$, the definition (\ref{eq: BIEMPeta-X}) of $\eta_{ij}$, and as usual dot denotes differentiation with respect to $t$.   By Lemma \ref{lem: shortlemma} with $\mu=1/\nu\neq 0$ (which applies since $R(t)$ is positive and differentiable), (\ref{eq: BIEMPdotf+fH=0withetaH}) shows that the function $f=\eta_{ij}R^{1/\nu}=\eta_{ij}X_1X_2\cdots X_d$ is constant for any pair $i,j$ (for the pair $i=j$, $f$ is clearly a constant function, namely zero).  Therefore the definition (\ref{eq: BIEMPDconstant-X}) of $D$ is also constant, being proportional to a sum of squares of these constant functions.  By the definitions (\ref{eq: BIEMPDconstant-X}) and (\ref{eq: BIEMPR-X}) of the constant $D$ and the function $R(t)$, we now rewrite (\ref{eq: BIEMPI1minusIiwitheta}) as
\begin{equation}\dot{H}_R=\frac{-\nu d \kappa}{(d-1)}\left[\dot{\phi}^2+\frac{2D}{R^{2/\nu}}+\left(\rho+p\right)
\right].\label{eq: BIEMPI1minusIiEFE}\end{equation}
This shows that $R(t), \phi(t), \rho(t)$ and $p(t)$ satisfy the hypothesis of Theorem \ref{thm: EFE-EMP}, applied with  constants $\epsilon, \varepsilon , N, A_0, \dots, A_{N}$ and functions $a(t), G_0(t),\dots,G_{N}(t)$ according to Table \ref{tb: BIEMP}.  Since $\tau(t), Y(\tau), Q(\tau)$ and $\varphi(\tau)$ defined in (\ref{eq: BIEMPdottau-X}), (\ref{eq: BIEMPY-aQ-varphi}) and (\ref{eq: BIEMPvarphi-phi})  are equivalent to that in the forward implication of Theorem \ref{thm: EFE-EMP}, by this theorem and by definition (\ref{eq: BIEMPvarrho-rhohookp-p}) of $\varrho(\tau), \textup{\textlhookp}(\tau)$, the generalized EMP equation (\ref{eq: AEMP}) holds for constants $B_0, \dots, B_{N}$ and functions $\lambda_0(\tau), \dots, \lambda_{N}(\tau)$ as indicated in Table \ref{tb: BIEMP}.  This proves the forward implication.  

Note that equation (\ref{eq: BIEMPI1minusIiEFE}) with $\nu=1/d$ is the same as the FRLW analogue equation  (\ref{eq: FRLWEMPiiminusi}) with $M=1$, $n_1=2d$, $k=0$ and by identifying $R(t)$ here with $a(t)$ in the FRLW model.   One can compare this observation with J. Lidsey's results on the $3+1$-dimensional Bianchi I model in \cite{L}.

To prove the converse implication, we assume to be given functions which solve the generalized EMP equation (\ref{eq: BIEMP}) and we begin by showing that $(I_0)$ is satisfied.   Differentiating the definition of $R(t)$ in (\ref{eq: BIEMPR-Yalpha-sigma}) and using the definition in (\ref{eq: BIEMPdottau-Yvarphi-Q}) of $\tau(t)$, we have
\begin{eqnarray}\dot{R}(t)&=&\frac{1}{q}Y(\tau(t))^{\frac{1}{q}-1}Y'(\tau(t))\dot{\tau}(t)\notag\\
&=&\frac{\theta}{q}Y(\tau(t))^{1/q}Y'(\tau(t)).\end{eqnarray}
so that
\begin{equation}H_R(t)\stackrel{def.}{=}\frac{\dot{R}(t)}{R(t)}=\frac{\theta}{q}Y'(\tau(t)).\label{eq: BIEMPH-Y}\end{equation}
Differentiating the definition (\ref{eq: BIEMPphi-varphi}) of $\phi(t)$ and using the definition (\ref{eq: BIEMPdottau-Yvarphi-Q}) of  $\tau(t)$ gives
\begin{equation}\dot{\phi}(t)=\varphi '(\tau(t))\dot{\tau}(t)=\theta\varphi '(\tau(t))Y(\tau(t)).\label{eq: BIEMPdotphi-Y}\end{equation}
Using (\ref{eq: BIEMPH-Y}) and (\ref{eq: BIEMPdotphi-Y}), and also the definitions (\ref{eq: BIEMPR-Yalpha-sigma}) and (\ref{eq: BIEMPrho-varrhop-hookp}) of $R(t)$ and $\rho(t)$ respectively, the definition (\ref{eq: BIEMPVphi-Y}) of $V\circ\phi$ can be written as
\begin{equation}
V\circ\phi=\frac{1}{\kappa}\left(
\frac{(d-1)}{2\nu^2 d}H_R^2-\frac{D\kappa}{R^{2/\nu}}
\right)-\frac{1}{2}\dot\phi^2-\rho-\frac{\Lambda}{\kappa}
\label{eq: BIEMPVphi-a}.\end{equation}

\noindent The quantity in parenthesis here is in fact equal to the left-hand-side of equation $(I_0)$.  To see this, we differentiate the definition (\ref{eq: BIEMPX-Ralpha}) of $X_l(t)$, divide the result by $X_l$ and use the definition (\ref{eq: BIEMPH-Y}) of $H_R$ to obtain
\begin{equation}H_l\stackrel{def.}{=}\frac{\dot{X}_l}{X_l}=\frac{\frac{1}{\nu d}R^{{1}/{\nu d}-1}\dot{R}e^{\alpha_l}+\dot{\alpha}_lR^{{1}/{\nu d}}e^{\alpha_l}}{R^{{1}/{\nu d}}e^{\alpha_l}}=\frac{1}{\nu d}H_R+\dot\alpha_l\label{eq: BIEMPdotX/X-Halpha}\end{equation}
so that 
\begin{equation}\displaystyle\sum_{l<k}H_lH_k
=\displaystyle\sum_{l<k}\left(\frac{1}{\nu^2 d^2}H_R^2+\frac{1}{\nu d}(\dot\alpha_l+\dot\alpha_k)H_R+\dot\alpha_l \dot\alpha_k\right).
\label{eq: BIEMPdotXldotXk/XlXk-Halpha}\end{equation}
The first term on the right-hand side of (\ref{eq: BIEMPdotXldotXk/XlXk-Halpha}) does not depend on the indices $l,k$, and is therefore equal to $\frac{1}{\nu^2 d^2}H_R^2$ times the quantity
\begin{equation}\displaystyle\sum_{l<k}1=\displaystyle\sum_{k=2}^d\displaystyle\sum_{l=1}^{k-1}1=\displaystyle\sum_{k=2}^d(k-1)=\displaystyle\sum_{j=1}^{d-1}j=\frac{d(d-1)}{2}.\label{eq: BIEMPsum1l<k}\end{equation}
The second term on the right-hand side of (\ref{eq: BIEMPdotXldotXk/XlXk-Halpha}) sums to zero since 
\begin{eqnarray}
\displaystyle\sum_{l<k}(\dot\alpha_l+\dot\alpha_k)
&=&\displaystyle\sum_{l=1}^{d-1}\displaystyle\sum_{k=l+1}^d\dot\alpha_l+\displaystyle\sum_{k=2}^d\displaystyle\sum_{l=1}^{k-1}\dot\alpha_k\notag\\
&=&\displaystyle\sum_{l=1}^{d-1}(d-l)\dot\alpha_l +\displaystyle\sum_{k=2}^d (k-1)\dot\alpha_k\notag\\
&=&(d-1)\dot\alpha_1+\displaystyle\sum_{j=2}^{d-1}(d-j+j-1)\dot\alpha_j + (d-1)\dot\alpha_d\notag\\
&=&(d-1)\displaystyle\sum_{l=1}^d\dot\alpha_l\notag\\
&=&(d-1)\sigma(t)\displaystyle\sum_{l=1}^d c_l\notag\\
&=&0\label{eq: BIEMPalphal+alphak=0}\end{eqnarray}
where on the last lines, we have used the definition (\ref{eq: BIEMPR-Yalpha-sigma}) of $\alpha_l(t)$ and the condition (\ref{eq: BIEMPccondition}) on the constants $c_l$.  For the third term on the right-hand side of   (\ref{eq: BIEMPdotXldotXk/XlXk-Halpha}), we use the definitions of $\alpha_l(t), \sigma(t), \tau(t)$ and $R(t)$ in (\ref{eq: BIEMPR-Yalpha-sigma}), (\ref{eq: BIEMPdotsigma-dottau}) and (\ref{eq: BIEMPdottau-Yvarphi-Q}) to write
\begin{equation}\dot\alpha_l\dot\alpha_k=c_l c_k \dot\sigma^2=\frac{c_lc_k}{\dot\tau^{2/q\nu}}=\frac{c_lc_k}{\theta^{2/q\nu}(Y\circ\tau)^{2/q\nu}}=\frac{c_lc_k}{\theta^{2/q\nu}R^{2/\nu}}.\label{eq: BIEMPdotalphaldotalphak-R}\end{equation}
Therefore (\ref{eq: BIEMPdotXldotXk/XlXk-Halpha}) becomes
\begin{equation}\displaystyle\sum_{l<k}H_lH_k=\frac{(d-1)}{2\nu^2 d}H_R^2+\displaystyle\sum_{l<k}\frac{c_lc_k}{\theta^{2/q\nu}R^{2/\nu}}.
\label{eq: BIEMPdotXldotXk/XlXk-Rc}
\end{equation}
Then by the condition (\ref{eq: BIEMPccondition}) on the constants $c_l$, (\ref{eq: BIEMPdotXldotXk/XlXk-Rc}) becomes
\begin{equation}\displaystyle\sum_{l<k}H_lH_k=\frac{(d-1)}{2\nu^2 d}H_R^2-\frac{D\kappa}{R^{2/\nu}}.
\label{eq: BIEMPdotXldotXk/XlXk-RD}
\end{equation}
That is, the expression (\ref{eq: BIEMPVphi-a}) for $V$ can now be written as
\begin{equation}V\circ\phi=\frac{1}{\kappa}\displaystyle\sum_{l<k}H_lH_k-\frac{1}{2}\dot{\phi}^2-\rho-\frac{\Lambda}{\kappa},\label{eq: BIEMPVphi-X}
\end{equation}
showing that $(I_0)$ holds under the assumptions of the converse implication.

To conclude the proof we must also show that the equations $(I_1), \dots, (I_d)$ hold.   In the converse direction the hypothesis of the converse of Theorem \ref{thm: EFE-EMP} holds, applied with constants $N, B_0, \dots, B_{N}$ and functions $\lambda_0(\tau), \dots, \lambda_{N}(\tau)$ as indicated in Table \ref{tb: BIEMP}.  Since  $\tau(t), \varphi(\tau), R(t)$ and $\phi(t)$ defined in (\ref{eq: BIEMPdottau-Yvarphi-Q}), (\ref{eq: BIEMPR-Yalpha-sigma}) and (\ref{eq: BIEMPphi-varphi}) are consistent with the converse implication of Theorem \ref{thm: EFE-EMP}, applied with $a(t), \delta  $ and $\varepsilon $ as in Table \ref{tb: BIEMP}, by this theorem and by the definition (\ref{eq: BIEMPrho-varrhop-hookp}) of $\rho(t), p(t)$ the  scale factor  equation (\ref{eq: AEFE}) holds for constants $\delta  , \varepsilon , A_0, \dots, A_{N}$ and functions $G_0(t),\dots,G_{N}(t)$ according to Table \ref{tb: BIEMP}.  That is, we have regained (\ref{eq: BIEMPI1minusIiEFE}).  Now solving (\ref{eq: BIEMPVphi-a}) for $\rho(t)$ and substituting this into (\ref{eq: BIEMPI1minusIiEFE}), we obtain
\begin{equation}\dot{H}_R=\frac{-\nu d\kappa}{(d-1)}\left[ \frac{1}{2}\dot\phi^2 -V\circ\phi+\frac{D}{R^{2/\nu}}+p +\frac{(d-1)}{2\nu^2 d\kappa}H_R^2-\frac{\Lambda}{\kappa} \right].\end{equation}
Multiplying by $\frac{(d-1)}{\nu d}$ and rearranging, we get that
\begin{equation}\frac{(d-1)}{\nu d}\dot{H}_R+  \frac{(d-1)}{2\nu^2 d}H_R^2   +\frac{\kappa D}{R^{2/\nu}}=-\kappa\left[\frac{1}{2}\dot\phi^2-V\circ\phi+p\right] + \Lambda\label{eq: BIEMPlhs=lhsofIi}\end{equation}
The left-hand side of this equation is in fact equal to the left-hand-side of $(I_i)$ for any $i\in\{1,\dots, d\}$. 
  To see this, first recall (\ref{eq: BIEMPdotX/X-Halpha}) and write
  \begin{equation}\dot{H}_l+H_l^2=\frac{1}{\nu d}\dot{H}_R+\frac{1}{\nu^2 d^2}H_R^2+\ddot{\alpha}_l+\frac{2}{\nu d}\dot{\alpha}_lH_R+\dot{\alpha}_l^2,\label{eq: BIEMPddotX/X-Halpha}
\end{equation}
therefore for any fixed $i$
\begin{equation}\displaystyle\sum_{l\neq i}(\dot{H}_l+H_l^2)=\displaystyle\sum_{l\neq i}\left(\frac{1}{\nu d}\dot{H}_R+\frac{1}{\nu^2 d^2}H_R^2+\ddot{\alpha}_l+\frac{2}{\nu d}\dot{\alpha}_lH_R+\dot{\alpha}_l^2\right).\label{eq: BIEMPsumlneqiddotX/X}\end{equation}
Since the first two terms on the right-hand side of (\ref{eq: BIEMPsumlneqiddotX/X}) do not depend on the indices $l,k$, and
also using the definitions (\ref{eq: BIEMPR-Yalpha-sigma}) and (\ref{eq: BIEMPccondition}) of $\alpha_l$ and the constants $c_l$ to write $\sum_l \dot\alpha_l=0 \Rightarrow \sum_{l\neq i}\dot\alpha_l=-\dot\alpha_i$, (\ref{eq: BIEMPsumlneqiddotX/X}) becomes
\begin{equation}\displaystyle\sum_{l\neq i}(\dot{H}_l+H_l^2)=\frac{(d-1)}{\nu d}\dot{H}_R+\frac{(d-1)}{\nu^2 d^2}H_R^2- \frac{2}{\nu d}H_R\dot\alpha_i+\displaystyle\sum_{l\neq i}\left(\ddot\alpha_l+\dot\alpha_l^2\right) .
\label{eq: BIEMPddotX/X-H}\end{equation}
By the definitions (\ref{eq: BIEMPR-Yalpha-sigma}), (\ref{eq: BIEMPdotsigma-dottau}) and (\ref{eq: BIEMPdottau-Yvarphi-Q}) of $\alpha_l(t), \sigma(t), \tau(t)$ and $R(t)$, we see that
\begin{eqnarray}\dot\alpha_l(t)R(t)^{1/\nu}&=&c_l\dot\sigma(t)R(t)^{1/\nu}\notag\\
&=&\frac{c_l}{\dot\tau(t)^{1/q\nu}}R(t)^{1/\nu}\notag\\
&=&\frac{c_l}{\theta^{1/q\nu}Y(\tau(t))^{1/q\nu}}R(t)^{1/\nu}\notag\\
&=&\frac{c_l}{\theta^{1/q\nu}}\mbox{ is a constant.}\label{eq: BIEMPdotalphaRconstant}\end{eqnarray}
By Lemma \ref{lem: shortlemma} with $\mu=1/\nu\neq 0$, equation (\ref{eq: BIEMPdotalphaRconstant}) shows that 
\begin{equation}\ddot\alpha_l+\frac{1}{\nu}\dot\alpha_lH_R=0\end{equation}
for all $l\in\{1, \dots, d\}$.  Therefore in total, we have that (\ref{eq: BIEMPsumlneqiddotX/X}) is
\begin{equation}\displaystyle\sum_{l\neq i}(\dot{H}_l+H_l^2)=\frac{(d-1)}{\nu d}\dot{H}_R+\frac{(d-1)}{\nu^2 d^2}H_R^2+\frac{(2-d)}{\nu d}H_R\displaystyle\sum_{l\neq i}\dot\alpha_l+\displaystyle\sum_{l\neq i}\dot\alpha_l^2.
\label{eq: BIEMPddotX/X-alphaH}\end{equation}

\noindent To form the rest of the left-hand side of $(I_i)$, again use (\ref{eq: BIEMPdotX/X-Halpha}) to obtain
\begin{equation}H_lH_k=\frac{1}{\nu^2 d^2}H_R^2+\frac{1}{\nu d}(\dot\alpha_l+\dot\alpha_k)H_R+\dot\alpha_l\dot\alpha_k.\end{equation}
Therefore for any fixed $i$, we have
\begin{equation}\displaystyle\sum_{\stackrel{l<k}{l,k\neq i}}H_lH_k
=\displaystyle\sum_{\stackrel{l<k}{l,k\neq i}}\left(\frac{1}{\nu^2 d^2}H_R^2+\frac{1}{\nu d}(\dot\alpha_l+\dot\alpha_k)H_R+\dot\alpha_l\dot\alpha_k\right).\label{eq: BIEMPsuml<klkneqidotXldotXk/XlXk-alhpaH}\end{equation}
As we saw in (\ref{eq: BIEMPsum1l<k}), $\sum_{l<k}1=\frac{d(d-1)}{2}$ therefore the first term on the right-side of  (\ref{eq: BIEMPsuml<klkneqidotXldotXk/XlXk-alhpaH}), which does not depend on the indices $l,k$, is equal to $\frac{1}{\nu^2 d^2}H_R^2$ times
\begin{equation}\displaystyle\sum_{\stackrel{l<k}{l,k\neq i}}1=\displaystyle\sum_{l<k}1-\displaystyle\sum_{l\neq i}1=\frac{d(d-1)}{2}-(d-1)=\frac{(d-1)(d-2)}{2}.\label{eq: BIEMPsum1l<klkneqi}\end{equation}
As we saw in (\ref{eq: BIEMPalphal+alphak=0}), $\sum_{l<k}(\dot\alpha_l+\dot\alpha_k)=0$ therefore the second term on the right-hand side of (\ref{eq: BIEMPsuml<klkneqidotXldotXk/XlXk-alhpaH}) sums to $\frac{1}{\nu d}H_R$ times the quantity
\begin{eqnarray}
\sum_{\stackrel{l<k}{l,k\neq i}}(\dot\alpha_l+\dot\alpha_k)
&=&\sum_{{l<k}}(\dot\alpha_l+\dot\alpha_k)-\sum_{l\neq i}(\dot\alpha_l+\dot\alpha_i)\notag\\
&=&-\sum_{l\neq i}\dot\alpha_l-(d-1)\dot\alpha_i\notag\\
&=&\dot\alpha_i-(d-1)\dot\alpha_i\notag\\
&=&(2-d)\dot\alpha_i\label{eq: BIEMPsuml<klkneqidotalphal+dotalphak=2-ddotalphai}
\end{eqnarray}
where we have used the definitions
(\ref{eq: BIEMPR-Yalpha-sigma}) and (\ref{eq: BIEMPccondition}) of $\alpha_l$ and the constants $c_l$ to write
$\sum_{l\neq i}\dot\alpha_l=-\dot\alpha_i$.  Considering the third term on the right-hand side of (\ref{eq: BIEMPsuml<klkneqidotXldotXk/XlXk-alhpaH}), we have that
\begin{eqnarray}
\sum_{\stackrel{l<k}{l,k\neq i}}\dot\alpha_l\dot\alpha_k
&=&\sum_{l<k}\dot\alpha_l\dot\alpha_k - c_i\sum_{l\neq i}c_l\notag\\
&=&\sum_{l<k}\dot\alpha_l\dot\alpha_k + \left(\sum_{l\neq i}c_l\right)^2\notag\\
&=&\sum_{l<k}\dot\alpha_l\dot\alpha_k + 2\sum_{\stackrel{l<k}{l,k\neq i}}c_lc_k + \sum_{l\neq i}c_l^2\notag\\
&=&-\sum_{l<k}\dot\alpha_l\dot\alpha_k -\sum_{l\neq i}c_l^2\notag\\
&=&-\sum_{l<k}\frac{c_lc_k}{\theta^{2/q\nu}R^{2/\nu}}-\sum_{l\neq i}c_l^2
\label{eq: BIEMPsuml<klkneqidotalphaldotalphak=-suml<kdotalphaldotalphak}\end{eqnarray}
where again we have used that $\sum_{l\neq i}\dot\alpha_l=-\dot\alpha_i$,
 and on the last line we recall (\ref{eq: BIEMPdotalphaldotalphak-R}).  So by  (\ref{eq: BIEMPsum1l<klkneqi}), (\ref{eq: BIEMPsuml<klkneqidotalphal+dotalphak=2-ddotalphai}) and (\ref{eq: BIEMPsuml<klkneqidotalphaldotalphak=-suml<kdotalphaldotalphak}), in total (\ref{eq: BIEMPsuml<klkneqidotXldotXk/XlXk-alhpaH}) becomes
\begin{equation}\displaystyle\sum_{\stackrel{l<k}{l,k\neq i}}H_lH_k=\frac{(d-1)(d-2)}{2\nu^2 d^2}H_R^2+\frac{(2-d)}{\nu d}\dot\alpha_iH_R-\sum_{l<k}\frac{c_lc_k}{\theta^{2/q\nu}R^{2/\nu}}-\sum_{l\neq i}c_l^2.
\label{eq: BIEMPdotXldotXk/XlXk-Ralphac}\end{equation}
Summing (\ref{eq: BIEMPddotX/X-alphaH}) and (\ref{eq: BIEMPdotXldotXk/XlXk-Ralphac}), the left-hand side of any $(I_i)$ Einstein equation is
\begin{equation}\displaystyle\sum_{l\neq i}(\dot{H}_l+H_l^2)+\displaystyle\sum_{\stackrel{l<k}{l,k\neq i}}H_lH_k=\frac{(d-1)}{\nu d}\dot{H}_R+\frac{d(d-1)}{2\nu^2 d^2}H_R^2-\sum_{l<k}\frac{c_lc_k}{\theta^{2/q\nu}R^{2/\nu}}.
\label{eq: BIEMPBIlhsIi}\end{equation}
Then by the condition (\ref{eq: BIEMPccondition}) on the constants $c_l$, (\ref{eq: BIEMPdotXldotXk/XlXk-Ralpha}) becomes
\begin{equation}\displaystyle\sum_{l\neq i}(\dot{H}_l+H_l^2)+\displaystyle\sum_{\stackrel{l<k}{l,k\neq i}}H_lH_k=\frac{(d-1)}{\nu d}\dot{H}_R+\frac{d(d-1)}{2\nu^2 d^2}H_R^2+\frac{D\kappa}{R^{2/\nu}}.
\label{eq: BIEMPdotXldotXk/XlXk-Ralpha}\end{equation}
Therefore by (\ref{eq:  BIEMPlhs=lhsofIi}) and (\ref{eq: BIEMPdotXldotXk/XlXk-Ralpha}), we obtain $(I_i)$ for all $i\in\{1,\dots,d\}$.  This proves the theorem.
\end{proof}

\break 

\subsection{Reduction to classical EMP: pure scalar field}

To compute some exact solutions to Einstein's equations for the Bianchi I metric, we take $\rho=p=0$ and choose parameter $\nu=1/q$ in Theorem \ref{thm: BIEMP}.  Therefore solving the Bianchi I Einstein equations
\begin{equation}
\displaystyle\sum_{l< k}H_lH_k
\stackrel{(I_0)'}{=}
\kappa\left[\frac{1}{2}\dot\phi^2+V\circ\phi\right]+\Lambda\end{equation}
\begin{equation}
\displaystyle\sum_{l\neq i}(\dot{H}_l+H_l^2)+\displaystyle\sum_{\stackrel{l < k}{l,k\neq i}}H_lH_k\stackrel{(I_i)'}{=}-\kappa\left[\frac{1}{2}\dot\phi^2-V\circ\phi\right]+\Lambda\notag
\end{equation}
for $l,k,i\in\{1, \dots, d\}$ is equivalent to solving the classical EMP equation
\begin{equation}Y''(\tau)+Q(\tau)Y(\tau)=\frac{-2 d\kappa D}{\theta^2(d-1) Y(\tau)^3}\label{eq: nomatterBIEMP}\end{equation}
for constants $\theta, D>0$.  The solutions of $(I_0)', (I_1)',\dots,(I_d)'$ and (\ref{eq: nomatterBIEMP}) are related by
\begin{equation}R(t)=Y(\tau(t))^{1/q} \qquad\mbox{ and }\qquad \varphi '(\tau)^2= \frac{(d-1)}{d\kappa} Q(\tau)\label{eq: nomatterBIclassEMPa-Yvarphi-Q}\end{equation}
for $q\neq 0, \phi(t)=\varphi(\tau(t)), R(t)\stackrel{def.}{=} \left(X_1(t)\cdots X_d(t)\right)^{1/q}$ and 
\begin{equation}\dot\tau(t)=\theta R(t)^q=\theta Y(\tau(t))\label{eq: nomatterBIclassEMPdottau-a-Y}\end{equation}
for any $\theta>0$.  Also the quantity
\begin{equation}D\stackrel{def.}{=}\frac{1}{2d\kappa}X_1^2X_2^2\cdots X_d^2\left(\displaystyle\sum_{l<k}\eta_{lk}^2\right)\label{eq: nomatterBIEMPDconstant-X}\end{equation}
is constant for
\begin{equation}\eta_{lk}\stackrel{def.}{=}H_l-H_k, \mbox{\small \    $l\neq k, l, \ k\in\{1,\dots, d\}$.}\label{eq: nomatterBIEMPeta-X}\end{equation}
In the converse direction we have 
\begin{equation}X_l(t)=R(t)^{q/ d}e^{\alpha_l(t)}\label{eq: nomatterBIEMPX-Ralpha}\end{equation}
for $\alpha_l(t)\stackrel{def.}{=}c_l\sigma(t)$ where $\sigma(t)$ satisfies 
\begin{equation}\dot\sigma(t)=1/\dot\tau(t)\label{eq: nomatterBIEMPdotsigma-dottau}\end{equation}
 and $c_l$ are any constants for which both
\begin{equation}\displaystyle\sum_{l=1}^d c_l=0\qquad\mbox{ and }\qquad\qquad\displaystyle\sum_{l<k}c_lc_k=-\theta^{2/q\nu} D\kappa.\label{eq: BIEMPcconditionY=superpowerb=0r=1s=1}\end{equation}
Also $V$ is taken to be defined as
\begin{equation}
V(\phi(t))=\left[\frac{(d-1)\theta^2}{2 d\kappa }(Y')^2-\frac{D}{Y^2}-\frac{\theta^2}{2}Y^2(\varphi ')^2-\frac{\Lambda}{\kappa}\right]\circ\tau(t).\label{eq: nomatterBIEMPVphi-Y}
\end{equation}

We now refer to Appendix D for solutions of the classical  EMP equation (\ref{eq: nomatterFRLWEMP}), which we will map to a solution of the Bianchi I Einstein equations.  By comparing (\ref{eq: nomatterBIclassEMPdottau-a-Y}) and (\ref{eq: dottau-Ytau^s0}), we note to only consider solutions of (\ref{eq: dottau-Ytau^s0}) in Appendix D corresponding to $r_0=1$.  Also by comparing (\ref{eq: nomatterBIEMPdotsigma-dottau}) to (\ref{eq: dotsigma=1/dottau^s0}), we see to take $s_0=1$ in Appendix D.   Of course when $D=0$ in (\ref{eq: nomatterBIEMPDconstant-X}), $\eta_{lk}=0$ for all pairs $l,k$ so that  $X_1(t), X_2(t), \dots, X_d(t)$ agree up to a constant multiple. In this case, if we take $R(t)=X_1(t)=\cdots=X_d(t)\stackrel{def.}{=}a(t)$ and $q=d$ we obtain the FRLW cosmology with curvature  $k=0$ and $n=2d, D=0$ so that we may refer to sections 3.1.1 and 3.2.2 for exact solutions to the Bianchi I Einstein equations if $D=0$.  Here we consider solutions to the classical EMP (\ref{eq: nomatterBIEMP}) by referring to Table \ref{tb: exactEMP} with $\lambda_1 = \frac{-2d\kappa D}{\theta^2(d-1)}<0$.

\begin{example}
For $\theta=1$ and choice of constant $D=(d-1)/2d\kappa$, we consider solution 5 in Table \ref{tb: exactEMP} with $d_0=b_0=0$ and $c_0=1$.  That is, we have solution $Y(\tau)=(a_0+2\tau)^{1/2}$ to the classical EMP $Y''(\tau)+Q(\tau)Y(\tau)=-1/Y(\tau)^3$ for $Q(\tau)=0$ and $a_0\in\mathds{R}$.  By (\ref{eq: taueqnforY=superpowerb=0r=1}) - (\ref{eq: dottauforY=superpowerb=0r=1}) we have $\tau(t)=\frac{1}{2}\left((t-t_0)^2-a_0\right)$ and 
\begin{equation}R(t)=Y(\tau(t))^{1/q}=(t-t_0)^{1/q}\end{equation}
for any $q\neq 0$ and $t_0\in\mathds{R}$.  Also by (\ref{eq: sigmaforY=superpowerb=0s=1}) we have $\sigma(t)=\ln(t-t_0)$ so that 
\begin{eqnarray}X_i(t)
&=&R(t)^{q/d}e^{c_i \sigma(t)}\notag\\
&=&(t-t_0)^{1/d}e^{c_i \ln(t-t_0)}\notag\\
&=&(t-t_0)^{1/d+c_i}\label{eq: XiforBIclassicalEMPY=superpowerd=b=0}
\end{eqnarray}
for $1\leq i\leq d$ and for constants $c_1, \dots, c_d$ that satisfy
\begin{equation}\displaystyle\sum_{l=1}^d c_l=0\qquad\mbox{ and }\qquad\qquad\displaystyle\sum_{l<k}c_lc_k=- \frac{(d-1)}{2d} \label{eq: BIEMPcconditionY=superpowerb=d=0}\end{equation}
by (\ref{eq: BIEMPccondition}).  
Since $Q(\tau)=0=\varphi'(\tau)$ the scalar field
$\phi(t)\stackrel{def.}{=}\varphi(\tau(t)) = \phi_0\in\mathds{R}\label{eq: phiforBIclassicalEMPY=superpowerb=d=0}$
is constant.  Finally, by (\ref{eq: nomatterBIEMPVphi-Y}), (\ref{eq: YprimetauforY=superpowerb=d=0r=1}) and (\ref{eq: dottauforY=superpowerb=0r=1}), we obtain constant potential $V(\phi(t))=-\Lambda/\kappa$.

One can verify by hand that $X_1(t), \dots, X_d(t), \phi(t)$ and $V(\phi(t))$ in (\ref{eq: XiforBIclassicalEMPY=superpowerd=b=0}), (\ref{eq: phiforBIclassicalEMPY=superpowerb=d=0}) and (\ref{eq: VphiforBIclassicalEMPY=superpowerb=d=0}) satisfy the vacuum Bianchi I equations $(I_0)'$ and $(I_i)'$ for $1\leq i\leq d$ and  for $c_i$ as in (\ref{eq: BIEMPcconditionY=superpowerb=d=0}), with use of the identity in equation (\ref{eq: BIEMPsuml<klkneqidotalphaldotalphak=-suml<kdotalphaldotalphak}).  This example is the well-known Kasner solution $ds^2=-dt^2+\sum_{i=1}^d t^{2p_i}dx_i^2$, in which the constants $p_i$ must satisfy the Kasner conditions $\sum_{i=1}^d p_i=\sum_{i=1}^d p_i^2=1$.  By setting $p_i=1/d+c_i$, we see that our conditions (\ref{eq: BIEMPcconditionY=superpowerb=d=0}) are equivalent to the Kasner conditions since $\sum_{i=1}^d p_i=\sum_{i=1}^d(1/d+c_i)=1+\sum_{i=1}^dc_i=1+0=1$ and  $\sum_{i=1}^d p_i^2=\sum_{i=1}^d(1/d+c_i)^2=(1/d)+\sum_{i=1}^dc_i^2=(1/d)-2\sum_{l<k}c_lc_k=(1/d)+2((d-1)/2d)=1$.
\end{example}

\begin{example}
For $\theta=q=1$ and choice of constant $D=(d-1)/2d\kappa$, we consider solution 5 in Table \ref{tb: exactEMP} with $d_0=a_0=0$, $b_0=4$ and $c_0=1$.  That is, we have solution $Y(\tau)=(4\tau(t)^2+2\tau(t))^{1/2}$ to the classical EMP $Y''(\tau)+Q(\tau)Y(\tau)=-1/Y(\tau)^3$ for $Q(\tau)=0$.  By (\ref{eq: taueqnforY=superpowerr=1}) - (\ref{eq: dottauforY=superpowerr=1}) we have $\tau(t)=\frac{1}{8}\left(e^{2(t-t_0)} + e^{-2(t-t_0)}-8\right)$ and 
\begin{eqnarray}R(t)&=&Y(\tau(t))
=\frac{1}{2}\sinh(2(t-t_0))
\end{eqnarray}
for any  $t_0\in\mathds{R}$.  Also by (\ref{eq: sigmaforY=superpowerr=1s=1}) we have that 
\begin{eqnarray}\sigma(t)&=&
- 2  Arccoth\left( e^{2(t-t_0)}  \right)= \ln\left(tanh(t-t_0)\right)
\label{eq: sigmaforBIEMPY=superpowerr=1s=1}
\end{eqnarray}
so that 
\begin{eqnarray}
X_i(t)&=&R(t)^{1/d}e^{c_i\sigma(t)}=\frac{1}{2^{1/d}}\sinh(2(t-t_0))^{1/d}e^{c_i\sigma(t)}\notag\\
&=&\frac{1}{2^{1/d}}\sinh(2(t-t_0))^{1/d}tanh^{c_i}(t-t_0)
\end{eqnarray}
for $\sigma(t)$ as in (\ref{eq: sigmaforBIEMPY=superpowerr=1s=1}) and constants that satisfy
\begin{equation}\displaystyle\sum_{l=1}^d c_l=0\qquad\mbox{ and }\qquad\qquad\displaystyle\sum_{l<k}c_lc_k=-\frac{(d-1)}{2d}.\label{eq: BIEMPcconditionY=superpowera=d=0c=1r=s=1}\end{equation}
Since $Q(\tau)=0=\varphi'(\tau)$, the scalar field is constant
$\phi(t)\stackrel{def.}{=}\varphi(\tau(t)) = \phi_0\in\mathds{R}$.  Finally, by (\ref{eq: nomatterBIEMPVphi-Y}), (\ref{eq: YprimetauforY=superpowerd=0b>0r=1}) and (\ref{eq: dottauforY=superpowerr=1}), we obtain constant potential $V(\phi(t))=\frac{2(d-1)}{d \kappa } -\frac{\Lambda}{\kappa}$.  One can compare this solution with the higher-dimensional solution of Lorenz-Petzold in \cite{LP}.

As an example for $d=3$, one can take $c_1=-\frac{1}{\sqrt{3}}, c_2=0$ and $c_3=\frac{1}{\sqrt{3}}$ so that
\begin{eqnarray}X_1(t)&=&\frac{1}{\sqrt[3]{2}} \sinh(2(t-t_0))^{1/3}tanh^{-1/\sqrt{3}}(t-t_0)\notag\\
X_2(t)&=&\frac{1}{\sqrt[3]{2}} \sinh(2(t-t_0))^{1/3}\notag\\
X_3(t)&=&\frac{1}{\sqrt[3]{2}} \sinh(2(t-t_0))^{1/3}tanh^{1/\sqrt{3}}(t-t_0)
\end{eqnarray}
and potential 
$V(\phi(t))= \frac{4}{ 3\kappa }-\frac{\Lambda}{\kappa}$
solve the Bianchi I equations in $3+1$ spacetime dimensions.  One can compare this solution to the Bali and Jain solution in \cite{BJ}.

\end{example}

\section{In terms of a Schr\"odinger-Type Equation}

To reformulate the Einstein field equations $(I_0),\dots, (I_d)$ in terms of a Schr\"odinger-type equation with one less non-linear term than the generalized EMP, one can apply Corollary \ref{cor: EFE-NLSAnonzeroEzero} to the difference $d(I_0)-\displaystyle\sum_{i=1}^d(I_i)$.   Below is the resulting statement.

\begin{thm}\label{thm: BINLS}
Suppose you are given twice differentiable functions $X_1(t), \dots, X_d(t)>0$, a once differentiable function $\phi(t)$, and also functions $\rho(t), p(t), V(x)$ which satisfy the Einstein equations $(I_0),\dots,(I_d)$ for some $\Lambda\in\mathds{R}, d\in\mathds{N}\backslash\{0,1\}, \kappa\in\mathds{R}\backslash\{0\}$.  
Let $g(\sigma)$ denote the inverse of a function $\sigma(t)$ which satisfies
\begin{equation}\dot{\sigma}(t)=\frac{1}{\theta \left(X_1(t)\cdots X_d(t)\right)}\label{eq: BINLSdotsigma-R}\end{equation}
for some $\theta>0$.  Then the following functions 
\begin{eqnarray}
u(\sigma)&=&\left[\frac{1}{X_1\cdots X_d}\right]\circ g(\sigma)\label{eq: BINLSu-R}\\
P(\sigma)&=&\frac{d\kappa}{(d-1)}\psi '(\sigma)^2\label{eq: BINLSP-psi}
\end{eqnarray}
solve the Schr\"odinger-type equation
\begin{equation}u''(\sigma)+\left[E-P(\sigma)\right]u(\sigma)=
\frac{\theta^2 d \kappa(\uprho(\sigma)+\mathrm{p}(\sigma))}{(d-1) u(\sigma)}
\label{eq: BINLS}\end{equation}
for 
\begin{equation}\psi(\sigma)=\phi(g(\sigma))\label{eq: BINLSpsi-phi}\end{equation}
\begin{equation}\uprho(\sigma)=\rho(g(\sigma)), \ \mathrm{p}(\sigma)=p(g(\sigma)) \label{eq: BINLSuprho-rhormp-p}\end{equation}
and where 
\begin{equation}E\stackrel{def.}{=}\frac{-\theta^2}{(d-1)}X_1^2X_2^2\cdots X_d^2\displaystyle\sum_{l<k}\eta_{lk}^2\label{eq: BINLSE-X}\end{equation}
is a constant for 
\begin{equation}\eta_{lk}\stackrel{def.}{=}H_l-H_k, \mbox{\small $l\neq k, l, k\in\{1, \dots, d\}$.}\label{eq: BINLSeta-X}\end{equation}

Conversely, suppose you are given a twice differentiable function $u(\sigma)>0$, and also functions $P(\sigma)$ and $\uprho(\sigma), \mathrm{p}(\sigma)$  which solve (\ref{eq: BINLS}) for some constants $E<0, \theta>0, \kappa\in\mathds{R}\backslash\{0\}$ and $d\in\mathds{N}\backslash\{0,1\}$.  In order to construct functions which solve $(I_0),\dots, (I_d)$, first find $\sigma(t), \psi(\sigma)$ which solve the differential equations
\begin{equation}\dot{\sigma}(t)=\frac{1}{\theta} u(\sigma(t))\qquad\mbox{ and }\qquad \psi '(\sigma)^2= \frac{(d-1)}{ d \kappa} P(\sigma).\label{eq: BINLSdotsigma-upsi-P}\end{equation}
Let
\begin{equation}R(t)=u(\sigma(t))^{-\nu}\label{eq: BINLSR-ualpha-sigma}\qquad\mbox{ and }\qquad \alpha_l(t)\stackrel{def.}{=}c_l\sigma(t), l\in\{1,\dots,d\} \end{equation}
where $c_l$ are any constants for which both
\begin{equation}\displaystyle\sum_{l=1}^d c_l=0\qquad\mbox{ and }\qquad\displaystyle\sum_{l<k}c_lc_k=\frac{(d-1)E}{2d}.\label{eq: BINLScconditions}\end{equation}
Then the functions
\begin{equation}X_l(t)=R(t)^{1/\nu d}e^{\alpha_l(t)}\label{eq: BINLSX-Ralpha}\end{equation}
\begin{equation}\phi(t)=\psi(\sigma(t))\label{eq: BINLSphi-psi}\end{equation}
\begin{equation}\rho(t)=\uprho(\sigma(t)), \ p(t)=\mathrm{p}(\sigma(t))\label{eq: BINLSrho-uprhop-rmp}\end{equation}
and
\begin{equation}
V(\phi(t))=\left[\frac{(d-1)}{2\theta^2 d\kappa}\left(  (u')^2 + Eu^2 \right)
-\frac{1}{2\theta^2}u^2(\psi ')^2-\uprho-\frac{\Lambda}{\kappa}\right]\circ\sigma(t)
\label{eq: BINLSVphi-u}
\end{equation}
satisfy the equations $(I_0),\dots, (I_d)$.\end{thm}

\begin{proof}
This proof will implement Corollary \ref{cor: EFE-NLSAnonzeroEzero} with constants and functions as indicated in the following table.

\begin{table}[ht]
\centering
\caption{{ Corollary \ref{cor: EFE-NLSAnonzeroEzero} applied to Bianchi I}}\label{tb: BINLS}
\vspace{.2in}
\begin{tabular}{r | l c r | l}
In Corollary & substitute & & In Corollary & substitute \\[4pt]
\hline
\raisebox{-5pt}{$a(t)$} & \raisebox{-5pt}{$R(t)$}      &&    \raisebox{-5pt}{$\varepsilon $} & \raisebox{-5pt}{${\nu d \kappa}/{(d-1)}$}\\[8pt]
$G(t)$ & $\mbox{constant } \nu E/\theta^2$    &&     $A$ & $2/\nu$ \\[8pt]
 $G_{1}(t)$ & $ \frac{-\nu d\kappa}{(d-1)}(\rho(t)+p(t))$  &&      $A_{1}$   &$0$\\[8pt]
$F_{1}(\sigma)$&$ \frac{\theta^2 d \kappa}{(d-1)} (\uprho(\sigma)+\mathrm{p}(\sigma)) $ &&        $C_{1}$&  $1$ \\[6pt]
\hline
\end{tabular}
\end{table}

\break
Much of this proof will rely on computations that are exactly the same as those seen in the proof of Theorem \ref{thm: BIEMP} (the generalized EMP formulation of Bianchi I).  Therefore we will restate the relevant results here, but point the reader to the details in the proof of Theorem \ref{thm: BIEMP}.

To prove the forward implication, we assume to be given functions which solve the Einstein field equations $(I_0), \dots, (I_d)$.  Forming the linear combination $d(I_0)-\displaystyle\sum_{i=1}^d(I_i)$ of Einstein's equations and simplifying, as was done in (\ref{eq: BIEMPdI1minusIi}) - (\ref{eq: BIEMPI1minusIiwitheta}), 
\begin{equation}\dot{H}_R+\frac{\nu}{(d-1)}\displaystyle\sum_{l<k}\eta_{lk}^2=\frac{-\nu d\kappa}{(d-1)}\left[\dot{\phi}^2+ \left(\rho+p\right)\right]\label{eq: BINLSI1minusIiwitheta}\end{equation}
where 
\begin{equation}H_R(t)\stackrel{def.}{=}\frac{\dot{R}(t)}{R(t)}\end{equation}
and
\begin{equation}R(t)\stackrel{def.}{=}(X_1(t)\cdots X_d(t))^\nu\label{eq: BINLSR-X}\end{equation}
for any $\nu\neq 0$.  Next we will confirm that $E$ is constant.  As was done in (\ref{eq: BIEMPlhsIi=lhsIj})-(\ref{eq: BIEMPdotf+fH=0withetaH}), since the right-hand sides of Einstein's equations $(I_i)$  are the same for all $i\in\{1, \dots, d\}$, by equating the left-hand sides of any two equations $(I_i)$ and $(I_j)$ for $i\neq j$, and after some rearranging we obtain
\begin{equation}\dot\eta_{ij}+\frac{1}{\nu}\eta_{ij}H_R=0\label{eq: BINLSdoteta+stuff=0}\end{equation}
for $\eta_{ij}$ defined in (\ref{eq: BINLSeta-X}).  Therefore the definition (\ref{eq: BINLSE-X}) of $E$ is constant, being proportional to a sum of squares of these constant functions.  By the definitions (\ref{eq: BINLSE-X}) and (\ref{eq: BINLSR-X}) of the constant $E$ and the function $R(t)$, we now rewrite (\ref{eq: BINLSI1minusIiwitheta}) from above as 
\begin{equation}\dot{H}_R=\frac{-\nu d \kappa}{(d-1)}\left[\dot\phi^2 +(\rho+p)\right]+\frac{\nu E}{\theta^2 R^{2/\nu}}.\label{eq: BINLSI0minusIi}\end{equation}
This shows that $R(t), \phi(t), \rho(t)$ and $p(t)$ satisfy the hypothesis of Corollary \ref{cor: EFE-NLSAnonzeroEzero}, applied with  constants $\varepsilon , N, A, A_1 \dots, A_{N}$ and functions $a(t), G(t), G_1(t),\dots,G_{N}(t)$ according to Table \ref{tb: BINLS}.  Since $\sigma(t), u(\sigma), P(\sigma)$ and $\psi(\sigma)$ defined in (\ref{eq: BINLSdotsigma-R}), (\ref{eq: BINLSu-R}), (\ref{eq: BINLSP-psi}) and (\ref{eq: BINLSpsi-phi})  are equivalent to that in the forward implication of Corollary \ref{cor: EFE-NLSAnonzeroEzero}, applied with constants and functions according to Table \ref{tb: BINLS}, by this corollary and by definition (\ref{eq: BINLSuprho-rhormp-p}) of $\uprho_i(\sigma), \mathrm{p}(\sigma)$, the Schr\"odinger-type equation (\ref{eq: CNLSANONZERO}) holds for constants $C_1, \dots, C_{N}$ and functions $F_1(\sigma), \dots, F_{N}(\sigma)$ as indicated in Table \ref{tb: BINLS}.  This proves the forward implication.

To prove the converse implication, we assume to be given functions which solve the Schr\"odinger-type equation (\ref{eq: BINLS}) and we begin by showing that $(I_0)$ is satisfied.  Differentiating the definition of $R(t)$ in (\ref{eq: BINLSR-ualpha-sigma}) and using the definition in (\ref{eq: BINLSdotsigma-upsi-P}) of $\sigma(t)$, we obtain
\begin{eqnarray}\dot{R}(t)
&=&-\nu u(\sigma(t))^{-\nu-1}u'(\sigma(t))\dot\sigma(t)\notag\\
&=&-\frac{\nu}{\theta}u(\sigma(t))^{-\nu}u'(\sigma(t)).\label{eq: BINLSdotR-u}\end{eqnarray}
so that 
\begin{equation}H_R\stackrel{def.}{=}\frac{\dot{R}}{R}=-\frac{\nu}{\theta}u'(\sigma(t)).\label{eq: BINLSH-u}\end{equation}
Differentiating the definition (\ref{eq: BINLSphi-psi}) of $\phi(t)$ and using the definition in (\ref{eq: BINLSR-ualpha-sigma}) of $\sigma(t)$, we obtain
\begin{equation}
\dot{\phi}(t)=\psi '(\sigma(t))\dot{\sigma}(t)=\frac{1}{\theta}\psi '(\sigma(t))u(\sigma(t)).\label{eq: BINLSphi-u}
\end{equation}
Using (\ref{eq: BINLSH-u}) and (\ref{eq: BINLSphi-u}), and also the definitions (\ref{eq: BINLSR-ualpha-sigma}) and (\ref{eq: BINLSrho-uprhop-rmp}) of $R(t)$ and $\rho(t)$ respectively, the definition (\ref{eq: BINLSVphi-u}) of $V\circ\phi$ can be written as
\begin{eqnarray}
V\circ\phi
&=&\frac{1}{\kappa}\left(\frac{(d-1)}{2\nu^2 d}H_R^2+\frac{(d-1)E}{2d\theta^{2}R^{2/\nu}}\right)-\frac{1}{2}\dot{\phi}^2-\rho-\frac{\Lambda}{\kappa}.\label{eq: BINLSVphi-RE}
\end{eqnarray}
The quantity in parenthesis here is in fact equal to the left-hand-side of equation $(I_0)$.  
To see this, first note that the definitions
$X_l(t)\stackrel{def.}{=}R(t)^{1/\nu d}e^{\alpha_l(t)}$ in (\ref{eq: BINLSX-Ralpha}) and
$H(t)\stackrel{def.}{=}\frac{\dot{R}(t)}{R(t)}$ in (\ref{eq: BINLSH-u}), and the condition $\sum_l c_l=0$, are the same as those in Theorem \ref{thm: BIEMP}.
Also by the definitions (\ref{eq: BINLSR-ualpha-sigma}), (\ref{eq: BINLSdotsigma-upsi-P}) and (\ref{eq: BINLSR-ualpha-sigma})  of $\alpha_l(t), \sigma(t)$ and $R(t)$, we obtain 
\begin{equation}\dot\alpha_l\dot\alpha_k=c_lc_k\dot\sigma^2=\frac{c_lc_k}{\theta^2}(u\circ\sigma)^2=\frac{c_lc_k}{\theta^2R(t)^{2/\nu}}\label{eq: BINLSdotalphaldotalphak-Rc},\end{equation}
which is a slightly modified version of (\ref{eq: BIEMPdotalphaldotalphak-R}) from our computation in the proof of Theorem \ref{thm: BIEMP}.  Therefore by the arguments in (\ref{eq: BIEMPdotXldotXk/XlXk-Halpha})-(\ref{eq: BIEMPdotXldotXk/XlXk-Rc}), and using (\ref{eq: BINLSdotalphaldotalphak-Rc}) to slightly modify the last term to apply here, we have that
\begin{equation}\displaystyle\sum_{l<k}H_lH_k=\frac{(d-1)}{2\nu^2 d}H_R^2+\displaystyle\sum_{l<k}\frac{c_lc_k}{\theta^{2}R^{2/\nu}}.
\label{eq: BINLSdotXldotXk/XlXk-Rc}
\end{equation}
Then by the condition (\ref{eq: BINLScconditions}) on the constants $c_l$, (\ref{eq: BINLSdotXldotXk/XlXk-Rc}) becomes
\begin{equation}\displaystyle\sum_{l<k}H_lH_k=\frac{(d-1)}{2\nu^2 d}H_R^2+\frac{(d-1)E}{2d\theta^{2}R^{2/\nu}}.
\label{eq: BINLSdotXldotXk/XlXk-RE}
\end{equation}
That is, the expression (\ref{eq: BINLSVphi-RE}) for $V$ can now be written as
\begin{equation}V\circ\phi=\frac{1}{\kappa}\displaystyle\sum_{l<k}H_lH_k-\frac{1}{2}\dot{\phi}^2-\rho-\frac{\Lambda}{\kappa},\label{eq: BIEMPVphi-X}
\end{equation}
showing that $(I_0)$ holds under the assumptions of the converse implication.

To conclude the proof we must also show that the equations $(I_1), \dots, (I_d)$ hold.   In the converse direction the hypothesis of the converse of Corollary \ref{cor: EFE-NLSAnonzeroEzero} holds, applied with constants $N, C_1, \dots, C_{N}$ and functions $F_1(\sigma), \dots, F_{N}(\sigma)$ as indicated in Table \ref{tb: BINLS}.  Since  $\sigma(t), \psi(\sigma), R(t)$ and $\phi(t)$ defined in (\ref{eq: BINLSdotsigma-upsi-P}), (\ref{eq: BINLSR-ualpha-sigma}) and (\ref{eq: BINLSphi-psi}) are consistent with the converse implication of Corollary \ref{cor: EFE-NLSAnonzeroEzero}, applied with $a(t)$ and $\varepsilon $ as in Table \ref{tb: BINLS}, by this corollary and by the definition (\ref{eq: BINLSrho-uprhop-rmp}) of $\rho(t), p(t)$ the  scale factor  equation (\ref{eq: CEFEANONZERO}) holds for constants $\varepsilon , A, A_1, \dots, A_{N}$ and functions $G(t),G_1(t),\dots,G_{N}(t)$ according to Table \ref{tb: BINLS}.  That is, we have regained (\ref{eq: BINLSI0minusIi}).  Now solving (\ref{eq: BINLSVphi-RE}) for $\rho(t)$ and substituting this into (\ref{eq: BINLSI0minusIi}), we obtain
\begin{equation}
\dot{H}_R
=
\frac{-\nu d\kappa}{(d-1)}
\left[ 
\frac{1}{2}\dot\phi^2 -V\circ\phi+\frac{(d-1)E}{2 d\theta^2\kappa R^{2/\nu}}+ p +\frac{(d-1)}{2\nu^2 d\kappa}H_R^2-\frac{\Lambda}{\kappa} \right].\end{equation}
Multiplying by $\frac{(d-1)}{\nu d}$ and rearranging, we see that
\begin{equation}
\frac{(d-1)}{\nu d}\dot{H}_R+\frac{(d-1)}{2\nu^2 d}H_R^2+\frac{(d-1)E}{2 d\theta^2 R^{2/\nu}}
=
-\kappa
\left[ 
\frac{1}{2}\dot\phi^2 -V\circ\phi+p \right]+\Lambda.\label{eq: BINLSlhs=lhsofIi}\end{equation}

The left-hand side of this equation is in fact equal to the left-hand-side of $(I_i)$ for any $i\in\{1,\dots, d\}$.  To see this, again we use that the definitions $X_l(t)\stackrel{def.}{=}R(t)^{1/\nu d}e^{\alpha_l(t)}$ in (\ref{eq: BINLSX-Ralpha}) and
$H_R(t)\stackrel{def.}{=}\frac{\dot{R}(t)}{R(t)}$ in (\ref{eq: BINLSH-u}), and the condition $\sum_l c_l=0$, are the same as those in Theorem \ref{thm: BIEMP}.  Also by the definitions (\ref{eq: BINLSR-ualpha-sigma}) and (\ref{eq: BINLSdotsigma-upsi-P}) of $\alpha_l(t), \sigma(t)$ and $R(t)$, we see that
\begin{eqnarray}
\dot\alpha_l(t) R(t)^{1/\nu}
&=&c_l\dot\sigma(t) R(t)^{1/\nu}\notag\\
&=&\frac{c_l}{\theta}u(\sigma(t))R(t)^{1/\nu}\notag\\
&=&\frac{c_l}{\theta }\mbox{ is a constant}
\label{eq: BINLSdotalphaRconstant}\end{eqnarray}
which shows that
\begin{equation}\ddot\alpha_l+\frac{1}{\nu}\dot\alpha_l H_R=0\label{eq: BINLSddotalpha+stuff=0}\end{equation}
holds here as it does in Theorem \ref{thm: BIEMP}.  Therefore by the arguments in  (\ref{eq: BIEMPddotX/X-Halpha})-(\ref{eq: BIEMPBIlhsIi}), and as above using (\ref{eq: BINLSdotalphaldotalphak-Rc}) to slightly modify the last term of (\ref{eq: BIEMPBIlhsIi}) to apply here, we have that
\begin{equation}\displaystyle\sum_{l\neq i}(\dot{H}_l+H_l^2)+\displaystyle\sum_{\stackrel{l<k}{l,k\neq i}}H_lH_k=\frac{(d-1)}{\nu d}\dot{H}_R+\frac{d(d-1)}{2\nu^2 d^2}H_R^2-\sum_{l<k}\frac{c_lc_k}{\theta^{2}R^{2/\nu}}.
\label{eq: BINLSBIlhsIi}\end{equation}
Then by the condition (\ref{eq: BINLScconditions}) on the constants $c_l$, (\ref{eq: BINLSBIlhsIi}) becomes
\begin{equation}\displaystyle\sum_{l\neq i}(\dot{H}_l+H_l^2)+\displaystyle\sum_{\stackrel{l<k}{l,k\neq i}}H_lH_k=\frac{1}{\nu d}(d-1)\dot{H}_R+\frac{d(d-1)}{2\nu^2 d^2}H_R^2-\sum_{l<k}\frac{(d-1)E}{2d\theta^{2}R^{2/\nu}}.
\label{eq: BINLSdotXldotXk/XlXk-Ralpha}\end{equation}
Combining (\ref{eq:  BINLSlhs=lhsofIi}) and (\ref{eq: BINLSdotXldotXk/XlXk-Ralpha}), we obtain $(I_i)$ for all $i\in\{1,\dots,d\}$.  This proves the theorem.

\end{proof}

\break 

\subsection{Reduction to linear Schr\"odinger: pure scalar field}
To show some examples we take $\rho=p=0$ so that Theorem \label{thm: BINLS} shows that solving the Bianchi I  Einstein equations 
\begin{equation}
\displaystyle\sum_{l< k}H_lH_k
\stackrel{(I_0)''}{=}
\kappa\left[\frac{1}{2}\dot\phi^2+V\circ\phi\right]+\Lambda\label{eq: BIEFErhopzero}\end{equation}
\begin{eqnarray}
\displaystyle\sum_{l\neq i}(\dot{H}_l+H_l^2)+\displaystyle\sum_{\stackrel{l < k}{l,k\neq i}}H_lH_k&\stackrel{(I_i)''}{=}&-\kappa\left[\frac{1}{2}\dot\phi^2-V\circ\phi\right]+\Lambda\notag
\end{eqnarray}
is equivalent to solving the linear Schr\"odinger equation
\begin{equation}u''(\sigma)+\left[E-P(\sigma)\right]u(\sigma)=0
.\label{eq: rhop=0BINLS}\end{equation}
The solutions of $(I_0)'',(I_i)''$ in (\ref{eq: BIEFErhopzero}) and the solutions of (\ref{eq: rhop=0BINLS}) are related by 
\begin{equation}R(t)=(X_1(t)\cdots X_d(t))^{\nu}=u(\sigma(t))^{-\nu}\mbox{ and }\psi '(\sigma)^2= \frac{(d-1)}{d\kappa} P(\sigma)\label{eq:  rhop=0BIclassLSR-uvarpsi-P}\end{equation}
for any $\nu\neq 0$ and where $\phi(t)=\psi(\sigma(t))$ and 
\begin{equation}\dot\sigma(t)=\frac{1}{\theta X_1(t)\cdots X_d(t)}=\frac{1}{\theta}u(\sigma(t)),\label{eq: rho=p=0BILSdotsigma-X-u}\end{equation}
for $\theta>0$.  Also the constant $E$ is
\begin{equation}E\stackrel{def.}{=}\frac{-\theta^2}{(d-1)}X_1^2X_2^2\cdots X_d^2\left(\displaystyle\sum_{l<k}\eta_{lk}^2\right)\label{eq: rho=p=0BILSEconstant-X}\end{equation}
for
$\eta_{lk}\stackrel{def.}{=}H_l-H_k$, $l, \ k\in\{1,\dots, d\}$.  In the converse direction
\begin{equation}X_l(t)=R(t)^{1/\nu d}e^{\alpha_l(t)}\label{eq: rho=p=0BILSX-Ralpha}\end{equation}
for $\alpha_l(t)\stackrel{def.}{=}c_l\sigma(t)$ and where the constants $c_l$ satisfy
\begin{equation}\displaystyle\sum_{l=1}^d c_l=0\qquad\mbox{ and }\qquad\qquad\displaystyle\sum_{l<k}c_lc_k=\frac{(d-1)E}{2d}.\label{eq: rho=p=0BILScconstantscondition}\end{equation}
Also $V$ is taken to be such that
\begin{equation}
V(\phi(t))=\left[\frac{(d-1)}{2\theta^2 d\kappa}\left( (u')^2
+Eu^2\right)-\frac{1}{2\theta^2}u^2(\psi ')^2
-\frac{\Lambda}{\kappa}\right]\circ\sigma(t).
\label{eq: rho=p=0BINLSVphi-u}
\end{equation}

We now refer to Appendix E for solutions of the linear Schr\"odinger equation (\ref{eq: rhop=0BINLS}).  We will map them using the theorem to exact solutions of Einstein's equations.  Of course when $E=0$ in (\ref{eq: rho=p=0BILSEconstant-X}), $\eta_{lk}=0$ for all pairs $l,k$ so that  $X_1(t), X_2(t), \dots, X_d(t)$ agree up to a constant multiple. In this case, if we take $R(t)=X_1(t)=\cdots=X_d(t)\stackrel{def.}{=}a(t)$ and $\nu=1/d$ we obtain the FRLW cosmology with curvature  $k=0, n=2d$ and $D=0$ so that we may refer to section 3.3.1 for exact solutions to the Bianchi I Einstein equations if $E=0$.  Here we consider solutions to the linear Schr\"odinger equation (\ref{eq:  rhop=0BINLS}) by referring to Table \ref{tb: exactNLS} with $E<0$.

\begin{example}
For $\theta=1$ and choice of constant $E=-1$, we take solution 4 in Table \ref{tb: exactNLS} with $c_0=-1$, $a_0=1$ and $b_0=0$ so that we have $u(\sigma)=e^{-\sigma}$, $P(\sigma)=0$.  By (\ref{eq: sigmaeqnforu=-expr=general}) - (\ref{eq: dotsigmausigmaforu=-expr=general}) with $r_0=1$ we obtain $\sigma(t)=\ln\left((t-t_0)\right)$ and 
\begin{equation}R(t)=\frac{1}{u(\sigma(t))^{\nu}}=(t-t_0)^{\nu}\end{equation}
for any $\nu\neq 0$ $t_0\in\mathds{R}$ so that
\begin{eqnarray}X_i(t)&=&R(t)^{1/\nu d}e^{c_i\sigma(t)}=(t-t_0)^{1/d+c_i}
\end{eqnarray}
for constants that satisfy
\begin{equation}
\displaystyle\sum_{l=0}^d c_l=0\mbox{ and }\displaystyle\sum_{l<k}c_lc_k = -\frac{(d-1)}{2d}.\end{equation}
Since $P=0=\psi'(\sigma)$,
$\psi(\sigma)=\psi_0$
for constant $\psi_0\in\mathds{R}$.  Finally, by (\ref{eq: rho=p=0BINLSVphi-u}), (\ref{eq: uprimesigmaforu=-expr=general}) and (\ref{eq: dotsigmausigmaforu=-expr=general}), we obtain constant potential $V(\phi(t))=-\Lambda/\kappa$.  This is alternate derivation of the Kasner vacuum solution seen above in Example 32.  


For $\nu=1/2$ and $t_0=0$, the solver was run with $u$ and $u'$ both perturbed by $.01$.  The graphs of $R(t)$ below show that this solution is unstable, since the absolute error grows up to two orders of magnitude over the graphed time interval.

\begin{figure}[htbp]
\centering
\caption{Instability of Bianchi I Example \theexample}
\includegraphics[width=4in]{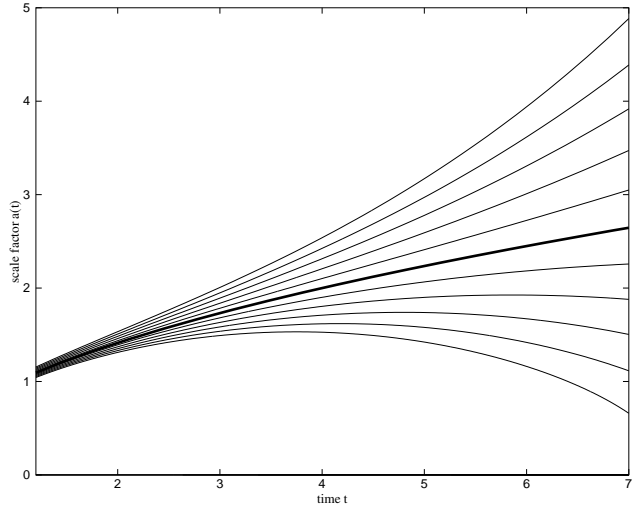}
\end{figure}

\end{example}

\begin{example}
For $\theta=1$ and choice of constant $E=-1$, we take solution 4 in Table \ref{tb: exactNLS} with $c_0=-1$ and $a_0,b_0>0$ so that we have $u(\sigma)=a_0e^{-\sigma}-b_0e^{\sigma}$ and $P(\sigma) = 0$.  By (\ref{eq: sigmaeqnforu=linearcomboexpr=1}) - (\ref{eq: dotsigmausigmaforu=linearcomboexpr=1}) we obtain $\sigma(t)=\ln\left(\sqrt{\frac{a_0}{b_0}}tanh(\sqrt{a_0b_0}(t-t_0))\right)$ and 
\begin{equation}R(t)=\frac{1}{u(\sigma(t))^{\nu}}=\frac{1}{(2\sqrt{a_0b_0})^{\nu}} \ sinh^{\nu}(2\sqrt{a_0b_0}(t-t_0))\end{equation}
for $\nu\neq 0$ and $t_0\in\mathds{R}$ so that
\begin{eqnarray}X_i(t)&=&R(t)^{1/\nu d}e^{c_i \sigma(t)}\\
&=&\frac{   \sqrt{a_0}^{c_i-1/d}}{\sqrt{b_0}^{c_i+1/d}}    \ sinh^{1/d+c_i}(\sqrt{a_0b_0}(t-t_0))
cosh^{1/d-c_i}(\sqrt{a_0b_0}(t-t_0))
.\notag\end{eqnarray}
   Since $P=0=\psi'(\sigma)$, the scalar field is constant
\begin{equation}\psi(\sigma)=\psi_0\in\mathds{R}.\end{equation}
Finally, by (\ref{eq: rho=p=0BINLSVphi-u}), (\ref{eq: uprimesigmaforu=linearcomboexpr=1}) and (\ref{eq: dotsigmausigmaforu=linearcomboexpr=1}), we obtain constant potential
\begin{eqnarray}
V(\phi(t))
&=&\frac{2(d-1)a_0b_0}{ d\kappa}-\frac{\Lambda}{\kappa}
\end{eqnarray}
since $coth^2(x)-csch^2(x)=1$.  

As an example for $d=3$, one can take $c_1=-\frac{1}{\sqrt{3}}, c_2=0$,  $c_3=\frac{1}{\sqrt{3}}$ and $a_0=b_0=1$ so that 
\begin{eqnarray}
X_1(t)&=&sinh^{(1-\sqrt{3})/3}((t-t_0))cosh^{(1+\sqrt{3})/3}((t-t_0))\notag\\
X_2(t)&=& sinh^{1/3}((t-t_0)) cosh^{1/3}((t-t_0))\notag\\
X_3(t)&=&sinh^{(1+\sqrt{3})/3}((t-t_0))cosh^{(1-\sqrt{3})/3}((t-t_0))\end{eqnarray}
and the potential $V(\phi(t))=\frac{4}{ 3\kappa}-\frac{\Lambda}{\kappa}$ solve the Bianchi I equations in $3+1$ spacetime dimensions.  One can also compare this solution with the Bali and Jain solution in section 2 of \cite{BJ}.

\end{example}

\begin{example}
For $\theta=1$ and choice of constant $E=-1$, we take solution 5 in Table \ref{tb: exactNLS} with $c_0=-1$ and $b_0 =0$ so that we have $u(\sigma)=(a_0/\sigma)e^{-\sigma^2/2}$ and $P(\sigma)=\sigma^2+2/\sigma^2$  for $a_0>0$.  By (\ref{eq: sigmaeqnforu=e^x^2/xr=1cnegative}) - (\ref{eq: dotsigmausigmaforu=e^x^2/xr=1cnegative}) we obtain $\sigma(t)=\sqrt{2\ln( a_0(t-t_0))}$ and 
\begin{eqnarray}
R(t)&=&\frac{1}{u(\sigma(t))^{\nu}}
=\left(\sqrt{2}(t-t_0)\sqrt{\ln( a_0(t-t_0))}\right)^{\nu}
\end{eqnarray}
for $t>t_0$ so that 
\begin{eqnarray}
X_i(t)&=&R(t)^{1/\nu d}e^{c_i\sigma(t)}\notag\\
&=&\left(\sqrt{2}(t-t_0)\sqrt{\ln( a_0(t-t_0))}\right)^{1/d}  e^{c_i \sqrt{2\ln( a_0(t-t_0))}}.
\end{eqnarray}
for constants $c_i$ that satisfy
\begin{equation}\displaystyle\sum_{l=1}^d c_l=0\mbox{ and }\displaystyle\sum_{l<k}c_lc_k=-\frac{(d-1)}{2d}.\end{equation}
By (\ref{eq:  phiforu=xe^x^2r=1cnegative}) with $\alpha_0=(d-1)/d\kappa$, we have scalar field
\begin{eqnarray}
\phi(t)&=&\psi(\sigma(t))\notag\\
&=&\sqrt{ \frac{ (d-1)}{2d\kappa} } \left(\sqrt{2 \ln^2( a_0(t-t_0)) +1} + \ln\left[2\ln( a_0(t-t_0))\right] \right.\notag\\
&&\left. - \ln\left[4 +  4\sqrt{2 \ln^2( a_0(t-t_0))   + 1} \right] \right)+\beta_0
\end{eqnarray}
for $\beta_0\in\mathds{R}$.  Finally, by (\ref{eq:  rho=p=0BINLSVphi-u}), (\ref{eq: uprimesigmaforu=e^x^2/xr=1cnegative}) and (\ref{eq: dotsigmausigmaforu=e^x^2/xr=1cnegative}), we obtain
\begin{eqnarray}
V(\phi(t))&=&
\left[\frac{(d-1)}{2 d\kappa}\left( (u')^2
-u^2\right)-\frac{1}{2\theta^2}u^2(\psi ')^2
-\frac{\Lambda}{\kappa}\right]\circ\sigma(t)\notag\\
 &=&\frac{(d-1)(2\ln(a_0(t-t_0))-1)}{8d\kappa  (t-t_0)^2 \ln^2( a_0(t-t_0))}
-\frac{\Lambda}{\kappa}.
\end{eqnarray}

For example when $d=3$, we can take $c_1=-\frac{1}{\sqrt{3}}, c_2=0$ and $c_3=\frac{1}{\sqrt{3}}$ to obtain
\begin{eqnarray}
X_1(t)&=&\left(\sqrt{2}(t-t_0)\sqrt{\ln( a_0(t-t_0))}\right)^{1/3}  e^{-\frac{1}{\sqrt{3}} \sqrt{2\ln( a_0(t-t_0))}}\notag\\
X_2(t)&=&\left(\sqrt{2}(t-t_0)\sqrt{\ln( a_0(t-t_0))}\right)^{1/3}\notag\\
X_3(t)&=&\left(\sqrt{2}(t-t_0)\sqrt{\ln( a_0(t-t_0))}\right)^{1/3}  e^{\frac{1}{\sqrt{3}}  \sqrt{2\ln( a_0(t-t_0))}},
\end{eqnarray}
\begin{eqnarray}
\phi(t)&=&\sqrt{\frac{1}{3\kappa}}\left(\sqrt{2 \ln^2( a_0(t-t_0)) +1}  + \ln\left[2\ln( a_0(t-t_0))\right] \right.\notag\\
&&\left. - \ln\left[4+ 4\sqrt{2 \ln^2( a_0(t-t_0))   + 1} \right] \right)+\beta_0
\end{eqnarray}
and 
\begin{equation}V(\phi(t))=\frac{(2\ln(a_0(t-t_0))-1)}{12\kappa  (t-t_0)^2 \ln^2( a_0(t-t_0))}
-\frac{\Lambda}{\kappa}\end{equation}
for $\beta_0, t_0 \in\mathds{R}$ and $a_0>0$.  Note that this is similar to the FRLW solution found in Example 21 by setting $n=2d$ and by identifying $a(t)$ in Example 21 with $R(t)$ here.

For $\nu=a_0=1$ and $t_0=0$, the solver was run with $u$ and $u'$ both perturbed by $.01$.  The graphs of $R(t)$ below show that this solution is unstable, since the absolute error grows three orders of magnitude over the graphed time interval.

\break

\begin{figure}[htbp]
\centering
\caption{Instability of Bianchi I Example \theexample}
\includegraphics[width=4in]{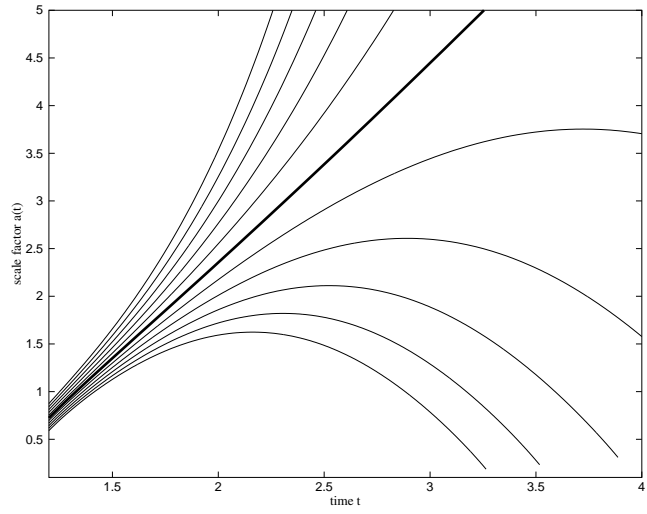}
\end{figure}

\end{example}

\chapter{Reformulations of a conformal Bianchi I model}

For a change of coordinates in comparison to chapter 4, we consider a Bianchi I metric of the form
\begin{equation} ds^2=-\left(a_1(t)a_2(t)\cdots a_d(t)\right)^2dt^2+a_1(t)^2 dx_1^2+a_2^2(t)dx_2^2+\cdots +a_d^2(t)dx_d^2\label{eq: confBImetric}
\end{equation}
in a $d+1-$dimensional spacetime.  The nonzero Einstein equations $g^{ij}G_{ij}=-\kappa g^{ij}T_{ij}+\Lambda$, multiplied by $|g_{00}|=(a_1 a_2 \cdots a_d)^2$ are 
\begin{equation}
\displaystyle\sum_{l<k}H_lH_k 
\stackrel{(I_0)}{=} 
\kappa\left[\frac{\dot{\phi}^2}{2}+(a_1a_2\cdots a_d)^2\left(V\circ\phi+\rho+\frac{\Lambda}{\kappa}\right) \right]\label{eq: confBIEFEI0Dd}\end{equation}
\begin{eqnarray}
\displaystyle\sum_{l\neq 1}\dot{H}_l-\displaystyle\sum_{l<k}H_lH_k 
&\stackrel{(I_1)}{=} &
\kappa\left[-\frac{\dot{\phi}^2}{2}+(a_1a_2\cdots a_d)^2\left(V\circ\phi-p+\frac{\Lambda}{\kappa}\right) \right]\notag\\
&\vdots&\notag\\
\displaystyle\sum_{l\neq i}\dot{H}_l-\displaystyle\sum_{l<k}H_lH_k 
&\stackrel{(I_i)}{=} &
\kappa\left[-\frac{\dot{\phi}^2}{2}+(a_1a_2\cdots a_d)^2\left(V\circ\phi-p+\frac{\Lambda}{\kappa}\right) \right]\notag\\
&\vdots&\notag\\
\displaystyle\sum_{l\neq d}\dot{H}_l-\displaystyle\sum_{l<k}H_lH_k 
& \stackrel{(I_d)}{=} &
\kappa\left[-\frac{\dot{\phi}^2}{2}+(a_1a_2\cdots a_d)^2\left(V\circ\phi-p+\frac{\Lambda}{\kappa}\right) \right]\notag
\end{eqnarray}
where $H_l(t)\stackrel{def.}{=}\frac{\dot{a}_l}{a_l}$.

\break

\section{In terms of a Generalized EMP}

\begin{thm}\label{thm: confBIEMP}  Suppose you are given twice differentiable functions $a_1(t),\dots, a_d(t)>0$, a once differentiable function $\phi(t)$ and also functions $\rho(t), p(t), V(x)$ which satisfy the Einstein equations $(I_0), \dots, (I_d)$ for some $\Lambda\in\mathds{R}, d\in\mathds{N}\backslash\{0,1\}, \kappa\in\mathds{R}\backslash\{0\}$.  Denote
\begin{equation}R(t)\stackrel{def.}{=}(a_1(t) a_2(t)\cdots a_d(t))^\nu\label{eq: confBIR-a}\end{equation}
for some $\nu\neq 0$.  If $f(\tau)$ is the inverse of a function $\tau(t)$ which satisfies
\begin{equation}\dot{\tau}(t)=\theta R(t)^{q+\frac{1}{\nu}},\label{eq: confBIdottau-R}\end{equation}
for some constants $\theta>0$ and $q\neq 0$, then
\begin{equation}Y(\tau)=R(f(\tau))^{q}\qquad\mbox{ and }\qquad Q(\tau)=\frac{q\nu d\kappa}{(d-1)}
\varphi'(\tau)^2 \label{eq: confBIY-RQ-varphi}\end{equation}
solve the generalized EMP equation
\begin{equation}Y''(\tau)+Q(\tau)Y(\tau)=\frac{q\nu L}{\theta^2Y(\tau)^{(q\nu+2)/q\nu}}- \frac{q\nu d \kappa\left(\varrho(\tau)+\textup{\textlhookp}(\tau)\right)}{\theta^2(d-1)Y(\tau)}\label{eq: confBIEMP}\end{equation}
for 
\begin{equation}\varphi(\tau)=\phi(f(\tau))\label{eq: confBIvarphi-phi}\end{equation}
\begin{equation}\varrho(\tau)=\rho(f(\tau)), \ \textup{\textlhookp}(\tau)=p(f(\tau)). \label{eq: confBIEMPvarrho-rhohookp-p}\end{equation}
\begin{equation}L\stackrel{def.}{=}\frac{2}{(d-1)}\displaystyle\sum_{1<l<k\leq d}\mu_l\mu_k-\displaystyle\sum_{j=2}^d\mu_j^2,\label{eq: confBIlambda0-mul}\end{equation}
where $\mu_i\in\mathds{R}$ are such that $a_i(t)=\omega_ie^{\mu_i t}a_1(t)$ for some $\omega_i\in\mathds{R},  i\in\{2,\dots, d\}$.

Conversely, suppose you given a twice differentiable function $Y(\tau)>0$, a continuous function $Q(\tau)$ and also functions $\varrho(\tau), \textup{\textlhookp}(\tau)$ which solve (\ref{eq: confBIEMP}) for some constants $\theta>0$ and $q, \nu, \kappa \in\mathds{R}\backslash\{0\}, L\in\mathds{R}, d\in\mathds{N}\backslash\{0,1\}$.  In order to construct functions which solve $(I_0), \dots, (I_d)$, first find $\tau(t), \varphi(\tau)$ which solve the differential equations
\begin{equation}\dot{\tau}(t)=\theta Y(\tau(t))^{(q\nu+1)/q\nu}\qquad\mbox{ and }\qquad \varphi'(\tau)^2=\frac{(d-1)}{q\nu d\kappa }Q(\tau).\label{eq: confBIdottau-Yvarphi-Q}\end{equation}
Next find constants $\mu_i, i\in\{2,\dots, d\}$ which satisfy
\begin{equation}L=\frac{2}{(d-1)}\displaystyle\sum_{2\leq l<k\leq d}\mu_l\mu_k-\displaystyle\sum_{j=2}^d\mu_j^2,\label{eq: confBIL-mulconverse}\end{equation}
and let
\begin{equation}R(t)=Y(\tau(t))^{1/q}\label{eq: confBIR-Y}.\end{equation}
Then the functions
\begin{equation}a_1(t)=R(t)^{1/\nu d}(\omega_2\cdots\omega_de^{(\mu_2+\cdots+\mu_d)t})^{-1/d}\label{eq: confBIa1-Rexp}\end{equation}
\begin{equation}a_i(t)=\omega_ie^{\mu_i t}a_1(t)\label{eq: confBIai-a1exp}\end{equation} 
\begin{equation}\phi(t)=\varphi(\tau(t))\label{eq: confBIphi-varphi}\end{equation}
\begin{equation}\rho(t)=\varrho(\tau(t)),\qquad\qquad p(t)=\textup{\textlhookp}(\tau(t))\label{eq: confBIrho-varrhop-textlhookp}\end{equation}
and
\begin{equation}V(\phi(t))=\left[\frac{(d-1)}{2d\kappa} \left(
\frac{\theta^2}{\nu^2  q^2} (Y')^2 + \frac{L}{Y^{2/q\nu}}\right)
 -\frac{\theta^2}{2} (\varphi')^2 Y^{2} 
  -\frac{\Lambda}{\kappa}-\varrho\right]\circ\tau(t)\label{eq: confBIdefnV}\end{equation}
satisfy the Einstein equations $(I_0), \dots, (I_d)$ for any $\omega_i>0 , 2\leq i\leq d$.
\end{thm}

\begin{proof}
This proof will implement Theorem \ref{thm: EFE-EMP} with constants and functions as indicated in the following table.

   \begin{table}[ht]
\centering
\caption{{ Theorem \ref{thm: EFE-EMP} applied to conformal Bianchi I}}\label{tb: confBIEMP}
\vspace{.2in}
\begin{tabular}{r | l c r | l}
In Theorem & substitute & & In Theorem& substitute \\[4pt]
\hline
\raisebox{-5pt}{$a(t)$} & \raisebox{-5pt}{$R(t)$}      &&    \raisebox{-5pt}{$N$} & \raisebox{-5pt}{$1$}\\[8pt]
$\delta  $ & $-1/\nu$      &&    $\varepsilon $ & ${\nu d\kappa}/{(d-1)}$\\[8pt]
   $G_0(t)$ & $\mbox{constant } \nu L$    &&     $A_0$ & $0$ \\[8pt]
 $G_{1}(t)$ & $ \frac{-\nu d\kappa}{(d-1)}(\rho(t)+p(t))$  &&      $A_{1}$   &$-2/\nu$\\[8pt]
 $ \lambda_0(\tau)$ &constant ${q\nu L}/{\theta^2}$       &&      $B_0$ &$(2+q\nu)/q\nu$ \\[8pt] 
$\lambda_{1}(\tau)$&$\frac{-q \nu d \kappa}{\theta^2(d-1)}(\varrho(\tau)+\textup{\textlhookp}(\tau))$ &&  $B_{1}$&  $1$ \\[6pt]
\hline
\end{tabular}
\end{table}

To prove the forward implication, we assume to be given functions which solve the Einstein field equations $(I_0), \dots, (I_d)$.   Since the right-hand sides of Einstein equations $(I_i)$ are the same for all $i\in\{1, \dots, d\}$, we begin by equating the left-hand side of $(I_1)$ with the left-hand side of any $(I_{j})$ for $j\in\{2,\dots, d\}$ since it will give us a simplifying relation among the scale factors $a_1(t), \dots, a_d(t)$.  Doing this, we obtain
\begin{equation}
\displaystyle\sum_{l\neq 1}\dot{H}_l
-\displaystyle\sum_{l<k}H_lH_k
=
\displaystyle\sum_{l\neq j}\dot{H}_l
-\displaystyle\sum_{l<k}H_lH_k.\label{eq: confBIEMPequatelhsIi}\end{equation}
All terms cancel except for the $j^{th}$ term from the first sum on the left-hand side, and also the $1^{st}$ term from the first sum on the right-hand side.  This leaves
\begin{equation}\dot{H}_j=\dot{H}_1, \label{eq: confBIEMPequatelhss}\end{equation}
which holds for all $j\in\{1, \dots, d\}$.  Integrating, we obtain
\begin{equation}H_j=H_1+\mu_j\label{eq: confBIEMPdota/a=dota/a+mu}\end{equation}
for $\mu_{j}\in\mathds{R}, j\in\{1,\dots, d\}$ ($\mu_1=0$).
Since in general $\frac{d}{dt}\ln(a_i)=\frac{\dot{a}_i}{a_i}=H_i$, (\ref{eq: confBIEMPdota/a=dota/a+mu}) can be written
\begin{equation}\frac{d}{dt}\ln(a_j)=\frac{d}{dt}\ln(a_1)+\mu_j.\label{eq: confBIEMPlogderiv=logderiv+mu}\end{equation}
Integrating again we get
\begin{equation}\ln(a_j)=\ln(a_1)+\mu_j t+c'_j\label{eq: confBIEMPlna=lna+mut+c}\end{equation}
for some $c_j\in\mathds{R}, j\in\{1,\dots, d\}$ ($c_1=0$). Exponentiating and letting $\omega_j\stackrel{def.}{=}e^{c_j}>0,$ we have that
\begin{equation}a_j(t)=\omega_je^{\mu_j t}a_1(t),\label{eq: confBIEMPaj=cjemuta1}\end{equation}
where of course this holds trivially for $j=1$ where $\omega_1=1$ and $\mu_1=0$.
By (\ref{eq: confBIEMPdota/a=dota/a+mu}), the left-hand side of $(I_0)$ can be written as
\begin{equation}
\displaystyle\sum_{l<k}H_lH_k
=\displaystyle\sum_{l<k}(H_1+\mu_l)(H_1+\mu_k)=\displaystyle\sum_{l<k}(H_1^2+(\mu_l+\mu_k)H_1+\mu_l\mu_k).
\label{eq: confBIEMPlhsI0wrta1}\end{equation}
The first term on the right-hand side of (\ref{eq: confBIEMPlhsI0wrta1}) does not depend on the indices $l,k$.  By using our computation in (\ref{eq: BIEMPsum1l<k}), this term is equal to $H_1^2$ times $\sum_{l<k}1=\frac{d(d-1)}{2}$.  The second term on the right-hand side of (\ref{eq: confBIEMPlhsI0wrta1}) sums to $H_1$ times the quantity
\begin{eqnarray}
\displaystyle\sum_{l<k}(\mu_l+\mu_k)&=&\displaystyle\sum_{l=1}^{d-1}\displaystyle\sum_{k=l+1}^d\mu_l+\displaystyle\sum_{k=2}^d\displaystyle\sum_{l=1}^{k-1}\mu_k\notag\\
&=&\displaystyle\sum_{l=1}^{d-1}(d-l)\mu_l+\displaystyle\sum_{k=2}^d (k-1)\mu_k\notag\\
&=&(d-1)\mu_1+\displaystyle\sum_{j=2}^{d-1}(d-j+j-1)\mu_j+(d-1)\mu_d\notag\\
&=&(d-1)\displaystyle\sum_{j=1}^d\mu_j.\end{eqnarray}
Therefore (\ref{eq: confBIEMPlhsI0wrta1}) becomes
\begin{equation}\displaystyle\sum_{l<k}H_lH_k 
=
\frac{d(d-1)}{2}H_1^2+(d-1)H_1\displaystyle\sum_{j=1}^d\mu_j+\displaystyle\sum_{l<k}\mu_l\mu_k
\label{eq: confBIEMPnewlhsI0}.\end{equation}
Next using the definition (\ref{eq: confBIR-a}) of $R(t)$, we further define  
\begin{eqnarray}
H_R
&\stackrel{def.}{=}&\frac{\dot{R}}{R}\notag\\
&=&\frac{\nu (a_1\cdots a_d(t))^{\nu -1} (\dot{a}_1\cdots a_d+ \cdots+ a_1\cdots \dot{a}_d)}{(a_1\cdots a_d)^\nu}\notag\\
&=&\nu\displaystyle\sum_{j=1}^d H_j\notag\\
&=&\nu\displaystyle\sum_{j=1}^d\left(H_1+\mu_j\right)\notag\\
&=&\nu\left(dH_1+\displaystyle\sum_{j=1}^d\mu_j\right)\label{eq: confBIHR-H1}
\end{eqnarray}
by (\ref{eq: confBIEMPdota/a=dota/a+mu}).  Therefore 
\begin{equation}H_1=\frac{1}{\nu d}H_R-\frac{1}{d}\displaystyle\sum_{j=1}^d\mu_j\label{eq: confBIEMPH1wrtHR}\end{equation}
and
\begin{equation}\dot{H}_1=\frac{1}{\nu d}\dot{H}_R\label{eq: confBIEMPdotH1wrtdotHR}\end{equation}
so that (\ref{eq: confBIEMPnewlhsI0}) becomes
\begin{equation}\displaystyle\sum_{l<k}H_lH_k=\quad\qquad\qquad\qquad\qquad\qquad\qquad\qquad\qquad\qquad\qquad\qquad\qquad\qquad\label{eq: confBIlhsI0HR}\end{equation}
\begin{equation}
\qquad\frac{d(d-1)}{2}\left( \frac{1}{\nu d}H_R-\frac{1}{d}\displaystyle\sum_{j=1}^d\mu_j    \right)^2+(d-1)\left( \frac{1}{\nu d}H_R-\frac{1}{d}\displaystyle\sum_{j=1}^d\mu_j    \right)\displaystyle\sum_{j=1}^d\mu_j+\displaystyle\sum_{l<k}\mu_l\mu_k.\notag
\end{equation}
Collecting terms (the $H_R$ terms sums to zero), we obtain
\begin{equation}\displaystyle\sum_{l<k}H_lH_k=
\frac{(d-1)}{2\nu^2d}H_R^2- \frac{(d-1)}{2d}\left(\displaystyle\sum_{j=1}^d\mu_j\right)^2+\displaystyle\sum_{l<k}\mu_l\mu_k.\label{eq: confBIlhsI0-HRwithsummulmuk}
\end{equation}
Also using that
\begin{equation}
\left(\displaystyle\sum_{j=1}^d\mu_j\right)^2=2\displaystyle\sum_{l<k}\mu_l\mu_k+\displaystyle\sum_{j=1}^d\mu_j^2
\end{equation}
the equation (\ref{eq: confBIlhsI0-HRwithsummulmuk}) becomes
\begin{equation}\displaystyle\sum_{l<k}H_lH_k=
\frac{(d-1)}{2\nu^2d}H_R^2  +\frac{1}{d}\displaystyle\sum_{l<k}\mu_l\mu_k
- \frac{(d-1)}{2d}\displaystyle\sum_{j=1}^d\mu_j^2 .\label{eq: confBIlhsI0-HRsimplified}
\end{equation}
 Defining the quantity
\begin{eqnarray}L&\stackrel{def.}{=}&\frac{2}{(d-1)}\displaystyle\sum_{l<k}\mu_l\mu_k-\displaystyle\sum_{j=1}^d\mu_j^2\notag\\
&=& \frac{2}{(d-1)}\displaystyle\sum_{2\leq l<k\leq d}\mu_l\mu_k-\displaystyle\sum_{j=2}^d\mu_j^2 \mbox{ \ \ \  \  (since $\mu_1=0$)}, \label{eq: confBIL-mu-inproof}\end{eqnarray}
we see that (\ref{eq: confBIlhsI0-HRsimplified}) becomes
\begin{equation}\displaystyle\sum_{l<k}H_lH_k=
\frac{(d-1)}{2\nu^2d}H_R^2  +\frac{(d-1)}{2d}L.\label{eq: confBIlhsI0-HRsimplifiedwithL}
\end{equation}
Similarly by (\ref{eq: confBIEMPequatelhss}), (\ref{eq: confBIEMPdotH1wrtdotHR}) and (\ref{eq: confBIlhsI0-HRsimplifiedwithL}), the left-hand side of $(I_{i})$ for $1\leq i\leq d$ can be written as
\begin{equation}
\displaystyle\sum_{l\neq i}\dot{H}_l-\displaystyle\sum_{l<k}H_lH_k 
=\frac{(d-1)}{\nu d}\dot{H}_R-\frac{(d-1)}{2\nu^2d}H_R^2-\frac{(d-1)}{2d}L.
\label{eq: confBIEMPnewlhsIi}  
\end{equation}
That is, by (\ref{eq: confBIlhsI0-HRsimplified}), (\ref{eq: confBIEMPnewlhsIi}) and the definition (\ref{eq: confBIR-a}) of $R(t)$, equations $(I_0)$ and $(I_i)$ for $1\leq i\leq d$  can be written as
\begin{equation}
\frac{(d-1)}{2\nu^2d}H_R^2  +\frac{(d-1)}{2d}L\stackrel{(I_0)'}{=}\kappa\left[\frac{\dot{\phi}^2}{2}+R^{2/\nu}\left(V\circ\phi+\rho+\frac{\Lambda}{\kappa}\right) \right]\label{eq: confBIEMPnewI0}\notag\end{equation}
\begin{equation}
\frac{(d-1)}{\nu d}\dot{H}_R-\frac{(d-1)}{2\nu^2d}H_R^2-\frac{(d-1)}{2d}L \stackrel{(I_i)'}{=}\kappa\left[-\frac{\dot{\phi}^2}{2}+R^{2/\nu}\left(V\circ\phi - p+\frac{\Lambda}{\kappa}\right) \right].\notag
\label{eq: confBIEMPnewIi}\end{equation}
Forming the linear combination $(I_0)'-(I_i)'$, and multiplying by $\frac{-\nu d}{(d-1)}$,
 \begin{equation}\dot{H}_R-\frac{1}{\nu}H_R^2-\nu L=\frac{-\nu d\kappa}{(d-1)}\left[\dot\phi^2+R^{2/\nu}(\rho+p)\right].\label{eq: confBIEMPI0minusIiwrtHRsimplified}\end{equation}
This shows that $R(t), \phi(t), \rho(t)$ and $p(t)$ satisfy the hypothesis of Theorem \ref{thm: EFE-EMP}, applied with  constants $\epsilon, \varepsilon , N, A_0, \dots, A_{N}$ and functions $a(t), G_0(t),\dots,G_{N}(t)$ according to Table \ref{tb: confBIEMP}.  Since $\tau(t), Y(\tau), Q(\tau)$ and $\varphi(\tau)$ defined in (\ref{eq: confBIdottau-R}), (\ref{eq: confBIY-RQ-varphi}) and (\ref{eq: confBIvarphi-phi})  are equivalent to that in the forward implication of Theorem \ref{thm: EFE-EMP}, by this theorem and by definition (\ref{eq: confBIEMPvarrho-rhohookp-p}) of $\varrho(\tau), \textup{\textlhookp}(\tau)$, the generalized EMP equation (\ref{eq: AEMP}) holds for constants $B_0, \dots, B_{N}$ and functions $\lambda_0(\tau), \dots, \lambda_{N}(\tau)$ as indicated in Table \ref{tb: confBIEMP}.  This proves the forward implication.

To prove the converse implication, we assume to be given functions which solve the generalized EMP equation (\ref{eq: confBIEMP}) and we begin by showing that $(I_0)$ is satisfied.  Differentiating the definition of $R(t)$ in  (\ref{eq: confBIR-Y}) and using the definition in (\ref{eq: confBIdottau-Yvarphi-Q}) of $\tau(t)$, we see that
\begin{eqnarray}
\dot{R}(t)
&=&\frac{1}{q}Y(\tau(t))^{\frac{1}{q}-1}Y'(\tau(t))\dot{\tau}(t)\\
&=&\frac{\theta}{q}Y(\tau(t))^{\frac{1}{q}(1+\frac{1}{\nu})}Y'(\tau(t))\end{eqnarray}
Dividing by $R(t)$, we obtain
\begin{equation}H_{R}(t)\stackrel{def.}{=}\frac{\dot{R}(t)}{R(t)}=\frac{\theta}{q}Y(\tau(t))^{1/q\nu}Y'(\tau(t)).\label{eq: confBIHR-Y}\end{equation}
Differentiating the definition (\ref{eq: confBIphi-varphi}) of $\phi(t)$ and using definition (\ref{eq: confBIdottau-Yvarphi-Q}) of $\tau(t)$ gives
\begin{equation}\dot\phi(t)=\varphi'(\tau(t))\dot\tau(t)=\theta \varphi'(\tau(t)) Y(\tau(t))^{(q\nu+1)/q\nu}.\label{eq: confBIdotphi-varphiprimeY}\end{equation}
Using (\ref{eq: confBIHR-Y}) and (\ref{eq: confBIdotphi-varphiprimeY}), and also the definitions (\ref{eq: confBIR-Y}) and (\ref{eq: confBIrho-varrhop-textlhookp}) of $R(t)$ and $\rho_i(t)$ respectively, the definition (\ref{eq: confBIdefnV}) of $V\circ\phi$ can be written as
\begin{equation}V\circ\phi= \frac{1}{R^{2/\nu}} \left[\frac{1}{\kappa}\left(\frac{(d-1)}{2\nu^2d}H_R^2+\frac{(d-1)}{2d}L \right) - \frac{\dot\phi^2}{2}\right]   - \frac{\Lambda}{\kappa} -\rho.\label{eq: confBIVphi-R}\end{equation}
The quantity in the inner parenthesis here is in fact equal to the left-hand-side of equation $(I_0)$.  To see this, we differentiate the definitions in (\ref{eq: confBIa1-Rexp}) and (\ref{eq: confBIai-a1exp})  of $a_i(t)$, divide the results by $a_i(t)$, and use the definition (\ref{eq: confBIHR-Y}) of $H_R$  to obtain
\begin{equation}H_1\stackrel{def.}{=}\frac{\dot{a}_1}{a_1}=
\frac{1}{\nu d}H_R-\frac{1}{d}(\mu_2+\cdots +\mu_d)\label{eq: confBIEMPH1-HR}\end{equation}
and
\begin{equation}H_i\stackrel{def.}{=}\frac{\dot{a}_i}{a_i}=\frac{\dot{a}_1}{a_1}+\mu_i=H_1+\mu_i \end{equation}
for $i\in\{1,\dots,d\}$ by taking $\mu_1\stackrel{def.}{=}0$.
This confirms that the identities (\ref{eq:  confBIEMPequatelhss}), (\ref{eq:  confBIEMPdota/a=dota/a+mu}), (\ref{eq: confBIEMPH1wrtHR}) and (\ref{eq: confBIEMPdotH1wrtdotHR}) hold in the converse direction, so that the
 computations (\ref{eq: confBIEMPnewlhsI0})-(\ref{eq: confBIlhsI0-HRsimplifiedwithL}) are also valid in the converse direction for $L$ in (\ref{eq: confBIL-mulconverse}). That is,
 \begin{eqnarray}
\displaystyle\sum_{l<k}H_lH_k=\frac{(d-1)}{2\nu^2d}H_R^2+\frac{(d-1)}{2d}L\label{eq: confBIEMPHlHk-HR}\end{eqnarray}
which shows that (\ref{eq: confBIVphi-R}) can be written
\begin{equation}V\circ\phi= \frac{1}{R^{2/\nu}} \left[\frac{1}{\kappa}\left( \displaystyle\sum_{l<k}H_lH_k \right) -\frac{\dot\phi^2}{2}\right]   -\frac{\Lambda}{\kappa}-\rho\label{eq: confBIVphi-HlHk}\end{equation}
so that $(I_0)$ holds.

To conclude the proof we must also show that the equations $(I_1), \dots, (I_d)$ hold.  In the converse direction the hypothesis of the converse of Theorem \ref{thm: EFE-EMP} holds, applied with constants 
$N, B_0, \dots, B_N$ and functions $\lambda_0(\tau), \dots, \lambda_N(\tau)$ as indicated in Table \ref{tb: confBIEMP}.  Since $\tau(t), \varphi(\tau), R(t)$ and $\phi(t)$ defined in (\ref{eq: confBIdottau-Yvarphi-Q}), (\ref{eq: confBIR-Y}) and (\ref{eq: confBIphi-varphi}) are consistent with the converse implication of Theorem \ref{thm: EFE-EMP}, applied with $a(t), \delta  $ and $\varepsilon $ as in Table \ref{tb: confBIEMP}, by this theorem and by the definition (\ref{eq: confBIrho-varrhop-textlhookp}) of $\rho(t), p(t)$ the  scale factor equation (\ref{eq: AEFE}) holds for constants $\delta  , \varepsilon , A_0, \dots, A_N$ and functions $G_0(t), \dots, G_N(t)$ according to Table \ref{tb: confBIEMP}.  That is, we have regained equation (\ref{eq: confBIEMPI0minusIiwrtHRsimplified}).  Now solving (\ref{eq: confBIVphi-R}) for $R^{2/\nu}\rho(t)$ and substituting this into (\ref{eq: confBIEMPI0minusIiwrtHRsimplified}), we obtain
\begin{equation}
\dot{H}_R-\frac{1}{2\nu}H_R^2-\frac{\nu }{2}L=\frac{-\nu d\kappa}{(d-1)}\left[\frac{\dot\phi^2}{2}+R^{2/\nu}\left(-V\circ\phi+ p-\frac{\Lambda}{\kappa}\right)\right].
\end{equation}
Multiplying by $\frac{(d-1)}{\nu d}$ and rearranging, we get that
\begin{equation}\frac{(d-1)}{\nu d}\dot{H}_R-\frac{(d-1)}{2\nu^2 d}H_R^2-\frac{(d-1)}{2 d}L=-\kappa
\left[\frac{\dot\phi^2}{2}+R^{2/\nu}\left(-V\circ\phi+p-\frac{\Lambda}{\kappa}\right)\right].\label{eq: confBIIiwrtR}
\end{equation}
As noted above, the computations  (\ref{eq: confBIEMPnewlhsI0})-(\ref{eq: confBIlhsI0-HRsimplifiedwithL})  still hold in the converse direction so that by (\ref{eq: confBIEMPnewlhsIi}), we see that the left-hand side of (\ref{eq: confBIIiwrtR}) is in fact equal to the left-hand side of $(I_i)$ for any $i\in\{1,\dots, d\}$.  Since $R=(a_1\cdots a_d)^\nu$, the right-hand side of (\ref{eq: confBIIiwrtR}) agrees with the right-hand side of $(I_i)$ for all $i\in\{1,\dots, d\}$. This proves the theorem.
\end{proof}

\subsection{Reduction to classical EMP: pure scalar field}

To show some examples we take $\rho=p=0$ and choose parameter $\nu=1/q$ in Theorem \ref{thm: confBIEMP} to find that solving the Bianchi I Einstein equations
\begin{equation}
\displaystyle\sum_{l<k}H_lH_k 
\stackrel{(I_0)'}{=} 
\kappa\left[\frac{\dot{\phi}^2}{2}+(a_1a_2\cdots a_d)^2 \left(V\circ\phi+\frac{\Lambda}{\kappa}\right) \right]\end{equation}
\begin{equation}
\displaystyle\sum_{l\neq i}\dot{H}_l-\displaystyle\sum_{l<k}H_lH_k 
\stackrel{(I_i)'}{=} 
\kappa\left[-\frac{\dot{\phi}^2}{2}+(a_1a_2\cdots a_d)^2 \left(V\circ\phi+\frac{\Lambda}{\kappa}\right)\right]\notag\\
\end{equation}
for $l,k,i\in\{1,\dots, d\}$ is equivalent to solving the classical EMP equation
\begin{equation}Y''(\tau)+Q(\tau)Y(\tau)=\frac{L}{\theta^2Y(\tau)^3}\label{eq: nomatterconfBIEMP}\end{equation}
for constants $\theta,L>0$.  The solutions of $(I_0)', (I_1)',\dots,(I_d)'$ and (\ref{eq: nomatterconfBIEMP}) are related by
\begin{equation}R(t)=Y(\tau(t))^{1/q} \qquad\mbox{ and }\qquad \varphi '(\tau)^2= \frac{(d-1)}{d\kappa} Q(\tau)\label{eq: nomatterconfBIclassEMPa-Yvarphi-Q}\end{equation}
for $q\neq 0, \phi(t)=\varphi(\tau(t)), R(t)\stackrel{def.}{=} \left(a_1(t)\cdots a_d(t)\right)^{1/q}$ and 
\begin{equation}\dot\tau(t)=\theta R(t)^{2q}=\theta Y(\tau(t))^2,\label{eq: nomatterconfBIclassEMPdottau-a-Y}\end{equation}
for any $\theta>0$.  Also we define the constant
\begin{equation}L\stackrel{def.}{=}\frac{2}{(d-1)}\displaystyle\sum_{2<l<k\leq d}\mu_l\mu_k-\displaystyle\sum_{j=2}^d\mu_j^2,\label{eq: nomatterconfBIlambda0-mul}\end{equation}
where $\mu_i\in\mathds{R}$ are such that $a_i(t)=c_ie^{\mu_i t}a_1(t)$ for some $c_i\in\mathds{R},  i\in\{2,\dots, d\}$.
In the converse direction we also define
\begin{equation}a_1(t)=R(t)^{q/ d}(\omega_2\cdots\omega_de^{(\mu_2+\cdots+\mu_d)t})^{-1/d}\label{eq: nomatterconfBIEMPX-Ralpha}\end{equation}
and  take $V$ to be
\begin{equation}V(\phi(t))=\left[\frac{(d-1)}{2d\kappa} \left(
\theta^2 (Y')^2 + \frac{L}{Y^{2}}\right)
 -\frac{\theta^2}{2} (\varphi')^2 Y^{2} 
  -\frac{\Lambda}{\kappa}\right]\circ\tau(t).\label{eq: nomatterconfBIdefnV}\end{equation}

\break

\section{In terms of a Schr\"odinger-Type Equation}
To reformulate the Bianchi I Einstein equations $(I_0), \dots, (I_d)$ in (\ref{eq: confBIEFEI0Dd}) in terms of an equation with one less non-linear term than that which is provided by the generalized EMP formulation, one can apply Corollary \ref{cor: EFE-NLSAnonzeroEzero} to the difference $d(I_0)-\displaystyle\sum_{i=1}^d(I_i)$ (and similar to above, define $V\circ\phi$ in $u-$notation to be such that $(I_0)$ holds).   Below is the resulting statement.

\begin{thm}\label{thm: confBINLS}
Suppose you are given twice differentiable functions $a_1(t), \dots, a_d(t)>0$, a once differentiable function $\phi(t)$, and also functions $\rho(t), p(t), V(x)$ which satisfy the Einstein equations $(I_0),\dots,(I_d)$ for some $\Lambda\in\mathds{R}, d\in\mathds{N}\backslash\{0,1\}, \kappa\in\mathds{R}\backslash\{0\}$.  
Denote \begin{equation}R(t)\stackrel{def.}{=} \left(a_1(t)\cdots a_d(t)\right)^\nu\label{eq: confBINLSR-a}\end{equation}for some $\nu\neq 0$, 
then the functions 
\begin{eqnarray}
u(\sigma)&=&R(\sigma+t_0)^{-1/\nu}\label{eq: confBINLSu-R}\\
P(\sigma)&=&\frac{d\kappa}{(d-1)}\psi '(\sigma)^2\label{eq: confBINLSP-psi}
\end{eqnarray}
solve the Schr\"odinger-type equation
\begin{equation}u''(\sigma)+\left[E-P(\sigma)\right]u(\sigma)=\frac{ d \kappa(\uprho(\sigma)+\mathrm{p}(\sigma))}{(d-1) u(\sigma)}
\label{eq: confBINLS}\end{equation}
for 
\begin{equation}\psi(\sigma)=\phi(\sigma+t_0)\label{eq: confBINLSpsi-phi}\end{equation}
\begin{equation}\uprho(\sigma)=\rho(\sigma+t_0), \ \mathrm{p}(\sigma)=p(\sigma+t_0). \label{eq: confBINLSuprho-rhormp-p}\end{equation}
and where 
\begin{equation}E\stackrel{def.}{=}\frac{2}{(d-1)}\displaystyle\sum_{l<k}\mu_l\mu_k-\displaystyle\sum_{j=1}^d\mu_j^2\label{eq: confBINLSE-mu}\end{equation}
for constants $\mu_j$ such that $a_j(t)=\omega_je^{\mu_j t}a_1(t)$ for some $\omega_j>0, j\in\{1, \dots, d\}$.

Conversely, suppose you are given a twice differentiable function $u(\sigma)>0$, and also functions $P(\sigma)$ and $\uprho(\sigma), \mathrm{p}(\sigma)$  which solve (\ref{eq: confBINLS}) for some constants $E<0,  \kappa\in\mathds{R}\backslash\{0\}$ and $d\in\mathds{N}\backslash\{0,1\}$.  
Then we define  $\psi(\sigma)$ such that
\begin{equation}\psi '(\sigma)^2= \frac{(d-1)}{ d \kappa} P(\sigma),\label{eq: confBINLSpsi-P}\end{equation}
and constants $\mu_i, i\in\{2, \dots, d\}$ which satisfy
\begin{equation}E{=}\frac{2}{(d-1)}\displaystyle\sum_{l<k}\mu_l\mu_k-\displaystyle\sum_{j=1}^d\mu_j^2,\label{eq: confBINLSE-muconverse}\end{equation}
and let
\begin{equation}R(t)=u(t-t_0)^{-\nu}.\label{eq: confBINLSR-u}\end{equation}
Then the functions
\begin{equation}a_1(t)=R(t)^{1/\nu d}(\omega_2\cdots\omega_de^{(\mu_2+\cdots+\mu_d)t})^{-1/d}\label{eq: confBINLSa1-R}\end{equation}
\begin{equation}a_i(t)=\omega_ie^{\mu_i t}a_1(t)\label{eq: confBINLSai-a1exp}\end{equation} 
\begin{equation}\phi(t)=\psi(t-t_0)\label{eq: confBINLSphi-psi}\end{equation}
\begin{equation}\rho(t)=\uprho(t-t_0),\qquad\qquad p(t)=\mathrm{p}(t-t_0)\label{eq: confBINLSrho-uprhop-mathrmp}\end{equation}
and
\begin{equation}
V(\phi(t))=\left[
\frac{(d-1)}{2 d\kappa}\left( (u')^2+ Eu^2 \right)
-\frac{1}{2}(\psi ')^2u^2-\frac{\Lambda}{\kappa} -\uprho\right]\circ(t-t_0)
\label{eq: confBINLSVphi-u}
\end{equation}
satisfy the equations $(I_0),\dots, (I_d)$ for any $\omega_i>0, 2\leq i\leq d$.\end{thm}

\begin{proof}
This proof will implement Corollary \ref{cor: EFE-NLSAzeroEnonzero} with constants and functions as indicated in the following table.

\begin{table}[ht]
\centering
\caption{{ Corollary \ref{cor: EFE-NLSAzeroEnonzero} applied to conformal Bianchi I}}\label{tb: confBINLS}
\vspace{.2in}
\begin{tabular}{r | l c r | l}
In Corollary  & substitute & & In Corollary & substitute \\[4pt]
\hline
\raisebox{-5pt}{$a(t)$} & \raisebox{-5pt}{$R(t)$}      &&    \raisebox{-5pt}{$N$} & \raisebox{-5pt}{$1$}\\[8pt]
$\delta  $ & $-1/\nu$      &&   $\varepsilon $ &${\nu d\kappa}/{(d-1)}$\\[8pt]
 $G(t)$ & $\mbox{constant } \nu E$    &&     $A$ & $0$ \\[8pt]
 $G_{1}(t)$ & $ \frac{-\nu d \kappa}{(d-1)}(\rho(t)+p(t))$  &&      $A_{1}$   &$-2/\nu$\\[8pt]
$F_{1}(\sigma)$&$ \frac{d \kappa}{(d-1)} (\uprho(\sigma)+\mathrm{p}(\sigma)) $ &&        $C_{1}$&  $1$ \\[6pt]
\hline
\end{tabular}
\end{table}

Much of this proof will rely on computations that are exactly the same as those seen in the proof of Theorem \ref{thm: confBIEMP} (the generalized EMP formulation of conformally Bianchi I).  Therefore we will restate the relevant results here, but point the reader to the details in the proof of Theorem \ref{thm: confBIEMP}.

To prove the forward implication, we assume to be given functions which solve the  Einstein field equations $(I_0),\dots, (I_d)$.  Since the right-hand sides of Einstein equation $(I_i)$ are all the same for $i\in\{1,\dots, d\}$, we begin by equating the left-hand side of $(I_{1})$ with the left-hand side of any $(I_j)$  for $j\in\{2,\dots, d\}$ since it will give us a simplifying relation among the scale factors $a_1(t), \dots, a_d(t)$.   Exactly this was done in (\ref{eq: confBIEMPequatelhsIi})-(\ref{eq: confBIEMPaj=cjemuta1}) so that again we obtain
\begin{equation}H_j=H_1+\mu_j\label{eq: confBINLSHj-H1+mu}\end{equation}
and
\begin{equation}a_j(t)=\omega_je^{\mu_j t}a_1(t)\label{eq: confBINLSaj-a1}\end{equation}
for $\omega_j>0, \mu_j\in\mathds{R}, j\in\{1, \dots, d\}$, and $\mu_1=0, \omega_1=1$.  This allows us to follow the arguments given in (\ref{eq: confBIEMPlhsI0wrta1})-(\ref{eq: confBIEMPnewlhsIi}), so that the Einstein equations $(I_0)$ and $(I_i)$ for $1\leq i\leq d$ can be written as
\begin{equation}
\frac{(d-1)}{2\nu^2d}H_R^2  +\frac{(d-1)}{2d}E\stackrel{(I_0)'}{=}\kappa\left[\frac{\dot{\phi}^2}{2}+R^{2/\nu}\left(V\circ\phi+\rho+\frac{\Lambda}{\kappa}\right) \right]\label{eq: confBINLSnewI0}\notag\end{equation}
\begin{equation}
\frac{(d-1)}{\nu d}\dot{H}_R-\frac{(d-1)}{2\nu^2d}H_R^2-\frac{(d-1)}{2d}E \stackrel{(I_i)'}{=}\kappa\left[-\frac{\dot{\phi}^2}{2}+R^{2/\nu}\left(V\circ\phi-p+\frac{\Lambda}{\kappa}\right) \right].\notag
\label{eq: confBINLSnewIi}\end{equation}
where again $R(t)\stackrel{def.}{=}(a_1(t)\cdots a_d(t))^{\nu}$ for $\nu\neq 0$ by (\ref{eq: confBINLSR-a}),
\begin{equation}L\stackrel{def.}{=}\frac{2}{(d-1)}\displaystyle\sum_{l<k}\mu_l\mu_k-\displaystyle\sum_{j=1}^d\mu_j^2=E\mbox{ (by (\ref{eq: confBINLSE-mu}))}\label{eq: confBINLSL-mu}\end{equation}
and 
\begin{equation}H_R\stackrel{def.}{=}\frac{\dot{R}}{R},\label{eq: confBINLSHR-R}\end{equation}
so that 
\begin{equation}H_1=\frac{1}{\nu d}H_R-\frac{1}{d}\displaystyle\sum_{j=1}^d\mu_j\label{eq: confBINLSH1-HR}\end{equation}
\begin{equation}\dot{H}_1=\frac{1}{\nu d}\dot{H}_R.\label{eq: confBINLSdotH1-dotHR}\end{equation}
Again we form the linear combination $(I_0)'-(I_i)'$, multiply by $\frac{-\nu d}{(d-1)}$, and obtain 
 \begin{equation}\dot{H}_R-\frac{1}{\nu}H_R^2-\nu E=\frac{-\nu d\kappa}{(d-1)}\left[\dot\phi^2+R^{2/\nu}(\rho+p)\right].\label{eq: confBINLSI0minusIiwrtHRsimplified}\end{equation}
This shows that $R(t), \phi(t), \rho(t)$ and $p(t)$ satisfy the hypothesis of Corollary \ref{cor: EFE-NLSAzeroEnonzero}, applied with  constants $\epsilon, \varepsilon , N, A, A_1 \dots, A_{N}$ and functions $a(t), G(t),G_1(t),\dots,G_{N}(t)$ according to Table \ref{tb: confBINLS}.  Since $
u(\sigma), P(\sigma)$ and $\psi(\sigma)$ defined in (\ref{eq: confBINLSu-R})-(\ref{eq: confBINLSP-psi}) and (\ref{eq: confBINLSpsi-phi})  are equivalent to that in the forward implication of Corollary \ref{cor: EFE-NLSAnonzeroEzero}, by this corollary and by definition (\ref{eq: confBINLSuprho-rhormp-p}) of $\uprho(\sigma), \mathrm{p}(\sigma)$, the Schr\"odinger-type equation (\ref{eq: CNLSAZERO}) holds for constants $C_1, \dots, C_{N}$ and functions $F_1(\sigma), \dots, F_{N}(\sigma)$ as indicated in Table \ref{tb: confBINLS}.  This proves the forward implication.

To prove the converse implication, we assume to be given functions which solve the Schr\"odinger-type equation (\ref{eq: confBINLS}) and we will show that equations $(I_0), \dots, (I_d)$ are satisfied. 

To show that $(I_0)$ is satisfied, we differentiate the definition of $R(t)$ in (\ref{eq: confBINLSR-u}) 
to obtain
\begin{eqnarray}\dot{R}(t)
&=&-\nu u(t-t_0)^{-\nu-1}u'(t-t_0).\label{eq: confBINLSdotR-u}\end{eqnarray}
Dividing by $R(t)$, we have 
\begin{equation}H_R\stackrel{def.}{=}\frac{\dot{R}}{R}=-\nu \frac{u'(t-t_0)}{u(t-t_0)}.\label{eq: confBINLSHR-u}\end{equation}
Differentiating the definition (\ref{eq: confBINLSphi-psi}) of $\phi(t)$, we get that
\begin{equation}
\dot{\phi}(t)=\psi '(t-t_0).\label{eq: confBINLSphi-u}
\end{equation}
Using (\ref{eq: confBINLSHR-u}) and (\ref{eq: confBINLSphi-u}), and also the definitions (\ref{eq: confBINLSR-u}) and (\ref{eq: confBINLSrho-uprhop-mathrmp}) of $R(t)$ and $\rho(t)$ respectively, the definition (\ref{eq: confBINLSVphi-u}) of $V\circ\phi$ can be written as
\begin{equation}
V\circ\phi
=\frac{1}{R^{2/\nu}}\left[\frac{1}{\kappa}\left(\frac{(d-1)}{2\nu^2 d}H_R^2+\frac{(d-1)E}{2d}\right)-\frac{1}{2}\dot{\phi}^2\right]-\frac{\Lambda}{\kappa}-\rho.\label{eq: confBINLSVphi-RE}
\end{equation}
The quantity in parenthesis here is in fact equal to the left-hand-side of equation $(I_0)$.  
To see this, note that the definitions of $a_1(t), a_i(t)$ in (\ref{eq: confBINLSa1-R}), (\ref{eq: confBINLSai-a1exp}) and also
$H_R\stackrel{def.}{=}\dot{R}/R$ are the same as those in Theorem \ref{thm: confBIEMP}.  Therefore we may follow the arguments given in (\ref{eq: confBIEMPnewlhsI0})-(\ref{eq: confBIlhsI0-HRsimplifiedwithL}) and (\ref{eq: confBIEMPH1-HR})-(\ref{eq: confBIEMPHlHk-HR}) to see that the identity 
 \begin{eqnarray}
\displaystyle\sum_{l<k}H_lH_k=\frac{(d-1)}{2\nu^2d}H_R^2+\frac{(d-1)}{2d}E\label{eq: confBINLSHlHk-HR}\end{eqnarray}
still holds in the converse direction (since (\ref{eq: confBINLSE-muconverse}) in Theorem \ref{thm: confBIEMP} and (\ref{eq: confBIL-mulconverse}) here show that $L=E$).
This shows that (\ref{eq: confBINLSVphi-RE}) can be written
\begin{equation}
V\circ\phi
=\frac{1}{R^{2/\nu}}\left[\frac{1}{\kappa}\left(\displaystyle\sum_{l<k}H_lH_k\right)-\frac{1}{2}\dot{\phi}^2\right]-\frac{\Lambda}{\kappa}-\rho\label{eq: confBINLSVphi-HlHk}
\end{equation}
so that $(I_0)$ holds under the assumptions of the converse implication.

To conclude the proof we must also show that the equations $(I_1), \dots, (I_d)$ hold.   In the converse direction the hypothesis of the converse of Corollary \ref{cor: EFE-NLSAzeroEnonzero} holds, applied with constants $N, C_1, \dots, C_{N}$ and functions $F_1(\sigma), \dots, F_{N}(\sigma)$ as indicated in Table \ref{tb: confBINLS}.  Since  $
\psi(\sigma), R(t)$ and $\phi(t)$ defined in (\ref{eq: confBINLSpsi-P}), (\ref{eq: confBINLSR-u}) and (\ref{eq: confBINLSphi-psi}) are consistent with the converse implication of Corollary \ref{cor: EFE-NLSAzeroEnonzero}, applied with $a(t)$ and $\delta  , \varepsilon $ as in Table \ref{tb: confBINLS}, by this corollary and by the definition (\ref{eq: confBINLSrho-uprhop-mathrmp}) of $\rho(t), p(t)$ the  scale factor  equation (\ref{eq: CEFEAZERO}) holds for constants $\delta  , \varepsilon , A, A_1, \dots, A_{N}$ and functions $G(t),G_1(t),\dots,G_{N}(t)$ according to Table \ref{tb: confBINLS}.  That is, we have regained (\ref{eq: confBINLSI0minusIiwrtHRsimplified}).  Now solving (\ref{eq: confBINLSVphi-RE}) for $R^{2/\nu}\rho(t)$ and substituting this into (\ref{eq: confBINLSI0minusIiwrtHRsimplified}), we obtain
\begin{equation}
\dot{H}_R-\frac{1}{2\nu}H_R^2-\frac{\nu}{2} E
=
\frac{-\nu d\kappa}{(d-1)}
\left[ 
\frac{1}{2}\dot\phi^2 +R^{2/\nu}\left(-V\circ\phi+p -\frac{\Lambda}{\kappa}\right) \right].\end{equation}
Multiplying by $\frac{(d-1)}{\nu d}$ and rearranging, we get that
\begin{equation}
\frac{(d-1)}{\nu d}\dot{H}_R-\frac{(d-1)}{2\nu^2 d}H_R^2-\frac{(d-1)}{2d}E
=
\kappa
\left[ 
-\frac{1}{2}\dot\phi^2 +R^{2/\nu}\left(V\circ\phi-p +\frac{\Lambda}{\kappa}\right)\right].\label{eq: confBINLSlhs=lhsofIi}\end{equation}
As noted above, the computations  (\ref{eq: confBIEMPnewlhsI0})-(\ref{eq: confBIlhsI0-HRsimplifiedwithL}) from Theorem \ref{thm: confBIEMP} still hold in this theorem, in the converse direction.  Therefore the left-hand side of (\ref{eq: confBINLSlhs=lhsofIi}) is in fact equal to the left-hand side of $(I_i)$ for any $i\in\{1, \dots, d\}$.  Since $R=(a_1\cdots a_d)^\nu$, the right-hand side of 
(\ref{eq: confBINLSlhs=lhsofIi}) agrees with the right-hand side of $(I_i)$ for any $i\in\{1, \dots, d\}$.  This proves the theorem.
\end{proof}

\subsection{Reduction to linear Schr\"odinger: pure scalar field}
By applying Theorem \ref{thm: confBINLS} with $\rho=p=0$, we obtain a linear Schr\"odinger equation.
We now refer to Appendix E for solutions of the linear Schr\"odinger equation.

\begin{example}
For $\vartheta=\nu=1$ and choice of constant $E=-1$, we take solution 4 in Table \ref{tb: exactNLS} with $c_0=-1$, $a_0=1$ and $b_0=0$ so that we have $u(\sigma)=e^{-\sigma}$ and  $P(\sigma)=0$.  By (\ref{eq: confBINLSR-u}),
\begin{equation}R(t)=\frac{1}{u(t-t_0)}=e^{t-t_0}\end{equation}
so that 
\begin{eqnarray}a_1(t)&=&R(t)^{1/d}e^{-(\mu_2+\cdots + \mu_d)t/d}\notag\\
&=&e^{(1-\mu_2-\cdots - \mu_d)(t-t_1) /d},
\end{eqnarray}
\begin{eqnarray}a_i(t)&=&e^{\mu_i t}a_1(t)\notag\\
&=&e^{((1-\mu_2-\cdots - \mu_d)/d+\mu_i )(t-t_1)}
\end{eqnarray}
for constants $\mu_i$ that satisfy
\begin{equation}
E=-1=\frac{2}{(d-1)}\displaystyle\sum_{l<k}\mu_l\mu_k-\displaystyle\sum_{j=1}^d\mu_j^2.\end{equation}
Since $P=0=\psi'(\sigma)$, we get
\begin{equation}\psi(\sigma)=\psi_0\end{equation}
for constant $\psi_0\in\mathds{R}$.  Finally, by (\ref{eq: confBINLSVphi-u}),  we obtain constant potential
\begin{eqnarray}
V(\phi(t))&=&=\left[
\frac{(d-1)}{2 d\kappa}\left( (u')^2 - u^2 \right)
-\frac{\Lambda}{\kappa} \right]\circ(t-t_0)\notag\\
&=&-\frac{\Lambda}{\kappa}.
\end{eqnarray}

For example when $d=3$, we can take $\mu_2=\frac{1}{\sqrt{3}}$ and $\mu_3=\frac{2}{\sqrt{3}}$ to obtain the vacuum solution 
\begin{eqnarray}
a_1(t)&=&e^{(1-\sqrt{3})(t-t_1) /3}\notag\\
a_2(t)&=&e^{(t-t_1) /3}\notag\\
a_3(t)&=&e^{(1+\sqrt{3})(t-t_1) /3}\notag\\
\phi(t)&=&\phi_0\notag\\
V&=&-\Lambda/\kappa.
\end{eqnarray}

For $t_0=0$, the solver was run with $u, u'$ and $\sigma$ perturbed by $.001$.  The graphs of $R(t)$ below show that the solution is unstable.  The absolute error grows three orders of magnitude over the graphed time interval.

\begin{figure}[htbp]
\centering
\vspace{-.1in}
\caption{Instability of conformal Bianchi I Example \theexample}
\includegraphics[width=4in]{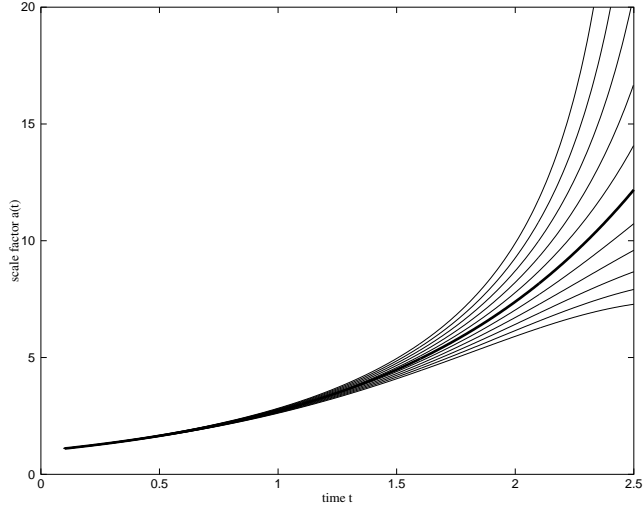}
\end{figure}
\end{example}

\break

\begin{example}
For $\vartheta=1$, $n=6$, $\nu=1$ and choice of constant $E=-1$, we take the negative of solution 4 in Table \ref{tb: exactNLS} with $c_0=-1$ and  $a_0,b_0>0$ so that we have $u(\sigma)=b_0e^{\sigma}-a_0 e^{-\sigma}$ and  $P(\sigma)=0$.  By (\ref{eq: confBINLSR-u})
\begin{equation}R(t)=\frac{1}{u(t-t_0)}= \frac{1}{b_0e^{(t-t_0)}-a_0 e^{-(t-t_0)}}  \end{equation}
so that 
\begin{eqnarray}a_1(t)&=&R(t)^{1/ d}e^{-(\mu_2+\cdots + \mu_d)t/d}\notag\\
&=&\frac{e^{-(\mu_2+\cdots + \mu_d)t/d}}{\left(     b_0e^{(t-t_0)}-a_0 e^{-(t-t_0)}   \right)^{1/d}},
\end{eqnarray}
\begin{eqnarray}a_i(t)&=&e^{\mu_i t}a_1(t)\notag\\
&=&\frac{e^{-(\mu_2+\cdots + \mu_d)t/d+\mu_i t}}{\left(    b_0e^{(t-t_0)} - a_0 e^{-(t-t_0)}   \right)^{1/d}}
\end{eqnarray}
for constants $\mu_i$ that satisfy
\begin{equation}
E=-1=\frac{2}{(d-1)}\displaystyle\sum_{l<k}\mu_l\mu_k-\displaystyle\sum_{j=1}^d\mu_j^2.\end{equation}
Since $P=0=\psi'(\sigma)$, we have
\begin{equation}\psi(\sigma)=\psi_0\end{equation}
for constant $\psi_0\in\mathds{R}$.  Finally, by (\ref{eq: confBINLSVphi-u}), 
we obtain constant potential
\begin{eqnarray}
V(\phi(t))&=&=\left[
\frac{(d-1)}{2 d\kappa}\left( (u')^2 - u^2 \right)
-\frac{\Lambda}{\kappa} \right]\circ(t-t_0)\notag\\
&=&\frac{2(d-1)a_0b_0}{ d\kappa}-\frac{\Lambda}{\kappa}
\end{eqnarray}
since $\cosh^2(x)-\sinh^2(x)=1$.

For example when $d=3$, we can take $a_0=b_0=1$, $\mu_2=\frac{1}{\sqrt{3}}$ and $\mu_3=\frac{2}{\sqrt{3}}$ to obtain the solution 
\begin{eqnarray}
a_1(t)&=&\left(\frac{1}{2}  csch(t-t_0)\right)^{1/3}e^{-t/\sqrt{3}}\notag\\
a_2(t)&=&\left(\frac{1}{2}  csch(t-t_0)\right)^{1/3}\notag\\
a_3(t)&=&\left(\frac{1}{2}  csch(t-t_0)\right)^{1/3}e^{t/\sqrt{3}}\notag\\
\phi(t)&=&\phi_0\notag\\
V&=&\frac{4}{3\kappa}-\Lambda/\kappa.
\end{eqnarray}

For $a_0=b_0=1$ and $t_0=0$, the solver was run with $u, u'$ and $\sigma$ perturbed by $.01$.  The graphs of $R(t)$ below show that the solution is stable.  
\begin{figure}[htbp]
\centering
\vspace{-.1in}
\caption{Instability of conformal Bianchi I Example \theexample}
\includegraphics[width=4in]{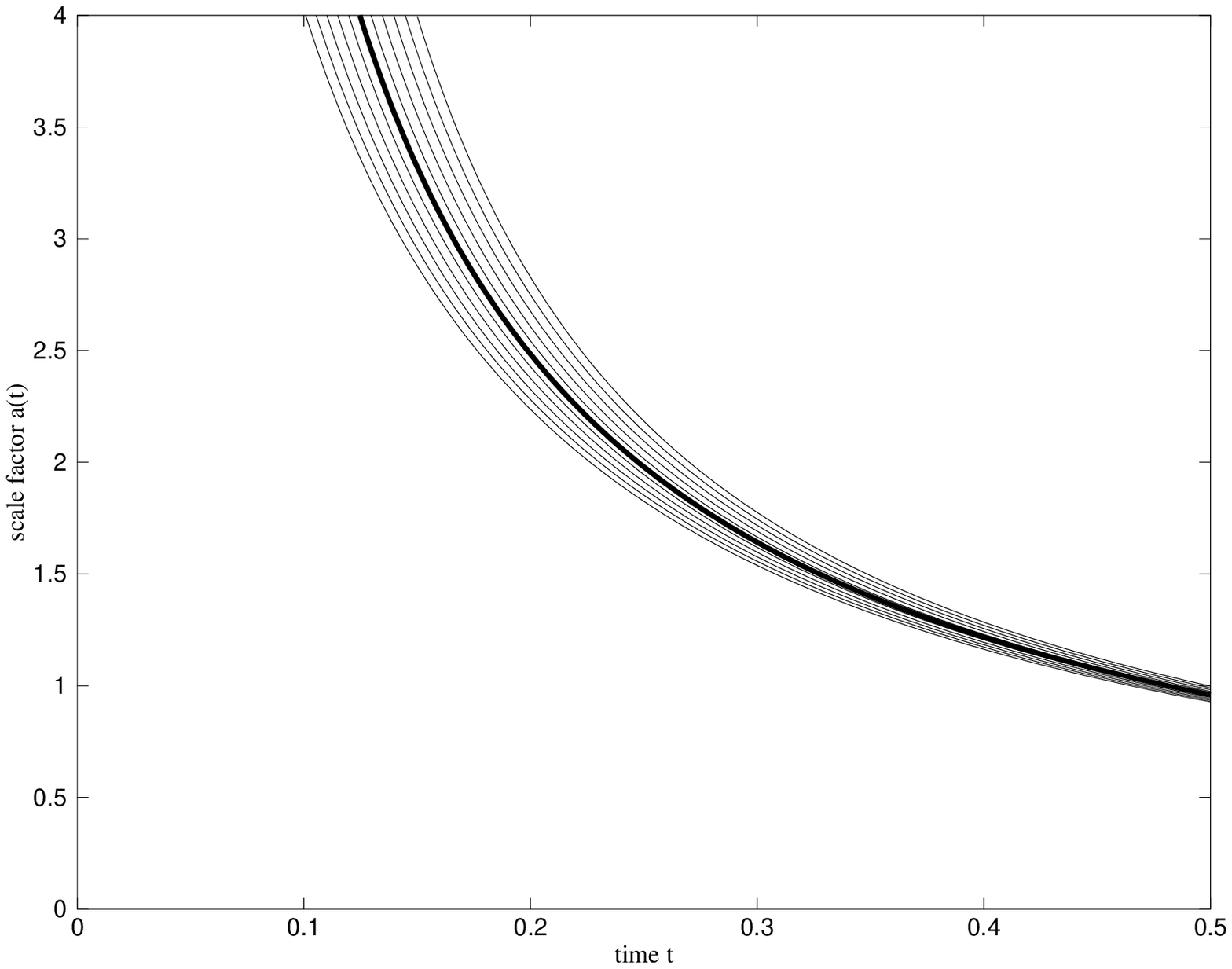}
\end{figure}
\end{example}

\break
\begin{example}
For $\vartheta=\nu=1$ and choice of constant $E=-1$, we take solution 5 in table \ref{tb: exactNLS} with $c_0=-1$, $a_0=1$ and $b_0=0$ so that we have $u(\sigma)=(1/\sigma)e^{-\sigma^2/2}$ and $P(\sigma)=\sigma^2+2/\sigma^2$.  By (\ref{eq: confBINLSR-u})
\begin{equation}R(t)=\frac{1}{u(t-t_0)}=(t-t_0) e^{(t-t_0)^2/2}
\end{equation}
so that 
\begin{eqnarray}a_1(t)&=&R(t)^{1/d}e^{-(\mu_2+\cdots + \mu_d)t/d}\notag\\
&=&(t-t_0)^{1/d} e^{(t-t_0)^2/2d}e^{-(\mu_2+\cdots + \mu_d)t/d},\end{eqnarray}
\begin{eqnarray}a_i(t)&=&e^{\mu_it}a_1(t)\notag\\
&=&(t-t_0)^{1/d} e^{(t-t_0)^2/2d}e^{(-(\mu_2+\cdots + \mu_d)/d+\mu_i)t}
\end{eqnarray}
for constants $\mu_i$ that satisfy
\begin{equation}E=-1=\frac{2}{(d-1)}\displaystyle\sum_{l<k}\mu_l\mu_k-\displaystyle\sum_{j=1}^d\mu_j^2.\end{equation}
By (\ref{eq: psiforP=c^2x^2+2/x^2+b}) with  $\alpha_0=(d-1)/d\kappa$, we have scalar field
\begin{eqnarray}
&&\phi(t)=\psi(t-t_0)\notag\\
&&=\frac{ \sqrt{(d-1)}}{2\sqrt{d\kappa}}\left(\sqrt{(t-t_0)^4+2} + \sqrt{2}\ln\left[\frac{(t-t_0)^2}{2\sqrt{2}\sqrt{(t-t_0)^4+2} +4}\right] \right)+\beta_0\notag\\
\end{eqnarray}
and finally by (\ref{eq: confBINLSVphi-u}), 
\begin{eqnarray}
V(\phi(t))&=&\left[
\frac{(d-1)}{2 d\kappa}\left( (u')^2 - u^2 \right)
-\frac{1}{2}(\psi ')^2u^2-\frac{\Lambda}{\kappa} \right]\circ(t-t_0)\notag\\
&=&\frac{(d-1)}{2 d\kappa}e^{-(t-t_0)^2}\left(\frac{1}{(t-t_0)^2}-\frac{1}{(t-t_0)^4}\right)    
 - \frac{\Lambda}{\kappa}.
\end{eqnarray}

For example when $d=3$, we can take $t_0=0$, $\mu_2=\frac{1}{\sqrt{3}}$ and $\mu_3=\frac{2}{\sqrt{3}}$ to obtain the solution
\begin{eqnarray}
a_1(t)&=&t^{1/3} e^{t^2/6-t/\sqrt{3}}\notag\\
a_2(t)&=&t^{1/3} e^{t^2/6}\notag\\
a_3(t)&=&t^{1/3} e^{t^2/6+t/\sqrt{3}}\notag\\
\phi(t)&=&\frac{1}{\sqrt{6\kappa}}\left(t\sqrt{t^2+2} + \sqrt{2}\ln\left[\frac{t^2}{2\sqrt{2}\sqrt{t^4+2} +4}\right] \right)\notag\\
V(\phi(t))&=&\frac{1}{ 3\kappa}e^{-t^2}\left(\frac{1}{t^2}-\frac{1}{t^4}\right)    
 - \frac{\Lambda}{\kappa}.
\end{eqnarray}

For $t_0=0$, the solver was run with $u, u'$ and $\sigma$ perturbed by $.001$.  The graphs of $R(t)$ below show that the solution is unstable.  The absolute error grows up to three orders of magnitude in the small graphed time interval.  
\begin{figure}[htbp]
\centering
\vspace{-.1in}
\caption{Instability of conformal Bianchi I Example \theexample}
\includegraphics[width=4in]{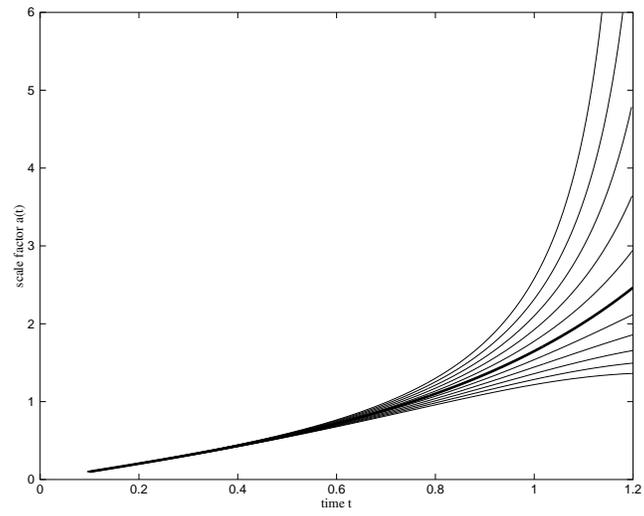}
\vspace{-.5in}
\end{figure}

\end{example}

\chapter{Reformulations of a Bianchi V model}
\label{ch: BV}

For the inhomogeneous, anisotropic Bianchi V metric
\begin{equation} ds^2=-dt^2+X_1^2(t)dx_1^2+\displaystyle\sum_{i=2}^de^{2\beta x_1}X_i(t)^2dx_i^2\label{eq: BVmetric}
\end{equation}
in a $d+1-$dimensional spacetime for $d\neq 0,1$ and $\beta\neq 0$, the nonzero  Einstein's equations $g^{ij}G_{ij}=-\kappa g^{ij}T_{ij}+\Lambda$ are 
\begin{equation}
\displaystyle\sum_{l< k}H_lH_k-\frac{d(d-1)\beta^2}{2X_1^2}
\stackrel{(I_0)}{=}
\kappa\left[\frac{1}{2}\dot\phi^2+V\circ\phi+\rho\right]+\Lambda\label{eq: BVEFEI0Id}\end{equation}
\begin{eqnarray}\displaystyle\sum_{l\neq 1}(\dot{H}_l+H_l^2)+\displaystyle\sum_{\stackrel{l < k}{l,k\neq 1}}H_lH_k-\frac{(d-1)(d-2)\beta^2}{2X_1^2}&\stackrel{(I_1)}{=}&-\kappa\left[\frac{1}{2}\dot\phi^2-V\circ\phi+p\right]+\Lambda\notag\\
&\vdots&\notag\\
\displaystyle\sum_{l\neq i}(\dot{H}_l+H_l^2)+\displaystyle\sum_{\stackrel{l < k}{l,k\neq i}}H_lH_k-\frac{(d-1)(d-2)\beta^2}{2X_1^2}&\stackrel{(I_i)}{=}&-\kappa\left[\frac{1}{2}\dot\phi^2-V\circ\phi+ p\right]+\Lambda\notag\\
&\vdots&\notag\\
\displaystyle\sum_{l\neq d}(\dot{H}_l+H_l^2)+\displaystyle\sum_{\stackrel{l < k}{l,k\neq d}}H_lH_k-\frac{(d-1)(d-2)\beta^2}{2X_1^2}&\stackrel{(I_d)}{=}&-\kappa\left[\frac{1}{2}\dot\phi^2-V\circ\phi+p\right]+\Lambda\notag\end{eqnarray}
where $H_l(t)\stackrel{def.}{=}\dot{a}_l/a_l$ and $i,l,k\in\{1,\dots,d\}$.
There is one more equation, the off-diagonal entry $G_{01}=-\kappa T_{01}+\Lambda g_{01}$ which states
\begin{equation}\beta \displaystyle\sum_{l=2}^d H_l-(d-1)\beta H_1\stackrel{(I_{01})}{=}0.\label{eq: otherBVeqn}\end{equation}

\break

\section{In terms of a Generalized EMP}
\begin{thm}\label{thm: BVEMP}
Suppose you are given twice differentiable functions $X_1(t), \dots, X_d(t)>0$, a once differentiable function $\phi(t)$, and also functions $\rho(t), p(t), V(x)$ which satisfy the Einstein equations $(I_0),\dots,(I_d)$ in (\ref{eq: BVEFEI0Id}) and $(I_{01})$ in (\ref{eq: otherBVeqn}) for some $\Lambda\in\mathds{R}, d\in\mathds{N}\backslash\{0,1\}, \kappa\in\mathds{R}\backslash\{0\}$ and $M\in\mathds{N}$.  Denote 
\begin{equation}R(t)\stackrel{def.}{=} \left(X_1(t)\cdots X_d(t)\right)^\nu\label{eq: BVBIEMPR-X}\end{equation}
for some $\nu\neq 0$.  If $f(\tau)$ is the inverse of a function $\tau(t)$ which satisfies
\begin{equation}\dot{\tau}(t)=\theta R(t)^q\label{eq: BVBIEMPdottau-X}\end{equation}
for some constants $\theta> 0$ and $q\neq 0$, then
\begin{equation}Y(\tau)= R(f(\tau))^q  \qquad\mbox{ and }\qquad Q(\tau)= \frac{q\nu d\kappa}{(d-1)} \varphi'(\tau)^2\label{eq: BVBIEMPY-aQ-varphi}\end{equation}
solve the generalized EMP equation
\begin{equation}Y''(\tau)+Q(\tau)Y(\tau)=\frac{-2q\nu d\kappa D}{\theta^2(d-1) Y(\tau)^{(2+q\nu)/q\nu}}-\frac{q\nu d\beta^2}{\theta^2e^{2\alpha_1}Y(\tau)^{(2+q\nu d)/q\nu d}}\qquad\qquad\notag
\end{equation}
\begin{equation}\qquad\qquad-\displaystyle\sum_{i=1}^M \frac{q\nu d\kappa\left(\varrho_i(\tau)+\textup{\textlhookp}_i(\tau)\right)}{\theta^2(d-1)Y(\tau)}\label{eq: BVBIEMP}\end{equation}
for some $\alpha_1\in\mathds{R}$,
\begin{equation}\varphi(\tau)=\phi(f(\tau))\label{eq: BVBIEMPvarphi-phi}\end{equation}
\begin{equation}\varrho(\tau)=\rho(f(\tau)), \ \textup{\textlhookp}(\tau)=p(f(\tau)) \label{eq: BVBIEMPvarrho-rhohookp-p}\end{equation}
and where 
\begin{equation}D\stackrel{def.}{=}\frac{1}{2d\kappa}X_1^2X_2^2\cdots X_d^2\left(\displaystyle\sum_{l<k}\eta_{lk}^2\right)\label{eq: BVBIEMPDconstant-X}\end{equation}
is a constant for
\begin{equation}\eta_{lk}\stackrel{def.}{=}H_l-H_k, \mbox{\small \    $l\neq k, l, \ k\in\{1,\dots, d\}$.}\label{eq: BVBIEMPeta-X}\end{equation}

Conversely, suppose you are given a twice differentiable function $Y(\tau)>0$, a continuous function $Q(\tau)$, and also functions  $\varrho(\tau), \textup{\textlhookp}(\tau)$ which solve (\ref{eq: BVBIEMP}) for some constants $ \theta>0$ and $q, \nu, \kappa \in\mathds{R}\backslash\{0\}, k, D, \alpha_1,\beta\in\mathds{R}, d\in\mathds{N}\backslash\{0,1\}$.  In order to construct functions which solve $(I_0),\dots, (I_d), (I_{01})$, first find $\tau(t), \varphi(\tau)$ which solve the differential equations
\begin{equation}\dot{\tau}(t)=\theta Y(\tau(t))\qquad\mbox{ and }\qquad \varphi '(\tau)^2= \frac{(d-1)}{q\nu d\kappa} Q(\tau).\label{eq: BVBIEMPdottau-Yvarphi-Q}\end{equation}
Next find a function $\sigma(t)$ such that 
\begin{equation}\dot\sigma(t)=\frac{1}{\dot\tau(t)^{1/q\nu}}\label{eq: BVBIEMPdotsigma-dottau}\end{equation}
and let 
\begin{equation}R(t)=Y(\tau(t))^{1/q}\qquad\mbox{ and }\qquad\alpha_l(t)\stackrel{def.}{=}c_l\sigma(t), \ \mbox{\small $l\in\{2,\dots,d\}$}\label{eq: BVBIEMPR-Yalpha-sigma}\end{equation} 
where $c_l$ are any constants for which both
\begin{equation}\displaystyle\sum_{l=2}^d c_l=0\qquad\mbox{ and }\qquad\qquad\displaystyle\sum_{2\leq l<k\leq d}c_lc_k=-\theta^{2/q\nu} D\kappa.\label{eq: BVBIEMPccondition}\end{equation}
Then the functions
\begin{equation}X_1(t)=R(t)^{1/\nu d}e^{\alpha_1}, \ \alpha_1\in\mathds{R}\label{eq: BVX1-constR}\end{equation}
\begin{equation}X_l(t)=R(t)^{1/\nu d}e^{\alpha_l(t)},  \ \ 2\leq l \leq d \label{eq: BVBIEMPX-Ralpha}\end{equation}
\begin{equation}\phi(t)=\varphi(\tau(t))\label{eq: BVBIEMPphi-varphi}\end{equation}
\begin{equation}\rho(t)=\varrho(\tau(t)), \ p(t)=\textup{\textlhookp}(\tau(t))\label{eq: BVBIEMPrho-varrhop-hookp}\end{equation}
and
\begin{equation}
V(\phi(t))=\left[\frac{(d-1)\theta^2}{2\nu^2 d\kappa q^2}(Y')^2-\frac{D}{Y^{2/q\nu}}-\frac{d(d-1)\beta^2}{2\kappa e^{2\alpha_1} Y^{2/q\nu d}}-\frac{\theta^2}{2}Y^2(\varphi ')^2-\varrho-\frac{\Lambda}{\kappa}\right]\circ\tau(t)\label{eq: BVBIEMPVphi-Y}
\end{equation}
satisfy the Einstein equations $(I_0),\dots,(I_d)$.
\end{thm}

\begin{proof}
This proof will implement Theorem \ref{thm: EFE-EMP} with constants and functions as indicated in the following table.

   \begin{table}[ht]
\centering
\caption{{ Theorem \ref{thm: EFE-EMP} applied to Bianchi V}}\label{tb: BVEMP}
\vspace{.2in}
\begin{tabular}{r | l c r | l}
In Theorem & substitute & & In Theorem& substitute \\[4pt]
\hline
\raisebox{-5pt}{$a(t)$} & \raisebox{-5pt}{$R(t)$} && \raisebox{-5pt}{$N$ }& \raisebox{-5pt}{$1$}\\[8pt]
$\delta  $ &$0$      &&    $\varepsilon $ & ${\nu d\kappa}/{(d-1)}$\\[8pt]
$G_0(t)$ & $\mbox{constant } -2\nu d\kappa D /(d-1)$    &&     $A_0$ & $2/\nu$ \\[8pt]
 $G_{1}(t)$ & $ \frac{-\nu d\kappa}{(d-1)}(\rho(t)+p(t))$  &&      $A_{1}$   &$0$\\[8pt]
 $G_{2}(t)$ & constant $-\nu d\beta^2 e^{-2\alpha_1}$  &&      $A_{2}$   &$2/\nu d$\\[8pt]
 $ \lambda_0(\tau)$ &constant ${-2q\nu d \kappa D}/{\theta^2(d-1)}$       &&      $B_0$ &$(2+q\nu)/q\nu$ \\[8pt] 
$\lambda_{1}(\tau)$&$\frac{-q\nu d\kappa}{\theta^2(d-1)}(\varrho(\tau)+\textup{\textlhookp}(\tau))$ &&        $B_{1}$&  $1$ \\[6pt]
$\lambda_{2}(\tau)$&constant $-q\nu d\beta^2/\theta^2 e^{2\alpha_1}$ &&        $B_{2}$&  $(2+q\nu d)/q\nu d$ \\[6pt]
\hline
\end{tabular}
\end{table}

To prove the forward implication, we assume to be given functions which solve the Einstein field equations $(I_0),\dots,(I_d)$.  Forming the linear combination $d(I_0)-\displaystyle\sum_{i=1}^d(I_i)$ of Einstein's equations, we have that
\begin{equation}
d\displaystyle\sum_{l < k}H_lH_k-\displaystyle\sum_{i=1}^d \displaystyle\sum_{l\neq i}(\dot{H}_l+H_l^2)-\displaystyle\sum_{i=1}^d\displaystyle\sum_{\stackrel{l<k}{l,k\neq i}}H_lH_k-\frac{d(d-1)\beta^2}{X_1^2}=d\kappa\left[\dot\phi^2+ \left(\rho+ p\right)\right]\label{eq: BVBIEMPdI1minusIi}
\end{equation}
since $-d^2(d-1)+d(d-1)(d-2)=d(d-1)(-d+d-2)=-2d(d-1)$ and
where the summing indices $l, k\in\{1, \dots, d\}$.  The first double sum in (\ref{eq: BVBIEMPdI1minusIi}) contains the quantity $\frac{\ddot{X}_l}{X_l}$  $(d-1)$-times for any fixed $l$, and the second double sum contains the quantity $\frac{\dot{X}_l\dot{X}_k}{X_lX_k}$ $(d-2)$-times for any fixed $l,k$ pair with $l<k$.  Therefore (\ref{eq: BVBIEMPdI1minusIi}) simplifies to\\
\begin{equation}d\displaystyle\sum_{l<k}H_lH_k-(d-1)\displaystyle\sum_{l=1}^d(\dot{H}_l+H_l^2)-(d-2)\displaystyle\sum_{l<k}H_lH_k-\frac{d(d-1)\beta^2}{X_1^2}\qquad\qquad\qquad\notag\end{equation}
\begin{equation}\qquad\qquad\qquad\qquad\qquad\qquad\qquad=d\kappa\left[\dot\phi^2+ \left(\rho+p\right)\right].\end{equation}
Collecting the first and third sums, we have 
\begin{equation}2\displaystyle\sum_{l<k}H_lH_k-(d-1)\displaystyle\sum_{l=1}^d(\dot{H}_l+H_l^2)-\frac{d(d-1)\beta^2}{X_1^2}=d\kappa\left[\dot\phi^2+\left(\rho+p\right)\right].\label{eq: BVBIEMPI1minusIicollected}\end{equation}
Using the definition (\ref{eq: BVBIEMPR-X}) of $R(t)$, we define
\begin{equation}H_R\stackrel{def.}{=}\frac{\dot{R}}{R}=\frac{\nu \left(X_1\cdots X_d\right)^{\nu-1}\left(\dot{X}_1X_2\cdots X_d + \cdots + X_1X_2\cdots \dot{X}_d\right)}{\left(X_1\cdots X_d\right)^\nu}=\nu\displaystyle\sum_{l=1}^d H_l.\label{eq: BVBIEMPH-X}\end{equation}
Differentiating $H_R$ gives
\begin{equation}\dot{H}_R=\nu\displaystyle\sum_{l=1}^d\dot{H}_l.\label{eq: BVBIEMPdotH-X}\end{equation}
Therefore (\ref{eq: BVBIEMPI1minusIicollected}) can be written as
\begin{equation}2\displaystyle\sum_{l<k}H_lH_k-(d-1)\left(\frac{1}{\nu}\dot{H}_R+\displaystyle\sum_{l=1}^dH_l^2
\right)-\frac{d(d-1)\beta^2}{X_1^2}=d\kappa\left[\dot\phi^2+ \left(\rho+p\right)\right].\end{equation}
Multiplying by $\frac{-\nu}{(d-1)}$ and rearranging, we get
\begin{equation}\dot{H}_R+\frac{\nu}{(d-1)}\left((d-1)\displaystyle\sum_{l=1}^dH_l^2-2\displaystyle\sum_{l<k}H_lH_k\right)+\frac{\nu d\beta^2}{X_1^2}=\frac{-\nu d\kappa}{(d-1)}\left[\dot{\phi}^2+\left(\rho+p\right)
\right].\label{eq: BVBIEMPI1minusIiwithdotH}\end{equation}
By definition (\ref{eq: BVBIEMPeta-X}) of the quantities $\eta_{lk}$, we have 
\begin{equation}
\displaystyle\sum_{l<k}\eta_{lk}^2=\displaystyle\sum_{l<k}\left(H_l^2-2H_lH_k+H_k^2\right).\label{eq: BVBIEMPetalksqrddefn}\end{equation}
The first and last terms on the right-hand side of (\ref{eq: BVBIEMPetalksqrddefn}) sum to
\begin{eqnarray}
\displaystyle\sum_{l<k}\left(H_l^2+H_k^2\right)
&=&\displaystyle\sum_{l=1}^{d-1}\displaystyle\sum_{k=l+1}^dH_l^2
+\displaystyle\sum_{k=2}^d\displaystyle\sum_{l=1}^{k-1}H_k^2\notag\\
&=&\sum_{l=1}^{d-1}(d-l)H_l^2+\sum_{k=2}^d(k-1)H_k^2\notag\\
&=&(d-1)H_1^2+ \sum_{j=2}^{d-1}H_j^2 +\sum_{j=2}^{d-1}(j-1)H_j^2+(d-1)H_d^2\notag\\
&=&(d-1)\sum_{j=1}^dH_j^2\end{eqnarray}
therefore (\ref{eq: BVBIEMPetalksqrddefn}) becomes
\begin{equation}
\displaystyle\sum_{l<k}\eta_{lk}^2
=(d-1)\displaystyle\sum_{l=1}^dH_l^2-2\displaystyle\sum_{l<k}H_lH_k.
\end{equation}
Using this to rewrite (\ref{eq: BVBIEMPI1minusIiwithdotH}) shows that
\begin{equation}\dot{H}_R+\frac{\nu}{(d-1)}\displaystyle\sum_{l<k}\eta_{lk}^2+\frac{\nu d\beta^2}{X_1^2}=\frac{-\nu d\kappa}{(d-1)}\left[\dot{\phi}^2+ \left(\rho+p\right)\right].\label{eq: BVBIEMPI1minusIiwitheta}
\end{equation}

Next we will confirm that $D$ is a constant.  Since the right-hand sides of Einstein equations $(I_i)$ are the same for all $i\in\{1, \dots, d\}$, by equating the left-hand sides of any two equations $(I_i)$ and $(I_j)$ for $i\neq j$,
\begin{equation}\displaystyle\sum_{l\neq i}(\dot{H}_l+H_l^2)+\displaystyle\sum_{\stackrel{l<k}{l,k\neq i}}H_lH_k=\displaystyle\sum_{l\neq j}(\dot{H}_l+H_l^2)+\displaystyle\sum_{\stackrel{l<k}{l,k\neq j}}H_lH_k\label{eq: BVBIEMPlhsIi=lhsIj}\end{equation}
where we recall that the sum indices $l,k\in\{1, \dots, d\}$.  For the first sum on each side, the left and the right-hand sides of (\ref{eq: BVBIEMPlhsIi=lhsIj}) contain all the same terms, except for the $j^{th}$ indexed term which appears on the left, and the $i^{th}$ indexed term which appears on the right.  Therefore many terms cancel and we are left with
\begin{equation}\dot{H}_j+H_j^2+\displaystyle\sum_{\stackrel{l<k}{l,k\neq i}}H_lH_k=\dot{H}_i+H_i^2+\displaystyle\sum_{\stackrel{l<k}{l,k\neq j}}H_lH_k.\label{eq: BVBIEMPlhsIi=lhsIjfirstsumcancelled}\end{equation}
For the second (double) sum on each side, the left and right-hand sides of (\ref{eq: BVBIEMPlhsIi=lhsIjfirstsumcancelled}) contain all the same terms, except for the terms where either $l,k=j$ which appear on the left, and the terms where either $l,k=i$ which appear on the right.  Therefore many terms cancel and we are left with 
\begin{equation}\dot{H}_j+H_j^2+\displaystyle\sum_{{l\neq j}}H_jH_l=\dot{H}_i+H_i^2+\displaystyle\sum_{{l\neq i}}H_iH_l.\label{eq: BVBIEMPlhsIi=lhsIjsumscancelled}\end{equation}
Adding $-H_j^2+H_j^2=0$ to the left and $-H_i^2+H_i^2=0$ to the right, we obtain
\begin{equation}\dot\eta_{ij}+\frac{1}{\nu}\eta_{ij}H_R=0\label{eq: BVBIEMPdotf+fH=0withetaH}\end{equation}
where we have used the expression (\ref{eq: BVBIEMPH-X}) for $H_R$, the definition (\ref{eq: BVBIEMPeta-X}) of $\eta_{ij}$, and as usual dot denotes differentiation with respect to $t$.   By Lemma \ref{lem: shortlemma} with $\mu=1/\nu\neq 0$, which applies since $R(t)$ is positive and differentiable, (\ref{eq: BVBIEMPdotf+fH=0withetaH}) shows that the function $f=\eta_{ij}R^{1/\nu}=\eta_{ij}X_1X_2\cdots X_d$ is constant for any pair $i,j$ (for the pair $i=j$, $f$ is clearly a constant function, namely zero).  Therefore the definition (\ref{eq: BVBIEMPDconstant-X}) of $D$ is also constant, being proportional to a sum of squares of these constant functions.  By the definitions (\ref{eq: BVBIEMPDconstant-X}) and (\ref{eq: BVBIEMPR-X}) of the constant $D$ and the function $R(t)$, we now rewrite (\ref{eq: BVBIEMPI1minusIiwitheta}) as
\begin{equation}\dot{H}_R+\frac{\nu d\beta^2}{X_1^2}=\frac{-\nu d \kappa}{(d-1)}\left[\dot{\phi}^2+\frac{2D}{R^{2/\nu}}+\left(\rho+p\right)
\right].\label{eq: BVBIEMPI1minusIiEFE}\end{equation}
Finally, we use the remaining Einstein equation $(I_{01})$ and the relation $H_R=\nu\displaystyle\sum_{l=1}^dH_l$ in (\ref{eq: BVBIEMPH-X}) to obtain
\begin{equation}H_R=\nu d H_1.\label{eq: BVEMPHR-H1}\end{equation}
That is, $\frac{d}{dt}\ln(R)=\nu d \frac{d}{dt}\ln(X_1) $ which implies $\ln(R)+c_0=\nu d \ln(X_1)$ for some integration constant $c_0\in\mathds{R}$ so that 
\begin{equation}X_1= e^{\alpha_1}R^{1/\nu d}\label{eq: BVEMPX1-R}\end{equation}
for arbitrary constant $\alpha_1=-c_0/\nu d \in\mathds{R}$.  Therefore
 (\ref{eq: BVBIEMPI1minusIiEFE}) becomes
\begin{equation}\dot{H}_R=\frac{-\nu d \kappa}{(d-1)}\left[\dot{\phi}^2+\frac{2D}{R^{2/\nu}}+\frac{(d-1)\beta^2}{\kappa e^{2\alpha_1}R^{2/\nu d}}+\left(\rho+p\right)
\right].\label{eq: BVBIEMPI1minusIiEFEnoX1}\end{equation}
This shows that $R(t), \phi(t), \rho(t)$ and $p(t)$ satisfy the hypothesis of Theorem \ref{thm: EFE-EMP}, applied with  constants $\epsilon, \varepsilon , N, A_0, \dots, A_{N}$ and functions $a(t), G_0(t),\dots,G_{N}(t)$ according to Table \ref{tb: BIEMP}.  Since $\tau(t), Y(\tau), Q(\tau)$ and $\varphi(\tau)$ defined in (\ref{eq: BVBIEMPdottau-X}), (\ref{eq: BVBIEMPY-aQ-varphi}) and (\ref{eq: BVBIEMPvarphi-phi})  are equivalent to that in the forward implication of Theorem \ref{thm: EFE-EMP}, by this theorem and by definition (\ref{eq: BVBIEMPvarrho-rhohookp-p}) of $\varrho(\tau), \textup{\textlhookp}(\tau)$, the generalized EMP equation (\ref{eq: AEMP}) holds for constants $B_0, \dots, B_{N}$ and functions $\lambda_0(\tau), \dots, \lambda_{N}(\tau)$ as indicated in Table \ref{tb: BIEMP}.  This proves the forward implication.

To prove the converse implication, we assume to be given functions which solve the generalized EMP equation (\ref{eq: BVBIEMP}) and we begin by showing that $(I_0)$ is satisfied.   Differentiating the definition of $R(t)$ in (\ref{eq: BVBIEMPR-Yalpha-sigma}) and using the definition in (\ref{eq: BVBIEMPdottau-Yvarphi-Q}) of $\tau(t)$ shows that
\begin{eqnarray}\dot{R}(t)&=&\frac{1}{q}Y(\tau(t))^{\frac{1}{q}-1}Y'(\tau(t))\dot{\tau}(t)\notag\\
&=&\frac{\theta}{q}Y(\tau(t))^{1/q}Y'(\tau(t)).\end{eqnarray}
Dividing by $R(t)$ gives
\begin{equation}H_R(t)\stackrel{def.}{=}\frac{\dot{R}(t)}{R(t)}=\frac{\theta}{q}Y'(\tau(t)).\label{eq: BVBIEMPH-Y}\end{equation}
Differentiating the definition (\ref{eq: BVBIEMPphi-varphi}) of $\phi(t)$ and using definition (\ref{eq: BVBIEMPdottau-Yvarphi-Q}) of  $\tau(t)$,   we get
\begin{equation}\dot{\phi}(t)=\varphi '(\tau(t))\dot{\tau}(t)=\theta\varphi '(\tau(t))Y(\tau(t)).\label{eq: BVBIEMPdotphi-Y}\end{equation}
Using (\ref{eq: BVBIEMPH-Y}) and (\ref{eq: BVBIEMPdotphi-Y}), and also the definitions (\ref{eq: BVBIEMPR-Yalpha-sigma}) and (\ref{eq: BVBIEMPrho-varrhop-hookp}) of $R(t)$ and $\rho(t)$ respectively, the definition (\ref{eq: BVBIEMPVphi-Y}) of $V\circ\phi$ can be written as
\begin{equation}
V\circ\phi=\frac{1}{\kappa}\left(
\frac{(d-1)}{2\nu^2 d}H_R^2-\frac{D\kappa}{R^{2/\nu}}-\frac{d(d-1)\beta^2}{2 e^{2\alpha_1}R^{2/\nu d}}
\right)-\frac{1}{2}\dot\phi^2-\rho-\frac{\Lambda}{\kappa}
\label{eq: BVBIEMPVphi-a}.\end{equation}

The quantity in parenthesis here is in fact equal to the left-hand-side of equation $(I_0)$.  To see this, we differentiate the definition (\ref{eq: BVBIEMPX-Ralpha}) of $X_l(t)$, divide the result by $X_l$ and use the definition (\ref{eq: BVBIEMPH-Y}) of $H_R$ to obtain
\begin{equation}H_l\stackrel{def.}{=}\frac{\dot{X}_l}{X_l}=\frac{\frac{1}{\nu d}R^{{1}/{\nu d}-1}\dot{R}e^{\alpha_l}+\dot{\alpha}_lR^{{1}/{\nu d}}e^{\alpha_l}}{R^{{1}/{\nu d}}e^{\alpha_l}}=\frac{1}{\nu d}H_R+\dot\alpha_l.\label{eq: BVBIEMPdotX/X-Halpha}\end{equation}
Therefore we have
\begin{equation}\displaystyle\sum_{l<k}H_lH_k
=\displaystyle\sum_{l<k}\left(\frac{1}{\nu^2 d^2}H_R^2+\frac{1}{\nu d}(\dot\alpha_l+\dot\alpha_k)H_R+\dot\alpha_l \dot\alpha_k\right).
\label{eq: BVBIEMPdotXldotXk/XlXk-Halpha}\end{equation}
The first term on the right-hand side of (\ref{eq: BVBIEMPdotXldotXk/XlXk-Halpha}) does not depend on the indices $l,k$, and is therefore equal to $\frac{1}{\nu^2 d^2}H_R^2$ times the quantity
\begin{equation}\displaystyle\sum_{l<k}1=\displaystyle\sum_{k=2}^d\displaystyle\sum_{l=1}^{k-1}1=\displaystyle\sum_{k=2}^d(k-1)=\displaystyle\sum_{j=1}^{d-1}j=\frac{d(d-1)}{2}.\label{eq: BVBIEMPsum1l<k}\end{equation}
The second term on the right-hand side of (\ref{eq: BVBIEMPdotXldotXk/XlXk-Halpha}) sums to zero since 
\begin{eqnarray}
\displaystyle\sum_{l<k}(\dot\alpha_l+\dot\alpha_k)
&=&\displaystyle\sum_{l=1}^{d-1}\displaystyle\sum_{k=l+1}^d\dot\alpha_l+\displaystyle\sum_{k=2}^d\displaystyle\sum_{l=1}^{k-1}\dot\alpha_k\notag\\
&=&\displaystyle\sum_{l=1}^{d-1}(d-l)\dot\alpha_l +\displaystyle\sum_{k=2}^d (k-1)\dot\alpha_k\notag\\
&=&(d-1)\dot\alpha_1+\displaystyle\sum_{j=2}^{d-1}(d-j+j-1)\dot\alpha_j + (d-1)\dot\alpha_d\notag\\
&=&(d-1)\displaystyle\sum_{l=1}^d\dot\alpha_l\notag\\
&=&(d-1)\displaystyle\sum_{l=2}^d\dot\alpha_l\mbox{ since $\alpha_1$ constant}\notag\\
&=&(d-1)\sigma(t)\displaystyle\sum_{l=2}^d c_l\notag\\
&=&0\label{eq: BVBIEMPalphal+alphak=0}\end{eqnarray}
where on the last lines, we have used the definition (\ref{eq: BVBIEMPR-Yalpha-sigma}) of $\alpha_l(t)$ for $l\in\{2,\dots, d\}$ and the condition (\ref{eq: BVBIEMPccondition}) on the constants $c_l$.  For the third term on the right-hand side of   (\ref{eq: BVBIEMPdotXldotXk/XlXk-Halpha}), we use the definitions of $\alpha_l(t), \sigma(t), \tau(t)$ and $R(t)$ in (\ref{eq: BVBIEMPR-Yalpha-sigma}), (\ref{eq: BVBIEMPdotsigma-dottau}) and (\ref{eq: BVBIEMPdottau-Yvarphi-Q}) to write
\begin{equation}\dot\alpha_l\dot\alpha_k=c_l c_k \dot\sigma^2=\frac{c_lc_k}{\dot\tau^{2/q\nu}}=\frac{c_lc_k}{\theta^{2/q\nu}(Y\circ\tau)^{2/q\nu}}=\frac{c_lc_k}{\theta^{2/q\nu}R^{2/\nu}}\label{eq: BVBIEMPdotalphaldotalphak-R}\end{equation}
for $l\in\{2, \dots, d\}$.
Therefore (\ref{eq: BVBIEMPdotXldotXk/XlXk-Halpha}) becomes
\begin{equation}\displaystyle\sum_{l<k}H_lH_k=\frac{(d-1)}{2\nu^2 d}H_R^2+\displaystyle\sum_{2\leq l<k\leq}\frac{c_lc_k}{\theta^{2/q\nu}R^{2/\nu}}
\label{eq: BVBIEMPdotXldotXk/XlXk-Rc}
\end{equation}
since $\dot\alpha_1=0$. 
Then by the condition (\ref{eq: BVBIEMPccondition}) on the constants $c_l$, (\ref{eq: BVBIEMPdotXldotXk/XlXk-Rc}) becomes
\begin{equation}\displaystyle\sum_{l<k}H_lH_k=\frac{(d-1)}{2\nu^2 d}H_R^2-\frac{D\kappa}{R^{2/\nu}}.
\label{eq: BVBIEMPdotXldotXk/XlXk-RD}
\end{equation}
Using this and the definition (\ref{eq: BVX1-constR}) of $X_1(t)=e^{\alpha_1}R(t)^{1/\nu d}$, the expression (\ref{eq: BVBIEMPVphi-a}) for $V$ can now be written as
\begin{equation}V\circ\phi=\frac{1}{\kappa}\left(\displaystyle\sum_{l<k}H_lH_k-\frac{d(d-1)\beta^2}{2X_1^2}\right)-\frac{1}{2}\dot{\phi}^2-\rho-\frac{\Lambda}{\kappa},\label{eq: BVBIEMPVphi-X}
\end{equation}
showing that $(I_0)$ holds under the assumptions of the converse implication.

To conclude the proof we must also show that the equations $(I_1), \dots, (I_d)$ hold.   In the converse direction the hypothesis of the converse of Theorem \ref{thm: EFE-EMP} holds, applied with constants $N, B_0, \dots, B_{N}$ and functions $\lambda_0(\tau), \dots, \lambda_{N}(\tau)$ as indicated in Table \ref{tb: BIEMP}.  Since  $\tau(t), \varphi(\tau), R(t)$ and $\phi(t)$ defined in (\ref{eq: BVBIEMPdottau-Yvarphi-Q}), (\ref{eq: BVBIEMPR-Yalpha-sigma}) and (\ref{eq: BVBIEMPphi-varphi}) are consistent with the converse implication of Theorem \ref{thm: EFE-EMP}, applied with $a(t), \delta  $ and $\varepsilon $ as in Table \ref{tb: BIEMP}, by this theorem and by the definition (\ref{eq: BVBIEMPrho-varrhop-hookp}) of $\rho(t), p(t)$ the  scale factor  equation (\ref{eq: AEFE}) holds for constants $\delta  , \varepsilon , A_0, \dots, A_{N}$ and functions $G_0(t),\dots,G_{N}(t)$ according to Table \ref{tb: BIEMP}.  That is, we have regained (\ref{eq: BVBIEMPI1minusIiEFE}).  Now solving (\ref{eq: BVBIEMPVphi-a}) for $\rho(t)$ and substituting this into (\ref{eq: BVBIEMPI1minusIiEFE}), we obtain
\begin{equation}\dot{H}_R=\frac{-\nu d\kappa}{(d-1)}\left[ \frac{1}{2}\dot\phi^2 -V\circ\phi+\frac{D}{R^{2/\nu}}-\frac{d(d-1)\beta^2}{2\kappa e^{2\alpha_1}R^{2/\nu d}}+p +\frac{(d-1)}{2\nu^2 d\kappa}H_R^2-\frac{\Lambda}{\kappa} \right].\end{equation}
Multiplying by $\frac{(d-1)}{\nu d}$ and rearranging we get
\begin{equation}\frac{(d-1)}{\nu d}\dot{H}_R+  \frac{(d-1)}{2\nu^2 d}H_R^2   +\frac{\kappa D}{R^{2/\nu}}-\frac{d(d-1)\beta^2}{2e^{2\alpha_1}R^{2/\nu d}}=-\kappa\left[\frac{1}{2}\dot\phi^2-V\circ\phi+p\right] + \Lambda\label{eq: BVBIEMPlhs=lhsofIi}\end{equation}
The left-hand side of this equation is in fact equal to the left-hand-side of $(I_i)$ for any $i\in\{1,\dots, d\}$. 
  To see this, first recall (\ref{eq: BVBIEMPdotX/X-Halpha}) and write
  \begin{equation}\dot{H}_l+H_l^2=\frac{1}{\nu d}\dot{H}_R+\frac{1}{\nu^2 d^2}H_R^2+\ddot{\alpha}_l+\frac{2}{\nu d}\dot{\alpha}_lH_R+\dot{\alpha}_l^2,\label{eq: BVBIEMPddotX/X-Halpha}
\end{equation}
therefore for any fixed $i$
\begin{equation}\displaystyle\sum_{l\neq i}(\dot{H}_l+H_l^2)=\displaystyle\sum_{l\neq i}\left(\frac{1}{\nu d}\dot{H}_R+\frac{1}{\nu^2 d^2}H_R^2+\ddot{\alpha}_l+\frac{2}{\nu d}\dot{\alpha}_lH_R+\dot{\alpha}_l^2\right).\label{eq: BVBIEMPsumlneqiddotX/X}\end{equation}
The last term on the right-hand side sums to zero since
\begin{eqnarray}
\displaystyle\sum_{l\neq i}\dot\alpha_l^2
&=&\sum_{l\neq i}\dot\alpha_l (-\dot\alpha_i-\sum_{k\neq i, k\neq l}\dot\alpha_k)\notag\\
&=&-\dot\alpha_i\sum_{l\neq i}\dot\alpha_l-2\sum_{\stackrel{l<k}{l,k\neq i}}\dot\alpha_l\dot\alpha_k\notag\\
&=&(\sum_{k\neq i}\dot\alpha_k)(\sum_{l\neq i}\dot\alpha_l)-2\sum_{\stackrel{l<k}{l,k\neq i}}\dot\alpha_l\dot\alpha_k\notag\\
&=&0.\end{eqnarray}
Since the first two terms on the right-hand side of (\ref{eq: BVBIEMPsumlneqiddotX/X}) do not depend on the indices $l,k$, and
also using the definitions (\ref{eq: BVBIEMPR-Yalpha-sigma}) and (\ref{eq: BVBIEMPccondition}) of $\alpha_l$ and the constants $c_l$ to write $\sum_l \dot\alpha_l=0 \Rightarrow \sum_{l\neq i}\dot\alpha_l=-\dot\alpha_i$, (\ref{eq: BVBIEMPsumlneqiddotX/X}) becomes
\begin{equation}\displaystyle\sum_{l\neq i}(\dot{H}_l+H_l^2)=\frac{(d-1)}{\nu d}\dot{H}_R+\frac{(d-1)}{\nu^2 d^2}H_R^2- \frac{2}{\nu d}H_R\dot\alpha_i+\displaystyle\sum_{l\neq i}\ddot\alpha_l .
\label{eq: BVBIEMPddotX/X-H}\end{equation}
By the definitions (\ref{eq: BVBIEMPR-Yalpha-sigma}), (\ref{eq: BVBIEMPdotsigma-dottau}) and (\ref{eq: BVBIEMPdottau-Yvarphi-Q}) of $\alpha_l(t), \sigma(t), \tau(t)$ and $R(t)$,
\begin{eqnarray}\dot\alpha_l(t)R(t)^{1/\nu}&=&c_l\dot\sigma(t)R(t)^{1/\nu}\notag\\
&=&\frac{c_l}{\dot\tau(t)^{1/q\nu}}R(t)^{1/\nu}\notag\\
&=&\frac{c_l}{\theta^{1/q\nu}Y(\tau(t))^{1/q\nu}}R(t)^{1/\nu}\notag\\
&=&\frac{c_l}{\theta^{1/q\nu}}\mbox{ is a constant}\label{eq: BVBIEMPdotalphaRconstant}\end{eqnarray}
for $l\in\{2,\dots,d\}$.
By Lemma \ref{lem: shortlemma} with $\mu=1/\nu\neq 0$, which applies since $R(t)$ is a positive differentiable function, (\ref{eq: BVBIEMPdotalphaRconstant}) shows that 
\begin{equation}\ddot\alpha_l+\frac{1}{\nu}\dot\alpha_lH_R=0\end{equation}
for all $l\in\{2, \dots, d\}$.  Therefore in total, we have that (\ref{eq: BVBIEMPsumlneqiddotX/X}) is
\begin{equation}\displaystyle\sum_{l\neq i}(\dot{H}_l+H_l^2)=\frac{(d-1)}{\nu d}\dot{H}_R+\frac{(d-1)}{\nu^2 d^2}H_R^2+\frac{(2-d)}{\nu d}H_R\displaystyle\sum_{l\neq i}\dot\alpha_l.
\label{eq: BVBIEMPddotX/X-alphaH}\end{equation}

\noindent To form the rest of the left-hand side of $(I_i)$, again use (\ref{eq: BVBIEMPdotX/X-Halpha}) to obtain
\begin{equation}H_lH_k=\frac{1}{\nu^2 d^2}H_R^2+\frac{1}{\nu d}(\dot\alpha_l+\dot\alpha_k)H_R+\dot\alpha_l\dot\alpha_k,\end{equation}
therefore for any fixed $i$
\begin{equation}\displaystyle\sum_{\stackrel{l<k}{l,k\neq i}}H_lH_k
=\displaystyle\sum_{\stackrel{l<k}{l,k\neq i}}\left(\frac{1}{\nu^2 d^2}H_R^2+\frac{1}{\nu d}(\dot\alpha_l+\dot\alpha_k)H_R+\dot\alpha_l\dot\alpha_k\right).\label{eq: BVBIEMPsuml<klkneqidotXldotXk/XlXk-alhpaH}\end{equation}
As we saw in (\ref{eq: BVBIEMPsum1l<k}), $\sum_{l<k}1=\frac{d(d-1)}{2}$ therefore the first term on the right-side of  (\ref{eq: BVBIEMPsuml<klkneqidotXldotXk/XlXk-alhpaH}), which does not depend on the indices $l,k$, is equal to $\frac{1}{\nu^2 d^2}H_R^2$ times
\begin{equation}\displaystyle\sum_{\stackrel{l<k}{l,k\neq i}}1=\displaystyle\sum_{l<k}1-\displaystyle\sum_{l\neq i}1=\frac{d(d-1)}{2}-(d-1)=\frac{(d-1)(d-2)}{2}.\label{eq: BVBIEMPsum1l<klkneqi}\end{equation}
As we saw in (\ref{eq: BVBIEMPalphal+alphak=0}), $\sum_{l<k}(\dot\alpha_l+\dot\alpha_k)=0$ therefore the second term on the right-hand side of (\ref{eq: BVBIEMPsuml<klkneqidotXldotXk/XlXk-alhpaH}) sums to $\frac{1}{\nu d}H_R$ times
\begin{eqnarray}
\sum_{\stackrel{l<k}{l,k\neq i}}(\dot\alpha_l+\dot\alpha_k)
&=&\sum_{{l<k}}(\dot\alpha_l+\dot\alpha_k)-\sum_{l\neq i}(\dot\alpha_l+\dot\alpha_i)\notag\\
&=&-\sum_{l\neq i}\dot\alpha_l-(d-1)\dot\alpha_i\notag\\
&=&\dot\alpha_i-(d-1)\dot\alpha_i\notag\\
&=&(2-d)\dot\alpha_i\label{eq: BVBIEMPsuml<klkneqidotalphal+dotalphak=2-ddotalphai}
\end{eqnarray}
where again we have used the definitions
(\ref{eq: BVBIEMPR-Yalpha-sigma}) and (\ref{eq: BVBIEMPccondition}) of $\alpha_l$ and the constants $c_l$ to write
$\sum_{l\neq i}\dot\alpha_l=-\dot\alpha_i$.  Considering the third term on the right-hand side of (\ref{eq: BVBIEMPsuml<klkneqidotXldotXk/XlXk-alhpaH}), we have that
\begin{eqnarray}
\sum_{\stackrel{l<k}{l,k\neq i}}\dot\alpha_l\dot\alpha_k
&=&-\sum_{\stackrel{l<k}{l,k\neq i}}\dot\alpha_l\dot\alpha_k+(\sum_{k\neq i}\dot\alpha_k)(\sum_{l\neq i}\dot\alpha_l)\notag\\
&=&-\sum_{\stackrel{l<k}{l,k\neq i}}\dot\alpha_l\dot\alpha_k-\dot\alpha_i\sum_{l\neq i}\dot\alpha_l\notag\\
&=&-\sum_{l<k}\dot\alpha_l\dot\alpha_k\notag\\
&=&-\sum_{2\leq l<k\leq d}\frac{c_lc_k}{\theta^{2/q\nu}R^{2/\nu}}
\label{eq: BVBIEMPsuml<klkneqidotalphaldotalphak=-suml<kdotalphaldotalphak}\end{eqnarray}
where again we have used that $\sum_{l\neq i}\dot\alpha_l=-\dot\alpha_i$,
 and on the last line we recall (\ref{eq: BVBIEMPdotalphaldotalphak-R}) and also that $\alpha_1$ is constant.  So by  (\ref{eq: BVBIEMPsum1l<klkneqi}), (\ref{eq: BVBIEMPsuml<klkneqidotalphal+dotalphak=2-ddotalphai}) and (\ref{eq: BVBIEMPsuml<klkneqidotalphaldotalphak=-suml<kdotalphaldotalphak}), in total (\ref{eq: BVBIEMPsuml<klkneqidotXldotXk/XlXk-alhpaH}) becomes
\begin{equation}\displaystyle\sum_{\stackrel{l<k}{l,k\neq i}}H_lH_k=\frac{(d-1)(d-2)}{2\nu^2 d^2}H_R^2+\frac{(2-d)}{\nu d}\dot\alpha_iH_R-\sum_{2\leq l<k\leq d}\frac{c_lc_k}{\theta^{2/q\nu}R^{2/\nu}}.
\label{eq: BVBIEMPdotXldotXk/XlXk-Ralphac}\end{equation}
By (\ref{eq: BVBIEMPddotX/X-alphaH}), (\ref{eq: BVBIEMPdotXldotXk/XlXk-Ralphac}) and the definition (\ref{eq: BVEMPX1-R}) of $X_1(t)$ , the left-hand side of any $(I_i)$ Einstein equation is
\begin{equation}\displaystyle\sum_{l\neq i}(\dot{H}_l+H_l^2)+\displaystyle\sum_{\stackrel{l<k}{l,k\neq i}}H_lH_k-\frac{(d-1)(d-2)\beta^2}{2X_1^{2}}\qquad\qquad\qquad\qquad\qquad\qquad\notag\end{equation}
\begin{equation}\qquad\qquad\qquad=\frac{(d-1)}{\nu d}\dot{H}_R+\frac{d(d-1)}{2\nu^2 d^2}H_R^2-\sum_{2\leq l<k\leq d}\frac{c_lc_k}{\theta^{2/q\nu}R^{2/\nu}}-\frac{(d-1)(d-2)\beta^2}{2e^{2\alpha_1}R^{2/\nu d}}.
\label{eq: BVBIEMPBIlhsIi}\end{equation}
Then by the condition (\ref{eq: BVBIEMPccondition}) on the constants $c_l$, (\ref{eq: BVBIEMPBIlhsIi}) becomes
\begin{equation}\displaystyle\sum_{l\neq i}(\dot{H}_l+H_l^2)+\displaystyle\sum_{\stackrel{l<k}{l,k\neq i}}H_lH_k-\frac{(d-1)(d-2)\beta^2}{2X_1^2}\qquad\qquad\qquad\qquad\qquad\qquad\notag\end{equation}
\begin{equation}\qquad\qquad\qquad=\frac{(d-1)}{\nu d}\dot{H}_R+\frac{d(d-1)}{2\nu^2 d^2}H_R^2+\frac{D\kappa}{R^{2/\nu}}-\frac{(d-1)(d-2)\beta^2}{2e^{2\alpha_1}R^{2/\nu d}}.
\label{eq: BVBIEMPdotXldotXk/XlXk-Ralpha}\end{equation}
Therefore by (\ref{eq: BVBIEMPlhs=lhsofIi}) and (\ref{eq: BVBIEMPdotXldotXk/XlXk-Ralpha}), we obtain $(I_i)$ for all $i\in\{1,\dots,d\}$.  This proves the theorem.
\end{proof}

\subsection{Reduction to classical EMP:  pure scalar field}

We take $\rho=p=D=0$ and choose parameter $q=1/d\nu$ in Theorem \ref{thm: BVEMP}.  Then $c_i=0$ for all $2\leq i\leq d$ and we take $\alpha_1=0$ so that $X_1(t)=\cdots=X_d(t)$ and by the theorem solving the Bianchi V Einstein equations
\begin{equation}
\displaystyle\sum_{l< k}H_lH_k-\frac{d(d-1)\beta^2}{2X_1^2}
\stackrel{(I_0)'}{=}
\kappa\left[\frac{1}{2}\dot\phi^2+V\circ\phi\right]+\Lambda\label{eq: BVEFEI0Idrhopzero}\end{equation}
\begin{equation}
\displaystyle\sum_{l\neq i}(\dot{H}_l+H_l^2)+\displaystyle\sum_{\stackrel{l < k}{l,k\neq i}}H_lH_k-\frac{(d-1)(d-2)\beta^2}{2X_1^2}\stackrel{(I_i)'}{=}-\kappa\left[\frac{1}{2}\dot\phi^2-V\circ\phi\right]+\Lambda\notag\end{equation}
for $l,k,i\in\{1,\dots, d\}$ and
\begin{equation}\beta \displaystyle\sum_{l=2}^d H_l-(d-1)\beta H_1\stackrel{(I_{01})'}{=}0.\end{equation}
is equivalent to solving the classical EMP equation
\begin{equation}Y''(\tau)+Q(\tau)Y(\tau)=\frac{-\beta^2}{\theta^2
Y(\tau)^{3}}\label{eq: D=0BVclassicalEMP}\end{equation}
for constant $\theta>0$.
The solutions of $(I_0)', (I_1)', \dots, (I_d)'$ and $(I_{01})'$ in (\ref{eq: BVEFEI0Idrhopzero}) and the solutions of (\ref{eq: D=0BVclassicalEMP}) are related by
\begin{equation}R(t)=Y(\tau(t))^{d\nu}\qquad\mbox{ and }\qquad \varphi '(\tau)^2= \frac{(d-1)}{\kappa} Q(\tau)\label{eq: D=0BVR-Yvarphi-Q}\end{equation}
for $\nu\neq0$, $\phi(t)=\varphi(\tau(t))$, $R(t)\stackrel{def.}{=} \left(X_1(t)\cdots X_d(t)\right)^\nu$ and 
\begin{equation}\dot\tau(t)=\theta R(t)^{1/d\nu}=\theta Y(\tau(t))\label{eq: nomatterisoBVclassEMPdottau-R-Y}\end{equation}
for any $\theta>0$.  In the converse direction
\begin{equation}X_l(t)=R(t)^{1/\nu d}
\label{eq: D=0BVBIEMPX-Ralpha}\end{equation}
for $1\leq l\leq d$ and where
 $V$ is taken to be
\begin{equation}
V(\phi(t))=\left[\frac{d(d-1)\theta^2}{2 \kappa }(Y')^2-\frac{d(d-1)\beta^2}{2\kappa
 Y^{2}}-\frac{\theta^2}{2}Y^2(\varphi ')^2-\frac{\Lambda}{\kappa}\right]\circ\tau(t)\label{eq: D=0BVBIEMPVphi-Y}
\end{equation}

We now refer to Appendix D with $\lambda_1= - \beta^2/\theta^2 
<0$ for solutions of the classical  EMP equation (\ref{eq: D=0BVclassicalEMP}).  
Since by reducing to classical EMP we have taken $D=0$ in (\ref{eq: BVBIEMPDconstant-X}),  $c_2=\cdots=c_d=0$ and in this case we do not require $\sigma(t)$ from Appendix D.  By comparing (\ref{eq: nomatterisoBVclassEMPdottau-R-Y}) and (\ref{eq: dottau-Ytau^s0}), we note to only consider solutions of (\ref{eq: dottau-Ytau^s0}) in Appendix D corresponding to $r_0=1$.     

\begin{example}
For $\theta=1$ and choice of constants $\beta=\nu=1$, we consider solution 5 in Table \ref{tb: exactEMP} with $d_0=b_0=0$ and $c_0=a_0=1$.  That is, we have solution $Y(\tau)=(1+2\tau)^{1/2}$ to the classical EMP $Y''(\tau)+Q(\tau)Y(\tau)=-1/Y(\tau)^3$ for $Q(\tau)=0$.  By (\ref{eq: taueqnforY=superpowerb=0r=1}) - (\ref{eq: dottauforY=superpowerb=0r=1}) we have $\tau(t)=\frac{1}{2}\left((t-t_0)^2-1\right)$ and 
\begin{equation}R(t)=Y(\tau(t))^{d}=(t-t_0)^{d}\end{equation}
for  $t_0\in\mathds{R}$ so that 
\begin{eqnarray}X_i(t)=R(t)^{1/d}=(t-t_0)\label{eq: XiforBVclassicalEMPY=superpowerd=b=0}
\end{eqnarray}
for $1\leq i\leq d$.  
Since $Q(\tau)=0=\varphi'(\tau)$,
\begin{equation}\phi(t)\stackrel{def.}{=}\varphi(\tau(t)) = \phi_0\label{eq: phiforBVclassicalEMPY=superpowerb=d=0}\end{equation}
for constant $\phi_0\in\mathds{R}$, and by (\ref{eq: D=0BVBIEMPVphi-Y}), (\ref{eq: YprimetauforY=superpowerb=d=0r=1}) and (\ref{eq: dottauforY=superpowerb=0r=1}), we obtain constant potential
\begin{eqnarray}
V(\phi(t))&=&\left[\frac{d(d-1)}{2 \kappa }\left((Y')^2-\frac{1}{Y^{2}}\right)-\frac{\Lambda}{\kappa}\right]\circ\tau(t)\notag\\
&=&-\frac{\Lambda}{\kappa}.\label{eq: VphiforBIclassicalEMPY=superpowerb=d=0}
\end{eqnarray}
Since $H_l(t)=\dot{X}_l(t)/X_l(t)=1/(t-t_0)$ for all $1\leq l \leq d$, $\sum_{l<k}H_lH_k = \sum_{l<k} 1 / (t-t_0)^2 = d(d-1)/2(t-t_0)^2$ so that Einstein equation $(I_0)'$ is satisfied.  Also for any fixed $i$, $\sum_{l\neq i}(\dot{H}_l+H_l^2)+\sum_{\stackrel{l<k}{l,k\neq i}} H_lH_k = \sum_{l\neq i} 0 + \left[ d(d-1)/2 - (d-1) \right] / (t-t_0)^2 = (d-1)(d-2)/2(t-t_0)^2$ therefore the Einstein equations $(I_i)'$ for $1\leq i \leq d$ are also satisfied.  Of course, this is a vacuum solution.
\end{example}

\begin{example}
For $\theta=1$ and choice of constants $\beta=\nu=1$, we consider solution 5 in Table \ref{tb: exactEMP} with $d_0=a_0=0$, $b_0=4$ and $c_0=1$.  That is, we have solution $Y(\tau)=(4\tau(t)^2+2\tau(t))^{1/2}$ to the classical EMP $Y''(\tau)+Q(\tau)Y(\tau)=-1/Y(\tau)^3$ for $Q(\tau)=0$.  By (\ref{eq: taueqnforY=superpowerr=1}) - (\ref{eq: dottauforY=superpowerr=1}) we have $\tau(t)=\frac{1}{8}\left(e^{2(t-t_0)} + e^{-2(t-t_0)}-8\right)$ and 
\begin{eqnarray}R(t)&=&Y(\tau(t))^{d}\notag\\
&=&\frac{1}{4^{d}}\left(e^{2(t-t_0)}-e^{-2(t-t_0)}\right)^{d}\notag\\
&=&\frac{1}{2^{d}}\sinh(2(t-t_0))^{d}
\end{eqnarray}
for any  $t_0\in\mathds{R}$ so that 
\begin{eqnarray}
X_i(t)=R(t)^{1/d}=\frac{1}{2}\sinh(2(t-t_0)).
\end{eqnarray}
Since $Q(\tau)=0=\varphi'(\tau)$,
\begin{equation}\phi(t)\stackrel{def.}{=}\varphi(\tau(t)) = \phi_0\end{equation}
for constant $\phi_0\in\mathds{R}$ and by (\ref{eq: D=0BVBIEMPVphi-Y}), (\ref{eq: YprimetauforY=superpowerd=0b>0r=1}) and (\ref{eq: dottauforY=superpowerr=1}), we obtain constant potential
\begin{eqnarray}
V(\phi(t))&=&\left[\frac{d(d-1)}{2 \kappa }\left((Y')^2-\frac{1}{Y^{2}}\right)-\frac{\Lambda}{\kappa}\right]\circ\tau(t)\notag\\
&=&  \frac{d(d-1)}{2 \kappa }\left(   \frac{16 cosh^2(2(t-t_0)) -16}{4sinh^2(2(t-t_0)) } \right) -\frac{\Lambda}{\kappa}\notag\\
&=& \frac{2d(d-1)}{ \kappa } -\frac{\Lambda}{\kappa}
\label{eq: D=0isoBVEMPVphi-Y}
\end{eqnarray}
Since $H_l(t)=\dot{X}_l(t)/X_l(t)=2 coth(t-t_0)$ for all $1\leq l \leq d$, $\sum_{l<k}H_lH_k = \sum_{l<k} 4 coth^2(t-t_0) = 2d(d-1) coth^2(t-t_0)$ so that the left side of the Einstein equation $(I_0)'$ is
\begin{equation}2d(d-1) coth^2(t-t_0)-2d(d-1)csch^2(t-t_0) = 2d(d-1)\end{equation}
by the identity $coth^2(x)-csch^2(x)=1$, and therefore $(I_0)'$ is satisfied by the solution.  Also for any fixed $i$, 
$\sum_{l\neq i}(\dot{H}_l+H_l^2)+\sum_{\stackrel{l<k}{l,k\neq i}} H_lH_k = (d-1) \left(- 4 csch^2(t-t_0)+4coth^2(t-t_0) \right) + \left[ d(d-1)/2 - (d-1) \right] 4 coth^2(t-t_0) = 4(d-1) + 2(d-1)(d-2)coth^2(t-t_0)$ therefore the left side of Einstein equations $(I_i)'$ for $1\leq i \leq d$  equal
\begin{eqnarray} &&4(d-1) + 2(d-1)(d-2)coth^2(t-t_0) - 2(d-1)(d-2)csch^2(t-t_0)\notag\\
&&\quad\quad =4(d-1) + 2(d-1)(d-2) =  2d(d-1)   \end{eqnarray}
so that $(I_i)'$ are also satisfied for $1\leq i\leq d$.  
\end{example}

\section{In terms of a Schr\"odinger-Type Equation}

If one would like to reformulate the Einstein field equations $(I_0),\dots, (I_d)$ in (\ref{eq:  BVEFEI0Id}) in terms of an equation with one less non-linear term than that which is provided by the generalized EMP formulation, one can apply Corollary \ref{cor: EFE-NLSAnonzeroEzero} to the difference $d(I_0)-\displaystyle\sum_{i=1}^d(I_i)$ (and similar to above, define $V\circ\phi$ in $u-$notation to be such that $(I_0)$ holds).   Below is the resulting statement.

\begin{thm}\label{thm: BVNLS}
Suppose you are given twice differentiable functions $X_1(t), \dots, X_d(t)>0$, a once differentiable function $\phi(t)$, and also functions $\rho(t), p(t), V(x)$ which satisfy the Einstein equations $(I_0),\dots,(I_d)$ for some $\Lambda,\beta\in\mathds{R}, d\in\mathds{N}\backslash\{0,1\}, \kappa\in\mathds{R}\backslash\{0\}$.  
Let $g(\sigma)$ denote the inverse of a function $\sigma(t)$ which satisfies
\begin{equation}\dot{\sigma}(t)=\frac{1}{\theta \left(X_1(t)\cdots X_d(t)\right)}\label{eq: BVBINLSdotsigma-R}\end{equation}
for some $\theta>0$.  Then the following functions 
\begin{eqnarray}
u(\sigma)&=&\left[\frac{1}{X_1\cdots X_d}\right]\circ g(\sigma)\label{eq: BVBINLSu-R}\\
P(\sigma)&=&\frac{d\kappa}{(d-1)}\psi '(\sigma)^2\label{eq: BVBINLSP-psi}
\end{eqnarray}
solve the Schr\"odinger-type equation
\begin{equation}u''(\sigma)+\left[E-P(\sigma)\right]u(\sigma)=
\frac{\theta^2 d \kappa(\uprho(\sigma)+\mathrm{p}(\sigma))}{(d-1) u(\sigma)} +\frac{\theta^2 d\beta^2}{e^{2\alpha_1}u(\sigma)^{1-2/d}}
\label{eq: BVBINLS}\end{equation}
for some $\alpha_1\in\mathds{R}$,
\begin{equation}\psi(\sigma)=\phi(g(\sigma))\label{eq: BVBINLSpsi-phi}\end{equation}
\begin{equation}\uprho(\sigma)=\rho(g(\sigma)), \ \mathrm{p}(\sigma)=p(g(\sigma)). \label{eq: BVBINLSuprho-rhormp-p}\end{equation}
and where 
\begin{equation}E\stackrel{def.}{=}\frac{-\theta^2}{(d-1)}X_1^2X_2^2\cdots X_d^2\displaystyle\sum_{l<k}\eta_{lk}^2\label{eq: BVBINLSE-X}\end{equation}
is a constant for 
\begin{equation}\eta_{lk}\stackrel{def.}{=}H_l-H_k, \mbox{\small $l\neq k, l, k\in\{1, \dots, d\}$.}\label{eq: BVBINLSeta-X}\end{equation}

Conversely, suppose you are given a twice differentiable function $u(\sigma)>0$, and also functions $P(\sigma)$ and $\uprho_i(\sigma), \mathrm{p}(\sigma)$ which solve (\ref{eq: BVBINLS}) for some constants $E<0, \theta>0, \kappa\in\mathds{R}\backslash\{0\}, \beta, \alpha\in\mathds{R}$ and $d\in\mathds{N}\backslash\{0,1\}$.  In order to construct functions which solve $(I_0),\dots, (I_d)$, first find $\sigma(t), \psi(\sigma)$ which solve the differential equations
\begin{equation}\dot{\sigma}(t)=\frac{1}{\theta} u(\sigma(t))\qquad\mbox{ and }\qquad \psi '(\sigma)^2= \frac{(d-1)}{ d \kappa} P(\sigma).\label{eq: BVBINLSdotsigma-upsi-P}\end{equation}
Let
\begin{equation}R(t)=u(\sigma(t))^{-\nu}\label{eq: BVBINLSR-ualpha-sigma}\qquad\qquad \alpha_l(t)\stackrel{def.}{=}c_l\sigma(t), l\in\{1,\dots,d\} \end{equation}
where $c_l$ are any constants for which both
\begin{equation}\displaystyle\sum_{l=1}^d c_l=0\qquad\mbox{ and }\qquad\displaystyle\sum_{l<k}c_lc_k=\frac{(d-1)E}{2d}.\label{eq: BVBINLScconditions}\end{equation}
Then the functions
\begin{equation}X_1(t)=R(t)^{1/\nu d}e^{\alpha_1}, \ \alpha_1\in\mathds{R}\label{eq: BVNLSX1-constR}\end{equation}
\begin{equation}X_l(t)=R(t)^{1/\nu d}e^{\alpha_l(t)},  \ \  2\leq l \leq d\label{eq: BVBINLSX-Ralpha}\end{equation}
\begin{equation}\phi(t)=\psi(\sigma(t))\label{eq: BVBINLSphi-psi}\end{equation}
\begin{equation}\rho(t)=\uprho(\sigma(t)), \ p(t)=\mathrm{p}(\sigma(t))\label{eq: BVBINLSrho-uprhop-rmp}\end{equation}
and\\
$V(\phi(t))$
\begin{equation}
=\left[\frac{(d-1)}{2\theta^2 d\kappa}(u')^2
+\frac{(d-1)E}{2d\theta^2\kappa }u^2
-\frac{d(d-1)\beta^2}{2\kappa e^{2\alpha_1}}u^{2/d}
-\frac{1}{2\theta^2}u^2(\psi ')^2-\uprho-\frac{\Lambda}{\kappa}\right]\circ\sigma(t)
\label{eq: BVBINLSVphi-u}
\end{equation}
satisfy the equations $(I_0),\dots, (I_d)$ in (\ref{eq: BVEFEI0Id}).\end{thm}

\begin{proof}
This proof will implement Corollary \ref{cor: EFE-NLSAnonzeroEzero} with constants and functions as indicated in the following table.

\break

\begin{table}[ht]
\centering
\caption{{ Corollary \ref{cor: EFE-NLSAnonzeroEzero} applied to Bianchi V}}\label{tb: BINLS}
\vspace{.2in}
\begin{tabular}{r | l c r | l}
In Corollary  & substitute & & In Corollary & substitute \\[4pt]
\hline
\raisebox{-5pt}{$a(t)$} & \raisebox{-5pt}{$R(t)$}      &&    \raisebox{-5pt}{$\varepsilon $} & \raisebox{-5pt}{${\nu d \kappa}/{(d-1)}$}\\[8pt]
$G(t)$ & $\mbox{constant } \nu E/\theta^2$    &&     $A$ & $2/\nu$ \\[8pt]
 $G_{1}(t)$ & $ \frac{-\nu d\kappa}{(d-1)}(\rho(t)+p(t))$  &&      $A_{1}$   &$0$\\[8pt]
 $G_{2}(t)$ & constant $-\nu d\beta^2e^{-2\alpha_1}$  &&      $A_{2}$   &$2/\nu d$\\[8pt]
$F_{1}(\sigma)$&$ \frac{\theta^2 d \kappa}{(d-1)} (\uprho(\sigma)+\mathrm{p}(\sigma)) $ &&        $C_{1}$&  $1$ \\[6pt]
$F_{2}(\sigma)$&constant $\theta^2 d\beta^2 e^{-2\alpha_1}$ &&        $C_{2}$&  $1-2/d$ \\[6pt]
\hline
\end{tabular}
\end{table}

Much of this proof will rely on computations that are exactly the same as those seen in the proof of Theorem \ref{thm: BVEMP} (the generalized EMP formulation of Bianchi V).  Therefore we will restate the relevant results here, but point the reader to the details in the proof of Theorem \ref{thm: BVEMP}.

To prove the forward implication, we assume to be given functions which solve the Einstein field equations $(I_0), \dots, (I_d)$ in (\ref{eq: BVEFEI0Id}).  Forming the linear combination $d(I_0)-\displaystyle\sum_{i=1}^d(I_i)$ of Einstein's equations and simplifying, as was done in (\ref{eq: BVBIEMPdI1minusIi}) - (\ref{eq: BVBIEMPI1minusIiwitheta}), 
\begin{equation}\dot{H}_R+\frac{\nu}{(d-1)}\displaystyle\sum_{l<k}\eta_{lk}^2+\frac{\nu d\beta^2}{X_1^2}=\frac{-\nu d\kappa}{(d-1)}\left[\dot{\phi}^2+ \left(\rho+p\right)\right].\label{eq: BVBINLSI1minusIiwitheta}\end{equation}
where 
\begin{equation}H_R(t)\stackrel{def.}{=}\frac{\dot{R}(t)}{R(t)}=\nu\displaystyle\sum_{l=1}^d H_l \label{eq: BVNLSHR-Hl}\end{equation}
and
\begin{equation}R(t)\stackrel{def.}{=}(X_1(t)\cdots X_d(t))^\nu\label{eq: BVBINLSR-X}\end{equation}
for any $\nu\neq 0$.  Next we will confirm that $E$ is constant.  As was done in (\ref{eq: BVBIEMPlhsIi=lhsIj})-(\ref{eq: BVBIEMPdotf+fH=0withetaH}), since the right-hand sides of Einstein's equations $(I_i)$  are the same for all $i\in\{1, \dots, d\}$, by equating the left-hand sides of any two equations $(I_i)$ and $(I_j)$ for $i\neq j$, and after some rearranging we obtain
\begin{equation}\dot\eta_{ij}+\frac{1}{\nu}\eta_{ij}H_R=0\label{eq: BVBINLSdoteta+stuff=0}\end{equation}
for $\eta_{ij}$ defined in (\ref{eq: BVBINLSeta-X}).  Therefore the definition (\ref{eq: BVBINLSE-X}) of $E$ is constant, being proportional to a sum of squares of these constant functions.  By the definitions (\ref{eq: BVBINLSE-X}) and (\ref{eq: BVBINLSR-X}) of the constant $E$ and the function $R(t)$, we now rewrite (\ref{eq: BVBINLSI1minusIiwitheta}) from above as 
\begin{equation}\dot{H}_R+\frac{\nu d\beta^2}{X_1^2}=\frac{-\nu d \kappa}{(d-1)}\left[\dot\phi^2 +(\rho+p)\right]+\frac{\nu E}{\theta^2 R^{2/\nu}}.\label{eq: BVBINLSI0minusIi}\end{equation}
Finally, we use the remaining Einstein equation $(I_{01})$ and the relation $H_R=\nu\displaystyle\sum_{l=1}^dH_l$ in (\ref{eq: BVNLSHR-Hl}) to obtain
\begin{equation}H_R=\nu  d H_1\label{eq: BVNLSHR-H1}\end{equation}
so that again we have
\begin{equation}X_1=e^{\alpha_1}R^{1/\nu d}\label{eq: BVNLSX1-alpha1R}\end{equation}
for $\alpha_1\in\mathds{R}$ and (\ref{eq: BVBINLSI0minusIi}) becomes
\begin{equation}\dot{H}_R=\frac{-\nu d \kappa}{(d-1)}\left[\dot\phi^2 +(\rho+p)\right]+\frac{\nu E}{\theta^2 R^{2/\nu}}-\frac{\nu d\beta^2}{e^{2\alpha_1}R^{2/\nu d}}.\label{eq: BVBINLSI0minusIinoX1}\end{equation}
This shows that $R(t), \phi(t), \rho(t)$ and $p(t)$ satisfy the hypothesis of Corollary \ref{cor: EFE-NLSAnonzeroEzero}, applied with  constants $\varepsilon , N, A, A_1 \dots, A_{N}$ and functions $a(t), G(t), G_1(t),\dots,G_{N}(t)$ according to Table \ref{tb: BINLS}.  Since $\sigma(t), u(\sigma), P(\sigma)$ and $\psi(\sigma)$ defined in (\ref{eq: BVBINLSdotsigma-R}), (\ref{eq: BVBINLSu-R}), (\ref{eq: BVBINLSP-psi}) and (\ref{eq: BVBINLSpsi-phi})  are equivalent to that in the forward implication of Corollary \ref{cor: EFE-NLSAnonzeroEzero}, applied with constants and functions according to Table \ref{tb: BINLS}, by this corollary and by definition (\ref{eq: BVBINLSuprho-rhormp-p}) of $\uprho(\sigma), \mathrm{p}(\sigma)$, the Schr\"odinger-type equation (\ref{eq: CNLSANONZERO}) holds for constants $C_1, \dots, C_{N}$ and functions $F_1(\sigma), \dots, F_{N}(\sigma)$ as indicated in Table \ref{tb: BINLS}.  This proves the forward implication.

To prove the converse implication, we assume to be given functions which solve the Schr\"odinger-type equation (\ref{eq: BVBINLS}) and we begin by showing that $(I_0)$ is satisfied.  Differentiating the definition of $R(t)$ in (\ref{eq: BVBINLSR-ualpha-sigma}) and using the definition in (\ref{eq: BVBINLSdotsigma-upsi-P}) of $\sigma(t)$, 
\begin{eqnarray}\dot{R}(t)
&=&-\nu u(\sigma(t))^{-\nu-1}u'(\sigma(t))\dot\sigma(t)\notag\\
&=&-\frac{\nu}{\theta}u(\sigma(t))^{-\nu}u'(\sigma(t)).\label{eq: BVBINLSdotR-u}\end{eqnarray}
Dividing by $R(t)$,
\begin{equation}H_R\stackrel{def.}{=}\frac{\dot{R}}{R}=-\frac{\nu}{\theta}u'(\sigma(t)).\label{eq: BVBINLSH-u}\end{equation}
Differentiating the definition (\ref{eq: BVBINLSphi-psi}) of $\phi(t)$ and using the definition in (\ref{eq: BVBINLSR-ualpha-sigma}) of $\sigma(t)$,
\begin{equation}
\dot{\phi}(t)=\psi '(\sigma(t))\dot{\sigma}(t)=\frac{1}{\theta}\psi '(\sigma(t))u(\sigma(t)).\label{eq: BVBINLSphi-u}
\end{equation}
Using (\ref{eq: BVBINLSH-u}) and (\ref{eq: BVBINLSphi-u}), and also the definitions (\ref{eq: BVBINLSR-ualpha-sigma}) and (\ref{eq: BVBINLSrho-uprhop-rmp}) of $R(t)$ and $\rho(t)$ respectively, the definition (\ref{eq: BVBINLSVphi-u}) of $V\circ\phi$ can be written as
\begin{equation}
V\circ\phi
=\frac{1}{\kappa}\left(\frac{(d-1)}{2\nu^2 d}H_R^2+\frac{(d-1)E}{2d\theta^{2}R^{2/\nu}} -\frac{d(d-1)\beta^2}{2e^{2\alpha_1}R^{2/\nu d}} \right)-\frac{1}{2}\dot{\phi}^2-\rho-\frac{\Lambda}{\kappa}.\label{eq: BVBINLSVphi-RE}
\end{equation}
The quantity in parenthesis here is in fact equal to the left-hand-side of equation $(I_0)$.  
To see this, first note that the definitions
$X_l(t)\stackrel{def.}{=}R(t)^{1/\nu d}e^{\alpha_l(t)}$ in (\ref{eq: BVBINLSX-Ralpha}) and
$H(t)\stackrel{def.}{=}\frac{\dot{R}(t)}{R(t)}$ in (\ref{eq: BVBINLSH-u}), and the condition $\sum_l c_l=0$, are the same as those in Theorem \ref{thm: BIEMP}.
Also by the definitions (\ref{eq: BVBINLSR-ualpha-sigma}), (\ref{eq: BVBINLSdotsigma-upsi-P}) and (\ref{eq: BVBINLSR-ualpha-sigma})  of $\alpha_l(t), \sigma(t)$ and $R(t)$, we obtain 
\begin{equation}\dot\alpha_l\dot\alpha_k=c_lc_k\dot\sigma^2=\frac{c_lc_k}{\theta^2}(u\circ\sigma)^2=\frac{c_lc_k}{\theta^2R(t)^{2/\nu}}\label{eq: BVBINLSdotalphaldotalphak-Rc},\end{equation}
which is a slightly modified version of (\ref{eq: BVBIEMPdotalphaldotalphak-R}) from our computation in the proof of Theorem \ref{thm: BIEMP}.  Therefore by the arguments in (\ref{eq: BVBIEMPdotXldotXk/XlXk-Halpha})-(\ref{eq: BVBIEMPdotXldotXk/XlXk-Rc}), and using (\ref{eq: BVBINLSdotalphaldotalphak-Rc}) to slightly modify the last term to apply here,
\begin{equation}\displaystyle\sum_{l<k}H_lH_k=\frac{(d-1)}{2\nu^2 d}H_R^2+\displaystyle\sum_{l<k}\frac{c_lc_k}{\theta^{2}R^{2/\nu}}.
\label{eq: BVBINLSdotXldotXk/XlXk-Rc}
\end{equation}
Then by the condition (\ref{eq: BVBINLScconditions}) on the constants $c_l$, (\ref{eq: BVBINLSdotXldotXk/XlXk-Rc}) becomes
\begin{equation}\displaystyle\sum_{l<k}H_lH_k=\frac{(d-1)}{2\nu^2 d}H_R^2+\frac{(d-1)E}{2d\theta^{2}R^{2/\nu}}.
\label{eq: BVBINLSdotXldotXk/XlXk-RE}
\end{equation}
Using this and the definition (\ref{eq: BVNLSX1-constR}) of $X_1(t)=e^{\alpha_1}R(t)^{1/\nu d}$, the expression (\ref{eq: BVBINLSVphi-RE}) for $V$ can now be written as
\begin{equation}V\circ\phi=\frac{1}{\kappa}\left(\displaystyle\sum_{l<k}H_lH_k-\frac{d(d-1)\beta^2}{2X_1^2}\right)-\frac{1}{2}\dot{\phi}^2-\rho-\frac{\Lambda}{\kappa},\label{eq: BVBIEMPVphi-X}
\end{equation}
showing that $(I_0)$ holds under the assumptions of the converse implication.

To conclude the proof we must also show that the equations $(I_1), \dots, (I_d)$ hold.   In the converse direction the hypothesis of the converse of Corollary \ref{cor: EFE-NLSAnonzeroEzero} holds, applied with constants $N, C_1, \dots, C_{N}$ and functions $F_1(\sigma), \dots, F_{N}(\sigma)$ as indicated in Table \ref{tb: BINLS}.  Since  $\sigma(t), \psi(\sigma), R(t)$ and $\phi(t)$ defined in (\ref{eq: BVBINLSdotsigma-upsi-P}), (\ref{eq: BVBINLSR-ualpha-sigma}) and (\ref{eq: BVBINLSphi-psi}) are consistent with the converse implication of Corollary \ref{cor: EFE-NLSAnonzeroEzero}, applied with $a(t)$ and $\varepsilon $ as in Table \ref{tb: BINLS}, by this corollary and by the definition (\ref{eq: BVBINLSrho-uprhop-rmp}) of $\rho(t), p(t)$ the  scale factor  equation (\ref{eq: CEFEANONZERO}) holds for constants $\varepsilon , A, A_1, \dots, A_{N}$ and functions $G(t),G_1(t),\dots,G_{N}(t)$ according to Table \ref{tb: BINLS}.  That is, we have regained (\ref{eq: BVBINLSI0minusIinoX1}).  Now solving (\ref{eq: BVBINLSVphi-RE}) for $\rho(t)$ and substituting this into (\ref{eq: BVBINLSI0minusIinoX1}), we obtain
\begin{equation}
\dot{H}_R
=
\frac{-\nu d\kappa}{(d-1)}
\left[ 
\frac{1}{2}\dot\phi^2 -V\circ\phi+\frac{(d-1)E}{2 d\theta^2\kappa R^{2/\nu}}-\frac{d(d-1)\beta^2}{2\kappa e^{2\alpha_1}R^{2/\nu d}}+ p\right.\notag\end{equation}
\begin{equation}\qquad\qquad\left. +\frac{(d-1)}{2\nu^2 d\kappa}H_R^2-\frac{\Lambda}{\kappa} \right]+\frac{\nu E}{\theta^2 R^{2/\nu}}-\frac{\nu d\beta^2}{e^{2\alpha_1}R^{2/\nu d}}.\end{equation}
Multiplying by $\frac{(d-1)}{\nu d}$ and rearranging,
\begin{equation}
\frac{(d-1)}{\nu d}\dot{H}_R+\frac{(d-1)}{2\nu^2 d}H_R^2-\frac{(d-1)E}{2 d\theta^2 R^{2/\nu}}-\frac{(d-1)(d-2)\beta^2}{2e^{2\alpha_1}R^{2/\nu d}}\qquad\qquad\qquad\qquad\qquad\notag\end{equation}
\begin{equation}
\qquad\qquad\qquad
=
-\kappa
\left[ 
\frac{1}{2}\dot\phi^2 -V\circ\phi+p \right]+\Lambda.\label{eq: BVBINLSlhs=lhsofIi}\end{equation}
The left-hand side of this equation is in fact equal to the left-hand-side of $(I_i)$ for any $i\in\{1,\dots, d\}$.  To see this, again we use that the definitions $X_l(t)\stackrel{def.}{=}R(t)^{1/\nu d}e^{\alpha_l(t)}$ in (\ref{eq: BVBINLSX-Ralpha}) and
$H_R(t)\stackrel{def.}{=}\frac{\dot{R}(t)}{R(t)}$ in (\ref{eq: BVBINLSH-u}), and the condition $\sum_l c_l=0$, are the same as those in Theorem \ref{thm: BIEMP}.  Also by the definitions (\ref{eq: BVBINLSR-ualpha-sigma}) and (\ref{eq: BVBINLSdotsigma-upsi-P}) of $\alpha_l(t), \sigma(t)$ and $R(t)$, we obtain 
\begin{eqnarray}
\dot\alpha_l(t) R(t)^{1/\nu}
&=&c_l\dot\sigma(t) R(t)^{1/\nu}\notag\\
&=&\frac{c_l}{\theta}u(\sigma(t))R(t)^{1/\nu}\notag\\
&=&\frac{c_l}{\theta }\mbox{ is a constant}
\label{eq: BVBINLSdotalphaRconstant},\end{eqnarray}
which shows that
\begin{equation}\ddot\alpha_l+\frac{1}{\nu}\dot\alpha_l H_R=0\label{eq: BVBINLSddotalpha+stuff=0}\end{equation}
holds here, as it does in Theorem \ref{thm: BIEMP}.  Therefore by the arguments in  (\ref{eq: BVBIEMPddotX/X-Halpha})-(\ref{eq: BVBIEMPBIlhsIi}), and as above using (\ref{eq: BVBINLSdotalphaldotalphak-Rc}) to slightly modify the last term of (\ref{eq: BVBIEMPBIlhsIi}) to apply here,
\begin{equation}\displaystyle\sum_{l\neq i}(\dot{H}_l+H_l^2)+\displaystyle\sum_{\stackrel{l<k}{l,k\neq i}}H_lH_k=\frac{(d-1)}{\nu d}\dot{H}_R+\frac{(d-1)}{2\nu^2 d}H_R^2-\sum_{l<k}\frac{c_lc_k}{\theta^{2}R^{2/\nu}}.
\label{eq: BVBINLSBIlhsIi}\end{equation}
Then by the condition (\ref{eq: BVBINLScconditions}) on the constants $c_l$ and the definition (\ref{eq: BVNLSX1-constR}) of $X_1(t)$, the left-hand side of any Einstein equation $(I_i), i>0$ is
\begin{equation}\displaystyle\sum_{l\neq i}(\dot{H}_l+H_l^2)+\displaystyle\sum_{\stackrel{l<k}{l,k\neq i}}H_lH_k-\frac{(d-1)(d-2)\beta^2}{2X_1^2}\qquad\qquad\qquad\qquad\notag\end{equation}
\begin{equation}\qquad\qquad\qquad=\frac{(d-1)}{\nu d}\dot{H}_R+\frac{(d-1)}{2\nu^2 d}H_R^2-\frac{(d-1)E}{2d\theta^{2}R^{2/\nu}}-\frac{(d-1)(d-2)\beta^2}{2e^{2\alpha_1}R^{2/\nu d}}.
\label{eq: BVBINLSdotXldotXk/XlXk-Ralpha}\end{equation}
Combining (\ref{eq: BVBINLSlhs=lhsofIi}) and (\ref{eq: BVBINLSdotXldotXk/XlXk-Ralpha}), we obtain $(I_i)$ for all $i\in\{1,\dots,d\}$.  This proves the theorem.\end{proof}

\section{In terms of an Alternate Schr\"odinger-Type Equation}

\begin{thm}\label{thm: altBVNLS}
Suppose you are given twice differentiable functions $X_1(t), \dots, X_d(t)>0$, a once differentiable function $\phi(t)$, and also functions $\rho(t), p(t), V(x)$ which satisfy the Einstein equations $(I_0),\dots,(I_d)$ in (\ref{eq: BVEFEI0Id}) for some $\Lambda,\beta\in\mathds{R}, d\in\mathds{N}\backslash\{0,1\}, \kappa\in\mathds{R}\backslash\{0\}$.  
Let $g(\sigma)$ denote the inverse of a function $\sigma(t)$ which satisfies
\begin{equation}\dot{\sigma}(t)=\frac{1}{\theta \left(X_1(t)\cdots X_d(t)\right)}\label{eq: altBVBINLSdotsigma-R}\end{equation}
for some $\theta>0$.  Then the following functions 
\begin{eqnarray}
u(\sigma)&=&\left[\frac{1}{X_1\cdots X_d}\right]\circ g(\sigma)\label{eq: altBVBINLSu-R}\\
P(\sigma)&=&\frac{d\kappa}{(d-1)}\psi '(\sigma)^2\label{eq: altBVBINLSP-psi}
\end{eqnarray}
solve the Schr\"odinger-type equation
\begin{equation}u''(\sigma)+\left[E-P(\sigma)\right]u(\sigma)=
\frac{\theta^2 d \kappa(\uprho(\sigma)+\mathrm{p}(\sigma))}{(d-1) u(\sigma)} -\frac{\theta^2 D}{\nu u(\sigma)^{1-d}}
\label{eq: altBVBINLS}\end{equation}
for some $\alpha_1\in\mathds{R}$,
\begin{equation}\psi(\sigma)=\phi(g(\sigma))\label{eq: altBVBINLSpsi-phi}\end{equation}
\begin{equation}\uprho(\sigma)=\rho(g(\sigma)), \ \mathrm{p}(\sigma)=p(g(\sigma)). \label{eq: altBVBINLSuprho-rhormp-p}\end{equation}
and where 
\begin{equation}E\stackrel{def.}{=}-\theta^2 d\beta^2e^{-2\alpha_1}, \alpha_1\in\mathds{R}\label{eq: altBVBINLSE-beta}\end{equation}
and
\begin{equation}D\stackrel{def.}{=}\frac{-\nu}{(d-1)}X_1^2X_2^2\cdots X_d^2\displaystyle\sum_{l<k}\eta_{lk}^2\label{eq: altBVBINLSD-X}\end{equation}
is a constant for 
\begin{equation}\eta_{lk}\stackrel{def.}{=}H_l-H_k, \mbox{\small $l\neq k, l, k\in\{1, \dots, d\}$.}\label{eq: altBVBINLSeta-X}\end{equation}

Conversely, suppose you are given a twice differentiable function $u(\sigma)>0$, and also functions $P(\sigma)$ and $\uprho(\sigma), \mathrm{p}(\sigma)$ fwhich solve (\ref{eq: altBVBINLS}) for some constants $E<0, \theta>0, \kappa\in\mathds{R}\backslash\{0\}, \beta, \alpha\in\mathds{R}$ and $d\in\mathds{N}\backslash\{0,1\}$.  In order to construct functions which solve $(I_0),\dots, (I_d)$, first find $\sigma(t), \psi(\sigma)$ which solve the differential equations
\begin{equation}\dot{\sigma}(t)=\frac{1}{\theta} u(\sigma(t))\qquad\mbox{ and }\qquad \psi '(\sigma)^2= \frac{(d-1)}{ d \kappa} P(\sigma).\label{eq: altBVBINLSdotsigma-upsi-P}\end{equation}
Let
\begin{equation}R(t)=u(\sigma(t))^{-\nu}\label{eq: altBVBINLSR-ualpha-sigma}\qquad\qquad \alpha_l(t)\stackrel{def.}{=}c_l\sigma(t), l\in\{1,\dots,d\} \end{equation}
where $c_l$ are any constants for which both
\begin{equation}\displaystyle\sum_{l=1}^d c_l=0\qquad\mbox{ and }\qquad\displaystyle\sum_{l<k}c_lc_k=\frac{(d-1)\theta^2 D}{2d\nu}.\label{eq: altBVBINLScconditions}\end{equation}
Then the functions
\begin{equation}X_1(t)=R(t)^{1/\nu d}e^{\alpha_1}, \ \alpha_1\in\mathds{R}\label{eq: altBVNLSX1-constR}\end{equation}
\begin{equation}X_l(t)=R(t)^{1/\nu d}e^{\alpha_l(t)},  \ \  2\leq l \leq d\label{eq: altBVBINLSX-Ralpha}\end{equation}
\begin{equation}\phi(t)=\psi(\sigma(t))\label{eq: altBVBINLSphi-psi}\end{equation}
\begin{equation}\rho(t)=\uprho(\sigma(t)), \ p(t)=\mathrm{p}(\sigma(t))\label{eq: altBVBINLSrho-uprhop-rmp}\end{equation}
and\\
$V(\phi(t))$
\begin{equation}
=\left[\frac{(d-1)}{2\theta^2 d\kappa}(u')^2
+\frac{(d-1)D}{2d\nu \kappa }u^2
-\frac{d(d-1)\beta^2}{2\kappa e^{2\alpha_1}}u^{2/d}
-\frac{1}{2\theta^2}u^2(\psi ')^2-\uprho-\frac{\Lambda}{\kappa}\right]\circ\sigma(t)
\label{eq: altBVBINLSVphi-u}
\end{equation}
satisfy the equations $(I_0),\dots, (I_d)$.\end{thm}

\begin{proof}
This proof will implement Corollary \ref{cor: EFE-NLSAnonzeroEzero} with constants and functions as indicated in the following table.

\begin{table}[ht]
\centering
\caption{{ Corollary \ref{cor: EFE-NLSAnonzeroEzero} applied to Bianchi V, alternate}}\label{tb: altBINLS}
\vspace{.2in}
\begin{tabular}{r | l c r | l}
In Corollary  & substitute & & In Corollary & substitute \\[4pt]
\hline
\raisebox{-5pt}{$a(t)$} & \raisebox{-5pt}{$R(t)$}      &&    \raisebox{-5pt}{$\varepsilon $} & \raisebox{-5pt}{${\nu d \kappa}/{(d-1)}$}\\[8pt]
$G(t)$ & $\mbox{constant } -\nu d\beta^2e^{-2\alpha_1}$    &&     $A$ & $2/\nu d$ \\[8pt]
 $G_{1}(t)$ & $ \frac{-\nu d\kappa}{(d-1)}(\rho(t)+p(t))$  &&      $A_{1}$   &$0$\\[8pt]
$G_2(t)$ & $\mbox{constant } D$    &&     $A_2$ & $2/\nu$ \\[8pt]
$F_{1}(\sigma)$&$ \frac{\theta^2 d \kappa}{(d-1)} (\uprho(\sigma)+\mathrm{p}(\sigma)) $ &&        $C_{1}$&  $1$ \\[6pt]
$F_{2}(\sigma)$&constant $-\theta^2 D/\nu $ &&        $C_{2}$&  $1-d$ \\[6pt]
\hline
\end{tabular}
\end{table}  

\break

Much of this proof will rely on computations that are exactly the same as those seen in the proof of Theorem \ref{thm: BVEMP} (the generalized EMP formulation of Bianchi V).  Therefore we will restate the relevant results here, but point the reader to the details in the proof of Theorem \ref{thm: BVEMP}.

To prove the forward implication, we assume to be given functions which solve the Einstein field equations $(I_0), \dots, (I_d)$ in (\ref{eq: BVEFEI0Id}).  Forming the linear combination $d(I_0)-\displaystyle\sum_{i=1}^d(I_i)$ of Einstein's equations and simplifying, as was done in (\ref{eq: BVBIEMPdI1minusIi}) - (\ref{eq: BVBIEMPI1minusIiwitheta}), 
\begin{equation}\dot{H}_R+\frac{\nu}{(d-1)}\displaystyle\sum_{l<k}\eta_{lk}^2+\frac{\nu d\beta^2}{X_1^2}=\frac{-\nu d\kappa}{(d-1)}\left[\dot{\phi}^2+ \left(\rho+p\right)\right].\label{eq: altBVBINLSI1minusIiwitheta}\end{equation}
where 
\begin{equation}H_R(t)\stackrel{def.}{=}\frac{\dot{R}(t)}{R(t)}=\nu\displaystyle\sum_{l=1}^d H_l \label{eq: altBVNLSHR-Hl}\end{equation}
and
\begin{equation}R(t)\stackrel{def.}{=}(X_1(t)\cdots X_d(t))^\nu\label{eq: altBVBINLSR-X}\end{equation}
for any $\nu\neq 0$.  Next we will confirm that $D$ is constant.  As was done in (\ref{eq: BVBIEMPlhsIi=lhsIj})-(\ref{eq: BVBIEMPdotf+fH=0withetaH}), since the right-hand sides of Einstein's equations $(I_i)$  are the same for all $i\in\{1, \dots, d\}$, by equating the left-hand sides of any two equations $(I_i)$ and $(I_j)$ for $i\neq j$, and after some rearranging we obtain
\begin{equation}\dot\eta_{ij}+\frac{1}{\nu}\eta_{ij}H_R=0\label{eq: altBVBINLSdoteta+stuff=0}\end{equation}
for $\eta_{ij}$ defined in (\ref{eq: altBVBINLSeta-X}).  Therefore the definition (\ref{eq: altBVBINLSD-X}) of $D$ is constant, being proportional to a sum of squares of these constant functions.  By the definitions (\ref{eq: altBVBINLSD-X}) and (\ref{eq: altBVBINLSR-X}) of the constant $D$ and the function $R(t)$, we now rewrite (\ref{eq: altBVBINLSI1minusIiwitheta}) from above as 
\begin{equation}\dot{H}_R+\frac{\nu d\beta^2}{X_1^2}=\frac{-\nu d \kappa}{(d-1)}\left[\dot\phi^2 +(\rho+p)\right]+\frac{D}{R^{2/\nu}}.\label{eq: altBVBINLSI0minusIi}\end{equation}
Finally, we use the remaining Einstein equation $(I_{01})$ and the relation $H_R=\nu\displaystyle\sum_{l=1}^dH_l$ in (\ref{eq: altBVNLSHR-Hl}) to obtain
\begin{equation}H_R=\nu  d H_1\label{eq: altBVNLSHR-H1}\end{equation}
so that again we have
\begin{equation}X_1=e^{\alpha_1}R^{1/\nu d}\label{eq: altBVNLSX1-alpha1R}\end{equation}
for $\alpha_1\in\mathds{R}$ and (\ref{eq: altBVBINLSI0minusIi}) becomes
\begin{equation}\dot{H}_R=\frac{-\nu d \kappa}{(d-1)}\left[\dot\phi^2 +\displaystyle\sum_{i=1}^M(\rho_i+p_i)\right]+\frac{D}{ R^{2/\nu}}-\frac{\nu d\beta^2}{e^{2\alpha_1}R^{2/\nu d}}.\label{eq: altBVBINLSI0minusIinoX1}\end{equation}
This shows that $R(t), \phi(t), \rho(t)$ and $p(t)$ satisfy the hypothesis of Corollary \ref{cor: EFE-NLSAnonzeroEzero}, applied with  constants $\varepsilon , N, A, A_1 \dots, A_{N}$ and functions $a(t), G(t), G_1(t),\dots,G_{N}(t)$ according to Table \ref{tb: altBINLS}.  Since $\sigma(t), u(\sigma), P(\sigma)$ and $\psi(\sigma)$ defined in (\ref{eq: altBVBINLSdotsigma-R}), (\ref{eq: altBVBINLSu-R}), (\ref{eq: altBVBINLSP-psi}) and (\ref{eq: altBVBINLSpsi-phi})  are equivalent to that in the forward implication of Corollary \ref{cor: EFE-NLSAnonzeroEzero}, applied with constants and functions according to Table \ref{tb: altBINLS}, by this corollary and by definition (\ref{eq: altBVBINLSuprho-rhormp-p}) of $\uprho(\sigma), \mathrm{p}(\sigma)$, the Schr\"odinger-type equation (\ref{eq: CNLSANONZERO}) holds for constants $C_1, \dots, C_{N}$ and functions $F_1(\sigma), \dots, F_{N}(\sigma)$ as indicated in Table \ref{tb: altBINLS}.  This proves the forward implication.

To prove the converse implication, we assume to be given functions which solve the Schr\"odinger-type equation (\ref{eq: altBVBINLS}) and we begin by showing that $(I_0)$ is satisfied.  Differentiating the definition of $R(t)$ in (\ref{eq: altBVBINLSR-ualpha-sigma}) and using the definition in (\ref{eq: altBVBINLSdotsigma-upsi-P}) of $\sigma(t)$, 
\begin{eqnarray}\dot{R}(t)
&=&-\nu u(\sigma(t))^{-\nu-1}u'(\sigma(t))\dot\sigma(t)\notag\\
&=&-\frac{\nu}{\theta}u(\sigma(t))^{-\nu}u'(\sigma(t)).\label{eq: altBVBINLSdotR-u}\end{eqnarray}
Dividing by $R(t)$,
\begin{equation}H_R\stackrel{def.}{=}\frac{\dot{R}}{R}=-\frac{\nu}{\theta}u'(\sigma(t)).\label{eq: altBVBINLSH-u}\end{equation}
Differentiating the definition (\ref{eq: altBVBINLSphi-psi}) of $\phi(t)$ and using the definition in (\ref{eq: altBVBINLSR-ualpha-sigma}) of $\sigma(t)$,
\begin{equation}
\dot{\phi}(t)=\psi '(\sigma(t))\dot{\sigma}(t)=\frac{1}{\theta}\psi '(\sigma(t))u(\sigma(t)).\label{eq: altBVBINLSphi-u}
\end{equation}
Using (\ref{eq: altBVBINLSH-u}) and (\ref{eq: altBVBINLSphi-u}), and also the definitions (\ref{eq: altBVBINLSR-ualpha-sigma}) and (\ref{eq: altBVBINLSrho-uprhop-rmp}) of $R(t)$ and $\rho_i(t)$ respectively, the definition (\ref{eq: altBVBINLSVphi-u}) of $V\circ\phi$ can be written as
\begin{equation}
V\circ\phi
=\frac{1}{\kappa}\left(\frac{(d-1)}{2\nu^2 d}H_R^2+\frac{(d-1)D}{2d\nu R^{2/\nu}} -\frac{d(d-1)\beta^2}{2e^{2\alpha_1}R^{2/\nu d}} \right)-\frac{1}{2}\dot{\phi}^2-\rho-\frac{\Lambda}{\kappa}.\label{eq: altBVBINLSVphi-RE}
\end{equation}
The quantity in parenthesis here is in fact equal to the left-hand-side of equation $(I_0)$.  
To see this, first note that the definitions
$X_l(t)\stackrel{def.}{=}R(t)^{1/\nu d}e^{\alpha_l(t)}$ in (\ref{eq: altBVBINLSX-Ralpha}) and
$H(t)\stackrel{def.}{=}\frac{\dot{R}(t)}{R(t)}$ in (\ref{eq: altBVBINLSH-u}), and the condition $\sum_l c_l=0$, are the same as those in Theorem \ref{thm: BIEMP}.
Also by the definitions (\ref{eq: altBVBINLSR-ualpha-sigma}), (\ref{eq: altBVBINLSdotsigma-upsi-P}) and (\ref{eq: altBVBINLSR-ualpha-sigma})  of $\alpha_l(t), \sigma(t)$ and $R(t)$, we obtain 
\begin{equation}\dot\alpha_l\dot\alpha_k=c_lc_k\dot\sigma^2=\frac{c_lc_k}{\theta^2}(u\circ\sigma)^2=\frac{c_lc_k}{\theta^2R(t)^{2/\nu}}\label{eq: altBVBINLSdotalphaldotalphak-Rc},\end{equation}
which is a slightly modified version of (\ref{eq: BVBIEMPdotalphaldotalphak-R}) from our computation in the proof of Theorem \ref{thm: BIEMP}.  Therefore by the arguments in (\ref{eq: BVBIEMPdotXldotXk/XlXk-Halpha})-(\ref{eq: BVBIEMPdotXldotXk/XlXk-Rc}), and using (\ref{eq: altBVBINLSdotalphaldotalphak-Rc}) to slightly modify the last term to apply here,
\begin{equation}\displaystyle\sum_{l<k}H_lH_k=\frac{(d-1)}{2\nu^2 d}H_R^2+\displaystyle\sum_{l<k}\frac{c_lc_k}{\theta^{2}R^{2/\nu}}.
\label{eq: altBVBINLSdotXldotXk/XlXk-Rc}
\end{equation}
Then by the condition (\ref{eq: altBVBINLScconditions}) on the constants $c_l$, (\ref{eq: altBVBINLSdotXldotXk/XlXk-Rc}) becomes
\begin{equation}\displaystyle\sum_{l<k}H_lH_k=\frac{(d-1)}{2\nu^2 d}H_R^2+\frac{(d-1)D}{2d\nu R^{2/\nu}}.
\label{eq: altBVBINLSdotXldotXk/XlXk-RE}
\end{equation}
Using this and the definition (\ref{eq: altBVNLSX1-constR}) of $X_1(t)=e^{\alpha_1}R(t)^{1/\nu d}$, the expression (\ref{eq: altBVBINLSVphi-RE}) for $V$ can now be written as
\begin{equation}V\circ\phi=\frac{1}{\kappa}\left(\displaystyle\sum_{l<k}H_lH_k-\frac{d(d-1)\beta^2}{2X_1^2}\right)-\frac{1}{2}\dot{\phi}^2-\rho-\frac{\Lambda}{\kappa},\label{eq: altBVBIEMPVphi-X}
\end{equation}
showing that $(I_0)$ holds under the assumptions of the converse implication.

To conclude the proof we must also show that the equations $(I_1), \dots, (I_d)$ hold.   In the converse direction the hypothesis of the converse of Corollary \ref{cor: EFE-NLSAnonzeroEzero} holds, applied with constants $N, C_1, \dots, C_{N}$ and functions $F_1(\sigma), \dots, F_{N}(\sigma)$ as indicated in Table \ref{tb: altBINLS}.  Since  $\sigma(t), \psi(\sigma), R(t)$ and $\phi(t)$ defined in (\ref{eq: altBVBINLSdotsigma-upsi-P}), (\ref{eq: altBVBINLSR-ualpha-sigma}) and (\ref{eq: altBVBINLSphi-psi}) are consistent with the converse implication of Corollary \ref{cor: EFE-NLSAnonzeroEzero}, applied with $a(t)$ and $\varepsilon $ as in Table \ref{tb: altBINLS}, by this corollary and by the definition (\ref{eq: altBVBINLSrho-uprhop-rmp}) of $\rho(t), p(t)$ the  scale factor  equation (\ref{eq: CEFEANONZERO}) holds for constants $\varepsilon , A, A_1, \dots, A_{N}$ and functions $G(t),G_1(t),\dots,G_{N}(t)$ according to Table \ref{tb: altBINLS}.  That is, we have regained (\ref{eq: altBVBINLSI0minusIinoX1}).  Now solving (\ref{eq: altBVBINLSVphi-RE}) for $\rho(t)$ and substituting this into (\ref{eq: altBVBINLSI0minusIinoX1}), we obtain\\
\\
\begin{equation}
\dot{H}_R
=
\frac{-\nu d\kappa}{(d-1)}
\left[ 
\frac{1}{2}\dot\phi^2 -V\circ\phi+\frac{(d-1)D}{2 d\nu \kappa R^{2/\nu}}-\frac{d(d-1)\beta^2}{2\kappa e^{2\alpha_1}R^{2/\nu d}}+p \right.\qquad\qquad\qquad\qquad\notag\end{equation}
\begin{equation}\qquad\qquad\qquad\qquad\qquad\qquad\qquad\left.+\frac{(d-1)}{2\nu^2 d\kappa}H_R^2-\frac{\Lambda}{\kappa} \right]+\frac{D}{R^{2/\nu}}-\frac{\nu d\beta^2}{e^{2\alpha_1}R^{2/\nu d}}.\end{equation}
Multiplying by $\frac{(d-1)}{\nu d}$ and rearranging,
\begin{equation}
\frac{(d-1)}{\nu d}\dot{H}_R+\frac{(d-1)}{2\nu^2 d}H_R^2-\frac{(d-1)D}{2 d\nu R^{2/\nu}}-\frac{(d-1)(d-2)\beta^2}{2e^{2\alpha_1}R^{2/\nu d}}\qquad\qquad\qquad\qquad\notag\end{equation}
\begin{equation}\qquad\qquad\qquad\qquad\qquad\qquad\qquad\
=
-\kappa
\left[ 
\frac{1}{2}\dot\phi^2 -V\circ\phi+ p \right]+\Lambda.\label{eq: altBVBINLSlhs=lhsofIi}\end{equation}
The left-hand side of this equation is in fact equal to the left-hand-side of $(I_i)$ for any $i\in\{1,\dots, d\}$.  To see this, again we use that the definitions $X_l(t)\stackrel{def.}{=}R(t)^{1/\nu d}e^{\alpha_l(t)}$ in (\ref{eq: altBVBINLSX-Ralpha}) and
$H_R(t)\stackrel{def.}{=}\frac{\dot{R}(t)}{R(t)}$ in (\ref{eq: altBVBINLSH-u}), and the condition $\sum_l c_l=0$, are the same as those in Theorem \ref{thm: BIEMP}.  Also by the definitions (\ref{eq: altBVBINLSR-ualpha-sigma}) and (\ref{eq: altBVBINLSdotsigma-upsi-P}) of $\alpha_l(t), \sigma(t)$ and $R(t)$, we obtain 
\begin{eqnarray}
\dot\alpha_l(t) R(t)^{1/\nu}
&=&c_l\dot\sigma(t) R(t)^{1/\nu}\notag\\
&=&\frac{c_l}{\theta}u(\sigma(t))R(t)^{1/\nu}\notag\\
&=&\frac{c_l}{\theta }\mbox{ is a constant}
\label{eq: altBVBINLSdotalphaRconstant},\end{eqnarray}
which shows that
\begin{equation}\ddot\alpha_l+\frac{1}{\nu}\dot\alpha_l H_R=0\label{eq: altBVBINLSddotalpha+stuff=0}\end{equation}
holds here, as it does in Theorem \ref{thm: BIEMP}.  Using the arguments in  (\ref{eq: BVBIEMPddotX/X-Halpha})-(\ref{eq: BVBIEMPBIlhsIi}) and also using (\ref{eq: altBVBINLSdotalphaldotalphak-Rc}) to slightly modify one term in (\ref{eq: BVBIEMPBIlhsIi}) to apply here, we obtain
\begin{equation}\displaystyle\sum_{l\neq i}(\dot{H}_l+H_l^2)+\displaystyle\sum_{\stackrel{l<k}{l,k\neq i}}H_lH_k=\frac{(d-1)}{\nu d}\dot{H}_R+\frac{(d-1)}{2\nu^2 d}H_R^2-\sum_{l<k}\frac{c_lc_k}{\theta^{2}R^{2/\nu}}.
\label{eq: altBVBINLSBIlhsIi}\end{equation}
Then by the condition (\ref{eq: altBVBINLScconditions}) on the constants $c_l$ and the definition (\ref{eq: altBVNLSX1-constR}) of $X_1(t)$, the left-hand side of any Einstein equation $(I_i), i>0$ is
\begin{equation}\displaystyle\sum_{l\neq i}(\dot{H}_l+H_l^2)+\displaystyle\sum_{\stackrel{l<k}{l,k\neq i}}H_lH_k-\frac{(d-1)(d-2)\beta^2}{2X_1^2}\qquad\qquad\qquad\qquad\notag\end{equation}
\begin{equation}\qquad\qquad\qquad=\frac{(d-1)}{\nu d}\dot{H}_R+\frac{(d-1)}{2\nu^2 d}H_R^2-\frac{(d-1)D}{2d\nu R^{2/\nu}}-\frac{(d-1)(d-2)\beta^2}{2e^{2\alpha_1}R^{2/\nu d}}.
\label{eq: altBVBINLSdotXldotXk/XlXk-Ralpha}\end{equation}
Combining (\ref{eq: altBVBINLSlhs=lhsofIi}) and (\ref{eq: altBVBINLSdotXldotXk/XlXk-Ralpha}), we obtain $(I_i)$ for all $i\in\{1,\dots,d\}$.  This proves the theorem.

\end{proof}

\chapter{Reformulations of a conformal Bianchi V model}

We now consider a  Bianchi V metric of the form
\begin{eqnarray} &&ds^2=-\left(a_2(t)\cdots a_d(t)\right)^ddt^2+(a_2(t)\cdots a_d(t)) dx_1^2+a_2^{d-1}(t)e^{2\beta x_1}dx_2^2+\notag\\
&& \ \ \ \ \ \ \ \ \ \ \ \ \cdots +a_d^{d-1}(t)e^{2\beta x_1}dx_d^2,\label{eq: confBVdmetric}
\end{eqnarray}
which represents a change of coordinate systems in comparison with Chapter 4.
In a $d+1-$dimensional spacetime the nonzero Einstein equations $g^{ij}G_{ij}=-\kappa g^{ij}T_{ij}+\Lambda$, multiplied by $|g_{00}|=(a_2\cdots a_d)^d$ and $\frac{4}{(d-1)}$, are 
\begin{equation}
\displaystyle\sum_{l=2}^d H_l^2 + (d+1)\displaystyle\sum_{l<k}H_lH_k - 2d\beta^2(a_2\cdots a_d)^{d-1}\notag
\end{equation}
\begin{equation}\label{eq: confBVEFEI0Id}
\stackrel{(I_0)}{=} 
\frac{4\kappa}{(d-1)}\left[\frac{\dot{\phi}^2}{2}+(a_2\cdots a_d)^d\left(V\circ\phi+\rho+\frac{\Lambda}{\kappa } \right) \right]\end{equation}
\begin{equation}\notag
\displaystyle\sum_{l=2}^d\left( 2\dot{H}_l - H_l^2 \right) - (d+1)\displaystyle\sum_{l<k}H_lH_k -2(d-2)\beta^2(a_2\cdots a_d)^{d-1}\end{equation}
\begin{equation}\notag
\stackrel{(I_1)}{=} 
\frac{ 4\kappa}{(d-1)}\left[-\frac{\dot{\phi}^2}{2}+(a_2\cdots a_d)^d\left(V\circ\phi-p+\frac{\Lambda}{\kappa }\right)  \right]\notag
\end{equation}
\begin{equation}\notag
-2\dot{H}_i-\displaystyle\sum_{l=2}^d H_l^2+\frac{2d}{(d-1)}\displaystyle\sum_{l=2}^d \dot{H}_l  -(d+1) \displaystyle\sum_{ l<k}H_lH_k -2(d-2)\beta^2(a_2\cdots a_d)^{d-1} 
\end{equation}
\begin{equation}
\stackrel{(I_i)}{=} 
\frac{ 4\kappa}{(d-1)}\left[-\frac{\dot{\phi}^2}{2}+(a_2\cdots a_d)^d\left(V\circ\phi - p+\frac{\Lambda}{\kappa }\right) \right]\notag
\end{equation}
where $H_l(t)\stackrel{def.}{=}\frac{\dot{a}_l}{a_l}$ and $i,l,k\in\{2,\dots ,d\}$.

\section{In terms of a Generalized EMP}

\begin{thm}\label{thm: confBVEMP}  Suppose you are given twice differentiable functions $a_2(t),\dots, a_d(t)>0$, a once differentiable function $\phi(t)$ and also functions $\rho(t), p(t), V(x)$ which satisfy the Einstein equations $(I_0), \dots, (I_d)$ in (\ref{eq: confBVEFEI0Id}) for some $\Lambda\in\mathds{R}, d\in\mathds{N}\backslash\{0,1\},$ and $ \kappa\in\mathds{R}\backslash\{0\}$.  Denote
\begin{equation}R(t)\stackrel{def.}{=}(a_2(t)a_3(t)\cdots a_d(t))^\nu\label{eq: confBVR-a}\end{equation}
for some $\nu\neq 0$.  If $f(\tau)$ is the inverse of a function $\tau(t)$ which satisfies
\begin{equation}\dot{\tau}(t)=\theta R(t)^{q+\frac{d}{2\nu}},\label{eq: confBVdottau-R}\end{equation}
for some constants $\theta>0$ and $q\neq 0$, then
\begin{equation}Y(\tau)=R(f(\tau))^{q}\qquad\mbox{ and }\qquad Q(\tau)= \frac{2q\nu\kappa}{(d-1)}
\varphi'(\tau)^2 \label{eq: confBVY-RQ-varphi}\end{equation}
solve the generalized EMP equation
\begin{equation}Y''(\tau)+Q(\tau)Y(\tau)=\frac{q \nu (d-2)L}{2\theta^2Y(\tau)^{\frac{q\nu+d}{q\nu} }} 
 - \frac{2q\nu  \beta^2}{\theta^2Y(\tau)^{(1+q\nu) /q\nu }}
- \frac{2q\nu \kappa(\varrho(\tau)+\textup{\textlhookp}(\tau)) }{(d-1)\theta^2 Y(\tau)} 
\label{eq: confBVEMP}\end{equation}
for 
\begin{equation}\varphi(\tau)=\phi(f(\tau))\label{eq: confBVvarphi-phi}\end{equation}
\begin{equation}\varrho(\tau)=\rho(f(\tau)), \ \textup{\textlhookp}(\tau)=p(f(\tau)). \label{eq: confBVEMPvarrho-rhohookp-p}\end{equation}
\begin{equation}L\stackrel{def.}{=}\frac{2}{(d-2)}\displaystyle\sum_{3\leq l<k\leq d}\mu_l\mu_k-\displaystyle\sum_{j=3}^d\mu_j^2,\label{eq: confBVlambda0-mul}\end{equation}
where $\mu_i\in\mathds{R}$ are such that $a_i(t)=\omega_ie^{\mu_i t}a_2(t)$ for some $\omega_i\in\mathds{R},  i\in\{3,\dots, d\}$.

Conversely, suppose you given a twice differentiable function $Y(\tau)>0$, a continuous function $Q(\tau)$ and also functions $\varrho(\tau), \textup{\textlhookp}(\tau)$ which solve (\ref{eq: confBVEMP}) for some constants $\theta>0$ and $q, \nu, \kappa \in\mathds{R}\backslash\{0\}, L\in\mathds{R}, d\in\mathds{N}\backslash\{0,1\}$.  In order to construct functions which solve $(I_0), \dots, (I_d)$, first find $\tau(t), \varphi(\tau)$ which solve the differential equations
\begin{equation}\dot{\tau}(t)=\theta Y(\tau(t))^{(2q\nu+d)/2q\nu}\qquad\mbox{ and }\qquad \varphi'(\tau)^2=\frac{(d-1)}{2q\nu\kappa}Q(\tau).\label{eq: confBVdottau-Yvarphi-Q}\end{equation}
Next find constants $\mu_i, i\in\{3,\dots, d\}$ which satisfy
\begin{equation}L=\frac{2}{(d-2)}\displaystyle\sum_{3\leq l<k\leq d}\mu_l\mu_k-\displaystyle\sum_{j=3}^d\mu_j^2,\label{eq:  confBVL-mulconverse}\end{equation}
and let
\begin{equation}R(t)=Y(\tau(t))^{1/q}\label{eq: confBVR-Y}.\end{equation}
Then the functions
\begin{equation}a_2(t)=R(t)^{1/\nu (d-1)}(\omega_3\cdots\omega_de^{(\mu_3+\cdots+\mu_d)t})^{-1/(d-1)}\label{eq: confBVa1-Rexp}\end{equation}
\begin{equation}a_i(t)=\omega_ie^{\mu_i t}a_2(t)\label{eq: confBVai-a1exp}\end{equation} 
\begin{equation}\phi(t)=\varphi(\tau(t))\label{eq: confBVphi-varphi}\end{equation}
\begin{equation}\rho(t)=\varrho(\tau(t)),\qquad\qquad p(t)=\textup{\textlhookp}(\tau(t))\label{eq: confBVrho-varrhop-textlhookp}\end{equation}
and 
\begin{equation}V(\phi(t))=\left[ \frac{(d-1)}{8\kappa}\left( \frac{d\theta^2}{q^2 \nu^2} (Y')^2
+\frac{(d-2)L}{ Y^{d/q\nu}} - \frac{4 d\beta^2}{ Y^{1/q \nu}}\right)
 -\frac{\theta^2}{2} (\varphi')^2 Y^{2} 
  -\varrho-\frac{\Lambda}{\kappa}\right]\circ\tau(t)\label{eq: confBVdefnV}\end{equation}
satisfy the Einstein equations $(I_0), \dots, (I_d)$ for any $\omega_i>0 , 2\leq i\leq d$.
\end{thm}

\begin{proof}
This proof will implement Theorem \ref{thm: EFE-EMP} with constants and functions as indicated in the following table.

   \begin{table}[ht]
\centering
\caption{{ Theorem \ref{thm: EFE-EMP} applied to conformal Bianchi V}}\label{tb: confBVEMP}
\vspace{.2in}
\begin{tabular}{r | l c r | l}
In Theorem & substitute & & In Theorem& substitute \\[4pt]
\hline
\raisebox{-5pt}{$a(t)$} & \raisebox{-5pt}{$R(t)$}      &&    \raisebox{-5pt}{$N$} & \raisebox{-5pt}{$2$}\\[8pt]
$\delta  $ & $-d/2\nu$      &&    $\varepsilon $ & ${2\nu \kappa}/{(d-1)}$\\[8pt]
   $G_0(t)$ & $\mbox{constant } \nu(d-2) L/2$    &&     $A_0$ & $0$ \\[8pt]
 $G_{1}(t)$ & $ \frac{-2\nu \kappa}{(d-1)}(\rho(t)+p(t))$  &&      $A_{1}$   &$-d/\nu$\\[8pt]
 $G_{2}(t)$ &constant $-2\nu\beta^2$  &&      $A_{2}$   &$-(d-1)/\nu$\\[8pt]
 $ \lambda_0(\tau)$ &constant ${q\nu (d-2)L}/{2\theta^2}$       &&      $B_0$ &$(d+q\nu)/q\nu$ \\[8pt] 
$\lambda_{1}(\tau)$&$\frac{-2q \nu  \kappa}{\theta^2(d-1)}(\varrho(\tau)+\textup{\textlhookp}(\tau))$ &&  $B_{1}$&  $1$ \\[6pt]
$\lambda_{2}(\tau)$&constant $-2q\nu\beta^2/\theta^2$ &&  $B_{2}$&  $(1+q\nu)/q\nu$ \\[6pt]
\hline
\end{tabular}
\end{table}


To prove the forward implication, we assume to be given functions which solve the Einstein field equations $(I_0), \dots, (I_d)$ in (\ref{eq: confBVEFEI0Id}).   Since the right-hand sides of Einstein equations $(I_i)$ are the same for all $i\in\{2, \dots, d\}$, we begin by equating the left-hand side of $(I_2)$ with the left-hand side of any $(I_j)$ for $j\in\{3,\dots, d\}$.  After dividing by $2$, we obtain
\begin{equation}\dot{H}_j=\dot{H}_2. \label{eq: confBVEMPequatelhss}\end{equation}
Integrating, we obtain
\begin{equation}H_j=H_2+\mu_j\label{eq: confBVEMPdota/a=dota/a+mu}\end{equation}
for $\mu_j\in\mathds{R}, j\in\{2, \dots, d\}$ and where $\mu_2\stackrel{def.}{=}0$.
Since in general $\frac{d}{dt}\ln(a_i)=\frac{\dot{a}_i}{a_i}=H_i$, (\ref{eq: confBVEMPdota/a=dota/a+mu}) can be written
\begin{equation}\frac{d}{dt}\ln(a_j)=\frac{d}{dt}\ln(a_2)+\mu_j.\label{eq: confBVEMPlogderiv=logderiv+mu}\end{equation}
Integrating,
\begin{equation}\ln(a_j)=\ln(a_2)+\mu_j t+c'_j\label{eq: confBVEMPlna=lna+mut+c}\end{equation}
for some $c_j\in\mathds{R}, j\in\{2,\dots, d\}$ and where $c'_1\stackrel{def.}{=}0$. Exponentiating and letting $\omega_i\stackrel{def.}{=}e^{c_i}>0,$
\begin{equation}a_j(t)=\omega_je^{\mu_j t}a_2(t),\label{eq: confBVEMPaj=cjemuta1}\end{equation}
where of course this holds trivially for $j=2$ where $\omega_2=1$ and $\mu_2=0$.

Forming the linear combination $d(I_0)-\displaystyle\sum_{i=1}^d(I_i)$ of Einstein's equations,
\begin{equation}2d\displaystyle\sum_{l=2}^d\left(H_l^2-\dot{H}_l\right)+2d(d+1)\displaystyle\sum_{l<k}H_lH_k-4d\beta^2(a_2\cdots a_d)^{d-1}\qquad\qquad\notag\end{equation}
\begin{equation}\qquad \qquad =\frac{4d\kappa}{(d-1)}\left[\dot\phi^2+(a_2\cdots a_d)^d(\rho+p)\right].\label{eq: confBVlinearcomboEFEHlHk}\end{equation}
Using the definition (\ref{eq: confBVR-a}) of $R(t)$, we define  
\begin{eqnarray}
H_R
&\stackrel{def.}{=}&\frac{\dot{R}}{R}\notag\\
&=&\frac{\nu (a_2\cdots a_d(t))^{\nu -1} (\dot{a}_2\cdots a_d+ \cdots+ a_2\cdots \dot{a}_d)}{(a_2\cdots a_d)^\nu}\notag\\
&=&\nu\displaystyle\sum_{j=2}^d H_j\notag\\
&=&\nu\displaystyle\sum_{j=2}^d\left(H_1+\mu_j\right)\notag\\
&=&\nu\left((d-1)H_2+\displaystyle\sum_{j=2}^d\mu_j\right),\label{eq: confBVHR-H1}
\end{eqnarray}
therefore we have that
\begin{equation}H_2=\frac{1}{\nu (d-1)}H_R-\frac{1}{(d-1)}\displaystyle\sum_{j=2}^d\mu_j\label{eq: confBVEMPH1wrtHR}\end{equation}
and
\begin{equation}\dot{H}_2=\frac{1}{\nu (d-1)}\dot{H}_R.\label{eq: confBVEMPdotH1wrtdotHR}\end{equation}
By equation (\ref{eq: confBVEMPdota/a=dota/a+mu}) we obtain
\begin{equation}
\displaystyle\sum_{l<k}H_lH_k
=\displaystyle\sum_{l<k}(H_2+\mu_l)(H_2+\mu_k)=\displaystyle\sum_{l<k}(H_2^2+(\mu_l+\mu_k)H_2+\mu_l\mu_k)
\label{eq: confBVEMPsumHlHk-H2}\end{equation}
for $l,k\in\{2,\dots ,d\}$.  The first term on the right-hand side of (\ref{eq: confBVEMPsumHlHk-H2}) does not depend on the indices $l,k$ and therefore is equal to $H_2^2$ times
\begin{equation}\displaystyle\sum_{l<k}1=\displaystyle\sum_{k=3}^d\displaystyle\sum_{l=2}^{k-1}1=\displaystyle\sum_{k=3}^d(k-2)=\displaystyle\sum_{j=1}^{d-2}j=\frac{(d-1)(d-2)}{2}.\label{eq: BVEMPsum1l<k}\end{equation}
 The second term on the right-hand side of (\ref{eq: confBVEMPsumHlHk-H2}) sums to $H_2$ times
\begin{eqnarray}
\displaystyle\sum_{l<k}(\mu_l+\mu_k)&=&\displaystyle\sum_{l=2}^{d-1}\displaystyle\sum_{k=l+1}^d\mu_l+\displaystyle\sum_{k=3}^d\displaystyle\sum_{l=2}^{k-1}\mu_k\notag\\
&=&\displaystyle\sum_{l=2}^{d-1}(d-l)\mu_l+\displaystyle\sum_{k=3}^d (k-2)\mu_k\notag\\
&=&(d-2)\mu_2+\displaystyle\sum_{j=3}^{d-1}(d-j+j-2)\mu_j+(d-2)\mu_d\notag\\
&=&(d-2)\displaystyle\sum_{j=2}^d\mu_j.\end{eqnarray}
Therefore (\ref{eq: confBVEMPsumHlHk-H2}) becomes
\begin{equation}\displaystyle\sum_{l<k}H_lH_k 
=
\frac{(d-1)(d-2)}{2}H_2^2+(d-2)H_2\displaystyle\sum_{j=2}^d\mu_j+\displaystyle\sum_{l<k}\mu_l\mu_k
\label{eq: confBVEMPnewlhsI0}.\end{equation}
By (\ref{eq: confBVEMPH1wrtHR}) and (\ref{eq: confBVEMPdotH1wrtdotHR}), we also have that
\begin{equation}\displaystyle\sum_{l<k}H_lH_k=\frac{(d-1)(d-2)}{2}\left( \frac{1}{\nu (d-1)}H_R-\frac{1}{(d-1)}\displaystyle\sum_{j=2}^d\mu_j    \right)^2\label{eq: confBIlhsI0HR}\end{equation}
\begin{equation}
\qquad+(d-2)\left( \frac{1}{\nu (d-1)}H_R-\frac{1}{(d-1)}\displaystyle\sum_{j=2}^d\mu_j    \right)\displaystyle\sum_{j=2}^d\mu_j+\displaystyle\sum_{l<k}\mu_l\mu_k.\notag
\end{equation}
Collecting terms (the $H_R$ terms sums to zero),
\begin{equation}\displaystyle\sum_{l<k}H_lH_k=
\frac{(d-2)}{2\nu^2(d-1)}H_R^2- \frac{(d-2)}{2(d-1)}\left(\displaystyle\sum_{j=2}^d\mu_j\right)^2+\displaystyle\sum_{l<k}\mu_l\mu_k.\label{eq: confBVlhsI0-HRwithsummulmuk}
\end{equation}
Also we have that
\begin{equation}
\left(\displaystyle\sum_{j=2}^d\mu_j\right)^2=2\displaystyle\sum_{l<k}\mu_l\mu_k+\displaystyle\sum_{j=2}^d\mu_j^2\label{eq: sumsquaredofmu=2mulmuk+summusquared}
\end{equation}
for $l,k\in\{2,\dots,d\}$ therefore (\ref{eq: confBVlhsI0-HRwithsummulmuk}) becomes
\begin{equation}\displaystyle\sum_{l<k}H_lH_k=
\frac{(d-2)}{2\nu^2(d-1)}H_R^2  +\frac{1}{(d-1)}\displaystyle\sum_{l<k}\mu_l\mu_k
- \frac{(d-2)}{2(d-1)}\displaystyle\sum_{j=2}^d\mu_j^2 .\label{eq: confBVlhsI0-HRsimplified}
\end{equation}   
 Defining the quantity
\begin{eqnarray}L&\stackrel{def.}{=}&\frac{2}{(d-2)}\displaystyle\sum_{l<k}\mu_l\mu_k-\displaystyle\sum_{j=2}^d\mu_j^2\notag\\
&=& \frac{2}{(d-2)}\displaystyle\sum_{3\leq l<k\leq d}\mu_l\mu_k-\displaystyle\sum_{j=3}^d\mu_j^2 \mbox{ \ \ \  \  (since $\mu_2=0$)}, \label{eq: confBVL-mu-inproof}\end{eqnarray}
we can now rewrite equation (\ref{eq: confBVlhsI0-HRsimplified}) to say that
\begin{equation}\displaystyle\sum_{l<k}H_lH_k=
\frac{(d-2)}{2\nu^2(d-1)}H_R^2  +\frac{(d-2)}{2(d-1)}L.\label{eq: confBVEMPlhsI0-HRsimplifiedwithL}
\end{equation}
Similarly by (\ref{eq: confBVEMPequatelhss}), (\ref{eq: confBVEMPdota/a=dota/a+mu}), (\ref{eq: confBVEMPH1wrtHR}), (\ref{eq: confBVEMPdotH1wrtdotHR}), (\ref{eq: sumsquaredofmu=2mulmuk+summusquared}) and (\ref{eq: confBVL-mu-inproof}), we have that
\begin{eqnarray}
\displaystyle\sum_{l=2}^d\dot{H}_l\notag
&=&\displaystyle\sum_{l=2}^d\dot{H}_2\\
&=&(d-1)\dot{H}_2\notag\\
&=&\frac{1}{\nu}\dot{H}_R \label{eq: confBVsumdotHl-HR}
\end{eqnarray}
and also
\begin{eqnarray}
\displaystyle\sum_{l=2}^dH_l^2\notag
&=&\displaystyle\sum_{l=2}^d(H_2+\mu_l)^2\\
&=&(d-1)H_2^2+2H_2\displaystyle\sum_{l=2}^d\mu_l +\displaystyle\sum_{l=2}^d\mu_l^2\notag\\
&=&\frac{1}{(d-1)}\left(\frac{1}{\nu }H_R-\displaystyle\sum_{j=2}^d\mu_j\right)^2 +\frac{2}{(d-1)}\left(\frac{1}{\nu }H_R-\displaystyle\sum_{j=2}^d\mu_j\right)\displaystyle\sum_{l=2}^d\mu_l +\displaystyle\sum_{l=2}^d\mu_l^2\notag
\\
&=&\frac{1}{\nu^2(d-1)}H_R^2-\frac{1}{(d-1)} \left(\displaystyle\sum_{j=2}^d \mu_j\right)^2+\displaystyle\sum_{l=2}^d\mu_l^2\notag\\
&=&\frac{1}{\nu^2(d-1)}H_R^2-\frac{2}{(d-1)}\displaystyle\sum_{l<k}\mu_l\mu_k+\frac{(d-2)}{(d-1)}\displaystyle\sum_{l=2}^d\mu_l^2\notag\\
&=&\frac{1}{\nu^2(d-1)}H_R^2-\frac{(d-2)}{(d-1)}L.\notag\\ \label{eq: confBVsumHl^2-HR}
\end{eqnarray}
That is, the linear combination $d(I_0)-\displaystyle\sum_{i=1}^d(I_i)$ of Einstein equations in (\ref{eq: confBVlinearcomboEFEHlHk}) can be written in terms of $R$ and $H_R$ by using   (\ref{eq: confBIlhsI0-HRsimplifiedwithL}), (\ref{eq: confBVsumdotHl-HR}) and (\ref{eq: confBVsumHl^2-HR}) so that we obtain
\begin{equation}\frac{d^2}{\nu^2}H_R^2+d(d-2)L-\frac{2d}{\nu}\dot{H}_R-4d\beta^2R^{(d-1)/\nu} =\frac{4d\kappa}{(d-1)}\left[\dot\phi^2+R^{d/\nu}(\rho+p)\right].\end{equation}
Multiplying by $-\frac{\nu}{2d}$,
\begin{equation}\dot{H}_R-\frac{d}{2\nu}H_R^2-\frac{\nu(d-2)}{2}L+2\nu\beta^2R^{(d-1)/\nu} =\frac{-2\nu\kappa}{(d-1)}\left[\dot\phi^2+R^{d/\nu}(\rho+p)\right].\label{eq: confBVEMPI0minusIiwrtHRsimplified}\end{equation}
This shows that $R(t), \phi(t), \rho(t)$ and $p(t)$ satisfy the hypothesis of Theorem \ref{thm: EFE-EMP}, applied with  constants $\epsilon, \varepsilon , N, A_0, \dots, A_{N}$ and functions $a(t), G_0(t),\dots,G_{N}(t)$ according to Table \ref{tb: confBVEMP}.  Since $\tau(t), Y(\tau), Q(\tau)$ and $\varphi(\tau)$ defined in (\ref{eq: confBVdottau-R}), (\ref{eq: confBVY-RQ-varphi}) and (\ref{eq: confBVvarphi-phi})  are equivalent to that in the forward implication of Theorem \ref{thm: EFE-EMP}, by this theorem and by definition (\ref{eq: confBVEMPvarrho-rhohookp-p}) of $\varrho(\tau), \textup{\textlhookp}(\tau)$, the generalized EMP equation (\ref{eq: AEMP}) holds for constants $B_0, \dots, B_{N}$ and functions $\lambda_0(\tau), \dots, \lambda_{N}(\tau)$ as indicated in Table \ref{tb: confBVEMP}.  This proves the forward implication.

To prove the converse implication, we assume to be given functions which solve the generalized EMP equation (\ref{eq: confBVEMP}) and we begin by showing that $(I_0)$ is satisfied.  Differentiating the definition of $R(t)$ in  (\ref{eq: confBVR-Y}) and using the definition in (\ref{eq: confBVdottau-Yvarphi-Q}) of $\tau(t)$,
\begin{eqnarray} 
\dot{R}(t)
&=&\frac{1}{q}Y(\tau(t))^{\frac{1}{q}-1}Y'(\tau(t))\dot{\tau}(t)\\
&=&\frac{\theta}{q}Y(\tau(t))^{\frac{1}{q}(1+d/2\nu)}Y'(\tau(t))\end{eqnarray}
Dividing by $R(t)$,
\begin{equation}H_{R}(t)\stackrel{def.}{=}\frac{\dot{R}(t)}{R(t)}=\frac{\theta}{q}Y(\tau(t))^{d/2q\nu}Y'(\tau(t)).\label{eq: confBVHR-Y}\end{equation}
Differentiating the definition (\ref{eq: confBVphi-varphi}) of $\phi(t)$ and using definition (\ref{eq: confBVdottau-Yvarphi-Q}) of $\tau(t)$,
\begin{equation}\dot\phi(t)=\varphi'(\tau(t))\dot\tau(t)=\theta \varphi'(\tau(t)) Y(\tau(t))^{1+d/2q\nu}.\label{eq: confBVdotphi-varphiprimeY}\end{equation}
Using (\ref{eq: confBVHR-Y}) and (\ref{eq: confBVdotphi-varphiprimeY}), and also the definitions (\ref{eq: confBVR-Y}) and (\ref{eq: confBVrho-varrhop-textlhookp}) of $R(t)$ and $\rho(t)$ respectively, the definition (\ref{eq: confBVdefnV}) of $V\circ\phi$ can be written as
\begin{equation}V\circ\phi= \frac{1}{R^{d/\nu}} \left[\frac{(d-1)}{4\kappa}\left(\frac{d}{2\nu^2}H_R^2 +\frac{(d-2)}{2}L-2d\beta^2R^{(d-1)/\nu} \right) -\frac{\dot\phi^2}{2}  \right]-\rho -\frac{\Lambda}{\kappa}.\label{eq: confBVEMPVphi-R}\end{equation}
The quantity in the inner parenthesis here is in fact equal to the left-hand-side of equation $(I_0)$.  To see this, we differentiate the definitions in (\ref{eq: confBVa1-Rexp}) and (\ref{eq: confBVai-a1exp})  of $a_i(t)$, divide the results by $a_i(t)$, and use the definition (\ref{eq: confBVHR-Y}) of $H_R$  to obtain
\begin{equation}H_2\stackrel{def.}{=}\frac{\dot{a}_2}{a_2}=
\frac{1}{\nu (d-1)}H_R-\frac{1}{(d-1)}(\mu_3+\cdots +\mu_d)\label{eq: confBVEMPH1-HR}\end{equation}
and
\begin{equation}H_i\stackrel{def.}{=}\frac{\dot{a}_i}{a_i}=\frac{\dot{a}_2}{a_2}+\mu_i=H_2+\mu_i \end{equation}
for $i\in\{2,\dots,d\}$ by taking $\mu_2\stackrel{def.}{=}0$.  Therefore we obtain
\begin{equation}\dot{H}_i=\dot{H}_2=\frac{1}{\nu (d-1)}\dot{H}_R.\label{eq: confBVdotHidotH2-dotHR}\end{equation}
This confirms that the identities (\ref{eq:  confBVEMPequatelhss}), (\ref{eq:  confBVEMPdota/a=dota/a+mu}), (\ref{eq: confBVEMPH1wrtHR}) and (\ref{eq: confBVEMPdotH1wrtdotHR}) hold in the converse direction, so that the
 computations (\ref{eq: confBVEMPsumHlHk-H2})-(\ref{eq: confBVsumHl^2-HR}) are also valid in the converse direction for $L$ in (\ref{eq: confBVL-mulconverse}). That is,
 \begin{equation}
\displaystyle\sum_{l<k}H_lH_k=\frac{(d-2)}{2\nu^2(d-1)}H_R^2+\frac{(d-2)}{2(d-1)}L\label{eq: confBVEMPHlHk-HR}\end{equation}
and 
\begin{equation}\displaystyle\sum_{l=2}^d H_l^2=\frac{1}{\nu^2(d-1)}H_R^2-\frac{(d-2)}{(d-1)}L\label{eq: confBVEMPsumHl^2-HRL}\end{equation}
which shows that the left-hand side of $(I_0)$ is equal to
\begin{equation}\displaystyle\sum_{l=2}^d H_l^2+(d+1)\displaystyle\sum_{l<k}H_lH_k-2d\beta^2(a_2\cdots a_d)^{d-1}\notag\end{equation}
\begin{eqnarray}&&=\frac{d}{2\nu^2}H_R^2 +\frac{(d-2)}{2}L-2d\beta^2R^{(d-1)/\nu}\\
\end{eqnarray}
therefore (\ref{eq: confBVVphi-R}) can be written $V\circ \phi=$
\begin{equation}\frac{1}{R^{2/\nu}} \left[\frac{(d-1)}{4\kappa}\left(\displaystyle\sum_{l=2}^d H_l^2+(d+1) \displaystyle\sum_{l<k}H_lH_k -2d\beta^2(a_2\cdots a_d)^{d-1} \right) -\frac{\dot\phi^2}{2}  \right]-\rho -\frac{\Lambda}{\kappa}\label{eq: confBIVphi-HlHk}\end{equation}
so that $(I_0)$ holds.

To conclude the proof we must also show that the equations $(I_1), \dots, (I_d)$ hold.  In the converse direction the hypothesis of the converse of Theorem \ref{thm: EFE-EMP} holds, applied with constants 
$N, B_0, \dots, B_N$ and functions $\lambda_0(\tau), \dots, \lambda_N(\tau)$ as indicated in Table \ref{tb: confBIEMP}.  Since $\tau(t), \varphi(\tau), R(t)$ and $\phi(t)$ defined in (\ref{eq: confBVdottau-Yvarphi-Q}), (\ref{eq: confBVR-Y}) and (\ref{eq: confBVphi-varphi}) are consistent with the converse implication of Theorem \ref{thm: EFE-EMP}, applied with $a(t), \delta  $ and $\varepsilon $ as in Table \ref{tb: confBVEMP}, by this theorem and by the definition (\ref{eq: confBVrho-varrhop-textlhookp}) of $\rho(t), p(t)$ the  scale factor  equation (\ref{eq: AEFE}) holds for constants $\delta  , \varepsilon , A_0, \dots, A_N$ and functions $G_0(t), \dots, G_N(t)$ according to Table \ref{tb: confBVEMP}.  That is, we have regained equation (\ref{eq: confBVEMPI0minusIiwrtHRsimplified}).  Now solving (\ref{eq: confBVEMPVphi-R}) for $R^{d/\nu}\rho(t)$ and substituting this into (\ref{eq: confBVEMPI0minusIiwrtHRsimplified}), we obtain
\begin{equation}
\dot{H}_R-\frac{d}{2\nu}H_R^2-\frac{\nu(d-2) }{2}L+2\nu\beta^2R^{(d-1)/\nu}=\frac{-2\nu \kappa}{(d-1)}\left[\frac{\dot\phi^2}{2}+\frac{(d-1)}{4\kappa}\left(\frac{d}{2\nu^2}H_R^2 \right.\right.\notag\end{equation}
\begin{equation}\left.\left. +\frac{(d-2)}{2}L-2d\beta^2R^{(d-1)/\nu} \right)+R^{d/\nu}\left(-V\circ\phi+p-\frac{\Lambda}{\kappa}\right)\right].
\end{equation}
Collecting terms, multiplying the equation times $\frac{2}{\nu}$, and using that by definitions (\ref{eq: confBVa1-Rexp}) and (\ref{eq: confBVai-a1exp}) of $a_2, \dots, a_d$ we have that $R=(a_2\cdots a_d)^\nu$,
\begin{equation}
\frac{2}{\nu}\dot{H}_R-\frac{d}{2\nu^2}H_R^2-\frac{(d-2) }{2}L-2(d-2)\beta^2R^{(d-1)/\nu}\qquad\qquad\qquad\qquad\qquad\notag\end{equation}
\begin{equation}\qquad\qquad\qquad\qquad= \frac{-4 \kappa}{(d-1)}\left[\frac{\dot\phi^2}{2}+(a_2\cdots a_d)^d\left(-V\circ\phi+p-\frac{\Lambda}{\kappa}\right)\right]   .
\label{eq: confBVlhsshouldequallhsIi>0}\end{equation}
The left-hand side of this equation is in fact equal to the left-hand side of each of the Einstein equations $(I_1)$ and $(I_i)$ for $i\in\{2,\dots, d\}$.  To see this, we use (\ref{eq: confBVdotHidotH2-dotHR})-(\ref{eq: confBVEMPsumHl^2-HRL}) to obtain 
\begin{equation}
\displaystyle\sum_{l=2}^d\left(2\dot{H}_l - H_l^2\right) -(d+1)\displaystyle\sum_{l<k}H_lH_k -2(d-2)\beta^2(a_2\cdots a_d)^{d-1}\qquad\qquad\qquad\notag\end{equation}
\begin{eqnarray}
&&=\frac{2}{\nu}\dot{H}_R - \frac{d}{2\nu^2}H_R^2 - \frac{(d-2)}{2}L - 2(d-2)\beta^2R^{(d-1)/\nu}\label{eq: confBVEMPlhsI1-HR}
\end{eqnarray}
and 
\begin{equation}
-2\dot{H}_i-\displaystyle\sum_{l=2}^d H_l^2+\frac{2d}{(d-1)}\displaystyle\sum_{l=2}^d \dot{H}_l  -(d+1) \displaystyle\sum_{ l<k}H_lH_k -2(d-2)\beta^2(a_2\cdots a_d)^{d-1} \qquad\qquad\qquad\notag
\end{equation}
\begin{eqnarray}
&&=\frac{2}{\nu}\dot{H}_R-\frac{d}{2\nu^2}H_R^2 - \frac{(d-2)}{2}L-2(d-2)\beta^2 R^{(d-1)/\nu}\label{eq: confBVlhsIi-HR}
\end{eqnarray}
for $i\in\{2,\dots,d\}$.  By (\ref{eq: confBVEMPlhsI1-HR})  and (\ref{eq: confBVlhsshouldequallhsIi>0}) equation $(I_1)$ holds, and by (\ref{eq: confBVlhsIi-HR}) and (\ref{eq: confBVlhsshouldequallhsIi>0}) equation $(I_i)$ holds for $i\in\{2,\dots,d\}$.
This proves the theorem.
\end{proof}

\subsection{Reduction to generalized EMP with classical term}

We take $\rho=p=0$ and choose parameter $q=1/2\nu$ in Theorem \ref{thm: confBVEMP} to find that solving the Bianchi V Einstein equations
\begin{equation}
\displaystyle\sum_{l=2}^d H_l^2 +(d+1)\displaystyle\sum_{l<k}H_lH_k-2d\beta^2(a_2\cdots a_d)^{d-1}\notag
\end{equation}
\begin{equation}\label{eq: confBVEFEI0Idrhopzero}
\stackrel{(I_0)'}{=} 
\frac{4\kappa}{(d-1)}\left[\frac{\dot{\phi}^2}{2}+(a_2\cdots a_d)^d\left( V\circ\phi+\frac{\Lambda}{\kappa}\right) \right]\end{equation}
\begin{equation}\notag
\displaystyle\sum_{l=2}^d\left(2\dot{H}_l - H_l^2\right) - (d+1)\displaystyle\sum_{l<k}H_lH_k - 2(d-2)\beta^2(a_2\cdots a_d)^{d-1}\end{equation}
\begin{equation}\notag
\stackrel{(I_1)'}{=} 
\frac{ 4\kappa}{(d-1)}\left[-\frac{\dot{\phi}^2}{2}+(a_2\cdots a_d)^d\left( V\circ\phi+\frac{\Lambda}{\kappa}\right) \right]\notag
\end{equation}
\begin{equation}\notag
- 2\dot{H}_i - \displaystyle\sum_{l=2}^d H_l^2 + \frac{2d}{(d-1)}\displaystyle\sum_{l=2}^d \dot{H}_l  - (d+1) \displaystyle\sum_{ l<k}H_lH_k - 2(d-2)\beta^2(a_2\cdots a_d)^{d-1} 
\end{equation}
\begin{equation}
\stackrel{(I_i)'}{=} 
\frac{ 4\kappa}{(d-1)}\left[-\frac{\dot{\phi}^2}{2}+(a_2\cdots a_d)^d \left( V\circ\phi +\frac{\Lambda}{\kappa}\right) \right]\notag
\end{equation}
is equivalent to solving the classical EMP equation
\begin{equation}Y''(\tau)+Q(\tau)Y(\tau)=\frac{(d-2)L}{4\theta^2Y(\tau)^{ 1+2d }} 
 - \frac{  \beta^2}{\theta^2Y(\tau)^3 }
\label{eq: confBVEMPq=1/2nu}\end{equation}
for constants $\theta,L>0$ and $\beta\in\mathds{R}$.  The solutions of $(I_0)', (I_1)',\dots,(I_d)'$ in (\ref{eq: confBVEFEI0Idrhopzero}) and of (\ref{eq: confBVEMPq=1/2nu}) are related by
\begin{equation}R(t)=Y(\tau(t))^{2\nu} \qquad\mbox{ and }\qquad \varphi '(\tau)^2= \frac{(d-1)}{\kappa} Q(\tau)\label{eq: confBVclassEMPa-Yvarphi-Qq=1/2nu}\end{equation}
for $q\neq 0, \phi(t)=\varphi(\tau(t)), R(t)\stackrel{def.}{=} \left(a_2(t)\cdots a_d(t)\right)^{2\nu}$ and 
\begin{equation}\dot\tau(t)=\theta R(t)^{(d+1)/2\nu}=\theta Y(\tau(t))^{d+1},\label{eq: confBIclassEMPdottau-a-Yq=1/2nu}\end{equation}
for any $\theta>0$.  Also the constant
\begin{equation}L\stackrel{def.}{=}\frac{2}{(d-2)}\displaystyle\sum_{3\leq l<k\leq d}\mu_l\mu_k-\displaystyle\sum_{j=3}^d\mu_j^2,\label{eq: confBVL-mulmukq=1/2nu}\end{equation}
where $\mu_i\in\mathds{R}$ are such that $a_i(t)=\omega_ie^{\mu_i t}a_2(t)$ for some $c_i\in\mathds{R},  i\in\{3,\dots, d\}$.
In the converse direction 
\begin{equation}a_2(t)=R(t)^{1/\nu(d-1)}(\omega_3\cdots\omega_de^{(\mu_3+\cdots+\mu_d)t})^{-1/(d-1)}\label{eq: confBVEMPX-Ralphaq=1/2nu}\end{equation}
and   
\begin{equation}V(\phi(t))=\left[ \frac{(d-1)}{8\kappa}\left( 4d\theta^2 (Y')^2
+\frac{(d-2)L}{ Y^{2d}} - \frac{4 d\beta^2}{ Y^{2}}\right)
 -\frac{\theta^2}{2} (\varphi')^2 Y^{2} 
  -\frac{\Lambda}{\kappa}\right]\circ\tau(t)\label{eq: twotermsconfBVdefnV}\end{equation}

\begin{example}
For $d=3$, $\theta=\beta=\nu=1$ and $L=-12$, (\ref{eq: confBVEMPq=1/2nu}) becomes $Y''(\tau)+Q(\tau)Y(\tau)= - 3/Y(\tau)^7 - 1/Y(\tau)^3$.  We take the solution $Y(\tau)=(4\tau^2-1)^{1/4}$ in equation  (\ref{eq: exactEMPtwotermssoln}) and refer to (\ref{eq:  taueqnforysolnwithtwoterms}) - (\ref{eq:  Ytauforysolntwoterms}) to obtain solution $\tau(t)=\frac{-1}{2}coth(2(t-t_0))$ of the differential equation $\dot\tau(t)=Y(\tau(t))^4=4\tau(t)^2-1$.  Therefore by (\ref{eq: Ytauforysolntwoterms}) we have
\begin{eqnarray}
R(t)&=&Y(\tau(t))^2\notag\\
&=&csch(2(t-t_0))
\end{eqnarray}
for $t>t_0$, and 
\begin{eqnarray}
a_2(t)&=&R(t)^{1/2}e^{-\sqrt{3}t} = \sqrt{ csch(2(t-t_0)) } e^{-\sqrt{3}t }\notag\\
a_3(t)&=&e^{2\sqrt{3}t}a_2(t)=\sqrt{ csch(2(t-t_0)) } e^{\sqrt{3}t }\end{eqnarray}
where we have taken $\mu_3=2\sqrt{3}$ so that $L=-12= - \mu_3^2$.  Since $Q(\tau)=0=\varphi'(\tau)$,
\begin{equation}\phi(t)\stackrel{def.}{=}\varphi(\tau(t)) = \phi_0\end{equation}
for constant $\phi_0\in\mathds{R}$ and by (\ref{eq: twotermsconfBVdefnV}), (\ref{eq: Ytauforysolntwoterms}), (\ref{eq: Yprimeforysolntwoterms})
\begin{eqnarray}
V(\phi(t))&=&\left[ \frac{3}{\kappa}\left(  (Y')^2
-\frac{1}{ Y^{6}} - \frac{1}{ Y^{2}}\right)
  -\frac{\Lambda}{\kappa}\right]\circ\tau(t)\notag\\
  &=& \frac{3}{\kappa}sinh(2(t-t_0)) \left(    cosh^2(2(t-t_0))  - sinh^2(2(t-t_0)) - 1 \right) -\frac{\Lambda}{\kappa}\notag\\
  &=&-\frac{\Lambda}{\kappa}.\end{eqnarray}
  That is, we obtain a vacuum solution.
\end{example}

\subsection{Another Bianchi V metric}

Although it is not included as a special case of the above metric (\ref{eq: confBVdmetric}), we will additionally consider the metric $ds^2= - X(t)Y(t)dt^2+X(t)Y(t)dx^2+X(t)^2e^{\beta x}dy^2+Y(t)^2e^{\beta x}dz^2$ for $X(t),Y(t)>0$ with the energy-momentum tensor $T_{ij}=T^{(1)}_{ij}$ to be that of a minimally coupled scalar field $\phi$ with a potential $V$ as in (\ref{eq:  Tijminimallycoupledphi}).  This will show how the methodology developed in this thesis can be applied to this additional form of the Bianchi V metric.  The Einstein equations $g^{ij}G_{ij}=-\kappa g^{ij}T_{ij}+\Lambda$, multiplied by $2|g_{00}|=2XY$, in this case are
\begin{eqnarray}\notag H_X^2+H_Y^2+4H_XH_Y-\frac{3}{2}\beta^2&\stackrel{(I_0)'}{=}& \kappa \left[ \dot\phi^2 + 2XY V\circ\phi\right]\\
\notag  2\dot{H}_X+H_X^2+2\dot{H}_Y+H_Y^2-\frac{1}{2}\beta^2 &\stackrel{(I_1)'}{=}& \kappa \left[  -\dot\phi^2 +2XY V\circ\phi\right]
\\
\notag \dot{H}_X+3\dot{H}_Y+2H_Y^2-\frac{1}{2}\beta^2 &\stackrel{(I_2)'}{=}&\kappa \left[  -\dot\phi^2 +2XY V\circ\phi\right]\\
\notag \dot{H}_Y+3\dot{H}_X+2H_X^2-\frac{1}{2}\beta^2 &\stackrel{(I_3)'}{=}&\kappa \left[  -\dot\phi^2 +2XY V\circ\phi\right]\end{eqnarray}
where as usual $H_X(t)=\dot{X}(t)/X(t)$ and $H_Y(t)=\dot{Y}(t)/Y(t)$.  Forming the linear combination $3(I_0)'-(I_1)'-(I_2)'-(I_3)'$ and dividing by 6, we obtain 
\begin{equation}2H_XH_Y-\frac{1}{2}\beta^2 -\dot{H}_X-\dot{H}_Y=\kappa\dot\phi^2.\label{eq: JacobslincombEFE}\end{equation}
Equating the left sides of $(I_1)$, $(I_2)$ we obtain $\dot\eta+\frac{1}{\nu}\eta H_R=0$ for $\eta(t)\stackrel{def.}{=}H_X(t)-H_Y(t)$, $R(t)\stackrel{def.}{=}\left(X(t)Y(t)\right)^{\nu}$ for some $\nu\neq 0$, and $H_R(t)=\dot{R}(t)/R(t)$.  By Lemma A.1 with $\mu=1/\nu\neq 0$ this shows that $\eta XY$ is a constant function.  Setting $D=\frac{\nu}{2}\eta^2X^2Y^2$ and writing (\ref{eq: JacobslincombEFE}) in terms of $R$, we obtain
\begin{equation}\dot{H}_R-\frac{1}{2\nu}H_R^2+\nu\kappa\dot\phi^2=-\frac{D}{R^{2/\nu}}-\frac{\beta^2\nu}{2}.\label{eq: JacobslincomboEFEwrtR}\end{equation}
Next we apply Theorem \ref{thm: EFE-EMP} with substitutions made according to the following table.
   \begin{table}[ht]
\centering
\caption{{ Theorem \ref{thm: EFE-EMP} applied to a third Bianchi V}\label{tb: JacobsBVEMP}}
\vspace{.2in}
\begin{tabular}{r | l c r | l}
In Theorem & substitute & & In Theorem& substitute \\[4pt]
\hline
\raisebox{-5pt}{$a(t)$} & \raisebox{-5pt}{$R(t)$} && \raisebox{-5pt}{$N$ }& \raisebox{-5pt}{$2$}\\[8pt]
$\delta  $ &$-\frac{1}{2\nu}$      &&    $\varepsilon $ & $\nu\kappa$\\[8pt]
$G_0$ & $-D$    &&     $A_0$ & $2/\nu$ \\[8pt]
 $G_{1}$ & $ -\beta^2\nu/2$  &&      $A_{1}$   &$0$\\[8pt]
 $ \lambda_0$ &$-{qD}/{\theta^2}$       &&      $B_0$ &$(3+q\nu)/q\nu$ \\[8pt] 
$\lambda_{1}$&${-\beta^2\nu q}/{2\theta^2}$ &&        $B_{1}$&  $(q\nu+1)/q\nu$ \\[6pt]
\hline
\end{tabular}
\end{table}   

\break

By Theorem \ref{thm: EFE-EMP} we obtain the generalized EMP equation
\begin{equation}Y''(\tau)+Q(\tau)Y(\tau)=\frac{-qD}{\theta^2Y(\tau)^{(3+q\nu)/q\nu}} - \frac{\beta^2\nu q}{2\theta^2 Y(\tau)^{(q\nu+1)/q\nu}}\label{eq: JacobsEMP}\end{equation}
for $Y(\tau)=R(f(\tau))^q$, $q\neq 0$ and $Q(\tau)=q\nu\kappa\varphi'(\tau)^2$ where $f(\tau)$ is the inverse of $\tau(t)$ satisfying $\dot\tau(t)=\theta R(t)^{(2q\nu +1)/2\nu}$ and $\varphi(\tau)=\phi(f(\tau))$.  Conversely, given a solution $Y(\tau)>0, Q(\tau)$ to (\ref{eq: JacobsEMP}) one solves for $\tau(t)$ and $\sigma(t)$ in the equations $\dot\tau(t)=\theta Y(\tau(t))^{(2\nu q+1) / 2\nu q}$ and $\dot\sigma(t)=1/\dot\tau(t)^{2/(2q\nu+1)}$ for some $\theta>0$.  Also in the converse direction we define $\phi(t)=\varphi(\tau(t))$, $R(t)=Y(\tau(t))^{1/q}$, $X(t)=R(t)^{1/2\nu}e^{c\sigma(t)}$, $Y(t)=R(t)^{1/2\nu}e^{-c\sigma(t)}$ for a constant $c$ such that $c^2=D\theta^{2/(1+2q\nu)}/2\nu$, and the potential
\begin{equation}V(\phi(t))=\left[\frac{1}{\kappa}\left( \frac{3\theta^2}{4\nu^2q^2}(Y')^2 - \frac{D}{2\nu Y^{3/q\nu}} - \frac{3\beta^2}{4Y^{1/q\nu}}\right) -\frac{\theta^2}{2}Y^2(\varphi')^2\right]\circ\tau(t).\label{eq: JacobsV}\end{equation}
Taking $\nu=1, q=1/2, D=6$ and $ \beta=-2$, (\ref{eq: JacobsEMP}) becomes $Y''(\tau)+Q(\tau)Y(\tau)=-3/Y(\tau)^7-1/Y(\tau)^3$.  We take solution $Y(\tau)=\left(\frac{4}{\theta^2}\tau^2-1\right)^{1/4}$ and use (\ref{eq: tauforysolnswithtwotermsr=2})-(\ref{eq: Ytauforysolntwotermsr=2}) to obtain the solution $\tau(t)=\frac{\theta}{2}\cosh(2(t-t_0)$ of the differential equation $\dot\tau(t)=\theta Y(\tau(t))^2$.  We also solve for $\dot\sigma(t)=1/\dot\tau(t)$ to obtain $\sigma(t)=\frac{1}{2\theta}\ln(tanh(t-t_0))$.  Therefore by (\ref{eq: Ytauforysolntwotermsr=2}) we have 
\begin{equation}R(t)=Y(\tau(t))^{2}=sinh(2(t-t_0))\end{equation}
for $t>t_0$, and 
\begin{equation}X(t)^2=R(t)e^{2\sqrt{3\theta}\sigma(t)}=\sinh(2(t-t_0))tanh^{\sqrt{3}/\sqrt{\theta}}(t-t_0)\end{equation}
\begin{equation}Y(t)^2=R(t)e^{-2\sqrt{3\theta}\sigma(t)}=\sinh(2(t-t_0))tanh^{-\sqrt{3}/\sqrt{\theta}}(t-t_0).\end{equation}
Since $Q(\tau)=0$, $\phi(t)=\phi_0$ is constant and by (\ref{eq: JacobsV}), (\ref{eq: Ytauforysolntwotermsr=2}) and (\ref{eq: Yprimeforysolntwotermsr=2}),
\begin{eqnarray}
V(\phi(t))&=&0
.\end{eqnarray}
By identifying $X(t)^2, Y(t)^2$ and $X(t)Y(t)$ here with $A_1(t), A_2(t)$ and $A_3(t)$ in \cite{Joseph}, we obtain the Joseph vacuum solution with the constants $k_1=k_2=k_3=1$.

\section{In terms of a Schr\"odinger-Type Equation}

To reformulate Einstein's field equations $(I_0), \dots, (I_d)$ in (\ref{eq: confBVEFEI0Id}) in terms of an equation with one less non-linear term than that which is provided by the generalized EMP formulation, one can apply Corollary \ref{cor: EFE-NLSAzeroEnonzero} to the difference $d(I_0)-\displaystyle\sum_{i=1}^d(I_i)$ (and similar to above, define $V\circ\phi$ in $u-$notation to be such that $(I_0)$ holds).   Below is the resulting statement.

\begin{thm}\label{thm: confBVNLS}
Suppose you are given twice differentiable functions $a_2(t), \dots, a_d(t)>0$, a once differentiable function $\phi(t)$, and also functions $\rho(t), p(t), V(x)$ which satisfy the Einstein equations $(I_0),\dots,(I_d)$ for some $\Lambda\in\mathds{R}, d\in\mathds{N}\backslash\{0,1\}, \kappa\in\mathds{R}\backslash\{0\}$.  
Denote \begin{equation}R(t)\stackrel{def.}{=} \left(a_2(t)\cdots a_d(t)\right)^\nu\label{eq: confBVNLSR-a}\end{equation}for some $\nu\neq 0$,  
then the functions
\begin{eqnarray}
u(\sigma)&=&R(\sigma+t_0)^{-d/2\nu}\label{eq: confBVNLSu-R}\\
P(\sigma)&=&\frac{d\kappa}{(d-1)}\psi '(\sigma)^2\label{eq: confBVNLSP-psi}
\end{eqnarray}
solve the Schr\"odinger-type equation
\begin{equation}u''(\sigma)+\left[E-P(\sigma)\right]u(\sigma)=\frac{d\beta^2}{u^{(d-2)/d}}+
\frac{ d \kappa(\uprho(\sigma)+\mathrm{p}(\sigma))}{(d-1) u(\sigma)}
\label{eq: confBVNLS}\end{equation}
for
\begin{equation}\psi(\sigma)=\phi(\sigma+t_0)\label{eq: confBVNLSpsi-phi}\end{equation}
\begin{equation}\uprho(\sigma)=\rho(\sigma+t_0), \ \mathrm{p}(\sigma)=p(\sigma+t_0). \label{eq: confBVNLSuprho-rhormp-p}\end{equation} 
and where 
\begin{equation}E\stackrel{def.}{=}\frac{d(d-2)}{4}\left(\frac{2}{(d-2)}\displaystyle\sum_{l<k}\mu_l\mu_k-\displaystyle\sum_{j=3}^d\mu_j^2\right)\label{eq: confBVNLSE-mu}\end{equation}
for constants $\mu_j$ such that $a_j(t)=\omega_je^{\mu_j t}a_2(t)$ for some $\omega_j>0, j\in\{2, \dots, d\}$.

Conversely, suppose you are given a twice differentiable function $u(\sigma)>0$, and also functions $P(\sigma)$ and $\uprho(\sigma), \mathrm{p}(\sigma)$ which solve (\ref{eq: confBVNLS}) for some constants $E<0,  \kappa\in\mathds{R}\backslash\{0\}$ and $d\in\mathds{N}\backslash\{0,1\}$.  In order to construct functions which solve $(I_0),\dots, (I_d)$, 
If $\psi(\sigma)$ is such that  
\begin{equation}\psi '(\sigma)^2= \frac{(d-1)}{ d \kappa} P(\sigma),\label{eq: confBVNLSpsi-P}\end{equation}
 constants $\mu_i, i\in\{3, \dots, d\}$  satisfy
\begin{equation}E{=}\frac{d(d-2)}{4}\left(\frac{2}{(d-2)}\displaystyle\sum_{l<k}\mu_l\mu_k-\displaystyle\sum_{j=3}^d\mu_j^2\right),\label{eq: confBINLSE-muconverse}\end{equation}
and
\begin{equation}R(t)=u(t-t_0)^{-2\nu/d}.\label{eq: confBVNLSR-u}\end{equation} 
Then the functions
\begin{equation}a_2(t)=R(t)^{1/\nu (d-1)}(\omega_3\cdots\omega_de^{(\mu_2+\cdots+\mu_d)t})^{-1/(d-1)}\label{eq: confBVNLSa1-R}\end{equation}
\begin{equation}a_i(t)=\omega_ie^{\mu_i t}a_2(t)\label{eq: confBVNLSai-a1exp}\end{equation} 
\begin{equation}\phi(t)=\psi(t-t_0)\label{eq: confBVNLSphi-psi}\end{equation}
\begin{equation}\rho(t)=\uprho(t-t_0),\qquad\qquad p(t)=\mathrm{p}(t-t_0)\label{eq: confBVNLSprho-t-t0}\end{equation}
and\\
$V(\phi(t))$
\begin{equation}
=\left[\frac{(d-1)}{2\kappa}\left(
\frac{1}{d}(u')^2
+u^2\frac{E}{d}-d\beta^2u^{2/d}\right)
-\frac{1}{2}(\psi ')^2u^2
-\uprho-\frac{\Lambda}{\kappa}
\right]\circ(t-t_0)
\label{eq: confBVNLSVphi-u}
\end{equation}
satisfy the equations $(I_0),\dots, (I_d)$ for any $\omega_i>0, 2\leq i\leq d$.\end{thm}

\break

\begin{proof}
This proof will implement Corollary \ref{cor: EFE-NLSAzeroEnonzero} with constants and functions as indicated in the following table.

\begin{table}[ht]
\centering
\caption{{ Corollary \ref{cor: EFE-NLSAzeroEnonzero} applied to conformal Bianchi V}}\label{tb: confBVNLS}
\vspace{.2in}
\begin{tabular}{r | l c r | l}
In Corollary  & substitute & & In Corollary & substitute \\[4pt]
\hline
\raisebox{-5pt}{$a(t)$} & \raisebox{-5pt}{$R(t)$}      &&    \raisebox{-5pt}{$N$} & \raisebox{-5pt}{$2$}\\[8pt]
$\delta  $ & $-d/2\nu$      &&   $\varepsilon $ &${2\nu \kappa}/{(d-1)}$\\[8pt]
 $G(t)$ & $\mbox{constant } 2\nu E/d$    &&     $A$ & $0$ \\[8pt]
 $G_{1}(t)$ & $ \frac{-2\nu  \kappa}{(d-1)}(\rho(t)+p(t))$  &&      $A_{1}$   &$-d/\nu$\\[8pt]
 $G_{2}(t)$ & constant $-2\nu\beta^2$  &&      $A_{2}$   &$-(d-1)/\nu$\\[8pt]
$F_{1}(\sigma)$&$ \frac{d \kappa}{(d-1)} (\uprho(\sigma)+\mathrm{p}(\sigma)) $ &&        $C_{1}$&  $1$ \\[6pt]
$F_{2}(\sigma)$&constant $ d\beta^2$ &&        $C_{2}$&  $  (d-2)/d $ \\[6pt]
\hline
\end{tabular}
\end{table}

Much of this proof will rely on computations that are exactly the same as those seen in the proof of Theorem \ref{thm: confBVEMP} (the generalized EMP formulation of conformally Bianchi V).  Therefore we will restate the relevant results here, but point the reader to the details in the proof of Theorem \ref{thm: confBVEMP}.

To prove the forward implication, we assume to be given functions which solve the  Einstein field equations $(I_0),\dots, (I_d)$.  Since the right-hand sides of Einstein equation $(I_i)$ are all the same for $i\in\{1,\dots, d\}$, we begin by equating the left-hand side of $(I_{2})$ with the left-hand side of any $(I_j)$  for $j\in\{2,\dots, d\}$ since it will give us a simplifying relation among the scale factors $a_2(t), \dots, a_d(t)$.   Exactly this was done in (\ref{eq: confBVEMPequatelhss})-(\ref{eq: confBVEMPaj=cjemuta1}) so that again we obtain
\begin{equation}H_j=H_2+\mu_j\label{eq: confBVNLSHj-H1+mu}\end{equation}
and
\begin{equation}a_j(t)=\omega_je^{\mu_j t}a_2(t)\label{eq: confBVNLSaj-a1}\end{equation}
for $\omega_j>0, \mu_j\in\mathds{R}, j\in\{2, \dots, d\}$, and $\mu_2=0, \omega_2=1$.  Following the arguments given in (\ref{eq: confBVlinearcomboEFEHlHk})-(\ref{eq: confBVEMPI0minusIiwrtHRsimplified}), we form the linear combination $d(I_0)-\displaystyle\sum_{i=1}^d(I_i)$ and write everything in terms of $R=(a_2\cdots a_d)^{\nu}$ as defined in (\ref{eq: confBVNLSR-a}) to obtain 
\begin{equation}\dot{H}_R-\frac{d}{2\nu}H_R^2-\frac{2 \nu}{d}E+2\nu\beta^2R^{(d-1)/\nu} =\frac{-2\nu\kappa}{(d-1)}\left[\dot\phi^2+R^{d/\nu}(\rho+p)\right].\label{eq: confBVNLSI0minusIiwrtHRsimplified}\end{equation}
Since 
\begin{equation}L\stackrel{def.}{=}\frac{2}{(d-2)}\displaystyle\sum_{l<k}\mu_l\mu_k-\displaystyle\sum_{j=2}^d\mu_j^2  \ \Rightarrow \ \frac{\nu (d-2)}{2}L=\frac{2\nu}{d}E\mbox{ (by (\ref{eq: confBVNLSE-mu}))}\label{eq: confBINLSL-mu}\end{equation}
and 
\begin{equation}H_R\stackrel{def.}{=}\frac{\dot{R}}{R},\label{eq: confBINLSHR-R}\end{equation}
we again have that
\begin{equation}H_2=\frac{1}{\nu (d-1)}H_R-\frac{1}{(d-1)}\displaystyle\sum_{j=2}^d\mu_j\label{eq: confBVNLSH2-HR}\end{equation}
\begin{equation}\dot{H}_2=\frac{1}{\nu (d-1)}\dot{H}_R.\label{eq: confBVNLSdotH1-dotHR}\end{equation}
This shows that $R(t), \phi(t), \rho(t)$ and $p(t)$ satisfy the hypothesis of Corollary \ref{cor: EFE-NLSAzeroEnonzero}, applied with  constants $\epsilon, \varepsilon , N, A, A_1 \dots, A_{N}$ and functions $a(t), G(t),G_1(t),\dots,G_{N}(t)$ according to Table \ref{tb: confBVNLS}.  Since $
u(\sigma), P(\sigma)$ and $\psi(\sigma)$ defined in (\ref{eq: confBVNLSu-R})-(\ref{eq: confBVNLSP-psi}) and (\ref{eq: confBVNLSpsi-phi})  are equivalent to that in the forward implication of Corollary \ref{cor: EFE-NLSAnonzeroEzero}, by this corollary and by definition (\ref{eq: confBVNLSuprho-rhormp-p}) of $\uprho(\sigma), \mathrm{p}(\sigma)$, the Schr\"odinger-type equation (\ref{eq: CNLSAZERO}) holds for constants $C_1, \dots, C_{N}$ and functions $F_1(\sigma), \dots, F_{N}(\sigma)$ as indicated in Table \ref{tb: confBVNLS}.  This proves the forward implication.

To prove the converse implication, we assume to be given functions which solve the Schr\"odinger-type equation (\ref{eq: confBVNLS}) and we will show that equations $(I_0), \dots, (I_d)$ are satisfied. 
To show that $(I_0)$ is satisfied, we differentiate the definition of $R(t)$ in (\ref{eq: confBVNLSR-u}) to obtain
\begin{eqnarray}\dot{R}(t)
&=&-\frac{2\nu}{d} u(t-t_0)^{-2\nu/d-1}u'(t-t_0).\label{eq: confBVNLSdotR-u}\end{eqnarray}
Dividing by $R(t)$,
\begin{equation}H_R\stackrel{def.}{=}\frac{\dot{R}}{R}=-\frac{2\nu}{d} \frac{u'(t-t_0)}{u(t-t_0)}.\label{eq: confBVNLSHR-u}\end{equation}
Differentiating the definition (\ref{eq: confBVNLSphi-psi}) of $\phi(t)$ 
\begin{equation}
\dot{\phi}(t)=\psi '(t-t_0).\label{eq: confBVNLSphi-u}
\end{equation}
Using (\ref{eq: confBVNLSHR-u}) and (\ref{eq: confBVNLSphi-u}), and also the definitions (\ref{eq: confBVNLSR-u}) and (\ref{eq: confBVNLSprho-t-t0}) of $R(t)$ and $\rho(t)$ respectively, the definition (\ref{eq: confBINLSVphi-u}) of $V\circ\phi$ can be written as
\begin{equation}V\circ\phi= \frac{1}{R^{d/\nu}} \left[\frac{(d-1)}{4\kappa}\left(\frac{d}{2\nu^2}H_R^2 +\frac{2}{d}E- 2d\beta^2R^{(d-1)/\nu} \right) -\frac{\dot\phi^2}{2}   \right]-\rho-\frac{\Lambda}{\kappa}.\label{eq: confBVVphi-R}\end{equation}
The quantity in parenthesis here is in fact equal to the left-hand-side of equation $(I_0)$.  
To see this, note that the definitions of $a_2(t), a_i(t)$ in (\ref{eq: confBVNLSa1-R}), (\ref{eq: confBVNLSai-a1exp}) and also $H_R\stackrel{def.}{=}\dot{R}/R$ are the same as those in Theorem \ref{thm: confBVEMP}.  Therefore we may follow the arguments given in (\ref{eq: confBVHR-H1})-(\ref{eq: confBVsumHl^2-HR}) to see that the identities
\begin{equation}\displaystyle\sum_{l<k}H_lH_k=
\frac{(d-2)}{2\nu^2(d-1)}H_R^2  +\frac{2}{d(d-1)}E,\label{eq: confBVNLSIlhsI0-HRsimplifiedwithL}\end{equation}
\begin{equation}
\displaystyle\sum_{l=2}^d\dot{H}_l=\notag
\frac{1}{\nu}\dot{H}_R \label{eq: confBVNLSsumdotHl-HR}
\end{equation}
and
\begin{equation}
\displaystyle\sum_{l=2}^dH_l^2=\frac{1}{\nu^2(d-1)}H_R^2-\frac{4}{d(d-1)}E \label{eq: confBVNLSsumHl^2-HR}
\end{equation}
hold in the converse direction (since (\ref{eq: confBVL-mulconverse}) in Theorem \ref{thm: confBVEMP} and (\ref{eq: confBINLSE-muconverse}) here show that $L=4E/d(d-2)$).
This shows that the left-hand side of $(I_0)$ is equal to
\begin{equation}\displaystyle\sum_{l=2}^d H_l^2+(d+1)\displaystyle\sum_{l<k}H_lH_k-2d\beta^2(a_2\cdots a_d)^{d-1}\notag\end{equation}
\begin{eqnarray}&&=\frac{d}{2\nu^2}H_R^2 +\frac{2}{d}E-2d\beta^2R^{(d-1)/\nu}
\end{eqnarray}
and therefore (\ref{eq: confBVVphi-R}) can be written as
$V\circ \phi=$
\begin{equation}\frac{1}{R^{d/\nu}} \left[\frac{(d-1)}{4\kappa}\left(\displaystyle\sum_{l=2}^d H_l^2+(d+1) \displaystyle\sum_{l<k}H_lH_k -2d\beta^2(a_2\cdots a_d)^{d-1} \right) -\frac{\dot\phi^2}{2}   \right]-\rho-\frac{\Lambda}{\kappa}\label{eq: confBIVphi-HlHk}\end{equation}
which shows that $(I_0)$ holds in the converse direction.

To conclude the proof we must also show that the equations $(I_1), \dots, (I_d)$ hold.   In the converse direction the hypothesis of the converse of Corollary \ref{cor: EFE-NLSAzeroEnonzero} holds, applied with constants $N, C_1, \dots, C_{N}$ and functions $F_1(\sigma), \dots, F_{N}(\sigma)$ as indicated in Table \ref{tb: confBVNLS}.  Since  $
\psi(\sigma), R(t)$ and $\phi(t)$ defined in (\ref{eq: confBVNLSpsi-P}), (\ref{eq: confBVNLSR-u}) and (\ref{eq: confBVNLSphi-psi}) are consistent with the converse implication of Corollary \ref{cor: EFE-NLSAzeroEnonzero}, applied with $a(t)$ and $\delta  , \varepsilon $ as in Table \ref{tb: confBVNLS}, by this corollary and by the definition (\ref{eq: confBVNLSprho-t-t0}) of $\rho(t), p(t)$ the  scale factor  equation (\ref{eq: CEFEAZERO}) holds for constants $\delta  , \varepsilon , A, A_1, \dots, A_{N}$ and functions $G(t),G_1(t),\dots,G_{N}(t)$ according to Table \ref{tb: confBVNLS}.  That is, we have regained (\ref{eq: confBVNLSI0minusIiwrtHRsimplified}).  Now solving (\ref{eq: confBVVphi-R}) for $R^{d/\nu}\rho(t)$ and substituting this into (\ref{eq: confBVNLSI0minusIiwrtHRsimplified}), we obtain
\begin{equation}\dot{H}_R-\frac{d}{2\nu}H_R^2-\frac{2\nu}{d}E+2\nu\beta^2R^{(d-1)/\nu} =\frac{-2\nu\kappa}{(d-1)}\left[\frac{\dot\phi^2}{2}+R^{d/\nu}\left(-V\circ\phi+p -\frac{\Lambda}{\kappa}\right) \right.\notag\end{equation}
\begin{equation}\left.+ \frac{(d-1)}{4\kappa}\left(\frac{d}{2\nu^2}H_R^2 + \frac{2}{d}E - 2d\beta^2R^{(d-1)/\nu}\right)\right].\label{eq: confBVNLSI0minusIiwrtHRnotsimplified}\end{equation}
Simplifying and multiplying by $2/\nu$,
\begin{equation}\frac{2}{\nu}\dot{H}_R-\frac{d}{2\nu^2}H_R^2-\frac{2}{d}E-2(d-2)\beta^2R^{(d-1)/\nu}\qquad\qquad\qquad\qquad \qquad\qquad\notag\end{equation}
\begin{equation}\qquad\qquad\qquad\qquad\qquad=\frac{-4\kappa}{(d-1)}\left[\frac{\dot\phi^2}{2}+R^{d/\nu}\left(-V\circ\phi+p -\frac{\Lambda}{\kappa}\right) \right] . \label{eq: confBVNLSrhsIiwrtHRsimplified}\end{equation}
As noted above, the computations  (\ref{eq: confBVEMPnewlhsI0})-(\ref{eq: confBVsumHl^2-HR}) from Theorem \ref{thm: confBVEMP} still hold in this theorem, in the converse direction.  Therefore by (\ref{eq:  confBVNLSIlhsI0-HRsimplifiedwithL})-(\ref{eq: confBVNLSsumHl^2-HR})  the left-hand side of  $(I_1)$ is equal to
\begin{equation}
\displaystyle\sum_{l=2}^d\left(2\dot{H}_l -  H_l^2\right) - (d+1)\displaystyle\sum_{l<k}H_lH_k - 2(d-2)\beta^2(a_2\cdots a_d)^{d-1}\qquad\qquad\qquad\notag\end{equation}
\begin{eqnarray}
&&=\frac{2}{\nu}\dot{H}_R - \frac{d}{2\nu^2}H_R^2 - \frac{2}{d}E - 2(d-2)\beta^2R^{(d-1)/\nu}\label{eq: confBVNLSlhsI1-HR}
\end{eqnarray}
and the left-hand side of $(I_i)$  is equal to
\begin{equation}
- 2\dot{H}_i - \displaystyle\sum_{l=2}^d H_l^2 + \frac{2d}{(d-1)}\displaystyle\sum_{l=2}^d \dot{H}_l   - (d+1) \displaystyle\sum_{ l<k}H_lH_k  - 2(d-2)\beta^2(a_2\cdots a_d)^{d-1} \qquad\qquad\qquad\notag
\end{equation}
\begin{eqnarray}
&&=\frac{2}{\nu}\dot{H}_R - \frac{d}{2\nu^2}H_R^2  -  \frac{2}{d}E - 2(d-2)\beta^2 R^{(d-1)/\nu}\label{eq: confBVNLSlhsIi-HR}
\end{eqnarray}
for $i\in\{2,\dots,d\}$.  By (\ref{eq: confBVNLSrhsIiwrtHRsimplified}), (\ref{eq: confBVNLSlhsI1-HR}) and (\ref{eq: confBVNLSlhsIi-HR}), $(I_i)$ hold in the converse direction for $i\in\{1,\dots,d\}$.  This proves the theorem.
\end{proof}

\appendix 
\setcounter{chapter}{0}

\chapter{A short lemma}  

\begin{appxlem}\label{lem: shortlemma}
For any differentiable function $f(t)$, any positive differentiable function $R(t)$, and constant $\mu\in\mathds{R}$,
\begin{equation}f(t)R(t)^\mu \mbox{ is a constant}\end{equation}
if and only if
\begin{equation}\dot{f}(t)+\mu f(t) H(t)=0\end{equation}
where $H(t)\stackrel{def.}{=}\frac{\dot{R}(t)}{R(t)}$.
\end{appxlem}

\begin{proof}
$f(t)R(t)^\mu$ is a constant function if and only if
\begin{equation}\frac{d}{dt}\left(f(t)R(t)^\mu\right)=0.\end{equation}
Or equivalently,
\begin{equation}\dot{f}(t)R(t)^\mu+\mu f(t)R(t)^{\mu-1}\dot{R}(t)=0.\label{eq: lemmadotfR^mu+stuff=0}\end{equation}
Since $R(t)$ is positive, (\ref{eq: lemmadotfR^mu+stuff=0}) holds if and only if the same equation divided by $R(t)^\mu$ holds.  That is,
\begin{equation}\dot{f}(t)+\mu f(t)H(t)=0.\label{eq: lemmadotf+mufH=0}\end{equation}
\end{proof}

\appendix 
\setcounter{chapter}{1}
\chapter{The Einstein tensor in $d+1$ dimensions}

In this thesis we consider a number of different metrics on pseudo-Riemannian spacetime manifolds of arbitrary dimension.  
For a fixed dimension, one can use a computer program, like Mathematica or Maple, to compute the Einstein tensor $G_{ij}\stackrel{def.}{=}R_{ij}-\frac{1}{2}Rg_{ij}$ in terms of the Ricci tensor $R_{ij}$ and the scalar curvature $R$.
However on a space of arbitrary dimension, the use of computer programs involves a bit of guesswork (to the best of the author's knowledge) since one must choose a few fixed dimensions in which to compute $G_{ij}$, then deduce a more general form for $G_{ij}$ for arbitrary dimension, and in all cases computing power will limit the ability to check one's formula for very high dimension.  

We will show a by-hand method for computing $G_{ij}$ on a manifold of arbitrary dimension, which is manageable to apply at least when the metric is diagonal with coefficient functions that depend only on a few of the coordinate variables.  As an example, we will use the conformal Bianchi V metric 
\begin{equation}ds^2=-(a_2\cdots a_d)^ddt^2+(a_2\cdots a_d)dx_1^2+a_2^{d-1}e^{2\beta x_1}dx_2^2+\cdots + a_d^{d-1}e^{2\beta x_1}dx_d^2\label{eq: AppxBVmetric}\end{equation}
for $a_i=a_i(t), \beta\neq 0$,  and $i,j\in\{0, 1, \dots, d\}$, as in Chapter 7.  

We will compute the Ricci tensor in three pieces, each computed directly from Christoffel symbols of the second kind $\Gamma_{ij}^k$, by
\begin{equation}R_{ij}=R_{ij}^{(1)}+R_{ij}^{(2)}-R_{ij}^{(3)}\label{eq: AppxRicci-Gamma}\end{equation}
for 
\begin{equation}R_{ij}^{(1)}\stackrel{def.}{=}\displaystyle\sum_{k=0}^d\left(\Gamma^k_{kj,i}-\Gamma^k_{ij,k}\right),\label{eq: AppxRicci1-Gamma}\end{equation}
\begin{equation}R_{ij}^{(2)}\stackrel{def.}{=}\displaystyle\sum_{m=0}^d\displaystyle\sum_{n=0}^d \Gamma^n_{im}\Gamma^m_{nj},\label{eq: AppxRicci2-Gamma}\end{equation}
and
\begin{equation}R_{ij}^{(3)}\stackrel{def.}{=}\displaystyle\sum_{m=0}^d\displaystyle\sum_{n=0}^d\Gamma_{ij}^m\Gamma^n_{nm}\label{eq: AppxRicci3-Gamma}\end{equation}
where $,i$ denotes differentiation $\frac{\partial}{\partial x_i}$ with respect to $x_i$.  $\Gamma_{ij}^k$ are given in terms of the metric components by 
\begin{equation}\Gamma_{ij}^k\stackrel{def.}{=}\frac{1}{2}\displaystyle\sum_{s=0}^d g^{sk}\left(g_{si,j}-g_{ij,s}+g_{js,i}\right).\label{eq: AppxGammaij-gij}\end{equation}

We will first form matrices of Christoffel symbols and their derivatives.   Since the metric (\ref{eq: AppxBVmetric}) is diagonal, the sum in (\ref{eq: AppxGammaij-gij}) reduces to
\begin{equation}\Gamma_{ij}^k\stackrel{def.}{=}\frac{1}{2}g^{kk}\left(g_{ki,j}-g_{ij,k}+g_{jk,i}\right),\label{eq: AppxGammaij-gijdiagonal}
\end{equation}
which is only nonzero if at least two of $i,j,k$ are equal.  Since the metric is symmetric, if $k=i$ then 
\begin{equation}\Gamma_{kj}^k=\Gamma_{jk}^k=\frac{1}{2}g^{kk}g_{kk,j}\label{eq: AppxGammakjk}\end{equation}
and if $i=j$ then
\begin{equation}\Gamma_{ii}^k=\frac{1}{2}g^{kk}(2g_{ki,i}-g_{ii,k}).\label{eq: Gammaiik}\end{equation}
We denote by $\left[\Gamma_{ij}^k\right]$ the matrix with rows indexed by $i$, columns indexed by $j$, and $k$ fixed, and we use (\ref{eq: Gammaiik}) to compute the diagonal entries  and  (\ref{eq: AppxGammakjk}) to compute the nonzero non-diagonal entries of this matrix to obtain

\begin{equation}\left[\Gamma_{ij}^0\right]=\left(\begin{array}{cccccc}
\frac{d}{2}\displaystyle\sum_{l=2}^dH_l&&&&&\\
&\frac{1}{2(a_2\cdots a_d)^{d-1}}\displaystyle\sum_{l=2}^dH_l&&&&\\
&&A_2&&&\\
&&&\ddots &&\\
&&&&A_d\\
\end{array}\right)\label{eq: AppxGammaij0}\end{equation}
for  $A_i\stackrel{def.}{=}\frac{(d-1)a_i^{d-1}H_ie^{2\beta x_1}}{2(a_2\cdots a_d)^d}$, $2\leq i\leq d$,

\begin{equation}\left[\Gamma_{ij}^1\right]=\left(
\begin{array}{cccccc}
0&\frac{1}{2}\displaystyle\sum_{l=2}^d H_l&&&&\\
\frac{1}{2}\displaystyle\sum_{l=2}^d H_l&0&&&&\\
&&B_2&&&\\
&&&\ddots &&\\
&&&&B_d\\
\end{array}
\right)\label{eq: AppxGammaij1}\end{equation}
for $B_i\stackrel{def.}{=}-\beta\frac{a_i^{d-1}e^{2\beta x_1}}{(a_2\cdots a_d)}$, $2\leq i\leq d$, and 

\begin{equation}\left[\Gamma_{ij}^k\right]=\left(\begin{array}{cccccc}
0                              &0                                                &\cdots &\frac{(d-1)}{2} H_k \mbox{\small \ in $k^{th}$ column }            &\cdots&0\\
0                              &0                                                &\cdots&\beta\mbox{\small \ in $k^{th}$ column }                          &\cdots&0                             \\
\vdots                      &\vdots                                        &&&&\\
\frac{(d-1)}{2} H_k  \mbox{\small \ in $k^{th}$ row}                      &\beta\mbox{\small \ in $k^{th}$ row}&    &                                                                                          &            &             \\
       \vdots            &\vdots                                          &            &&&\\
0                & 0                                              &                                   &                                                          &             &        \\
\end{array}\right)\notag\end{equation}
\begin{equation}\label{eq: AppxGammaijk}\end{equation}
for each $k\in\{2, \dots ,d\}$.  

Taking derivatives of the Christoffel symbols in (\ref{eq: AppxGammaij0}), (\ref{eq: AppxGammaij1}) and (\ref{eq: AppxGammaijk}) with respect to $x_0=t$ (denoted by dot), we get

\begin{equation}\left[\Gamma_{ij,0}^0\right]=\left(\begin{array}{cccccc}
\frac{d}{2}\displaystyle\sum_{l=2}^d\dot{H}_l&&&&&\\
&\hspace{-.4in}\frac{1}{2(a_2\cdots a_d)^{d-1}}\left(\displaystyle\sum_{l=2}^d \dot{H}_l -(d-1) \left(  \displaystyle\sum_{l=2}^d H_l\right)^2\right)  &&&&\\
&& C_2&&&\\
&&& &  \ddots&\\
&&&&&C_d\\
\end{array}\right)\label{eq: AppxGammaij0,0}
\end{equation}
for $C_i\stackrel{def.}{=}\dot{A}_i=\frac{(d-1)a_i^{d-1}e^{2\beta x_1}}{2(a_2\cdots a_d)^d}\left((d-1)H_i^2+\dot{H}_i-dH_i\displaystyle\sum_{l=2}^d H_l \right), 2\leq i\leq d$,

\begin{equation}\left[\Gamma_{ij,0}^1\right]=\left(
\begin{array}{cccccc}
0&\frac{1}{2}\displaystyle\sum_{l=2}^d \dot{H}_l&&&&\\
\frac{1}{2}\displaystyle\sum_{l=2}^d \dot{H}_l&0&&&&\\
&& D_2&&&\\
&&&&\ddots&\\
&&&&&D_d\\
\end{array}
\right)\label{eq: AppxGammaij1,0}
\end{equation}
for $D_i\stackrel{def.}{=}\dot{B}_i=\frac{-\beta a_i^{d-1}e^{2\beta x_1}}{(a_2\cdots a_d)}\left(  (d-1)H_i-\displaystyle\sum_{l=2}^dH_l \right), 2\leq i\leq d$ and for $2\leq k\leq d$,

\begin{equation}\left[\Gamma_{ij,0}^k\right]=\left(\begin{array}{cccccc}
0                              &0                                                &\cdots &\frac{(d-1)}{2} \dot{H}_k \mbox{\small \ in $k^{th}$ column }            &\cdots&0\\
0                              &0                                                &\cdots&0                  &\cdots&0                             \\
\vdots                      &\vdots                                        &&&&\\
\frac{(d-1)}{2} \dot{H}_k  \mbox{\small \ in $k^{th}$ row}                      &0&    &                                                                                          &            &             \\
       \vdots            &\vdots                                          &            &&&\\
0                & 0                                              &                                   &                                                          &             &        \\
\end{array}\right).\label{eq: AppxGammaijk,0}\end{equation}

Taking derivatives of the Christoffel symbols in (\ref{eq: AppxGammaij0}), (\ref{eq: AppxGammaij1}) and (\ref{eq: AppxGammaijk}) with respect to $x_1$, 

\begin{equation}\left[\Gamma_{ij,1}^0\right]=\left(\begin{array}{cccccc}
0&&&&&\\
&0&&&&\\
&& E_2&&&\\
&&&\ddots &&\\
&&&&E_d\\
\end{array}\right)\label{eq: AppxGammaij0,1}\end{equation}
for $E_i\stackrel{def.}{=}\frac{\partial}{\partial x_1}A_i=\frac{\beta(d-1)a_i^{d-1}H_ie^{2\beta x_1}}{(a_2\cdots a_d)^d}, 2\leq i\leq d$,

\begin{equation}\left[\Gamma_{ij,1}^1\right]=\left(
\begin{array}{cccccc}
0&&&&&\\
&0&&&&\\
&&F_2&&&\\
&&&\ddots &&\\
&&&&F_d\\
\end{array}
\right)\label{eq: AppxGammaij1,1}\end{equation}
for $F_i\stackrel{def.}{=}\frac{\partial}{\partial x_1}B_i=-2\beta^2\frac{ a_i^{d-1}e^{2\beta x_1}}{(a_2\cdots a_d)}, 2\leq i\leq d$, and also 
\begin{equation}\left[\Gamma_{ij,1}^k\right]=0.\label{eq: AppxGammaijk,1}\end{equation}
for each $k\in\{2, \dots ,d\}$.

To form the first sum in $R_{ij}^{(1)}$ we note that the metric, and therefore the Christoffel symbols, depend only on $x_0, x_1$ and we use the matrices (\ref{eq: AppxGammaij0,0}) and (\ref{eq: AppxGammaij1,1}) to obtain
$\left[\displaystyle\sum_{k=0}^d\Gamma_{ij,k}^k\right]=\left[\displaystyle\sum_{k=0}^d\Gamma_{ij,0}^0+\displaystyle\sum_{k=0}^d\Gamma_{ij,1}^1\right]=$

\begin{equation}\left(\begin{array}{cccccc}
\frac{d}{2}\displaystyle\sum_{l=2}^d\dot{H}_l&&&&&\\
&\hspace{-.4in}\frac{1}{2(a_2\cdots a_d)^{d-1}}\left(\displaystyle\sum_{l=2}^d \dot{H}_l -(d-1) \left(  \displaystyle\sum_{l=2}^d H_l\right)^2\right)  &&&&\\
&&& J_2&&\\
&&& &\ddots&\\
&&&&&J_d\\
\end{array}\right)\label{eq: AppxGammaij0,0+Gammaij1,1}
\end{equation}
for $J_i\stackrel{def.}{=}C_i+F_i=\frac{a_i^{d-1}e^{2\beta x_1}}{(a_2\cdots a_d)^d}\left(\frac{(d-1)^2}{2}H_i^2+\frac{(d-1)}{2}\dot{H}_i-\frac{d(d-1)}{2}H_i\displaystyle\sum_{l=2}^d H_l  - 2\beta^2(a_2\cdots a_d)^{d-1} \right) $ for $2\leq i\leq d$.
To form the second sum in $R_{ij}^{(1)}$, again since the Christoffel symbols depend only on the variables $x_0,x_1$, all nonzero entries of the resulting matrix are contained in the rows indexed by $i=0,1$.  Therefore
\begin{equation}\left[\displaystyle\sum_{k=0}^d\Gamma_{kj,i}^k\right]=
\left(
\begin{array}{c}
 \overrightarrow{ \displaystyle\sum_{k=0}^d\Gamma_{kj,0}^k }\\
\vspace{-.2in} \\
\overrightarrow{ \displaystyle\sum_{k=0}^d\Gamma_{kj,1}^k} \\
\overrightarrow{0}\\
\vdots\\
\overrightarrow{0}
\end{array}
\right).\label{eq: AppxGammakjk,iRowVectors}
\end{equation}
For a fixed column $j$,  $ \displaystyle\sum_{k=0}^d\Gamma_{kj,0}^k=\Gamma_{0j,0}^0+\cdots+\Gamma_{dj,0}^d$ is the sum of the $i=0$ entry of the $j^{th}$ column of (\ref{eq: AppxGammaij0,0}), the $i=1$ entry of the $j^{th}$ column of (\ref{eq: AppxGammaij1,0}), the $i=2$ entry of the $j^{th}$ column of (\ref{eq: AppxGammaijk,0}) with $k=2$, and so on until finally the $i=d$ entry of the $j^{th}$ column of (\ref{eq: AppxGammaijk,0}) with $k=d$.  A similar methodology of summing entries of (\ref{eq: AppxGammaij0,1}), (\ref{eq: AppxGammaij1,1}) and (\ref{eq: AppxGammaijk,1}) will give the column entries of the row vector $\overrightarrow{ \displaystyle\sum_{k=0}^d\Gamma_{kj,1}^k}$.   Therefore (\ref{eq: AppxGammakjk,iRowVectors}) becomes

\begin{equation}\left[\displaystyle\sum_{k=0}^d\Gamma_{kj,i}^k\right]=
\left(
\begin{array}{cccc}
d\displaystyle\sum_{l=2}^d\dot{H}_l    &           &           &   \\ 
                                                                                           &   0     &            & \\
                                                                                           &           &\ddots &\\
                                                                                           &            &           & 0
 \end{array}
\right).\label{eq: AppxGammakjk,i}\end{equation}

\noindent Subtracting (\ref{eq: AppxGammaij0,0+Gammaij1,1}) from (\ref{eq: AppxGammakjk,i}), we obtain $R^{(1)}\stackrel{def.}{=}\left[R_{ij}^{(1)}\right]=$
\begin{equation} \left(\begin{array}{cccccc}
\frac{d}{2}\displaystyle\sum_{l=2}^d\dot{H}_l&&&&&\\
&\hspace{-.4in}\frac{-1}{2(a_2\cdots a_d)^{d-1}}\left(\displaystyle\sum_{l=2}^d \dot{H}_l - (d-1) \left(  \displaystyle\sum_{l=2}^d H_l\right)^2\right)  &&&&\\
&&& -J_2&&\\
&&& &\ddots&\\
&&&&&-J_d\\
\end{array}\right).\label{eq: AppxRij1}
\end{equation}

Next we define matrices 
\begin{equation}\upgamma^{(j)}=\left[\upgamma^{(j)}_{mk}\right]\stackrel{def.}{=}\Gamma^m_{jk},\label{eq: AppxDefngamma-Gamma}\end{equation}
indexed by $j\in\{0, 1, \dots, d\}$, which will help to compute the double sums  $R_{ij}^{(2)}$ and $R_{ij}^{(3)}$ in (\ref{eq: AppxRicci2-Gamma}) and (\ref{eq: AppxRicci3-Gamma}). 
For example, the matrix $\upgamma^{(0)}$ has rows that are equal to the $i=0$ rows of the matrices (\ref{eq: AppxGammaij0}), (\ref{eq: AppxGammaij1}) and (\ref{eq: AppxGammaijk}) of Christoffel symbols, the matrix $\upgamma^{(1)}$ has rows that are equal to the $i=1$ rows of the matrices (\ref{eq: AppxGammaij0}), (\ref{eq: AppxGammaij1}) and (\ref{eq: AppxGammaijk}), etc.  Therefore we obtain
\begin{equation}
\upgamma^{(0)}=
\left(
\begin{array}{ccccc}
\frac{d}{2}\displaystyle\sum_{l=2}^d H_l &                                                                           &                                &              &\\
                                                                      &   \frac{1}{2}\displaystyle\sum_{l=2}^d H_l   &                                &             &\\
                                                                      &                                                                           & \frac{(d-1)}{2}H_2 &              &\\
                                                                      &                                                                           &                               &\ddots   &\\
                                                                      &                                                                           &                               &              & \frac{(d-1)}{2}H_d\\
\end{array}
\right),
\label{eq: Appxlambdamk(0)}
\end{equation}

\begin{equation}
\upgamma^{(1)}=
\left(
\begin{array}{ccccc}
0                                                                      &   \frac{1}{2(a_2\cdots a_d)^{d-1}}\displaystyle\sum_{l=2}^d H_l  &                                &              &\\
 \frac{1}{2}\displaystyle\sum_{l=2}^d H_l &0                                                                                                          &                                &             &\\
                                                                      &                                                                                                           & \beta &              &\\
                                                                      &                                                                                                           &                               &\ddots   &\\
                                                                      &                                                                                                            &                              &              & \beta\\
\end{array}
\right),
\label{eq: Appxlambdamk(1)}
\end{equation}
and 
\begin{equation}
\upgamma^{(j)}=
\left(
\begin{array}{cccccc}
0                                                                                           &  0                                                                 & \hspace{-.5in} \cdots  &  A_j\mbox{\small \ in the $j^{th}$ column}  &\cdots&0 \\
0                                                                                           &  0                                                                 &  \hspace{-.5in}\cdots  &  B_j\mbox{\small \ in the $j^{th}$ column}  &\cdots&0\\
\vdots                                                                                   & \vdots                                                          &               &                                                                            &            &\\
\frac{(d-1)}{2}H_j \mbox{\small \ in the $j^{th}$ row}  &\beta\mbox{\small \ in the $j^{th}$ row}  &               &                                                                            &            &\\
\vdots                                                                                   & \vdots                                                          &                &                                                                           &            &\\
0                                                                                           &0                                                                    &                &                                                                           &           &
\end{array}
\right)\notag
\end{equation}
\begin{equation}\label{eq: Appxlambdamk(i)}
\end{equation}
for each $2\leq j\leq d$.

Since $\Gamma_{ij}^k=\Gamma_{ji}^k$, the double sum $R_{ij}^{(2)}$ in (\ref{eq: AppxRicci2-Gamma}) is 
\begin{equation}
R_{ij}^{(2)}= \displaystyle\sum_{m=0}^d\displaystyle\sum_{n=0}^d \Gamma^n_{im}\Gamma^m_{nj}= \displaystyle\sum_{m=0}^d\displaystyle\sum_{n=0}^d \Gamma^n_{im}\Gamma^m_{jn}= \displaystyle\sum_{m=0}^d\displaystyle\sum_{n=0}^d \gamma^{(i)}_{nm}\gamma^{(j)}_{mn}=Tr \left(\upgamma^{(i)}\upgamma^{(j)}\right).
  \end{equation}
\newpage

That is,  $R^{(2)}\stackrel{def.}{=}\left[R_{ij}^{(2)}\right]=$

\begin{equation}
 \left(
 \begin{array}{ccccc}
 \frac{(d^2+1)}{4}\left(\displaystyle\sum_{l=2}^d H_l\right)^2+\frac{(d-1)^2}{4}\displaystyle\sum_{l=2}^dH_l^2 
           & \frac{\beta(d-1)}{2}\displaystyle\sum_{l=2}^d H_l  & & &\\
 \frac{\beta(d-1)}{2}\displaystyle\sum_{l=2}^d H_l  
           & \hspace{-.5in} \frac{1}{2(a_2\cdots a_d)^{d-1}}\left(\displaystyle\sum_{l=2}^d H_l\right)^2+\beta^2(d-1)  & & &\\
  & &K_2 &           & \\
  & &   &\ddots & \\
  & &   &            &K_d
 \end{array}
 \right)\notag
\end{equation}
\begin{equation}\label{eq: AppxRij2}\end{equation}
for $K_i=(d-1)H_iA_i+2\beta B_i=\frac{a_i^{d-1}e^{2\beta x_1}}{(a_2\cdots a_d)^d}\left(\frac{(d-1)^2}{2}H_i^2-2\beta^2(a_2\cdots a_d)^{d-1}\right), 2\leq i\leq d$. 

Again since $\Gamma^k_{ij}=\Gamma^k_{ji}$, the double sum $R_{ij}^{(3)}$ in (\ref{eq: AppxRicci3-Gamma}) is
\begin{equation}R_{ij}^{(3)}=\displaystyle\sum_{m=0}^d\displaystyle\sum_{n=0}^d\Gamma_{ij}^m\Gamma_{nm}^n
=   \displaystyle\sum_{m=0}^d\displaystyle\sum_{n=0}^d\Gamma_{ji}^m\Gamma_{mn}^n
=   \displaystyle\sum_{m=0}^d\displaystyle\sum_{n=0}^d   \upgamma^{(j)}_{mi}  \upgamma^{(m)}_{nn}
=  \displaystyle\sum_{m=0}^d  \upgamma^{(j)}_{mi} \  Tr \ \upgamma^{(m)}  .
\end{equation}
That is, the $j^{th}$ column of the matrix $R^{(3)}\stackrel{def.}{=}
\left[R_{ij}^{(3)}\right]$ is equal to $\left(\upgamma^{(j)}\right)^T\cdot\upgamma$ or equivalently $\upgamma^{T}\cdot\upgamma^{(j)}$
where $\upgamma$ is the column vector
\begin{equation}\upgamma\stackrel{def.}{=}\left(
\begin{array}{c}
Tr \ \upgamma^{(0)}\\
Tr \ \upgamma^{(1)}\\
Tr \ \upgamma^{(2)}\\
\vdots\\
Tr \ \upgamma^{(d)}
\end{array}
\right)
=
\left(
\begin{array}{c}
d\displaystyle\sum_{l=2}^d H_l \\
\beta(d-1)\\
0\\
\vdots\\
0
\end{array}
\right)
\end{equation}
and ${}^T$ denotes transposition.  Therefore $R^{(3)}=$
\begin{equation}
\left(
\begin{array}{ccccc}
\frac{d^2}{2}\left(\displaystyle\sum_{l=2}^d H_l\right)^2  & \frac{\beta(d-1)}{2}\displaystyle\sum_{l=2}^d H_l                                                     &        &              &        \\ 
\frac{\beta(d-1)}{2}\displaystyle\sum_{l=2}^d H_l                   & \frac{d}{2(a_2\cdots a_d)^{d-1}}\left(\displaystyle\sum_{l=2}^d H_l\right)^2      &        &              &        \\
                                                                                                         &                                                                                                                                            &L_2 &              &        \\
                                                                                                         &                                                                                                                                            &         &\ddots  &        \\
                                                                                                         &                                                                                                                                            &          &            & L_d\\
\end{array}
\right)\label{eq: AppxRij3}
\end{equation}
for
\begin{eqnarray}L_i&=&dA_i\displaystyle\sum_{l=2}^dH_l+\beta(d-1) B_i\notag\\
&=&\frac{a_i^{d-1}e^{2\beta x_1}}{(a_2\cdots a_d)^d}\left(\frac{d(d-1)}{2}H_i\displaystyle\sum_{l=2}^d H_l - \beta^2(d-1)(a_2\cdots a_d)^{d-1}\right).\notag\\
\end{eqnarray}
By (\ref{eq: AppxRij1}), (\ref{eq: AppxRij2}) and (\ref{eq: AppxRij3}), the coefficients of the Ricci tensor (\ref{eq: AppxRicci-Gamma})  are
\begin{equation}\left[R_{ij}\right]=R^{(1)}+R^{(2)}-R^{(3)}=\notag\end{equation}
\begin{equation}
 \left(
 \begin{array}{ccccc}
 M&&&&\\
&N & & &\\
  & &P_2  &  & \\
   &&         &\ddots & \\
   &&         &            &P_d
 \end{array}
 \right)
\end{equation}
for
\begin{equation}M\stackrel{def.}{=} \frac{d}{2}\displaystyle\sum_{l=2}^d\dot{H}_l+\frac{(1-d^2)}{2}\displaystyle\sum_{2\leq l<k\leq d}H_lH_k-\frac{(d-1)}{2}\displaystyle\sum_{l=2}^d H_l^2,\end{equation}
\begin{equation} 
N\stackrel{def.}{=}\frac{-1}{2(a_2\cdots a_d)^{d-1}}\displaystyle\sum_{l=2}^d\dot{H}_l +\beta^2(d-1) \end{equation} 
and 
\begin{equation}
P_i\stackrel{def.}{=}-J_i+K_i-L_i=\frac{a_i^{d-1}e^{2\beta x_1}}{(a_2\cdots a_d)^d}\left( -\frac{(d-1)}{2}\dot{H}_i+\beta^2(d-1)(a_2\cdots a_d)^{d-1} \right),
\end{equation}
$2\leq i\leq d$, and where we have used that $\left(\displaystyle\sum_{l=2}^d H_l\right)^2=\displaystyle\sum_{l=2}^d H_l^2+2\displaystyle\sum_{2\leq l<k\leq d} H_lH_k$.

The scalar curvature is defined as $R\stackrel{def.}{=}\displaystyle\sum_{k=0}^d R_k^k\stackrel{def.}{=}\displaystyle\sum _{k=0}^d\displaystyle\sum_{l=0}^d g^{kl}R_{lk}$ therefore for our example,
\begin{eqnarray}
R&=&\frac{-1}{(a_2\cdots a_d)^d}M+\frac{1}{(a_2\cdots a_d)}N+\displaystyle\sum_{i=2}^d \frac{P_i}{a_i^{d-1}e^{2\beta x_1}}\notag\\
&=&\frac{1}{(a_2\cdots a_d)^d}\left(
\frac{(d-1)}{2}\displaystyle\sum_{l=2}^d H_l^2-d\displaystyle\sum_{l=2}^d\dot{H}_l
+\frac{(d^2-1)}{2}\displaystyle\sum_{2\leq l<k\leq d}H_lH_k\right.\notag\\
&&\qquad\qquad\qquad\qquad\qquad\qquad\qquad\qquad\qquad\qquad \left.
+\beta^2d(d-1)(a_2\cdots a_d)^{d-1}
 \right)\notag\\
\end{eqnarray}
Finally, the nonzero coefficients of the Einstein tensor $G_{ij}=R_{ij}-\frac{1}{2}Rg_{ij}$ are
\begin{equation} G_{00}=-\frac{(d-1)}{4}\left(
\displaystyle\sum_{l=2}^d H_l^2 +(d+1)\displaystyle\sum_{l<k}H_lH_k-2d\beta^2(a_2\cdots a_d)^{d-1}\right)
\end{equation}
$G_{11}=\frac{-(d-1)}{4(a_2\cdots a_d)^{d-1}}\cdot$
\begin{equation}
\left(\displaystyle\sum_{l=2}^d (H_l^2-2\dot{H}_l) +(d+1)\displaystyle\sum_{ l<k}H_lH_k +2(d-2)\beta^2(a_2\cdots a_d)^{d-1}\right)\notag
\end{equation}
and $G_{ii}=\frac{-(d-1)a_i^{d-1}e^{2\beta x_1}}{4(a_2\cdots a_d)^d}\cdot$
\begin{equation}
\hspace{-.2in}\left(2\dot{H}_i+\displaystyle\sum_{l=2}^d \left(H_l^2-\frac{2d}{(d-1)} \dot{H}_l\right)  +(d+1) \displaystyle\sum_{l<k}H_lH_k +2(d-2)\beta^2(a_2\cdots a_d)^{d-1} \right)\notag\end{equation}
for $2\leq i,l, k\leq d$.

\appendix 
\setcounter{chapter}{2}
\chapter{The non-positivity of $\sum_{l<k}c_lc_k$}

For the Bianchi I and Bianchi V models considered in Chapters \ref{ch: BI} and \ref{ch: BV}, we make use of the quantity 
\begin{equation}y\stackrel{def.}{=}\hspace{-.3cm}\displaystyle\sum_{1\leq l<k\leq d}c_lc_k\label{eq: y-clck}\end{equation}
for arbitrary constants $c_1, \dots, c_d\in\mathds{R}$ and $d\in\mathds{N}$ such that 
\begin{equation}c_1+\cdots +c_d=0.\label{eq: sumcl=0}\end{equation}
 We will show that $y$ is non-positive for all values of $c_1, \dots, c_d, d$.  

By (\ref{eq: sumcl=0}), we obtain
\begin{eqnarray}y&=&\hspace{-.4cm}\displaystyle\sum_{1\leq l < k \leq d-1}\hspace{-.4cm}c_lc_k+c_d\displaystyle\sum_{l=1}^{d-1}c_l\notag\\
&=&\hspace{-.4cm}\displaystyle\sum_{1\leq l < k \leq d-1}\hspace{-.4cm}c_lc_k-\left( \displaystyle\sum_{l=1}^{d-1}c_l\right)^2\notag\\
&=&\hspace{-.4cm}\displaystyle\sum_{1\leq l < k \leq d-1}\hspace{-.4cm}c_lc_k-\displaystyle\sum_{l=1}^{d-1}c_l^2-2\hspace{-.4cm}\displaystyle\sum_{1\leq l<k\leq d-1}\hspace{-.4cm}c_lc_k\notag\\
&=&-\displaystyle\sum_{l=1}^{d-1}c_l^2-\hspace{-.4cm}\displaystyle\sum_{1\leq l<k\leq d-1}\hspace{-.4cm}c_lc_k
\end{eqnarray}
so that by the following Lemma, $y\leq 0$.

\begin{appxlem}\label{lem: nonnegativesum}
The sum
\begin{equation}z_{M,N}\stackrel{def.}{=}\displaystyle\sum_{l=1}^M c_l^2+\frac{2}{N}\hspace{-.1cm}\displaystyle\sum_{1\leq l< k\leq M}\hspace{-.4cm}c_lc_k\end{equation}
is non-negative for all $M, N\in\mathds{N}$ and $c_1,\dots, c_M\in\mathds{R}$.
\end{appxlem}

\begin{proof}
We will prove this Lemma by induction on $M$.  We begin with $M=1$, for which $z_{1,N}=c_1^2\geq 0$ for all $c_1\in\mathds{R}$ and for all $N\in\mathds{N}$.  For $M=2$, we complete the square to obtain
\begin{eqnarray}z_{2,N}&=&c_1^2+c_2^2+\frac{2}{N}c_1c_2\notag\\
&=&c_2^2 + \left(\frac{2}{N}c_1\right)c_2+\left(c_1^2\right)\notag\\
&=&\left(c_2+\frac{c_1}{N}\right)^2-\frac{c_1^2}{N^2}+c_1^2\notag\\
&=&\left(c_2+\frac{c_1}{N}\right)^2+\left(\frac{N^2-1}{N^2}\right)c_1^2.
\end{eqnarray}
Therefore $z_{2,N}\geq 0$ for all $N\geq 1$.   By induction on $M$, we assume that 
\begin{equation}z_{M-1,N}=\displaystyle\sum_{l=1}^{M-1}c_l^2+\frac{2}{N}\hspace{-.1cm}\displaystyle\sum_{1\leq l<k\leq M-1}\hspace{-.4cm}c_lc_k\geq 0\mbox{ for all $N$},\label{eq: positivityinductiveassumption}\end{equation}
and we consider $z_{M,N}$.  If $N=1$, then
\begin{eqnarray}z_{M,1}&=&\displaystyle\sum_{l=1}^M c_l^2+2\hspace{-.4cm}\displaystyle\sum_{1\leq l<k\leq M}\hspace{-.4cm}c_lc_k\notag\\
&=&\left(\displaystyle\sum_{l=1}^M c_l\right)^2\geq 0.\end{eqnarray}
If $N=2$,
\begin{equation}z_{M,2}=\displaystyle\sum_{l=1}^M c_l^2+\hspace{-.4cm}\displaystyle\sum_{1\leq l<k\leq M}\hspace{-.4cm}c_lc_k.\label{eq: zM2}
\end{equation}
Now, consider the quantity
\begin{equation}\left(\frac{c_1}{2}+\frac{c_2}{2}+\cdots+\frac{c_{M-1}}{2}+c_M\right)^2 = \displaystyle\sum_{l=1}^{M-1} \frac{c_l^2}{4}+c_M^2+2c_M\displaystyle\sum_{l=1}^{M-1}\left(\frac{c_l}{2}\right)+2\hspace{-.4cm}\displaystyle\sum_{1\leq l<k\leq M-1}\hspace{-.4cm}\left(\frac{c_l}{2}\right)\left(\frac{c_k}{2}\right).\label{eq: positivityconsideredquantity}\end{equation}
One can double check that (\ref{eq: positivityconsideredquantity}) contains the correct number of terms by noting that 
\begin{eqnarray}&&(M-1)+1+2(M-1)+2\frac{(M-2)(M-1)}{2}\notag\\
&&\qquad\qquad=M+2M-2+M^2-3M+2\notag\\
&&\qquad\qquad =M^2 .\label{eq: nonnegativitynumberofterms}\end{eqnarray}
The last term in (\ref{eq: nonnegativitynumberofterms}) is due to the fact that the sum $\hspace{-.1cm}\displaystyle\sum_{1\leq l<k\leq M-1}\hspace{-.1cm}$ contains $(M-2)(M-1)/2$ terms.  One can see this by
\begin{eqnarray}\displaystyle\sum_{1\leq l<k\leq M-1}1&=& \displaystyle\sum_{k=2}^{M-1}\displaystyle\sum_{l=1}^{k-1}1\notag\\
&=& \displaystyle\sum_{k=2}^{M-1}(k-1)\notag\\
&=&\displaystyle\sum_{i=1}^{M-2}i\notag\\
&=&\frac{(M-2)(M-1)}{2}.
\end{eqnarray}
Simplifying (\ref{eq: positivityconsideredquantity}) we obtain
\begin{equation}\left(\frac{c_1}{2}+\frac{c_2}{2}+\cdots+\frac{c_{M-1}}{2}+c_M\right)^2 = \displaystyle\sum_{l=1}^{M-1} \frac{c_l^2}{4}+c_M^2+c_M\displaystyle\sum_{l=1}^{M-1}c_l+\frac{1}{2}\hspace{-.1cm}\displaystyle\sum_{1\leq l<k\leq M-1}\hspace{-.4cm}c_l c_k\label{eq: positivityconsideredquantitysimplified}\end{equation}
so that we may write (\ref{eq: zM2}) as
\begin{eqnarray}z_{M,2}&=&\left(\frac{c_1}{2}+\frac{c_2}{2}+\cdots+\frac{c_{M-1}}{2}+c_M\right)^2+\frac{3}{4}\displaystyle\sum_{l=1}^{M-1}c_l^2+\frac{1}{2}\hspace{-0cm}\displaystyle\sum_{1\leq l<k\leq M-1}\hspace{-.4cm}c_lc_k\notag\\
&=&\left(\frac{c_1}{2}+\frac{c_2}{2}+\cdots+\frac{c_{M-1}}{2}+c_M\right)^2+\frac{3}{4}\left(\displaystyle\sum_{l=1}^{M-1}c_l^2+\frac{2}{3}\hspace{-0cm}\displaystyle\sum_{1\leq l<k\leq M-1}\hspace{-.4cm}c_lc_k\right).\notag\\
\label{eq: zM2reducedtoinduction}\end{eqnarray}
By the inductive assumption on $M$ in (\ref{eq: positivityinductiveassumption}), the second quantity in parenthesis on the right-hand side of (\ref{eq: zM2reducedtoinduction}) is non-negative.  Therefore $z_{M,2}$ is non-negative.  Finally, to rewrite
\begin{equation}z_{M,N}=\displaystyle\sum_{l=1}^Mc_l^2+\frac{2}{N}\hspace{-.1cm}\displaystyle\sum_{1\leq l<k\leq M}c_lc_k\label{eq: positivityzMN}\end{equation}
for general $N$ we consider the quantity 
\begin{equation}\left(\frac{c_1}{N}+\frac{c_2}{N}+\cdots+\frac{c_{M-1}}{N}+c_M\right)^2 = \displaystyle\sum_{l=1}^{M-1} \frac{c_l^2}{N^2}+c_M^2+2c_M\displaystyle\sum_{l=1}^{M-1}\left(\frac{c_l}{N}\right)+2\hspace{-.4cm}\displaystyle\sum_{1\leq l<k\leq M-1}\hspace{-.4cm}\left(\frac{c_l}{N}\right)\left(\frac{c_k}{N}\right).\label{eq: positivityconsideredquantityN}\end{equation}
As in (\ref{eq: nonnegativitynumberofterms}), one can double check that the right-hand side of (\ref{eq: positivityconsideredquantityN}) contains $M^2$ terms.
Simplifying (\ref{eq: positivityconsideredquantityN}) we obtain
\begin{equation}\left(\frac{c_1}{N}+\frac{c_2}{N}+\cdots+\frac{c_{M-1}}{N}+c_M\right)^2 = \frac{1}{N^2}\displaystyle\sum_{l=1}^{M-1} c_l^2+c_M^2+\frac{2}{N}c_M\displaystyle\sum_{l=1}^{M-1}c_l +\frac{2}{N^2}\hspace{-0cm}\displaystyle\sum_{1\leq l<k\leq M-1}\hspace{-.4cm} c_l c_k\label{eq: positivityconsideredquantityNsimplified}\end{equation}
so that we may write (\ref{eq: positivityzMN}) as $z_{M,N} = $
\begin{eqnarray}&&\hspace{-1cm}\left(\frac{c_1}{N}+\frac{c_2}{N}+\cdots+\frac{c_{M-1}}{N}+c_M\right)^2+\frac{(N^2-1)}{N^2}\displaystyle\sum_{l=1}^{M-1}c_l^2+2\left(\frac{1}{N}-\frac{1}{N^2}\right)\hspace{-0cm}\displaystyle\sum_{1\leq l<k\leq M-1}\hspace{-.4cm}c_lc_k\notag\\
&=&\left(\frac{c_1}{N}+\frac{c_2}{N}+\cdots+\frac{c_{M-1}}{N}+c_M\right)^2+\frac{(N^2-1)}{N^2}\left(\displaystyle\sum_{l=1}^{M-1}c_l^2+\frac{2}{(N+1)}\hspace{-0cm}\displaystyle\sum_{1\leq l<k\leq M-1}\hspace{-.4cm}c_lc_k\right).\notag\\
\label{eq: zMNreducedtoinduction}\end{eqnarray}
By the inductive assumption on $M$ in (\ref{eq: positivityinductiveassumption}), the second quantity in parenthesis on the right-hand side of (\ref{eq: zMNreducedtoinduction}) is non-negative.  This proves the Lemma.
\end{proof}

\appendix 
\setcounter{chapter}{3}

\chapter{Exact solutions to EMP equations}  

Below in Table  \ref{tb: exactEMP}, we list some exact solutions of the generalized EMP equation with one non-linear term
\begin{equation}Y''(\tau)+Q(\tau)Y(\tau)=\frac{\lambda_1}{Y(\tau)^{B_1}}\label{eq: N=0EMP}\end{equation}
for constants $\lambda_1, B_1$.  In Table  \ref{tb: exactEMP}, $a_0,b_0,c_0, d_0, Q_0\in\mathds{R}$ are constants.

In special cases when the EMP reduces to a classical EMP equation 
\begin{equation}Y''(\tau)+Q(\tau)Y(\tau)=\frac{\lambda}{Y^3},\label{eq: FRLWSpecCasesClassicalEMP}\end{equation}
for $\lambda\in\mathds{R}$, exact solutions to the classical EMP can be obtained by a superposition principle:  if $Y_1(\tau)$ and $Y_2(\tau)$ are two linearly independent solutions of the homogeneous equation $Y''(\tau)+Q(\tau)Y(\tau)=0$, then their Wronskian $W(Y_1,Y_2)\stackrel{def.}{=}\left|\begin{array}{cc}Y_1&Y_1'\\Y_2&Y_2'\end{array}\right|$ is constant and $Y(\tau)=\left(A Y_1^2(\tau)+B Y_2^2(\tau)+2C Y_1(\tau)Y_2(\tau)\right)^{1/2}$ is a solution to (\ref{eq: FRLWSpecCasesClassicalEMP}) where $A,B,C\in\mathds{R}$ are constants that satisfy $AB-C^2=\lambda/W(Y_1,Y_2)^2$.  Some of the entries in Table \ref{tb: exactEMP} are obtained from this superposition principle.

We also record solution
\begin{equation}Y(\tau)=\left(\frac{4}{\theta^2}\tau^2-1\right)^{1/4}\label{eq: exactEMPtwotermssoln}\end{equation}
to the generalized EMP with two non-linear terms
\begin{equation}Y''(\tau)+Q(\tau)Y(\tau)=\frac{\lambda_1=-1/\theta^2}{Y(\tau)^{B_1=3}}+\frac{\lambda_2=-3/\theta^2}{Y(\tau)^{B_2=7}}\label{eq: exactEMPtwoterms}\end{equation}
for $Q=0$ and constant $\theta>0$.

   \begin{table}[ht]
\centering
\caption{ Exact Solutions of EMP\label{tb: exactEMP}}
\vspace{.2in}
\begin{tabular}{c | r | l | l | l}
&$Y(\tau)$ & $Q(\tau)$ &$ \lambda_1$&$B_1$\\[4pt]
\hline
$1$&$cos(\sqrt{Q_0}\tau)$&$Q_0>0$&$0$&n/a\\[8pt]
$2$&$sin(\sqrt{Q_0}\tau)$&$Q_0>0$&$0$&n/a\\[8pt]
$3$&$\left(a_0 cos^2(\sqrt{Q_0}\tau)+b_0sin^2(\sqrt{Q_0}\tau)\right.\quad$&$$&$$&$$\\
$$&$\quad\quad\left.+2c_0 cos(\sqrt{Q_0}\tau)sin(\sqrt{Q_0}\tau)\right)^{1/2}$&$Q_0>0$&$Q_0(a_0b_0-c_0^2)$&3\\[8pt]
$4$&$a_0\tau^{d_0}$&$\frac{d_0(1-d_0)}{\tau^2}$&$0$&n/a\\[8pt]
$5$&$\left(a_0\tau^{2d_0}+b_0\tau^{2(1-d_0)}+2c_0\tau\right)^{1/2}$&$\frac{d_0(1-d_0)}{\tau^2}$&$(1-2d_0)^2(a_0b_0-c_0^2)$&3\\[8pt]
$6$&$\tau$&$\frac{\lambda_1}{\tau^{B_1+1}}$&$\lambda_1$&$B_1$\\[8pt]
$7$&$\left(a_0^2\tau^2\pm b_0^2\right)^{1/2}$&$\frac{\lambda_1\mp a_0^2b_0^2}{\left(a_0^2\tau^2\pm b_0^2\right)^2}$&$\lambda_1$&$3$\\[6pt]
\hline
\end{tabular}
\end{table}  

In order to map an exact solution $Y(\tau)$ of a generalized or classical EMP equation to an exact solution of Einstein's equations via the correspondences in this thesis, one must first solve for $\tau(t)$ in the differential equation 
\begin{equation}\dot\tau(t)=Y(\tau(t))^{r_0}\label{eq: dottau-Ytau^s0}\end{equation}
 for some constant $r_0\neq 0$ that is dependent on the cosmological model.  For the EMP reformulations of Bianchi I and Bianchi V one must then solve for $\sigma(t)$ in 
 \begin{equation}\dot\sigma(t)=\frac{1}{\dot\tau(t)^{s_0}}\label{eq: dotsigma=1/dottau^s0}\end{equation}
 for some constant $s_0\neq 0$.   Also, in each cosmological model, in order to find the scale factor $\phi(t)$, one must first integrate some constant multiple of the square root of the function $Q(\tau)$ (and then compose the result with $\tau(t)$).   That is, to find $\phi(t)$ one must first compute
 \begin{equation}\varphi(\tau) = \displaystyle\int \sqrt{\alpha_0 Q(\tau) } d\tau\label{eq: varphigeneral}\end{equation}
 for some constant $\alpha_0$.   Therefore for each solution in Table \ref{tb: exactEMP}, we now record some solutions to (\ref{eq: dottau-Ytau^s0}), (\ref{eq: dotsigma=1/dottau^s0}) and (\ref{eq: varphigeneral}).

For $Y(\tau)$ from line 1 in Table \ref{tb: exactEMP}, equation (\ref{eq: dottau-Ytau^s0}) for $r_0=1$ is
\begin{equation}\dot\tau(t)=\cos\left(\sqrt{Q_0}\tau(t)\right)\label{eq: taueqnforY=cosr=1}\end{equation}
which has solution
\begin{equation}\tau(t)=\frac{2}{\sqrt{Q_0}}Arctan\left(tanh\left(\frac{\sqrt{Q_0}}{2}(t-t_0)\right)\right)\label{eq: tauforY=cosr=1}\end{equation}
for $t_0\in\mathds{R}$.  One can check this by noting that
\begin{equation}\dot\tau(t)=Y(\tau(t))^1=\frac{sech^2\left(\frac{\sqrt{Q_0}}{2}(t-t_0)\right)}{1+tanh^2\left(\frac{\sqrt{Q_0}}{2}(t-t_0)\right)}=sech\left(\sqrt{Q_0}(t-t_0)\right).\label{eq: dottauforY=cosr=1}\end{equation}
Also, we record that
\begin{eqnarray}Y'(\tau(t))&=& - \sqrt{Q_0}\sin\left(2Arctan\left(tanh\left(\frac{\sqrt{Q_0}}{2}(t-t_0) \right)\right)\right)\notag\\
&=&\frac{-2\sqrt{Q_0}tanh\left(\frac{\sqrt{Q_0}}{2}(t-t_0) \right)}{1+tanh^2\left(\frac{\sqrt{Q_0}}{2}(t-t_0) \right)}\notag\\
&=&-\sqrt{Q_0}tanh\left(\sqrt{Q_0}(t-t_0) \right)\label{eq: YprimetauforY=cos}
\end{eqnarray}
 since $sin(2Arctan(x))=2cos(Arctan(x))sin(Arctan(x))=\frac{2x}{1+x^2}$.

For $Y(\tau)$ from line 2 in Table \ref{tb: exactEMP}, equation (\ref{eq: dottau-Ytau^s0}) for $r_0=1$ is
\begin{equation}\dot\tau(t)=\sin\left(\sqrt{Q_0}\tau(t)\right)\label{eq: taueqnforY=sinr=1}\end{equation}
which has solution
\begin{equation}\tau(t)=\frac{2}{\sqrt{Q_0}}Arctan\left(e^{\sqrt{Q_0}(t-t_0)}\right)\label{eq: tauforY=sinr=1}\end{equation}
for $t_0\in\mathds{R}$.  One can check this by noting that
\begin{equation}\dot\tau(t)=Y(\tau(t))^1=\frac{2e^{\sqrt{Q_0}(t-t_0)}}{1+e^{2\sqrt{Q_0}(t-t_0)}}=sech\left(\sqrt{Q_0}(t-t_0)\right).\label{eq: dottauforY=sinr=1}\end{equation}
Also, we record that
\begin{eqnarray}
Y'(\tau(t))&=&\sqrt{Q_0}\cos\left( 2Arctan\left(e^{\sqrt{Q_0}(t-t_0)}\right)  \right)\notag\\
&=&\frac{   \sqrt{Q_0}\left(1-  e^{2\sqrt{Q_0}(t-t_0)}     \right)}{\left(1+   e^{2\sqrt{Q_0}(t-t_0)}  \right)}   \notag\\
&=&\frac{   \sqrt{Q_0}\left(e^{-\sqrt{Q_0}(t-t_0)}-  e^{\sqrt{Q_0}(t-t_0)}     \right)}{\left(e^{-\sqrt{Q_0}(t-t_0)}+   e^{\sqrt{Q_0}(t-t_0)}  \right)} \notag\\
&=&- \sqrt{Q_0}tanh\left( \sqrt{Q_0}(t-t_0) \right)\label{eq: YprimetauforY=sin}
\end{eqnarray}
since $cos(2Arctan(x))=cos^2(Arctan(x))-sin^2(Arctan(x))=\frac{1-x^2}{1+x^2}$.

For $Q(\tau)=Q_0>0$ from lines 1-3, (\ref{eq: varphigeneral}) becomes
\begin{equation}\varphi(\tau)=\displaystyle\int \sqrt{\alpha_0 Q_0} d\tau = \sqrt{\alpha_0Q_0} \tau +\beta_0\label{eq: varphiforQ>0}\end{equation}
for integration constant $\beta_0\in\mathds{R}$. 

Composing (\ref{eq: varphiforQ>0}) with $\tau(t)$ in (\ref{eq: tauforY=cosr=1}),
\begin{equation}\phi(t)\stackrel{def.}{=}\varphi(\tau(t))= 
2\sqrt{\alpha_0}Arctan\left(tanh\left(\frac{\sqrt{Q_0}}{2}(t-t_0)\right)
\right) +\beta_0.\label{eq: phiforY=cosr=1}\end{equation}

Composing (\ref{eq: varphiforQ>0}) with $\tau(t)$ in (\ref{eq: tauforY=sinr=1}),
\begin{equation}\phi(t)\stackrel{def.}{=} \varphi(\tau(t)) = 2\sqrt{\alpha_0}Arctan\left(e^{\sqrt{Q_0}(t-t_0)}\right) + \beta_0.\label{eq: phiforY=sinr=1}\end{equation}

For $Y(\tau)$ from line 4 in Table \ref{tb: exactEMP}, equation (\ref{eq: dottau-Ytau^s0}) is
\begin{equation}\dot\tau(t)=\left(a_0\tau(t)^{d_0}\right)^{r_0}\label{eq: taueqnforY=powerr=general}\end{equation}
which has solution
\begin{equation}\tau(t)=\left((1-r_0d_0)a_0^{r_0}(t-t_0)\right)^{\frac{1}{1-r_0d_0}}\label{eq: tauforY=powerr=general}\end{equation}
for $t_0\in\mathds{R}$ and $r_0\neq 1/d_0$.  One can check this by noting that 
\begin{equation}\dot\tau(t)=Y(\tau(t))^{r_0}=a_0^{r_0}\left((1-r_0d_0)a_0^{r_0}(t-t_0)\right)^{\frac{r_0d_0}{1-r_0d_0}}.\label{eq: dottauforY=powerr=general}\end{equation}
Also, we record that
\begin{eqnarray}
Y'(\tau(t))&=&d_0a_0\tau(t)^{d_0-1}\notag\\
&=&d_0a_0\left((1-r_0d_0)a_0^{r_0}(t-t_0)\right)^{\frac{d_0-1}{1-r_0d_0}}.\label{eq: YprimetauforY=powerr=general}
\end{eqnarray}

For $Y(\tau)$ from line 4 in Table \ref{tb: exactEMP}, equation (\ref{eq: dottau-Ytau^s0}) for $r_0=1/d_0$ is
\begin{equation}\dot\tau(t)=\tau(t)\label{eq: taueqnforY=powerr0d0=1}\end{equation}
which has solution
\begin{equation}\tau(t)=a_0e^{t-t_0}\label{eq: tauforY=powerr0d0=1}\end{equation}
for $a_0, t_0\in\mathds{R}$.  One can check this by noting that
\begin{equation}\dot\tau(t)=Y(\tau(t))^{r_0=1/d_0}=a_0e^{t-t_0}.\label{eq: dottauforY=powerr0d0=1}\end{equation}
Also, we record that
\begin{eqnarray}
Y'(\tau(t))&=&d_0\tau(t)^{d_0-1}\notag\\
&=&d_0  a_0^{d_0-1}e^{(d_0-1)(t-t_0)}
.\end{eqnarray}

For  $Y(\tau)$ from line 5 in Table \ref{tb: exactEMP} with $b_0=d_0=0$, equation (\ref{eq: dottau-Ytau^s0}) for $r_0=1$ is
\begin{equation}\dot\tau(t)=(a_0+2c_0\tau(t))^{1/2}\label{eq: taueqnforY=superpowerb=0r=1}\end{equation}
which has solution
\begin{equation}\tau(t)=\frac{1}{2}\left(c_0(t-t_0)^2-\frac{a_0}{c_0}\right)\label{eq: tauforY=superpowerb=0r=1}\end{equation}
for $c_0>0$ and $t_0\in\mathds{R}$.  One can check this by noting that
\begin{equation}\dot\tau(t)=Y(\tau(t))=c_0(t-t_0).\label{eq: dottauforY=superpowerb=0r=1}\end{equation}
Also, we record that
\begin{eqnarray}
Y'(\tau(t))&=&\frac{c_0}{(a_0+2c_0\tau(t))^{1/2}}\notag\\
&=&\frac{1}{t-t_0}\label{eq: YprimetauforY=superpowerb=d=0r=1}
\end{eqnarray}
for $t>t_0$ and $c_0>0$.
Equation (\ref{eq: dotsigma=1/dottau^s0}) for $s_0=1$ then becomes
\begin{equation}\dot\sigma(t)=\frac{1}{c_0(t-t_0)}.\label{eq: sigmaeqnforY=superpowerb=0s=1}\end{equation}
Integrating, we obtain
\begin{equation}\sigma(t)=\frac{1}{c_0}\ln(c_0(t-t_0)).\label{eq: sigmaforY=superpowerb=0s=1}\end{equation}

For $Y(\tau)$ from line 5 in Table \ref{tb: exactEMP} with $b_0>0$ and $d_0=0$, equation (\ref{eq: dottau-Ytau^s0}) for $r_0=1$ is
\begin{equation}\dot\tau(t)=(a_0+b_0\tau(t)^2+2c_0\tau(t))^{1/2}\label{eq: taueqnforY=superpowerr=1}\end{equation}
which has solution
\begin{equation}\tau(t)=\frac{1}{4b_0^{3/2}}\left(b_0e^{\sqrt{b_0}(t-t_0)}-4\lambda e^{-\sqrt{b_0}(t-t_0)}-4\sqrt{b_0}c_0\right)\label{eq: tauforY=superpowerr=1}\end{equation}
for $t_0\in\mathds{R}$ and $\lambda\stackrel{def.}{=}a_0b_0-c_0^2$.  One can check this by noting that 
\begin{equation}\dot\tau(t)=Y(\tau(t))=\frac{1}{4}e^{\sqrt{b_0}(t-t_0)}+\frac{\lambda}{b_0}e^{-\sqrt{b_0}(t-t_0)}.\label{eq: dottauforY=superpowerr=1}\end{equation}
Also, we record that
\begin{eqnarray}
Y'(\tau(t))&=&\frac{b_0\tau(t)+c_0}{Y(\tau(t))}\notag\\
&=&\sqrt{b_0}\frac{\left(b_0e^{\sqrt{b_0}(t-t_0)}-4\lambda e^{-\sqrt{b_0}(t-t_0)}\right)}{\left(b_0e^{\sqrt{b_0}(t-t_0)}+4\lambda e^{-\sqrt{b_0}(t-t_0)}\right)}.\label{eq: YprimetauforY=superpowerd=0b>0r=1}
\end{eqnarray}

For $s_0=1$, (\ref{eq: dotsigma=1/dottau^s0})  shows that
\begin{eqnarray}
\sigma(t)&=&\left(\frac{b_0}{\lambda}\right) \displaystyle\int \frac{e^{\sqrt{b_0}(t-t_0)} }{\left(1 - \frac{b_0}{ -4\lambda}e^{2\sqrt{b_0}(t-t_0)}\right)}dt \notag\\
&=&\left. \left(\frac{\sqrt{b_0}}{\lambda}\right) \displaystyle\int \frac{1 }{\left(1 - \frac{b_0}{ -4\lambda}u^2\right)}du \right|_{u=e^{\sqrt{b_0}(t-t_0)}}\notag\\
&=&\left. \sqrt{ \frac{-4\lambda}{ b_0}}  \left(\frac{\sqrt{b_0}}{\lambda}\right) Arccoth\left( \sqrt{  \frac{b_0}{ -4\lambda}}u\right)  \right|_{u=e^{\sqrt{b_0}(t-t_0)}}\notag\\
&=& \frac{-2}{\sqrt{-\lambda}} Arccoth\left( \sqrt{  \frac{b_0}{ -4\lambda}}e^{\sqrt{b_0}(t-t_0)}\right)
\label{eq: sigmaforY=superpowerr=1s=1}
\end{eqnarray}
for $\lambda<0$.

For $Y(\tau)$ from line 5 in Table \ref{tb: exactEMP} with $a_0=b_0=0$ and $c_0=1/2$, the equation $\dot\tau(t)=\theta Y(\tau(t))$ is
\begin{equation}\dot\tau(t)=\theta\sqrt{ \tau(t) }\label{eq: taueqnforY=sqrttau}\end{equation}
which has solution
\begin{equation}\tau(t)=\frac{\theta^2}{4}(t-t_0)^2\label{eq: tauforY=sqrttau}\end{equation}
for $t>t_0\in\mathds{R}$ and $\theta>0$.  One can check this by calculating
\begin{equation}\dot\tau(t)=\theta Y(\tau(t))=\frac{\theta^2}{2}(t-t_0).\label{eq: dottauforY=sqrttau}\end{equation}
Also, we record that 
\begin{equation}Y'(\tau(t))=\frac{1}{2\sqrt{\tau(t)}}=\frac{1}{\theta (t-t_0)}.\label{eq: YprimetauforY=sqrttau}\end{equation}

For $Q(\tau)$ from lines 4-5 in Table \ref{tb: exactEMP}, (\ref{eq: varphigeneral}) becomes
\begin{equation}\varphi(\tau)=\displaystyle\int \frac{\sqrt{ \alpha_0 d_0(1-d_0)}}{\tau} d\tau = \sqrt{ \alpha_0 d_0(1-d_0)}\ln(\tau) +\beta_0\label{eq: varphiforQ=1/tau^2}\end{equation}
for $\tau>0$, $0\leq d_0\leq 1$ and integration constant $\beta_0\in\mathds{R}$.

Composing (\ref{eq: varphiforQ=1/tau^2}) with $\tau(t)$ in (\ref{eq: tauforY=powerr=general}),
\begin{equation}\phi(t)\stackrel{def.}{=}\varphi(\tau(t)) =\frac{ \sqrt{ \alpha_0 d_0(1-d_0)} }{ (1-r_0d_0) }\ln\left[    \left((1-r_0d_0)(t-t_0)\right) \right] + \beta_0\label{eq: phiforY=powerr=general}\end{equation}
for $r_0<d_0$ and $t>t_0$.

Composing (\ref{eq: varphiforQ=1/tau^2}) with $\tau(t)$ in (\ref{eq: tauforY=sqrttau}),
\begin{equation}\phi(t)\stackrel{def.}{=}\varphi(\tau(t))=2\sqrt{ \alpha_0 d_0(1-d_0)}\ln(t-t_0) + \beta_0
\label{eq: phiforY=sqrttau}\end{equation}
for $t>t_0$.

For $Y(\tau)$ from line 6 in Table \ref{tb: exactEMP}, we solve for $\tau(t)$ in the equation
\begin{equation}\dot\tau(t)=\theta Y(\tau(t))=\theta \tau(t)\label{eq: taueqnforY=tauwithB}\end{equation}
which has solution
\begin{equation}\tau(t)=a_0e^{\theta (t-t_0)}\label{eq: tauforY=tauwithB}\end{equation}
for $a_0,t_0\in\mathds{R}$ and $\theta>0$.  For $Q(\tau)$ from line 6 in Table \ref{tb: exactEMP}, 
(\ref{eq: varphigeneral}) becomes
\begin{equation}\varphi(\tau)=\displaystyle\int \frac{\sqrt{\alpha_0 \lambda_1}}{\tau^{(B_1+1)/2}} d\tau = \frac{2\sqrt{\alpha_0\lambda_1}}{(1-B_1)\tau^{(B_1-1)/2}}+\beta_0\label{eq: varphiforQ=L/tauB} \end{equation}
for $\beta_0\in\mathds{R}$.  Composing (\ref{eq: varphiforQ=L/tauB}) with $\tau(t)$ in (\ref{eq: tauforY=tauwithB}),
\begin{equation}\phi(t)\stackrel{def.}{=}\varphi(\tau(t))= \frac{2\sqrt{\alpha_0\lambda_1}}{(1-B_1)}a_0^{(1-B_1)/2}e^{\theta(1-B_1)(t-t_0)/2}+\beta_0.\label{eq: phiforY=tauwithB}\end{equation}

For $Y(\tau)$ from line 7 in Table \ref{tb: exactEMP} with the plus sign and $a_0,b_0>0$, equation (\ref{eq: dottau-Ytau^s0}) for $r_0=1$ is
\begin{equation}\dot\tau(t)=\left(a_0^2\tau(t)^2+b_0^2\right)^{1/2}\label{eq: taueqnforY=sqrttau^2+b^2}\end{equation}
which has solution
\begin{equation}\tau(t)=\frac{b_0}{a_0}\sinh(a_0 (t-t_0))\label{eq: tauforY=sqrttau^2+b^2}\end{equation}
for $t_0\in\mathds{R}$.  One can check this by computing
\begin{equation}\dot\tau(t)=Y(\tau(t))=b_0 \cosh(a_0 (t-t_0)).\label{eq: dottauforY=sqrttau^2+b^2}\end{equation}
Also, we record that
\begin{equation}Y'(\tau(t))= a_0 \tanh(a_0(t-t_0)).\label{eq: YprimetauforY=sqrttau^2+b^2}\end{equation}

For $Q(\tau)$ from line 7 in Table \ref{tb: exactEMP} with the plus sign, (\ref{eq: varphigeneral}) becomes
\begin{equation}\varphi(\tau)=\displaystyle\sqrt{\alpha_0(\lambda_1- a_0^2b_0^2)}\int\frac{d\tau}{a_0^2\tau^2 + b_0^2}=\frac{\sqrt{\alpha_0(\lambda_1- a_0^2b_0^2)}}{a_0b_0}Arctan\left(\frac{a_0}{b_0}\tau\right)+\beta_0  \label{eq: varphiforQ=const/quadraticplus}\end{equation}
for $\tau > - b_0/a_0$ and assuming $a_0,b_0, (\lambda_1 - a_0^2b_0^2)>0$ and $\beta_0\in\mathds{R}$.  Composing (\ref{eq:  varphiforQ=const/quadraticplus}) with $\tau(t)$ in (\ref{eq: tauforY=sqrttau^2+b^2}),
\begin{equation}\phi(t)\stackrel{def.}{=}\varphi(\tau(t))=\frac{\sqrt{\alpha_0(\lambda_1- a_0^2b_0^2)}}{a_0b_0}Arctan\left(\sinh(a_0 (t-t_0))\right)+\beta_0. \label{eq: phiforY=sqrttau^2+b^2r0=1}\end{equation}

For $Y(\tau)$ from line 7 in Table \ref{tb: exactEMP} with the minus sign and $a_0,b_0>0$, equation (\ref{eq: dottau-Ytau^s0}) for $r_0=1$ is
\begin{equation}\dot\tau(t)=\left(a_0^2\tau(t)^2-b_0^2\right)^{1/2}\label{eq: taueqnforY=sqrttau^2-b^2}\end{equation}
which has solution
\begin{equation}\tau(t)=\frac{b_0}{a_0}\cosh(a_0 (t-t_0))\label{eq: tauforY=sqrttau^2-b^2}\end{equation}
for $t_0\in\mathds{R}$.  One can check this by computing
\begin{equation}\dot\tau(t)=Y(\tau(t))=b_0 \sinh(a_0 (t-t_0)).\label{eq: dottauforY=sqrttau^2-b^2}\end{equation}
Also, we record that
\begin{equation}Y'(\tau(t))= a_0 \coth(a_0(t-t_0))\label{eq: YprimetauforY=sqrttau^2+b^2}\end{equation}
for $t>t_0$.

For $Q(\tau)$ from line 7 in Table \ref{tb: exactEMP} with the minus sign, (\ref{eq: varphigeneral}) becomes
\begin{equation}\varphi(\tau)=\displaystyle\sqrt{\alpha_0(\lambda_1+ a_0^2b_0^2)}\int\frac{d\tau}{a_0^2\tau^2 - b_0^2}=-\frac{\sqrt{\alpha_0(\lambda_1+ a_0^2b_0^2)}}{a_0b_0}Arctanh\left(\frac{a_0}{b_0}\tau\right)+\beta_0 \label{eq: varphiforQ=const/quadraticminus}\end{equation}
for $\tau >  b_0/a_0$ and assuming $a_0,b_0, (\lambda_1 + a_0^2b_0^2)>0$ and $\beta_0\in\mathds{R}$.  Composing (\ref{eq:  varphiforQ=const/quadraticminus}) with $\tau(t)$ in (\ref{eq: tauforY=sqrttau^2-b^2}),
\begin{equation}\phi(t)\stackrel{def.}{=}\varphi(\tau(t))=-\frac{\sqrt{\alpha_0(\lambda_1+ a_0^2b_0^2)}}{a_0b_0}Arctanh\left(\cosh(a_0 (t-t_0))\right)+\beta_0. \label{eq: phiforY=sqrttau^2-b^2r0=1}\end{equation}

For $Y(\tau)$ in (\ref{eq: exactEMPtwotermssoln}) with $\theta=1$, equation (\ref{eq: dottau-Ytau^s0}) for $r_0=4$ is 
\begin{equation}\dot\tau(t)=4\tau(t)^2-1\label{eq: taueqnforysolnwithtwoterms}\end{equation}
which has solution
\begin{equation}\tau(t)=\frac{-1}{2}coth(2(t-t_0))\label{eq: tauforysolnswithtwoterms}\end{equation}
for $t_0\in\mathds{R}$.  One can check this by noting that 
\begin{equation}\dot\tau(t)=Y(\tau(t))^4= csch^2(2(t-t_0)) \label{eq: Ytauforysolntwoterms}\end{equation}
by the identity $coth^2(x)-1=csch^2(x)$.    Also, we record that
\begin{eqnarray}
Y'(\tau(t))&=&\frac{2\tau(t)}{(4\tau(t)^2-1)^{3/4}}\notag\\
&=&\frac{ - coth(2(t-t_0))}{csch^{3/2}(2(t-t_0))}\notag\\
&=& - cosh(2(t-t_0)) \sqrt{sinh(2(t-t_0))}\label{eq: Yprimeforysolntwoterms}
\end{eqnarray}
for $t>t_0$.

For $Y(\tau)$ in (\ref{eq: exactEMPtwotermssoln}), we solve for $\tau(t)$ in the equation 
\begin{equation}\dot\tau(t)=\theta Y(\tau(t))^{2}=\theta \left(4\tau(t)^2-1\right)^{1/2}\label{eq: taueqnforysolnwithtwotermsr=2}\end{equation}
which has solution
\begin{equation}\tau(t)=\frac{\theta}{2}cosh(2(t-t_0))\label{eq: tauforysolnswithtwotermsr=2}\end{equation}
for $t_0\in\mathds{R}$.  One can check this by noting that 
\begin{equation}\dot\tau(t)=\theta Y(\tau(t))^2= \theta sinh(2(t-t_0)) \label{eq: Ytauforysolntwotermsr=2}\end{equation}
by the identity $cosh^2(x)-1=sinh^2(x)$.    Also, we record that
\begin{eqnarray}
Y'(\tau(t))&=&\frac{2\tau(t)}{\theta^2 \left(\frac{4}{\theta^2}\tau(t)^2-1\right)^{3/4}}\notag\\
&=&\frac{ cosh(2(t-t_0))}{\theta sinh^{3/2}(2(t-t_0))}\notag\\
&=& \frac{ coth(2(t-t_0))}{\theta \sqrt{sinh(2(t-t_0))}}\label{eq: Yprimeforysolntwotermsr=2}
\end{eqnarray}
for $t>t_0$.

\appendix 
\setcounter{chapter}{4}

\chapter{Exact solutions to NLS equations}  

In the table below, we list some exact solutions of the NLS
\begin{equation}u''(\sigma)+[E-P(\sigma)]u(\sigma)=\frac{F_1}{u(\sigma)^{C_1}}\label{eq: N=2NLS}\end{equation}
for constants $E, F_1, C_1$.  Note that in Table \ref{tb: exactNLS}, $a_0,b_0,c_0,d_0\in\mathds{R}$ are constants, and any one solution $u(\sigma)$ may solve (\ref{eq: N=2NLS}) for a few distinct sets of $P(\sigma), E, F_1, C_1$.
   \begin{table}[ht]
\centering
\caption{ Exact Solutions of NLS\label{tb: exactNLS}}
\vspace{.2in}
\begin{tabular}{c | r | l | l | l | l}
&$u(\sigma)$ & $P(\sigma)$ &$E$&$ F_1$&$ C_1$\\[4pt]
\hline
1&\raisebox{-5pt}{$a_0\sigma^2+b_0\sigma+c_0$} & \raisebox{-5pt}{$(2a_0+d_0)/(a_0\sigma^2+b_0\sigma+c_0)$} &\raisebox{-5pt}{$0$}&    \raisebox{-5pt}{$-d_0$} &\raisebox{-5pt}{$0$} \\[8pt]
2&$a_0\cos^2(b_0\sigma)$ &$2b_0^2 tan^2(b_0\sigma)$&$2b_0^2$&$0$&n/a\\[8pt]
&&$4b_0^2 tan^2(b_0\sigma)$&$0$& $-2b_0^2a_0$&$0$\\[8pt]
3&$a_0 tanh(b_0\sigma)$&$c_0$&$c_0+2b_0^2$&$2b_0^2/a_0^2$&$-3$\\[8pt]
4&$a_0e^{-\sqrt{-c_0}\sigma}-b_0e^{\sqrt{-c_0}\sigma}$&$0$&$c_0<0$&$0$&n/a\\[8pt]
5&$(a_0/\sigma)e^{c_0\sigma^2/2}$&$c_0^2\sigma^2+2/\sigma^2+b_0$&$c_0+b_0$&$0$&n/a\\[8pt]
6&$ - a_0\cosh^2(b_0\sigma)$&$2b_0^2\tanh(b_0\sigma)^2+c_0$&$c_0-2b_0^2$&$0$&n/a\\[8pt]
7&$a_0 / \sigma^{b_0}$&  $  \frac{b_0(b_0+1)}{\sigma^2}+c_0  $  &$c_0$&$0$&n/a\\[6pt]
\hline
\end{tabular}
\end{table}

In order to map an exact solution $u(\sigma)$ of a non-linear (or linear) Schr\"odinger-type equation to an exact solution of Einstein's equations via the correspondences in this thesis, one may first need to solve for $\sigma(t)$ in the differential equation
\begin{equation}\dot\sigma(t)=\frac{1}{\theta}u(\sigma(t))\label{eq: dotsigma-usigma^r0}\end{equation}
for $\theta>0$.
Also the scalar field $\phi(t)$ is obtained by integrating a constant multiple of the square root of  $P(\sigma)$ (and then composing the result with $\sigma(t)$ or $\sigma(f_\sigma(\tau(t)))$ ).  That is, to find $\phi(t)$ one must first compute
\begin{equation}\psi (\sigma)=\displaystyle\int \sqrt{\alpha_0 P(\sigma)} d\sigma\label{eq: psigeneral}\end{equation}
for some constant $\alpha_0$.  Therefore for each solution in Table \ref{tb: exactNLS}, we now record some solutions to (\ref{eq: dotsigma-usigma^r0}) 
and (\ref{eq: psigeneral}).

For $u(\sigma)$ from line 1 in Table \ref{tb: exactNLS} with $a_0=0$, equation (\ref{eq: dotsigma-usigma^r0}) with $\theta=1$ is
\begin{equation}\dot\sigma(t)=b_0\sigma(t)+c_0\label{eq: sigmaeqnforu=linearr=1}\end{equation}
which has solution
\begin{equation}\sigma(t)=e^{b_0(t-t_0)}-\frac{c_0}{b_0}\label{eq: sigmaforu=linear}\end{equation}
for $b_0\neq 0$ and $t_0\in\mathds{R}$.  One can check this by computing that
\begin{equation}\dot\sigma(t)=u(\sigma(t))=b_0e^{b_0(t-t_0)}.\label{eq: dotsigmaanduforu=linearr=1}\end{equation}
Also, we note that
\begin{equation}u'(\sigma(t))=b_0.\label{eq: uprimesigmaforu=linearr=1}\end{equation}

For $u(\sigma)$ from line 1 in Table \ref{tb: exactNLS}, equation (\ref{eq: dotsigma-usigma^r0}) with $\theta=1$  is
\begin{equation}\dot\sigma(t)=a_0\sigma(t)^2+b_0\sigma(t)+c_0\label{eq: sigmaeqnforu=quadraticr=1}\end{equation}
which has solution
\begin{equation}\sigma(t)=\frac{1}{2a_0}\left(\sqrt{-\Delta} \ tan\left[\frac{\sqrt{-\Delta}}{2}(t-t_0)\right]-b_0\right)
\label{eq: sigmaforu=quadraticr=1}\end{equation}
for $t_0\in\mathds{R}$, $a_0>0$ and discriminant $\Delta=b_0^2-4a_0c_0<0$.  One can check this by computing that
\begin{equation}\dot\sigma(t)=u(\sigma(t))=\frac{-\Delta}{4a_0}sec^2\left[\frac{\sqrt{-\Delta}}{2}(t-t_0)\right].\label{eq: dotsigmaandusigmaforu=quadraticr=1}\end{equation}
Also, we record that
\begin{equation}u'(\sigma(t))=\sqrt{-\Delta} \ tan\left[\frac{\sqrt{-\Delta}}{2}(t-t_0)\right].\label{eq: uprimesigmaforu=quadraticr=1}\end{equation}

For $P(\sigma)$ from line 1 in Table \ref{tb: exactNLS} with $a_0>0$, (\ref{eq: psigeneral}) becomes
\begin{eqnarray}\psi(\sigma)&=& \displaystyle\int \sqrt{ \frac{\alpha_0(2a_0+d_0)}{a_0\sigma^2 +b_0\sigma+c_0}} d\sigma  \notag\\
&=& \sqrt{\frac{\alpha_0}{a_0}(2a_0+d_0)} \ \ln\left[2\left(\sigma + \frac{b_0}{2a_0} + \sqrt{\sigma^2+\frac{b_0}{a_0}\sigma+\frac{c_0}{a_0}}\right)\right] + \beta_0\label{eq: psiforP=1/quadratic}\end{eqnarray}
for $\Delta = b_0^2-4a_0c_0 < 0$  and integration constant $\beta_0\in\mathds{R}$.

\noindent Composing (\ref{eq: psiforP=1/quadratic}) with (\ref{eq: sigmaforu=quadraticr=1}),
\begin{eqnarray}\phi(t)&\stackrel{def.}{=}&\psi(\sigma(t))=
 \sqrt{\frac{\alpha_0}{a_0}(2a_0+d_0)} \ \ln\left[\frac{1}{a_0}\left(\sqrt{-\Delta} \ tan\left[\frac{\sqrt{-\Delta}}{2}(t-t_0)\right] \right.\right.\notag\\
 &&\qquad\qquad\qquad\left.\left.+ \sqrt{ -\Delta} sec\left[\frac{\sqrt{-\Delta}}{2}(t-t_0)\right] \right)\right] + \beta_0.\label{eq: phiforu=quadraticr=1}\end{eqnarray}

For $u(\sigma)$ from line 2 in Table \ref{tb: exactNLS}, equation (\ref{eq: dotsigma-usigma^r0})  is
\begin{equation}\dot\sigma(t)=\frac{a_0}{\theta} \cos^2(b_0\sigma(t))\label{eq: sigmaeqnforu=cos^2r=1}\end{equation}
which has solution
\begin{equation}\sigma(t)=\frac{1}{b_0}Arctan\left(\frac{b_0a_0}{\theta}(t-t_0)\right)\label{eq: sigmaforu=cos^2}\end{equation}
for $t_0\in\mathds{R}$.  One can check this by noting that
\begin{equation}\dot\sigma(t)=\frac{1}{\theta}u(\sigma(t))=\frac{a_0 \theta}{\theta^2+a_0^2b_0^2(t-t_0)^2}.\label{eq: dotsigmaforu=cos^2r=1}\end{equation}
For the record, we also compute
\begin{eqnarray}u'(\sigma(t))&=& -a_0b_0  \sin(2b_0\sigma(t))\notag\\
&=& -a_0b_0  \sin\left(2Arctan\left(\frac{b_0a_0}{\theta}(t-t_0)\right)\right)\notag\\
&=&   \frac{-2b_0^2a_0^2(t-t_0)\theta}{\theta^2+a_0^2b_0^2(t-t_0)^2}\label{eq: uprimesigmaforu=cos^2r=1}
\end{eqnarray}
 since $sin(2Arctan(x))=2cos(Arctan(x))sin(Arctan(x))=\frac{2x}{1+x^2}$.
 
For $P(\sigma)$ from line 2 in Table \ref{tb: exactNLS}, (\ref{eq: psigeneral}) becomes
\begin{eqnarray}\psi(\sigma)
&=&\sqrt{\alpha_0} \displaystyle\int \sqrt{   2b_0^2\tan^2(b_0\sigma)  } d\sigma  \notag\\
&=&\sqrt{\alpha_0} \displaystyle\int \sqrt{   2}b_0 \tan(b_0\sigma)   d\sigma  \notag\\
&=&\sqrt{2\alpha_0} \ln\left[ \sec(b_0\sigma) \right] + \beta_0
\label{eq: psiforcos^2with2btanandc=0}\end{eqnarray}
for $b_0,\sigma>0$ and integration constant $\beta_0\in\mathds{R}$.

\noindent Composing (\ref{eq: psiforcos^2with2btanandc=0}) with (\ref{eq: sigmaforu=cos^2}),
\begin{eqnarray}
\phi(t)&=&\psi(\sigma(t))\notag\\
&=&\sqrt{2\alpha_0} \ln\left[ \sec\left(Arctan(b_0a_0(t-t_0)) \right) \right] + \beta_0\notag\\
&=&\sqrt{2\alpha_0}\ln\left[   \sqrt{b_0^2a_0^2(t-t_0)^2+1}   \right]+\beta_0\notag\\
&=&\sqrt{\frac{\alpha_0}{2}}\ln\left[   \frac{b_0^2a_0^2}{\theta^2}(t-t_0)^2+1   \right]+\beta_0.\label{eq: phiforu=cos^2r=1withc=0}
\end{eqnarray}
 since $sec(Arctan(x))=1/cos(Arctan(x))=\sqrt{1+x^2}$
 
 For $P(\sigma)=4b_0^2\tan^2(b_0\sigma)$ for solution 2 in the table, (\ref{eq: psigeneral}) becomes
 \begin{eqnarray}\psi(\sigma)=\sqrt{\alpha_0}2b_0 \displaystyle\int tan(b_0\sigma) d\sigma=-2\sqrt{\alpha_0}\ln(\cos(b_0 \sigma))+\beta_0 \label{eq: psiforu=cos^2P=tan}
 \end{eqnarray}
for $0<b_0\sigma<\pi/2$ and $\beta_0\in\mathds{R}$.

\noindent Composing (\ref{eq: psiforu=cos^2P=tan}) with (\ref{eq: sigmaforu=cos^2}),
\begin{eqnarray}
\phi(t)&=&-2\sqrt{\alpha_0}\ln\left(\cos\left(Arctan\left(\frac{b_0a_0}{\theta}(t-t_0)\right)\right)\right)\notag\\
&=&\sqrt{\alpha_0}\ln\left(  1+ \frac{b_0^2a_0^2}{\theta^2}(t-t_0)^2 \right)\label{eq: phiforu=cos^2P=tan}
\end{eqnarray}

For $u(\sigma)$ from line 3 in Table \ref{tb: exactNLS}, equation (\ref{eq: dotsigma-usigma^r0})  with $\theta=1$ is
\begin{equation}\dot\sigma(t)=a_0 \tanh(b_0\sigma(t))\label{eq: sigmaeqnforu=tanh}\end{equation}
which has solution
\begin{equation}\sigma(t)=\frac{1}{b_0}Arcsinh\left(e^{a_0b_0(t-t_0)}\right)\label{eq: sigmaforu=tanhr=1}\end{equation}
for $t_0\in\mathds{R}$.  One can check this by computing that
\begin{equation}\dot\sigma(t)=u(\sigma(t))=a_0tanh\left(Arcsinh\left(e^{a_0b_0(t-t_0)}\right)\right)=\frac{a_0e^{a_0b_0(t-t_0)}}{\sqrt{1+e^{2a_0b_0(t-t_0)}}}.\label{eq: dotsigmausigmaforu=tanh}\end{equation}
Also, we will use that 
\begin{eqnarray}u'(\sigma(t))&=&a_0b_0 sech^2(b_0\sigma(t))\notag\\
&=&\frac{a_0b_0}{ cosh^2(Arcsinh\left(e^{a_0b_0(t-t_0)}\right))}\notag\\
&=&\frac{a_0b_0}{ 1+e^{2a_0b_0(t-t_0)}}\notag\\
&=&\frac{a_0b_0}{2} e^{-a_0b_0(t-t_0)}  sech(a_0b_0(t-t_0)) \label{eq: uprimesigmaforu=tanhr=1}
\end{eqnarray}
since $cosh^2(Arcsinh(x))=1+x^2$.

For $P(\sigma)=c_0$ from line 3 in Table \ref{tb: exactNLS}, (\ref{eq: psigeneral}) becomes
\begin{equation}\psi(\sigma)=\displaystyle\int \sqrt{\alpha_0c_0}d\sigma = \sqrt{\alpha_0c_0}\sigma+\beta_0\label{eq: psiforu=tanh}\end{equation}
for integration constant $\beta_0\in\mathds{R}$.

\noindent Composing (\ref{eq: psiforu=tanh}) with (\ref{eq: sigmaforu=tanhr=1}),
\begin{equation}
\phi(t)=\psi(\sigma(t))=\frac{\sqrt{\alpha_0c_0}}{b_0}Arcsinh\left(e^{a_0b_0(t-t_0)}\right)+\beta_0.\label{eq: phiforu=tanhr=1}
\end{equation}

For $u(\sigma)$ from line 4 in Table \ref{tb: exactNLS} with $b_0=0$, equation (\ref{eq: dotsigma-usigma^r0})  with $\theta=1$  is
\begin{equation}\dot\sigma(t)=a_0e^{- \sqrt{-c_0}\sigma(t)}\label{eq: sigmaeqnforu=-expr=general}\end{equation}
which has solution
\begin{equation}\sigma(t)=\frac{1}{\sqrt{-c_0}}\ln\left(\sqrt{-c_0}a_0(t-t_0)\right)\label{eq: sigmaforu=-exprgeneral}\end{equation}
for $a_0>0$ and $t_0\in\mathds{R}$.  One can check this by noting that
\begin{equation}\dot\sigma(t)=u(\sigma(t))=\frac{1}{\sqrt{-c_0}(t-t_0)}.\label{eq: dotsigmausigmaforu=-expr=general}\end{equation}
Also, we record that
\begin{eqnarray}
u'(\sigma(t))&=&-a_0\sqrt{-c_0}e^{-\sqrt{-c_0}\sigma(t)}\notag\\
&=&\frac{- 1}{ (t-t_0)}.\label{eq: uprimesigmaforu=-expr=general}
\end{eqnarray}

For $u(\sigma)$ from line 4 in Table \ref{tb: exactNLS}, equation (\ref{eq: dotsigma-usigma^r0}) for $\theta=1$ is 
\begin{equation}\dot\sigma(t)=a_0e^{-\sqrt{-c_0} \sigma(t)}-b_0e^{\sqrt{-c_0}\sigma(t)}\label{eq: sigmaeqnforu=linearcomboexpr=1}\end{equation}
which has solution
\begin{equation}\sigma(t)=\frac{1}{\sqrt{-c_0}}\ln\left(\sqrt{\frac{a_0}{b_0}}tanh(\sqrt{-a_0b_0c_0}(t-t_0))\right)\label{eq: sigmaforu=linearcomboexpr=1}\end{equation}
for $a_0,b_0>0$ and $t_0\in\mathds{R}$.  One can check this by computing that
\begin{equation}\dot\sigma(t)=u(\sigma(t))=2\sqrt{a_0b_0} \ csch(2\sqrt{-a_0b_0c_0}(t-t_0)).\label{eq: dotsigmausigmaforu=linearcomboexpr=1}\end{equation}
We also record that
\begin{eqnarray}
u'(\sigma(t))&=&-\sqrt{-c_0}\left(a_0e^{-\sqrt{-c_0} \sigma(t)}+b_0e^{\sqrt{-c_0}\sigma(t)}\right)\notag\\
&=&-\sqrt{-c_0a_0b_0}\left(coth(\sqrt{-a_0b_0c_0}(t-t_0)) +tanh(\sqrt{-a_0b_0c_0}(t-t_0))   \right)\notag\\
&=&-2\sqrt{-c_0a_0b_0} \ coth(2\sqrt{-a_0b_0c_0}(t-t_0))\label{eq: uprimesigmaforu=linearcomboexpr=1}
\end{eqnarray}
since $coth(x)+tanh(x)=2coth(2x)$.

For $u(\sigma)$ from line 5 in Table \ref{tb: exactNLS}, equation (\ref{eq: dotsigma-usigma^r0})  with $\theta=1$  is 
\begin{equation}\dot\sigma(t)=\frac{a_0}{\sigma(t)}e^{ c_0 \sigma(t)^2/2}\label{eq: sigmaeqnforu=e^x^2/xr=1cnegative}\end{equation}
which has solution
\begin{equation}\sigma(t)=\sqrt{\frac{-2}{c_0}\ln(-c_0 a_0(t-t_0))}\label{eq: sigmaforu=xe^x^2r=1}\end{equation}
for $c_0<0$ and $t_0\in\mathds{R}$.  One can check this by computing that
\begin{equation}\dot\sigma(t)=u(\sigma(t))=\frac{1}{\sqrt{-2c_0}(t-t_0)\sqrt{\ln(-c_0 a_0(t-t_0))}}.\label{eq: dotsigmausigmaforu=e^x^2/xr=1cnegative}\end{equation}
We also record that 
\begin{eqnarray}
u'(\sigma(t))&=&a_0 e^{c_0\sigma(t)^2/2}\left( c_0 - \frac{1}{\sigma(t)^2}  \right)\notag\\
&=&\frac{-1}{(t-t_0)}  \left( 1 + \frac{1}{2\ln(-c_0 a_0(t-t_0))}  \right).   \label{eq: uprimesigmaforu=e^x^2/xr=1cnegative}
\end{eqnarray}

For $P(\sigma)$ from line 5 in Table \ref{tb: exactNLS}, (\ref{eq: psigeneral}) becomes
\begin{eqnarray}\psi(\sigma)
&=&\sqrt{\alpha_0} \displaystyle\int \sqrt{  c_0^2\sigma^2+\frac{2}{\sigma^2}+b_0   } d\sigma  \notag\\
&=&\frac{ \sqrt{\alpha_0}}{2}\left(\sqrt{c_0^2\sigma^4+2+b_0\sigma^2} + \sqrt{2}\ln[\sigma^2] - \sqrt{2} \ln\left[ 2\sqrt{2}\sqrt{c_0^2\sigma^4+b_0\sigma^2+2}\right.\right.\notag\\
&&\quad \left. \left. +b_0\sigma^2+4\right]+\frac{b_0}{2c_0}\ln\left[2\sqrt{c_0^2\sigma^4+b_0\sigma^2+2}+\frac{b_0}{c_0}+2c_0\sigma^2\right]\right)+\beta_0
\label{eq: psiforP=c^2x^2+2/x^2+b}
\end{eqnarray}
for $\sigma>0$ and integration constant $\beta_0\in\mathds{R}$.

\noindent Composing (\ref{eq: psiforP=c^2x^2+2/x^2+b}) with (\ref{eq: sigmaforu=xe^x^2r=1}),
\begin{eqnarray}
&&\phi(t)=\psi(\sigma(t))=\notag\\
&=&\sqrt{ \frac{ \alpha_0}{2} } \left(\sqrt{2 \ln^2(-c_0 a_0(t-t_0)) + \frac{b_0}{-c_0}\ln(-c_0 a_0(t-t_0)) + 1} \right.\notag\\
&&\quad + \ln\left[\frac{2}{-c_0}\ln(-c_0 a_0(t-t_0))\right] - \ln\left[\frac{2b_0}{-c_0}\ln(-c_0 a_0(t-t_0))+4\right. \notag\\
&&\quad  + \left. 4\sqrt{2 \ln^2(-c_0 a_0(t-t_0)) + \frac{b_0}{-c_0}\ln(-c_0 a_0(t-t_0)) + 1} \right] \notag\\
&&\quad +\frac{b_0}{c_02\sqrt{2}}\ln\left[2\sqrt{2}\sqrt{  2 \ln^2(-c_0 a_0(t-t_0)) + \frac{b_0}{-c_0}\ln(-c_0 a_0(t-t_0)) + 1} \right.\notag\\
&&\left.\left.+\frac{b_0}{c_0}-4\ln(-c_0 a_0(t-t_0))\right]\right)+\beta_0
\label{eq: phiforu=xe^x^2r=1cnegative}\end{eqnarray}
for $c_0<0$ and $t> - 1/c_0a_0 +t_0$.

For $u(\sigma)$ from line 5 in Table \ref{tb: exactNLS} with $c_0=0$, equation (\ref{eq: dotsigma-usigma^r0}) with $\theta=1$  is
\begin{equation}\dot\sigma(t)=\frac{a_0}{\sigma(t)}\label{eq: sigmaeqnforu=e^x^2/x}
\end{equation}
which has solution
\begin{equation}\sigma(t)=\sqrt{ 2a_0(t-t_0) }\label{eq: sigmaforu=e^x^2/x}\end{equation}
for $a_0>0$ and $t>t_0\in\mathds{R}$.  One can check this by computing that
\begin{equation}\dot\sigma(t)=u(\sigma(t))=\frac{a_0}{ \sqrt{ 2a_0(t-t_0) } }.\label{eq: dotsigmausigmaforu=e^x^2/x}\end{equation}
We also record that
\begin{eqnarray}
u'(\sigma(t))&=&\frac{-a_0}{\sigma(t)^2}\notag\\
&=&\frac{-1}{ 2(t-t_0) }\label{eq: uprimesigmaforu=e^x^2/x}
\end{eqnarray}

For $P(\sigma)$ from line 5 in Table \ref{tb: exactNLS} with $c_0=0$, (\ref{eq: psigeneral}) becomes
\begin{eqnarray}\psi(\sigma)
&=&\sqrt{\alpha_0} \displaystyle\int \sqrt{\frac{2}{\sigma^2}+b_0   } d\sigma  \notag\\
&=&\sqrt{\alpha_0} \left(\sqrt{b_0\sigma^2+2}+\sqrt{2}\ln\left[\frac{\sigma}{\sqrt{2}+\sqrt{b_0\sigma^2+2} }\right]
\right)+\beta_0\label{eq: psiforP=2/x^2+b}
\end{eqnarray}
for integration constant $\beta_0\in\mathds{R}$.

\noindent Composing (\ref{eq: psiforP=2/x^2+b}) with (\ref{eq: sigmaforu=e^x^2/x}) for $c_0=0$,
\begin{eqnarray} 
\psi(t)&=&\psi(\sigma(t))\notag\\
&=&\sqrt{\alpha_0} \left(\sqrt{g(t)}+\frac{1}{\sqrt{2}}\ln\left[ 2a_0(t-t_0) \right]-\sqrt{2}\ln\left[\sqrt{2}+\sqrt{g(t)}\right]
\right)+\beta_0\label{eq: phiforu=1/sigmargeneral} \ \ \ \ \ \ \ \ \ \ 
\end{eqnarray}
for $g(t)= 2a_0b_0(t-t_0) +2$.

For $u(\sigma)$ from line 6 in Table \ref{tb: exactNLS}, equation (\ref{eq: dotsigma-usigma^r0})  with $\theta=1$ is
\begin{equation}\dot\sigma(t)= - a_0 \cosh^2(b_0\sigma(t))\label{eq: sigmaeqnforu=cosh^2r=1}\end{equation}
which has solution
\begin{equation}\sigma(t)=\frac{- 1}{b_0}Arctanh(b_0a_0(t-t_0))\label{eq: sigmaforu=cosh^2}\end{equation}
for $t_0\in\mathds{R}$.  One can check this by noting that
\begin{equation}\dot\sigma(t)=u(\sigma(t))=\frac{a_0}{a_0^2b_0^2(t-t_0)^2-1}.\label{eq: dotsigmaforu=cosh^2r=1}\end{equation}
For the record, we also compute
\begin{eqnarray}u'(\sigma(t))&=&  - a_0b_0  \sinh(2b_0\sigma(t))\notag\\
&=& - a_0b_0  \sinh(2Arctanh(b_0a_0(t-t_0)))\notag\\
&=&   \frac{2b_0^2a_0^2(t-t_0)}{a_0^2b_0^2(t-t_0)^2-1}\label{eq: uprimesigmaforu=cosh^2r=1}
\end{eqnarray}
 since $sinh(2Arctanh(x))=2cosh(Arctanh(x))sinh(Arctanh(x))=\frac{2x}{1-x^2}$.
 
 For $P(\sigma)$ from line 6 in Table \ref{tb: exactNLS}, (\ref{eq: psigeneral}) becomes
\begin{eqnarray}\psi(\sigma)
&=&\sqrt{\alpha_0} \displaystyle\int \sqrt{   2b_0^2\tanh^2(b_0\sigma)+c_0  } d\sigma  \notag\\
&=&-\sqrt{2\alpha_0} \ln\left[ 2 \left( \sqrt{2} b_0 \tanh(b_0\sigma) + \sqrt{ c_0 + 2 b_0^2 \tanh^2(b_0\sigma)}\right)\right]\notag\\
&& \quad + \frac{\sqrt{\alpha_0}}{b_0} \sqrt{c_0 + 2b_0^2} Arctanh\left[ \frac{ \sqrt{c_0 + 2b_0^2} tanh(b_0\sigma) }{\sqrt{ c_0  + 2b_0^2 tanh^2(b_0\sigma)} }\right] + \beta_0\label{eq: psiforcosh^2with2btan}\end{eqnarray}
for $c_0\geq - 2b_0^2$ and integration constant $\beta_0\in\mathds{R}$.

 \noindent Composing (\ref{eq: psiforcosh^2with2btan}) with (\ref{eq: sigmaforu=cosh^2}),
\begin{eqnarray}
&&\phi(t)=\psi(\sigma(t)) = \notag\\
&&-\sqrt{2\alpha_0} \ln\left[ 2 \left(  \sqrt{ c_0 + 2 b_0^4  a_0^2(t-t_0)^2  } -\sqrt{2} b_0^2a_0(t-t_0) \right)\right]\notag\\
&& \quad - \frac{\sqrt{\alpha_0}}{b_0} \sqrt{c_0 + 2b_0^2} Arctanh\left[ \sqrt{ \frac{ (c_0 + 2b_0^2)b_0^2a_0^2(t-t_0)^2 }{ c_0 + 2b_0^4a_0^2(t-t_0)^2} }\right] + \beta_0.\label{eq: phiforu=cosh^2r=1}\end{eqnarray}

For $u(\sigma)$ from line 7  in Table \ref{tb: exactNLS}, equation (\ref{eq: dotsigma-usigma^r0}) is
\begin{equation}\dot\sigma(t)=\frac{a_0}{\theta \sigma(t)^{b_0}},\label{eq: sigmaeqnforu=x^b}\end{equation}
which has solution
\begin{equation}\sigma(t)=\left( \frac{(1+b_0)a_0}{\theta}(t-t_0) \right)^{1/(1+b_0)}\label{eq: sigmaforu=x^b}\end{equation}
for $\theta>0$ and $t_0\in\mathds{R}$.  One can check this by computing 
\begin{equation}\dot\sigma(t)=\frac{1}{\theta}u(\sigma(t))=\frac{a_0}{\theta}\left( \frac{(1+b_0)a_0}{\theta}(t-t_0) \right)^{-b_0/(1+b_0)}.\label{eq: dotsigmaforu=x^b}\end{equation}
Also, we will use that 
\begin{equation}u'(\sigma(t))=\frac{-b_0\theta}{(1+b_0)(t-t_0)}.\label{eq: uprimesigmaforu=x^b}\end{equation}

For $P(\sigma)$ from line 7, (\ref{eq: psigeneral}) is
\begin{eqnarray}\psi(\sigma)
&=&\sqrt{\alpha_0}\displaystyle\int \sqrt{ \frac{b_0(b_0+1)}{\sigma^2}+c_0} \ d\sigma \notag\\
&=&\beta_0+\sqrt{\alpha_0}\left( \sqrt{b_0(b_0+1)+c_0\sigma^2} \right.\notag\\
&&\left.- \sqrt{b_0(b_0+1)} \log\left( \frac{\sqrt{b_0(b_0+1)}+\sqrt{b_0(b_0+1)+c_0\sigma^2}}{\sigma} \right) \right) \label{eq: psiforu=x^b}\end{eqnarray}
for $b_0>0$ and $\beta_0\in\mathds{R}$.
Composing (\ref{eq: psiforu=x^b}) with (\ref{eq: sigmaforu=x^b}),
\begin{eqnarray}\phi(t)&=&  \beta_0+\sqrt{\alpha_0}\left( \sqrt{b_0(b_0+1)+c_0\sigma(t)^2} \right.\notag\\
&&\left.- \sqrt{b_0(b_0+1)} \log\left( \frac{\sqrt{b_0(b_0+1)}+\sqrt{b_0(b_0+1)+c_0\sigma(t)^2}}{\sigma(t)} \right) \right) 
\label{eq: phiforu=x^b}\end{eqnarray}
for $\sigma(t)$ in (\ref{eq: sigmaforu=x^b}).

\appendix 
\setcounter{chapter}{5}

\chapter{Extra conservation equation}  

As noted in section 1.2, the correspondences in this thesis do not rely on $T_{ij}^{(2)}$ in (\ref{eq: T2matrix}) satisfying the conservation equation (\ref{eq: T2divzero}).  For each theorem in Chapters 3-7, we now record the EMP or NLS analogue of (\ref{eq: T2divzero}) for $l=0$.  For an arbitrary diagonal metric $g_{ij}$ one can compute, using the definition of the Christoffel symbols in (\ref{eq: christoffel}), that the equation 
\begin{equation}div(T^{(2)})_{l=0}=0\label{eq: divzeroinappx}\end{equation}
for $T_{ij}^{(2)}$ in (\ref{eq:  T2matrix}) takes the form
\begin{equation}\dot\rho+\left(\frac{1}{2}\displaystyle\sum_{i=1}^d g^{ii}\dot{g}_{ii}\right)(\rho+p)=0\label{eq: wrtdiaggij}\end{equation}
where dot denotes differentiation with respect to $x_0=t$.

For the FRLW metric (\ref{eq: FRLWmetric}), the conservation equation (\ref{eq: wrtdiaggij}) is
\begin{equation}\dot\rho+dH(\rho+p)=0\label{eq: FRLWEFEcons}\end{equation}
for $H(t)=\dot{a}(t)/a(t)$.  

Composing (\ref{eq: FRLWEFEcons}) with $f(\tau)$ from the EMP Theorem \ref{thm: FRLWEFE-EMP} and multiplying  by $f'(\tau)$ , we obtain
$\dot\rho(f(\tau))f'(\tau)+dH(f(\tau))f'(\tau)(\rho(f(\tau))+p(f(\tau)) )=0.$
That is, by the relation (\ref{eq: FRLWEMPH-Ytau}) and the definition (\ref{eq: FRLWEMPdottau-Yvarphi-Q}) of $\tau(t)$, we obtain the analogue conservation equation
\begin{equation}
\varrho'+\frac{d}{q}H_Y(\varrho+\textup{ \textlhookp})=0
\end{equation}
in the EMP variables, for $H_Y(\tau)\stackrel{def.}{=}Y'(\tau)/Y(\tau)$, $Y(\tau)=a(f(\tau))^q$ ($q\neq 0$) and $\varrho(\tau)=\rho(f(\tau)), \textup{ \textlhookp}(\tau)=p(f(\tau))$.

Composing (\ref{eq: FRLWEFEcons}) with $g(\sigma)$ from the NLS Theorem \ref{thm: FRLWEFE-NLS} and multiplying by $g'(\sigma)$, we obtain 
$\dot\rho(g(\sigma))g'(\sigma)+dH(g(\sigma))g'(\sigma)(\rho(g(\sigma))+p(g(\sigma)) )=0.$
That is, by the relation (\ref{eq: FRLWNLSH-u}) and the definition (\ref{eq: FRLWNLSdotsigma-upsi-P}) of $\sigma(t)$, we obtain the analogue conservation equation
\begin{equation}
\uprho' - dH_u(\uprho+\mathrm{p})=0
\end{equation}
in the NLS variables, for $H_u(\sigma)\stackrel{def.}{=}u'(\sigma)/u(\sigma)$, $u(\sigma)=1/a(g(\sigma))$ and $\uprho(\sigma)=\rho(g(\sigma)), \mathrm{p}(\sigma)=p(g(\sigma))$.

Composing (\ref{eq: FRLWEFEcons}) with $g(\sigma)$ from the alternate NLS Theorem \ref{thm: FRLWEFE-altNLS} and multiplying by $g'(\sigma)$ we obtain 
$\dot\rho(g(\sigma))g'(\sigma)+dH(g(\sigma))g'(\sigma)(\rho(g(\sigma))+p(g(\sigma)) )=0.$
That is, by the relation (\ref{eq: FRLWNLSaltH-u}) and the definition (\ref{eq: FRLWNLSaltdotsigma-upsi-P}) of $\sigma(t)$, we obtain the analogue equation
\begin{equation}
\uprho' - \frac{2d}{n_j}H_u(\uprho+\mathrm{p})=0
\end{equation}
in the alternate NLS variables for $H_u(\sigma)\stackrel{def.}{=}u'(\sigma)/u(\sigma)$, $u(\sigma)=1/a(g(\sigma))^{n_j/2}$ ($n_j\neq 0$) and $\uprho(\sigma)=\rho(g(\sigma)), \mathrm{p}(\sigma)=p(g(\sigma))$.

For the Bianchi I metric (\ref{eq: BImetric}), the conservation equation (\ref{eq: wrtdiaggij}) is
\begin{equation}\dot\rho+\frac{1}{\nu}H_R(\rho+p)=0\label{eq: BIEFEcons}\end{equation}
for $H_R(t)=\dot{R}(t)/R(t)$ and $R(t)=(X_1(t)\cdots X_d(t))^{\nu}$.

Composing (\ref{eq: BIEFEcons}) with $f(\tau)$ from the EMP Theorem \ref{thm: BIEMP} and multiplying by $f'(\tau)$ we obtain
$\dot\rho(f(\tau))f'(\tau)+\frac{1}{\nu}H_R(f(\tau))f'(\tau)(\rho(f(\tau))+p(f(\tau)) )=0.$
That is, by the relation (\ref{eq: BIEMPH-Y}) and the definition (\ref{eq: BIEMPdottau-Yvarphi-Q}) of $\tau(t)$ we obtain
\begin{equation}
\varrho' + \frac{1}{\nu q}H_Y(\varrho+\textup{\textlhookp})=0
\end{equation}
in the EMP variables for $H_Y(\tau)\stackrel{def.}{=}Y'(\tau)/Y(\tau)$, $Y(\tau)=R(f(\tau))^q$ ($q\neq 0$) and $\varrho(\tau)=\rho(f(\tau)), \textup{ \textlhookp}(\tau)=p(f(\tau))$.

Composing (\ref{eq: BIEFEcons}) with $g(\sigma)$ from the NLS Theorem \ref{thm: BINLS}, multiplying by $g'(\sigma)$ and using (\ref{eq: BINLSH-u}) and (\ref{eq: BINLSdotsigma-upsi-P}) we obtain
\begin{equation}
\uprho' - H_u( \uprho+\mathrm{p})=0
\end{equation}
for $H_u(\sigma)\stackrel{def.}{=}u'(\sigma)/u(\sigma)$, $u(\sigma)=1/R(g(\sigma))^{1/\nu}$ ($\nu\neq 0$) and $\uprho(\sigma)=\rho(g(\sigma)), \mathrm{p}(\sigma)=p(g(\sigma))$.

For the conformal version (\ref{eq: confBImetric}) of the Bianchi I metric, the conservation equation (\ref{eq: wrtdiaggij}) is
\begin{equation}\dot\rho+\frac{1}{\nu}H_R(\rho+p)=0\label{eq: confBIEFEcons}\end{equation}
for $H_R(t)=\dot{R}(t)/R(t)$ and $R(t)=(a_1(t)\cdots a_d(t))^{\nu}$.

Composing (\ref{eq: confBIEFEcons}) with $f(\tau)$ from Theorem \ref{thm: confBIEMP} and using (\ref{eq: confBIdottau-Yvarphi-Q}), (\ref{eq: confBIHR-Y}) we get the analogue conservation equation
\begin{equation}
\varrho'+\frac{1}{q\nu}H_Y(\varrho+\textup{\textlhookp})=0
\end{equation}
in the EMP variables for $H_Y(\tau)\stackrel{def.}{=}Y'(\tau)/Y(\tau)$, $Y(\tau)=R(f(\tau))^q$ ($q\neq 0$) and $\varrho(\tau)=\rho(f(\tau)), \textup{ \textlhookp}(\tau)=p(f(\tau))$.

Composing (\ref{eq: confBIEFEcons}) with $g(\sigma)\stackrel{def.}{=}\sigma+t_0$ from Theorem \ref{thm: confBINLS}
we obtain 
$\dot\rho(\sigma+t_0)+\frac{1}{\nu}H_R(\sigma+t_0)(\rho(\sigma+t_0)+p(\sigma+t_0) )=0.$
By (\ref{eq: confBINLSHR-u}), we obtain
\begin{equation}
\uprho' - H_u(\uprho+\mathrm{p})=0
\end{equation}
for the conservation equation in NLS variables for the Bianchi I metric (\ref{eq: confBImetric}), where $H_u(\sigma)\stackrel{def.}{=}u'(\sigma)/u(\sigma)$, $u(\sigma)=1/R(\sigma+t_0)^{1/\nu}$ and $\uprho(\sigma)=\rho(\sigma+t_0), \mathrm{p}(\sigma)=p(\sigma+t_0)$.

For the Bianchi V metric (\ref{eq: BVmetric}), the conservation equation (\ref{eq: wrtdiaggij}) is
\begin{equation}\dot\rho+\frac{1}{\nu}H_R(\rho+p)=0\label{eq: BVEFEcons}\end{equation}
for $H_R(t)=\dot{R}(t)/R(t)$ and $R(t)=(X_1(t)\cdots X_d(t))^{\nu}$.

Composing (\ref{eq: BVEFEcons}) with $f(\tau)$ from the EMP Theorem \ref{thm: BVEMP} and multiplying by $f'(\tau)$ we obtain
$\dot\rho(f(\tau))f'(\tau)+\frac{1}{\nu}H_R(f(\tau))f'(\tau)(\rho(f(\tau))+p(f(\tau)) )=0.$
That is, by the relation (\ref{eq: BVBIEMPH-Y}) and the definition (\ref{eq: BVBIEMPdottau-Yvarphi-Q}) of $\tau(t)$ we obtain
\begin{equation}
\varrho' + \frac{1}{\nu q}H_Y(\varrho+\textup{\textlhookp})=0
\end{equation}
in the EMP variables for $H_Y(\tau)\stackrel{def.}{=}Y'(\tau)/Y(\tau)$, $Y(\tau)=R(f(\tau))^q$ ($q\neq 0$) and $\varrho(\tau)=\rho(f(\tau)), \textup{ \textlhookp}(\tau)=p(f(\tau))$.

Composing (\ref{eq: BVEFEcons}) with $g(\sigma)$ from the NLS Theorem \ref{thm: BVNLS}, multiplying by $g'(\sigma)$ and using (\ref{eq: BVBINLSH-u}) and (\ref{eq: BVBINLSdotsigma-upsi-P}) we obtain
\begin{equation}
\uprho' - H_u( \uprho+\mathrm{p})=0
\end{equation}
for $H_u(\sigma)\stackrel{def.}{=}u'(\sigma)/u(\sigma)$, $u(\sigma)=1/R(g(\sigma))^{1/\nu}$ ($\nu\neq 0$) and $\uprho(\sigma)=\rho(g(\sigma)), \mathrm{p}(\sigma)=p(g(\sigma))$.

Composing (\ref{eq: BVEFEcons}) with $g(\sigma)$ from the alternate Bianchi V NLS Theorem \ref{thm: altBVNLS} and multiplying by $g'(\sigma)$ we obtain 
$\dot\rho(g(\sigma))g'(\sigma)+\frac{1}{\nu}H_R(g(\sigma))g'(\sigma)(\rho(g(\sigma))+p(g(\sigma)) )=0.$
That is, by the relation (\ref{eq: altBVBINLSH-u}) and the definition (\ref{eq: altBVBINLSdotsigma-upsi-P}) of $\sigma(t)$, we obtain the analogue equation
\begin{equation}
\uprho' - H_u(\uprho+\mathrm{p})=0
\end{equation}
in for $H_u(\sigma)\stackrel{def.}{=}u'(\sigma)/u(\sigma)$, $u(\sigma)=1/R(g(\sigma))^{1/\nu}$ ($\nu\neq 0$) and $\uprho(\sigma)=\rho(g(\sigma)), \mathrm{p}(\sigma)=p(g(\sigma))$.

For the conformal version (\ref{eq: confBVdmetric}) of the Bianchi V metric, the conservation equation (\ref{eq: wrtdiaggij}) is
\begin{equation}\dot\rho+\frac{d}{2\nu}H_R(\rho+p)=0\label{eq: confBVEFEcons}\end{equation}
for $H_R(t)=\dot{R}(t)/R(t)$ and $R(t)=(a_2(t)\cdots a_d(t))^{\nu}$.

Composing (\ref{eq: confBVEFEcons}) with $f(\tau)$ from Theorem \ref{thm: confBVEMP} and using (\ref{eq: confBVdottau-Yvarphi-Q}), (\ref{eq: confBVHR-Y}) we get the analogue conservation equation
\begin{equation}
\varrho'+\frac{d}{2q\nu}H_Y(\varrho+\textup{\textlhookp})=0
\end{equation}
in the EMP variables for $H_Y(\tau)\stackrel{def.}{=}Y'(\tau)/Y(\tau)$, $Y(\tau)=R(f(\tau))^q$ ($q\neq 0$) and $\varrho(\tau)=\rho(f(\tau)), \textup{ \textlhookp}(\tau)=p(f(\tau))$.

Composing (\ref{eq: confBVEFEcons}) with $g(\sigma)\stackrel{def.}{=}\sigma+t_0$ from Theorem \ref{thm: confBVNLS}
we obtain 
$\dot\rho(\sigma+t_0)+\frac{d}{2\nu}H_R(\sigma+t_0)(\rho(\sigma+t_0)+p(\sigma+t_0) )=0.$
By (\ref{eq: confBVNLSHR-u}), we obtain
\begin{equation}
\uprho' - H_u(\uprho+\mathrm{p})=0
\end{equation}
for the conservation equation in NLS variables for the Bianchi V metric in (\ref{eq: confBImetric}) for $H_u(\sigma)\stackrel{def.}{=}u'(\sigma)/u(\sigma)$, $u(\sigma)=1/R(\sigma+t_0)^{2\nu/d}$ and $\uprho(\sigma)=\rho(\sigma+t_0), \mathrm{p}(\sigma)=p(\sigma+t_0)$.

\end{document}